\documentclass[10pt,a4paper]{article}
\usepackage[utf8]{inputenc}
\usepackage{amsmath}
\usepackage{amsfonts}
\usepackage{amssymb}
\usepackage{graphicx}
\usepackage{multirow}
\usepackage[autostyle]{csquotes}
\usepackage[colorlinks,citecolor = blue,linkcolor=violet,anchorcolor = blue]{hyperref}
\usepackage{tikz}
\usepackage{xcolor, colortbl}
\definecolor{Gray}{gray}{0.90}
\definecolor{LightCyan}{rgb}{0.65,1,1}
\definecolor{Cyan}{rgb}{0.95,1,1}
\definecolor{GRAY}{gray}{0.75}

\usepackage[T1]{fontenc}
\usepackage[utf8]{inputenc}
\usepackage[font=small,labelfont=bf]{caption}

\newcommand{\bsub}{\begin{subequations}}
  \newcommand{\esub}{\end{subequations}$\!$}

\renewcommand\theequation{\thesection.\arabic{equation}}

\renewcommand{\eqref}[1]{(\ref{#1})}

\renewcommand{\theequation}{\arabic{section}.\arabic{equation}}

\usepackage[numbers,sort&compress]{natbib}
\newcommand{\R}{{\mathbb{R}}}

\newcommand{\highgreen}[1]{{\color{green}#1}}

\usepackage{parskip}
\usepackage{amsthm}
\usepackage[toc,page]{appendix}
\usepackage{wrapfig}
\usepackage{mathtools}
\usepackage{graphicx}
\graphicspath{{figures/}} 
\usepackage{subcaption}
\usepackage{dirtytalk}
\usepackage{arydshln} 
\numberwithin{equation}{subsection}

\usepackage[format=plain,
            labelfont={it},
            textfont=it]{caption}

\usepackage{cleveref}
\crefformat{section}{\S#2#1#3} 
\crefformat{subsection}{\S#2#1#3}
\crefformat{subsubsection}{\S#2#1#3}

\newtheorem{remark}{Remark}
\newtheorem{prop}{Proposition}

\usepackage{geometry} \author{Sarafa A. Iyaniwura and Michael J. Ward} \title{Synchrony and
  Oscillatory Dynamics for a 2-D PDE-ODE Model of Diffusion-Sensing
  with Small Signaling Compartments}
\begin{document}

\maketitle

\begin{abstract}
  We analyze a class of cell-bulk coupled PDE-ODE models, motivated by
  quorum and diffusion sensing phenomena in microbial systems, that
  characterize communication between localized spatially segregated
  dynamically active signaling compartments or ``cells'' that have a
  permeable boundary. In this model, the cells are disks of a common
  radius $\varepsilon \ll 1$ and they are spatially coupled through a
  passive extracellular bulk diffusion field with diffusivity $D$ in a
  bounded 2-D domain. Each cell secretes a signaling chemical into the
  bulk region at a constant rate and receives a feedback of the bulk
  chemical from the entire collection of cells. This global feedback,
  which activates signaling pathways within the cells, modifies the
  intracellular dynamics according to the external environment. The
  cell secretion and global feedback are regulated by permeability
  parameters across the cell membrane.  For arbitrary
  reaction-kinetics within each cell, the method of matched asymptotic
  expansions is used in the limit $\varepsilon\ll 1$ of small cell
  radius to construct steady-state solutions of the PDE-ODE model, and
  to derive a globally coupled nonlinear matrix eigenvalue problem
  (GCEP) that characterizes the linear stability properties of the
  steady-states. The analysis and computation of the nullspace of the
  GCEP as parameters are varied is central to the linear stability
  analysis. In the limit of large bulk diffusivity $D={D_0/\nu}\gg 1$,
  where $\nu\equiv {-1/\log\varepsilon}$, an asymptotic analysis of
  the PDE-ODE model leads to a limiting ODE system for the spatial
  average of the concentration in the bulk region that is coupled to
  the intracellular dynamics within the cells.  Results from the
  linear stability theory and ODE dynamics are illustrated for Sel'kov
  reaction-kinetics, where the kinetic parameters are chosen so that
  each cell is quiescent when uncoupled from the bulk medium. For
  various specific spatial configurations of cells, the linear
  stability theory is used to construct phase diagrams in parameter
  space characterizing where a switch-like emergence of intracellular
  oscillations can occur through a Hopf bifurcation. The effect of the
  membrane permeability parameters, the reaction-kinetic parameters,
  the bulk diffusivity, and the spatial configuration of cells on both
  the emergence and synchronization of the oscillatory intracellular
  dynamics, as mediated by the bulk diffusion field, is analyzed in
  detail. The linear stability theory is validated from full numerical
  simulations of the PDE-ODE system, and from the reduced ODE model
  when $D$ is large.
\end{abstract}



\setcounter{equation}{0}
\setcounter{section}{0}
\section{Introduction}\label{sec:intro}

Cell-to-cell communication is an important aspect of microbial systems
that is often achieved through the diffusion of an extracellular
signaling molecule, referred to as an autoinducer, in their common
environment (cf.  \cite{dunny1997cell}, \cite{taga2003chemical},
\cite{dano2}, \cite{QS_whiteley}). This form of bulk-mediated
communication involves the secretion and feedback of signaling
molecules from and into the cells, respectively, which enables each
cell to adjust its intracellular dynamics based on the signals it
receives from the autoinducer field that depends on the entire
collection of cells. Specific autoinducers responsible for such an
intercellular communication have been identified in many biological
systems including, cyclic adenosine monophosphate (cAMP) that triggers
intracellular oscillations for a collection of social amoebae
\textit{Dictyostelium discoideum} and guides them to aggregation in
low nutrient environments (cf.~\cite{gregor2010}, \cite{nandy1998},
\cite{goldbeter1997biochemical}), acetaldehyde (Ace) that leads to
glycolytic oscillations in a colony of yeast cells (cf.~\cite{dano},
\cite{dano2}, \cite{de2007dynamical}), and acylated homoserine
lactones (AHLs) that induces bioluminescence for certain species of
squid due to colonies of the marine bacterium {\em Vibrio fischeri}
located in the squid's light organ (cf.~\cite{taga2003chemical}).

In this context of intercellular bulk-mediated communication,
quorum-sensing (QS) refers to the onset of collective dynamics in the
cells that occurs when the cell density increases past a
threshold. There are two main categories of QS systems, for which
mathematical models have been developed. The first main type, which
includes yeast cells and social amoeba, involves a switch-like onset
of synchronized oscillatory intracellular dynamics as the cell density
increases (cf.~\cite{dano}, \cite{dano2}, \cite{de2007dynamical},
\cite{gregor2010}, \cite{mina}, \cite{review-yeast}). On the
macroscale, triggered synchronous temporal oscillations also occur in
physicochemical systems involving small catalyst-loaded particles
immersed in a BZ chemical mixture (cf.~\cite{taylor1}, \cite{taylor2},
\cite{tinsley1}, \cite{tinsley2}). The second main type of QS system,
for which the marine bacterium {\em Vibrio fischeri} and the human
pathogen {\em Pseudomonas aeruginosa} are prototypical examples, is
where an increase in the cell density leads to a sudden transition
between bistable steady-states (cf.~\cite{dock}, \cite{king2001},
\cite{qs_jump}).

QS systems that involve a switch-like onset of synchronous
intracellular oscillations have most typically been studied for
well-stirred systems, where the bulk diffusivity is taken to be
infinite. This well-mixed limit leads to a globally coupled ODE
system, where the global coupling arises from the spatially
homogeneous bulk diffusion field (cf.~\cite{gregor2010},
\cite{de2007dynamical}, \cite{taylor2}, \cite{tinsley1},
\cite{Rossler}, \cite{dicty}, \cite{smjw_quorum}). In this more
classical ODE setting, well-established mathematical tools such as ODE
bifurcation theory, phase reduction methods and the Kuramoto order
parameter (cf.~\cite{sync1}, \cite{sync2}) can be used to analyze the
onset of QS behavior and predict the degree of synchronization of
intracellular oscillations as the cell density increases.

However, when the bulk diffusivity is finite, diffusion-sensing
behavior associated with the spatial configuration of cells, the
spatial gradients of the bulk signal, reflecting versus absorbing
boundary conditions, and the mass transport properties of the bulk
medium, have all been shown experimentally to play an important role
for some QS systems (cf.~\cite{gao_diffusion},
\cite{dilanji_diffusion}, \cite{reflect}, \cite{DSQS} and the
references therein).  In contrast to the study of QS behavior through
an ODE theoretical framework, there have been relatively few
theoretical and modeling studies of how spatial diffusion triggers the
onset of collective intracellular oscillations in spatially
segregated, but localized, dynamically active reaction sites, and
these studies have typically been considered in 1-D spatial contexts
(cf.~\cite{gomez2007}, \cite{hess1}, \cite{gou2015}, \cite{gou2016},
\cite{gou2017}, \cite{hess2}, \cite{glass}, \cite{xu_bressloff},
\cite{paquin_1d}.)

{The goal of this paper is to study the emergence and
  synchronization of intracellular oscillations for the coupled
  cell-bulk PDE-ODE model of \cite{jia2016} in a 2-D domain. The
  formulation of this model was inspired by the 3-D PDE-ODE cell-based
  model of \cite{muller2006} with a single intracellular species (see
  also \cite{muller2013} and \cite{Mueller2014uecke}). In contrast to
  the analysis in \cite{jia2016} that was restricted either to the
  well-mixed limit or to the simple case of a ring pattern of
  identical cells, our study will focus on analyzing diffusion-sensing
  behavior, resulting from a passive bulk diffusion field with finite
  diffusivity, for various spatial configurations of a collection of
  heterogeneous cells.}

\begin{figure}[!h]
\begin{center}
\includegraphics[scale=01.10]{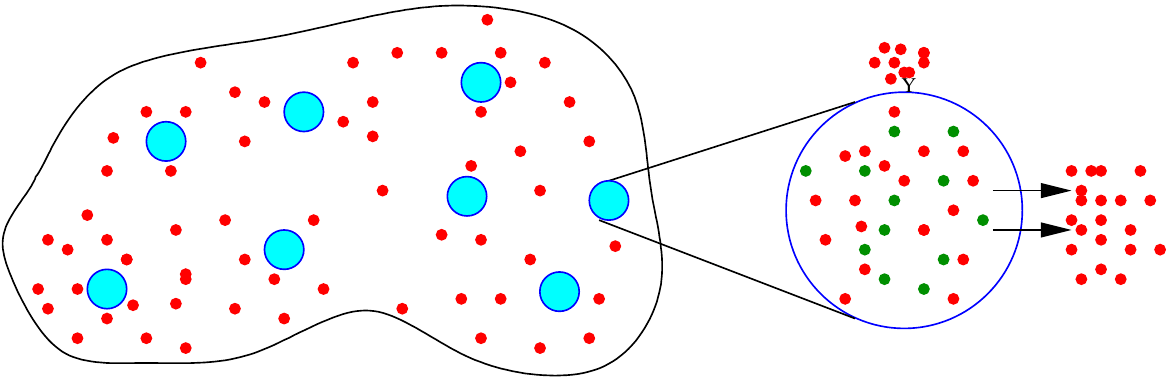}
\end{center}
\caption{A schematic diagram showing dynamically active signaling
  compartments (in cyan) in an arbitrary bounded 2D domain. The green
  and red dot represent the signaling chemicals in the cells, where
  only the red is secreted into the extracellular bulk region. On the
  right: A zoomed-in illustration of the intracellular concentration
  of chemicals within each signaling compartment, the secretion of
  signaling molecules into the bulk region, and the feedback of
  chemical into the cells.}
\label{onebulk_scheme}
\end{figure}

The coupled PDE-ODE model of \cite{jia2016} is formulated as follows:
Within $\Omega$ we assume that there are $m$ dynamically active
circular signalling compartments or ``cells'' of a common radius
$R_0$, denoted by $\Omega_j$ and centered at $\pmb{X}_j \in \Omega$
for $j=1,\dots,m$. In the bulk, or extracellular, region
$\Omega \setminus \cup_{j=1}^{m}\,\Omega_j$, the concentration
$\mathcal{U}(\pmb{X},T)$ of the autoinducer or bulk signal, which is
confined within $\partial\Omega$, is assumed to satisfy the passive
diffusion equation
\begin{subequations}\label{Dim_bulk} 
\begin{align}
  \mathcal{U}_T = & D_B \, \Delta \,
 \mathcal{U} - k_B \, \mathcal{U}\,, \quad T>0\,,\quad \pmb{X} \in \Omega
  \setminus \cup_{j=1}^{m}\,\Omega_j\,; \label{Dim_bulka}\\
  \, \partial_{n_{\pmb{X}}} \, \mathcal{U} = \, 0, \quad \pmb{X} \in
  \partial \Omega\,; & \qquad D_B \, \partial_{n_{\pmb{X}}} \, \mathcal{U}  =
  \beta_{1j} \, \mathcal{U} - \beta_{2j} \, \mu_j^1\,, \quad
  \pmb{X} \in \partial \Omega_j\,,\qquad j = 1, \ldots, m\,.\label{Dim_bulkb}
\end{align}
Here $D_B > 0$ and $k_B >0$ are the dimensional bulk diffusivity and
rate of degradation of the bulk signal, respectively. In the Robin
boundary condition \eqref{Dim_bulkb} on the cell membrane,
$\beta_{1j}>0$ and $\beta_{2j}>0$ are dimensional parameters modeling
the influx and efflux of chemical into and out of the $j^{\text{th}}$
cell, while $\partial_{n_{\pmb{X}}}$ denotes the outer normal
derivative on the cell boundary that points into the bulk region.
Within each cell we assume that there are $n$ interacting species
represented by the vector $\pmb{\mu}_j = (\mu_j^1,\ldots,\mu_j^n)^T$.
Assuming that the cells are sufficiently small so that there are no
spatial chemical gradients within them, the intracellular
reaction-kinetics $\pmb{F}_j$ for the $j^{\text{th}}$ cell is coupled
to the bulk medium via the integration of the flux across the cell
membrane. In this way, the intracellular dynamics within the
$j^{\text{th}}$ cell is coupled to the bulk signal \eqref{Dim_bulk} by
\begin{equation}\label{Dim_Intra}
\begin{split}
  \frac{d \pmb{\mu}_j}{dT} & = k_R\, \mu_c \, \pmb{F}_j
  \left( \pmb{\mu}_j/\mu_c  \right)  + \pmb{e}_1 \int_{\partial \Omega_j}\,
  \left( \beta_{1j} \, \mathcal{U} - \beta_{2j} \, \mu_j^1 \right) \,
  \text{d}S_{\pmb{X}}\,,\qquad j = 1, \ldots, m\,.
\end{split}
\end{equation}
\end{subequations}
Here $\pmb{e}_1 = (1,0,\ldots,0)^T$, $k_R > 0$ is the dimensional
reaction rate for the intracellular kinetics, and $\mu_c > 0$ is a
typical value for $\pmb{\mu}_j$. In this model, one signaling
chemical, labeled by $\mu_j^1$ can permeate the cell membrane with an
efflux parameter $\beta_{2j}$ and, by diffusion through the bulk
medium, can communicate with spatially distant cells.  The influx
permeability parameter $\beta_{1j}$ controls the global feedback into
the $j^{\mbox{th}}$ cell from the bulk diffusion field, which is
determined by the entire collection of cells. In
Fig.~\ref{onebulk_scheme} we schematically illustrate the cell-bulk
coupling for the case of $n=2$ intracellular species.

We assume that the radius $R_0$ of the signaling compartments is small
relative to the domain length-scale $L$, and so we introduce a small
parameter $ \varepsilon \equiv R_0/L \ll 1$. Then, by non-dimensionalizing
the coupled PDE-ODE model \eqref{Dim_bulk} using the approach of
\cite{jia2016}, we obtain that the dimensionless concentration of
chemical $U(\pmb{x},t)$ in the bulk region satisfies
\begin{subequations}\label{DimLess_bulk}
\begin{align}
  \tau \frac{\partial U}{\partial t} =&\,  D \, \Delta U -  \, U\,, \quad
   t > 0\,,\quad \pmb{x} \in \Omega \setminus \cup_{j=1}^{m}\,
   \Omega_{\varepsilon_j}\,;\label{DimLess_bulka} \\
  \partial_n \, U =\, 0\,, \quad \pmb{x} \in \partial \Omega\,;
  & \qquad \varepsilon  D\, \partial_n U  = d_{1j} \, U - d_{2j} \, u_j^1\,,
    \quad \pmb{x} \in \partial \Omega_{\varepsilon_j}\,,\qquad j = 1, \ldots, m\,,
    \label{DimLess_bulkb}
\end{align}
which is coupled to the dimensionless dynamics within the $j^{\text{th}}$ cell  by
\begin{align}\label{DimLess_Intra}
  \frac{\text{d} \pmb{u}_j}{\text{d}t} & =  \,
 \pmb{F}_j \left( \pmb{u}_j  \right)  + \frac{\pmb{e}_1}{\varepsilon \tau}
  \int_{\partial \Omega_{\varepsilon_j}}\, ( d_{1j} \, U - d_{2j} \, u_j^1 ) \,\,
              \text{d}s\,,\qquad j = 1, \ldots, m \,,
\end{align}
\end{subequations}
where $\pmb{u}_j=(u_j^1,\ldots,u_j^n)^T$ is the dimensionless vector
representing the $n$ chemical species in the $j^{\text{th}}$ cell, labeled by
$\Omega_{\varepsilon_j}\equiv \lbrace{\pmb{x} \, \vert \, \,
  |\pmb{x}-\pmb{x}_j| \leq \varepsilon\rbrace}$. We assume that the
centers of the cells are well-separated in the sense that
$\mbox{dist}(\pmb{x}_j,\pmb{x}_k)={\mathcal O}(1)$ for $j\neq k$ and
$\mbox{dist}(\pmb{x}_j,\partial\Omega)={\mathcal O}(1)$ as
$\varepsilon\to 0$. In this dimensionless formulation
\eqref{DimLess_bulk}, the key dimensionless parameters are
\begin{equation}\label{dim:param}
  D \equiv \frac{D_B}{k_B L^2} \,, \qquad
  d_{1j} \equiv \varepsilon \frac{\beta_{1j}}{k_B L} = {\mathcal O}(1) \,, \qquad
  d_{2j} \equiv \varepsilon \frac{\beta_{2j} L}{k_B} = {\mathcal O}(1) \,, \qquad
   \tau \equiv \frac{k_R}{k_B} \,.
\end{equation}
We refer to $D$ and $\tau$ as the effective bulk diffusivity and
reaction-time parameter, respectively. In \eqref{dim:param}, the
permeability parameters $\beta_{1j}$ and $\beta_{2j}$ are chosen as
$\mathcal{O}(\varepsilon^{-1})$ in order to ensure that there is an
${\mathcal O}(1)$ transport across the membrane of the small
cells. The parameter $\tau$ measures the relative rate of the
intracellular dynamics to the time-scale for degradation of the bulk
chemical. When the intracellular reactions proceed slowly, $\tau$ is
small, and little communication between the cells occurs. When the
effective bulk diffusivity $D$ is large, the cells are readily able to
communicate through the bulk medium, and in the well-mixed limit
$D\to \infty$ the bulk signal becomes spatially
homogeneous. Alternatively, for smaller values of $D$, only those
cells that are in close spatial proximity should be able to
communicate through the bulk diffusion field.

For arbitrary intracellular reaction kinetics and for an arbitrary
spatial arrangement of cells, in \S \ref{Analysis} we use strong
localized perturbation theory (cf.~\cite{ward2018spots},
\cite{WHK1993}) in the limit $\varepsilon\to 0$ to construct
steady-state solutions of \eqref{DimLess_bulk} and to derive the
linear stability problem for these steady-states. Unstable eigenvalues
of the linearization of a steady-state are shown to correspond to
roots $\lambda$ in $\mbox{Re}(\lambda)>0$ for which a certain globally
coupled nonlinear matrix eigenvalue problem (GCEP)
${\mathcal M}(\lambda)\pmb{c}=0$ (see \eqref{full_gcep}) has a
nontrivial solution. The $m\times m$ matrix ${\mathcal M}$, which
couples all the cells through an eigenvalue-dependent Green's
matrix, depends on the dimensionless parameters in
\eqref{dim:param}. The components of the corresponding normalized
eigenvector $\pmb{c}$ determines the relative magnitude of the
spatial gradient of the bulk signal near the cell membranes and it
determines the relative phases and amplitudes of small-scale
oscillations within the cells at the onset of a Hopf bifurcation when
$\lambda=i\lambda_I$, with $\lambda_I>0$, is a root of
$\det{\mathcal M}(\lambda)=0$ (see \eqref{nstabform:c}).

In \S \ref{largeD_ODE}, strong localized perturbation theory is used
to reduce the PDE-ODE model \eqref{DimLess_bulk} to an ODE
differential algebraic system with global coupling for the distinguished
limit $D={D_0/\nu}\gg 1$ of large bulk diffusivity, where
$\nu\equiv {-1/\log\varepsilon}$. In contrast to the simpler ODE
system derived in \cite{jia2016} for the well-mixed limit $D\to\infty$
(i.e.~$D\gg {\mathcal O}(\nu^{-1})$), the new derivation in \S
\ref{largeD_ODE}, as summarized in Proposition \ref{prop:dyn}, leads
to an ODE system that depends on the scaled bulk diffusivity parameter
$D_0$ and it depends weakly on the specific spatial configuration of
cells within the domain.

Our asymptotic theory is applied for the case of Sel'kov reaction
kinetics, which has been used as a conceptual model for glycolysis
oscillations in yeast cells \cite{Selkov}. With Sel'kov kinetics,
\eqref{DimLess_bulk} has a unique steady-state when $\varepsilon\ll
1$. As indicated from the experimental and modeling studies of
collective behavior in yeast cells (cf.~\cite{dano2}, \cite{dano},
\cite{de2007dynamical}), individual yeast cells are typically
non-oscillatory when isolated, but readily become synchronized in a
population of such cells. In qualitative agreement with this
observation, the Sel'kov parameters are chosen to be close to the
threshold for the onset of limit-cycle oscillations for an isolated
cell. As a result, the switch-like emergence of intracellular
oscillations, resulting from a Hopf bifurcation and illustrated
through various phase diagrams, is inherently due to the cell-cell
interaction, as mediated by the bulk diffusion field. In our study we
will analyze how the onset of intracellular oscillations and
synchronization depends on the membrane permeability parameters, the
reaction-time parameter, the bulk diffusivity, a Sel'kov kinetic
parameter, and the spatial pattern of cells. Diffusion-sensing
behavior, whereby cells initiate oscillations as a result of certain
spatial effects such as cell-clustering or the buildup of large
spatial gradients of the autoinducer field near the cell boundary, are
illustrated.

In \S \ref{SpecConfigSec} we illustrate our theory with Sel'kov
reaction kinetics for a ring and center-cell pattern of cells in the
unit disk (see Fig.~\ref{def_Mcells}), where the ring cells are taken
to have common parameters but where the parameters for the center cell
can be different. For the unit disk, the Green's matrices needed in
the steady-state and linear stability theory are available
analytically. For a ring and center-cell pattern, the GCEP matrix
${\mathcal M}(\lambda)$ has a cyclic sub-block and so in \S
\ref{RingCenterHole} we can analytically identify particular spatial
modes for which $\mbox{det}{\mathcal M}(\lambda)=0$.  By using
arclength continuation in $D$, Hopf bifurcation boundaries in the
$(D,\tau)$ parameter space for which
$\mbox{det}{\mathcal M}(i\lambda_I)=0$ can be computed for each of
these modes. In open regions of the $(D,\tau)$ parameter plane, we
show how to use a winding number criterion numerically on the roots of
$\mbox{det}{\mathcal M}(\lambda)=0$ in $\mbox{Re}(\lambda)>0$, so as
to compute the number of unstable eigenvalues of the linearization of
the unique steady-state.  In \S \ref{Subsec:SelKov_Exaample} these
phase diagrams are shown for the case of two ring cells. The
triggering effect on the emergence of intracellular oscillations due
to a center cell with either different permeability parameters or a
different Sel'kov kinetic parameter is studied in \S
\ref{sec:def_center_permea} and \S \ref{sec:alpha}, respectively.
Diffusion-sensing behavior as a result of changes in the ring radius
are studied in \S \ref{sec:ring}. The linear stability theory, which
predicts the onset of intracellular oscillations together with the
amplitude and phase differences near the Hopf bifurcation point, is
validated through large-scale simulations of the PDE-ODE model
\eqref{DimLess_bulk} using FlexPDE
\cite{flexpde2015solutions}. Moreover, we show that in certain cases
the new ODE system in Proposition \ref{prop:dyn}, derived for
$D={D_0/\nu}\gg 1$, can still provide a decent agreement with the full
numerical results even when $D={\mathcal O}(1)$.

In \S \ref{sec:LargePopulation} we study how the onset of
intracellular oscillations depends on the specific spatial arrangement
of ten cells in the unit disk that differ only in their influx
permeabilites $d_{1j}$, for $j=1,\ldots,m$. The cell configurations
considered include two clusters of cells (see
Fig.~\ref{fig:twoclusters}), two rings of cells with two isolated
cells (see Fig.~\ref{fig:tworing}), and arbitrarily placed cells (see
Fig.~\ref{fig:arbit}). For computational simplicity, we primarily
focus on the regime $D={D_0/\nu}\gg 1$, where matrix perturbation
theory can be used to asymptotically calculate the spectrum of the
GCEP matrix ${\mathcal M}(\lambda)$, as summarized in Proposition
\ref{lemma:large}. Since this analysis shows that only one matrix
eigenvalue of ${\mathcal M}(\lambda)$ can cross through zero, in \S
\ref{subsec:LargeCellExample} a numerical root-finding is readily
implemented on this eigenvalue to determine Hopf bifurcation
boundaries in the $(D_0,\tau)$ plane where intracellular oscillations
originate for the various spatial configurations of cells and
permeability parameter sets $d_{1j}$ for $j=1,\ldots,10$.  Our linear
stability results, validated through FlexPDE simulations of
\eqref{DimLess_bulk} and the ODE dynamics of Proposition
\ref{prop:dyn}, shows that when $D={D_0/\nu}$ the intracellular
dynamics depend sensitively on the influx permeability set, but only
weakly on the cell locations. In \S \ref{large:quiet} we implement the
linear stability theory based on the GCEP matrix for
$D={\mathcal O}(1)$ to predict that small-scale intracellular
oscillations can be highly heterogeneous in terms of amplitude and
phase when there are isolated cells. The linear stability theory is
confirmed from full PDE simulations. Finally, in \S
\ref{sec:discussion} we briefly summarize our study and discuss some
biological modeling problems that are well-aligned with the cell-based
PDE-ODE framework of \eqref{Dim_bulk}.

\setcounter{equation}{0}
\setcounter{section}{1}
\section{Asymptotic analysis of the dimensionless coupled model}\label{Analysis}

In this section, strong localized perturbation theory in the limit
$\varepsilon\to 0$ is applied to the dimensionless PDE-ODE model
\eqref{DimLess_bulk} for the regime $D = \mathcal{O}(1)$. This theory
is used to asymptotically approximate the steady-state solution of the
coupled model, and also to formulate a globally coupled eigenvalue
problem (GCEP) for studying the linear stability properties of the
derived steady-state solution.

\subsection{Asymptotic construction of the steady-state
  solution} \label{construct_ss}

We construct the steady-state solution for \eqref{DimLess_bulk} using
matched asymptotics. In the $j^{\text{th}}$ inner region,
defined within an $\mathcal{O} (\varepsilon)$ neighborhood of the
boundary of the $j^{\text{th}}$ cell, we introduce the local variables
$\pmb{y} =\varepsilon^{-1} (\pmb{x} - \pmb{x}_j)$ and
$U_j(\pmb{x}) = U_j(\varepsilon \pmb{y} + \pmb{x_j})$, where
$\rho \equiv |\pmb{y}|$.  Upon writing \eqref{DimLess_bulk} in terms
of the inner variables, for $ \varepsilon \to 0$ the steady-state
problem in the $j^{\text{th} }$ inner region is
$\Delta U_j  = 0$ for $\rho  \geq 1$, subject to
$D\, \partial_{\rho} U_j = d_{1j} U_j - d_{2j} u_j^1$ on
$\rho=1$. The radially symmetric solution to this problem is 
\begin{equation} \label{Inner_Solution}
\begin{split}
  U_j(\rho) &= A_j \log\rho + \frac{1}{d_{1j}} \left( D\, A_j + d_{2j}
    u_j^1 \right)\,, \qquad j=1, \ldots ,m \,,
\end{split}
\end{equation}
where $A_j$ for $j=1,\ldots,m $ are constants to be determined.  Upon
substituting \eqref{Inner_Solution} into the steady-state problem of
\eqref{DimLess_Intra}, we obtain that the steady-state intracellular
dynamics $\pmb{u}_j$ of the $j^{\text{th} }$ cell satisfies
\begin{equation}\label{Intra}
\begin{split}
  \pmb{F}_j \left( \pmb{u}_j  \right)  + \frac{2 \pi D}{\tau} \, A_j
  \pmb{e}_1 \,= \pmb{0}\,,\qquad j = 1,  \ldots, m\,.
\end{split}
\end{equation}
This determines $\pmb{u}_{j}$ in terms of the unknown constant
$A_j$. To proceed, we must derive another algebraic system for the
constants $A_j$ for $j=1,\ldots,m$, which is then coupled to
\eqref{Intra}.

By matching the far-field behaviour of the inner solution
\eqref{Inner_Solution} to an outer steady-state solution for
\eqref{DimLess_bulka}, we obtain that the outer solution must satisfy
a specific singularity behaviour as $\pmb{x} \to \pmb{x}_j$. In this
way, the steady-state outer approximation for the bulk solution satisfies
\begin{equation}\label{Outer_p} 
\begin{split}
  \Delta U - \varphi^2  \,U& =0\,, \quad \pmb{x} \in \Omega \setminus
  \{\pmb{x}_1,\ldots,\pmb{x}_m  \} \,; \qquad \partial_n U =0 \,,\quad
  \pmb{x} \in \partial \Omega \,; \\
  U \sim  A_j \log |\pmb{x} - \pmb{x}_j| & + \frac{A_j}{\nu} +
  \frac{1}{d_{1j}}(DA_j + d_{2j} u_j^1)\,, \quad \text{as} \quad \pmb{x} \to
  \pmb{x}_j\,, \qquad j =1,\ldots,m\,,
\end{split}
\end{equation}
where $\nu \equiv {-1/\log\varepsilon}$, with $\varepsilon \ll 1$,
and $\varphi \equiv \sqrt{1/D}$. The pre-specification of the regular
part of each singularity structure in \eqref{Outer_p} yields a
constraint. Overall these constraints provide an algebraic system for
$A_j$ for $j=1,\ldots,m$.

To determine this algebraic system, we next represent the solution to
\eqref{Outer_p} as
\begin{equation}\label{U_Green}
\begin{split}
U =  -2 \pi \sum_{i=1}^{m} A_i G(\pmb{x};\pmb{x}_i) \,,
\end{split}
\end{equation}
where the reduced-wave Green's function $G(\pmb{x};\pmb{x}_j)$ satisfies
\begin{equation}\label{R_Green}
\begin{split}
  \Delta G\, - \,& \varphi^2  G = -\delta(\pmb{x} - \pmb{x}_j)\,, \quad
  \pmb{x} \in \Omega\,; \qquad  \partial_n G =0 \,,\quad \pmb{x} \in
  \partial \Omega\,;\\
  G(\pmb{x};\pmb{x}_j) & \sim - \frac{1}{2\pi} \log|\pmb{x} - \pmb{x}_j| +
  R(\pmb{x}_j) + o(1)\,,\quad \text{as} \quad \pmb{x} \to \pmb{x}_j\,.
\end{split}
\end{equation}
Here $R_j \equiv R(\pmb{x}_j)$ is the regular part of
$G(\pmb{x};\pmb{x}_j)$ at $ \pmb{x}= \pmb{x}_j$. By expanding
\eqref{U_Green} as $\pmb{x} \to \pmb{x}_j$, we simply require that the
non-singular terms of the resulting expression agree with that
specified in \eqref{Outer_p} for each $j=1,\ldots,m$. This leads to an
algebraic system for the vector
$\pmb{\mathcal{A}} \equiv (A_1, \ldots, A_m)^T$, which is given in
matrix form as
\begin{equation}\label{Linear_System} 
\begin{split}
  \Big{(}   I + 2\pi \nu \mathcal{G} + \nu D P_1   \Big{)}
  \pmb{\mathcal{A}} =  - \nu \,P_2\, \pmb{u}^1\,,
\end{split}
\end{equation}
where $\pmb{u}^1\equiv (u_1^1,u_2^1, \ldots, u_m^1)^T$ denotes the vector of
chemicals that is secreted into the bulk region by the cells. In
\eqref{Linear_System}, $\mathcal{G}$ is the symmetric reduced-wave
Green's interaction matrix, while $P_1$ and $P_2$ are $m\times m$
diagonal matrices, defined by
\begin{equation}\label{GreensMatrix}
\begin{split}
\mathcal{G} \equiv
\begin{pmatrix}
R_1    & G_{12} & \dots & G_{1m}\\
G_{21} &  R_2   & \dots & G_{2m}\\
\vdots & \vdots & \ddots &  \vdots \\
G_{m1} & G_{m2} & \dots & R_{m}  
\end{pmatrix}\,, \quad  P_1 \equiv \text{diag}\Big{(}\frac{1}{d_{11}}, \ldots
, \frac{1}{d_{1m}} \Big{)} \,, \quad P_2 \equiv \text{diag}\Big{(}
\frac{d_{21}}{d_{11}}, \ldots , \frac{d_{2m}}{d_{1m}} \Big{)}\,.
\end{split}
\end{equation}
Here $G_{ji} = G_{ij} \equiv G(\pmb{x}_j;\pmb{x}_i)$ for $i \neq j$,
and $R_{j} \equiv R(\pmb{x}_j)$ for $j = 1, \ldots,m$, are obtained
from the solution to \eqref{R_Green}.

Overall, the asymptotic steady-state solution is determined in terms
of the solution $\pmb{\mathcal{A}}$ and $\pmb{u}_j$ for
$j=1,\ldots,m$, to the $n\times m$ dimensional nonlinear algebraic
system (NAS) given by \eqref{Intra} and \eqref{Linear_System}. This system
applies to arbitrary local reaction kinetics $\pmb{F}_j$ and
permeabilities parameters $d_{1j}>0$ and $d_{2j}>0$ for
$j=1,\ldots,m$. When the kinetics and permeability parameters are
identical for all the cells, the NAS reduces to the system given in
equations $(2.4)$ and $(2.9)$ of \cite{jia2016}.

Depending on the specific reaction kinetics assumed, the solution
structure to \eqref{Intra} and \eqref{Linear_System} as parameters are
varied can involve solution multiplicity, saddle-node points, and
other bifurcations. However, to illustrate our asymptotic theory we
will focus on the two-component Sel'kov reaction kinetics for which
\eqref{Intra} and \eqref{Linear_System} has a unique solution.

\subsection{Linear stability analysis} \label{LinearAnalysisDefect}

In the previous subsection, we characterized steady-state solutions of
the coupled model \eqref{DimLess_bulk} using strong perturbation
theory.  Suppose that the NAS \eqref{Intra} and \eqref{Linear_System}
has a solution for a given set of parameters. This then yields an
approximation to the steady-state solution solution $U_e(\pmb{x})$ and
$\pmb{u}_{ej}$, for $j=1,\ldots,m$, to \eqref{DimLess_bulk}. To
determine the linear stability of this steady-state we begin by
introducing the perturbation
\begin{equation}\label{stab:pert}
\begin{split}
 U(\pmb{x},t) = U_{e}(\pmb{x}) + e^{\lambda t}\xi (\pmb{x})\qquad \text{and} \qquad
 \pmb{u}_j(t) &= \pmb{u}_{ej} +  e^{\lambda t} \pmb{\phi}_j\,, \quad j=1,\ldots,m
 \,.
\end{split}
\end{equation}
where $\lambda $ is the eigenvalue of the linearization, and
$\xi (\pmb{x})$ and
$ \pmb{\phi}_j \equiv ( \phi_j^1,\ldots,\phi_j^n)^T$ are the
corresponding eigenfunctions in the bulk region and in the
$j^{\text{th} }$ cell, respectively.  Upon substituting this
perturbation into the PDE-ODE model \eqref{DimLess_bulk}, in the bulk
region we obtain the linearized problem
\begin{subequations}\label{Linear_bulk}
\begin{align}
  \tau \lambda \xi =&\,  D \, \Delta \xi -  \, \xi\,, \quad \pmb{x} \in
 \Omega \setminus \cup_{j=1}^{m}\,\Omega_{\varepsilon_j}\,;\label{Linear_bulkA}\\
  \partial_n \, \xi =\, 0, \quad \pmb{x} \in \partial \Omega\,;& \qquad
   \varepsilon  D\, \partial_{nj} \xi  = d_{1j} \, \xi - d_{2j} \, \phi_j^1\,,
  \quad \pmb{x} \in \partial \Omega_{\varepsilon_j}\,,\qquad j = 1, \ldots, m\,,
                            \label{Linear_bulkB}
\end{align}
which is coupled to the linearized intracellular dynamics of the
$j^{\text{th}}$ cell given in terms of
$\pmb{\phi}_j\equiv(\phi_j^1,\ldots,\phi_j^n)^T$, by
\begin{align}\label{Linear_Intra}
  \lambda \pmb{\phi}_j & =  J_j\, \pmb{\phi}_j   + \frac{\pmb{e}_1}
 {\varepsilon \tau} \int_{\partial \Omega_{\varepsilon_j}}( d_{1j} \, \xi - d_{2j} \,
                \phi_j^1 ) \,\text{d}s\,,\qquad j = 1, \ldots, m\,.
\end{align}
\end{subequations}
Here $J_j \equiv J(\pmb{u}_{ej}) $ is the Jacobian matrix of the local
kinetics $\pmb{F}_j$ evaluated at the steady-state $\pmb{u}_{ej}$.

Next, we use strong localized perturbation theory to analyze the
eigenvalue problem \eqref{Linear_bulk} in the limit
$\varepsilon \to 0$. This analysis leads to a limiting globally
coupled eigenvalue problem (GCEP) for $\lambda$ in the form of a
nonlinear matrix eigenvalue problem. This GCEP will be used to
investigate instabilities of the steady-state solution for the PDE-ODE
system \eqref{DimLess_bulk}.

To derive this GCEP, we first construct an inner region in an
$\mathcal{O}(\varepsilon)$ neighborhood of the $j^{\text{th}}$ cell by
introducing the local variables
$ \pmb{y} = \varepsilon^{-1}(\pmb{x} - \pmb{x}_j)$ and
$\xi_j(\pmb{x}) \equiv \xi_j(\pmb{x}_j + \varepsilon \pmb{y})$ with
$\rho = |\pmb{y}|$.  From \eqref{Linear_bulk}, we obtain for
$\varepsilon \to 0$ that
\begin{align}\label{leading_Inner}
 \Delta \, \xi_j  & = 0, \quad 1< \rho < \infty;\qquad
D\,\partial_{\rho }\, \xi_j = d_{1j} \xi_j - d_{2j} \phi_j^1, \quad \text{on} \quad \rho=1,
\end{align}
in the $j^{\text{th}}$ inner region.  The radially symmetric solution to
\eqref{leading_Inner} is
\begin{equation}\label{leading_Inner_Sol}
\begin{split}
  \xi_j &= c_j \log\rho + \frac{1}{d_{1j}} \left( D\, c_j  + d_{2j} \phi_j^1
  \right)\,, \qquad j=1,\ldots,m\,, 
\end{split}
\end{equation}
where $\rho = |\pmb{y}|$ and $c_j$ for $ j=1,\ldots,m$ are constants to
be determined.  Upon substituting \eqref{leading_Inner_Sol} into the
linearized intracellular dynamics of the $j^{\text{th}}$ cell
\eqref{Linear_Intra}, we obtain a linear relation between
$\pmb{\phi}_j$ and $c_j$ given by
\begin{equation}\label{Jac_sys}
\begin{split}
  (J_j - \lambda I)\pmb{\phi}_j = -\frac{2\pi D}{\tau} \,c_j \pmb{e}_1 \,.
    \qquad j=1,\ldots,m \,.
\end{split}
\end{equation}

Next, by analyzing the outer solution in the bulk region, we will
derive another linear system, which will be coupled to
\eqref{Jac_sys}. These two systems will provide the GCEP that is
needed to study the linear stability of the steady-state solution.

To determine this additional linear system, we first match the far-field
behaviour of the inner solution \eqref{leading_Inner_Sol} to the outer
solution in order to obtain the singularity behaviour of the outer
solution as $\pmb{x} \to \pmb{x}_j$. This yields in the bulk region that
\begin{equation}\label{Kinear_OuterDelta} 
\begin{split}
  \Delta \xi - \varphi_{\lambda}^2  \,\xi& =0\,, \quad \pmb{x} \in
  \Omega \setminus
  \{\pmb{x}_1,\ldots,\pmb{x}_m  \} \,; \qquad \partial_n \xi=0 \,,\quad
  \pmb{x} \in \partial \Omega \,; \\
  \xi \sim  c_j \log |\pmb{x} - \pmb{x}_j| & + \frac{c_j}{\nu} +
  \frac{1}{d_{1j}}(Dc_j + d_{2j} \phi_j^1)\,, \quad \text{as} \quad \pmb{x} \to
  \pmb{x}_j\,, \qquad j =1,\ldots,m\,,
\end{split}
\end{equation}
where $\nu \equiv {-1/\log\varepsilon}$ and
$\varphi_{\lambda} = \sqrt{(1 + \tau \lambda)/D}$. The solution to
\eqref{Kinear_OuterDelta} is represented as
\begin{equation}\label{Linear_OuterSol}
\begin{split}
\xi(\pmb{x}) = - 2\pi \sum_{i=1}^{m} c_i\, G_{\lambda}(\pmb{x};\pmb{x}_i) \,,
\end{split}
\end{equation}
where the eigenvalue-dependent reduced-wave Green's function
$G_\lambda(\pmb{x};\pmb{x}_j)$ satisfies
\begin{subequations}\label{EigGreen}
\begin{align}
  \Delta G_{\lambda} - & \varphi_{\lambda}^2  G_{\lambda} = -
  \delta(\pmb{x} - \pmb{x}_j)\,, \quad \pmb{x} \in \Omega\,;\qquad
   \partial_n G_{\lambda} =0 \,,\quad \pmb{x} \in \partial \Omega\,;
                         \label{EigGreenA} \\
  G_{\lambda}(\pmb{x};\pmb{x}_j) & \sim - \frac{1}{2\pi} \log|\pmb{x} -
 \pmb{x}_j| + R_{\lambda}(\pmb{x}_j) + o(1)\,,\quad \text{as} \quad
    \pmb{x} \to \pmb{x}_j\,.\label{EigGreenB}
\end{align}
\end{subequations}
Here $R_{\lambda j} \equiv R_{\lambda}(\pmb{x}_j)$ is the regular part
of $G_{\lambda}(\pmb{x};\pmb{x}_j)$ at $\pmb{x} = \pmb{x}_j$. In
\eqref{Linear_OuterSol} we have specified the principal branch of
$\varphi_{\lambda}$ to ensure that $G_{\lambda}$ is analytic in
$\mbox{Re}(\lambda)>0$ and that $G_{\lambda}$ decays far away
from the cells.  

By expanding \eqref{Linear_OuterSol} as $\pmb{x} \to \pmb{x}_j$, we
equate the non-singular terms of the resulting expression with those
specified in \eqref{Kinear_OuterDelta} for each $j=1,\ldots,m$.  This
yields a linear system for $\pmb{c}\equiv (c_1,\ldots,c_m)^T$ given in
matrix form by
\begin{equation}\label{OSystem2}
\begin{split}
  \big{(} I + 2\pi \nu \mathcal{G}_{\lambda} + \nu\,D P_1   \big{)}\, \pmb{c}=
  - \nu\,P_2\, \pmb{\phi}^1 \,,
\end{split}
\end{equation}
where $\pmb{\phi}^1 \equiv (\phi_1^1, \ldots, \phi_m^1)^T$ and
$\mathcal{G}_{\lambda}$ is the eigenvalue-dependent reduced-wave
Green's matrix whose entries are defined by
\begin{equation}\label{GreenMat}
  (\mathcal{G}_{\lambda})_{ij} = (\mathcal{G}_{\lambda})_{ji} \equiv
  G_{\lambda}(\pmb{x}_j,\pmb{x}_i) \quad \text{for} \quad i \neq j\,,
  \quad \text{and} \quad\ (\mathcal{G}_{\lambda})_{jj} =
  R_{\lambda j} \equiv R_\lambda(\pmb{x}_j) \,.
\end{equation}
In \eqref{OSystem2}, the diagonal matrices $P_1$ and $P_2$ are 
defined in  \eqref{GreensMatrix}.

Next, we will combine the algebraic system \eqref{OSystem2} with
\eqref{Jac_sys} in order to derive the GCEP for $\lambda$ and 
$\pmb{c} =(c_1,\ldots,c_m)^T$. We first use \eqref{Jac_sys} to determine
$\phi_j^1$ in terms of the constant $c_j$ as
$\phi_j^1 = 2\pi D\tau^{-1} \pmb{e}_1^T ( \lambda I - J_j )^{-1} \pmb{e}_1 \,
c_j$ for $j=1,\ldots,m$, provided that $\lambda$ is not an eigenvalue of $J_j$
for any $j=1,\ldots,m$. In matrix form this yields
\begin{equation}\label{phi_sys}
\begin{split}
 \pmb{\phi}^1 = \frac{2\pi D}{\tau }\mathcal{K} \pmb{c} \,,
\end{split}
\end{equation}
where $\mathcal{K} \equiv \mathcal{K}(\lambda)$ is an $m \times m$
diagonal matrix $\mathcal{K}\equiv \mbox{diag}\left( \mathit{K}_1,\ldots,
  \mathit{K}_m\right)$, whose entries are given by
\begin{subequations}\label{K_info}
\begin{equation}\label{K_mat_entry}
\begin{split}
  \mathit{K}_j =  \pmb{e}_1^T ( \lambda I - J_j)^{-1} \pmb{e}_1 =
  \frac{\pmb{e}_1^T N_j \pmb{e}_1 }{\det( \lambda I - J_j)} =
  \frac{ (N_{j})_{11} }{\det( \lambda I - J_j)} \,.
\end{split}
\end{equation}
Here $N_{j}$ is the $n \times n$ matrix of cofactors of
$( \lambda I - J_j)$, while $(N_{j})_{11}$ is its entry in the first
row and first column given by
\begin{equation}\label{Mj_11_matrix}
(N_{j})_{11} \equiv (N_{j}(\lambda))_{11}= \det
\begin{pmatrix}
  \lambda - \left. \frac{\partial F_j^2}{\partial u_2}\right|_{\pmb{u}=\pmb{u}_{e,j}}
  & \dots & \left. -
    \frac{\partial F_j^2}{\partial u_n}\right|_{\pmb{u}=\pmb{u}_{e,j}} \\
\vdots &  \ddots & \vdots \\
\left. - \frac{\partial F_j^n}{\partial u_2}\right|_{\pmb{u}=\pmb{u}_{e,j}}
& \dots &  \lambda - \left.
  \frac{\partial F_j^n}{\partial u_n}\right|_{\pmb{u}=\pmb{u}_{e,j}} 
\end{pmatrix}\,.
\end{equation}
\end{subequations}
The functions $F_j^2, \ldots, F_j^n$ are the components of the local reaction
kinetics $\pmb{F_j} \equiv (F_j^1, \ldots, F_j^n)^T$ of the $j^{\text{th}}$ cell.

Upon substituting \eqref{phi_sys} into \eqref{OSystem2}, we obtain the
$m$-dimensional homogeneous algebraic system
\begin{subequations} \label{full_gcep}
\begin{equation}\label{Global_System}
\begin{split}
  \mathcal{M}(\lambda) \pmb{c} &= \pmb{0}\,, \qquad \mbox{where} \qquad
 \mathcal{M}(\lambda) \equiv I + 2\pi \nu \mathcal{G}_{\lambda}  + \nu\,D\,P_1 +
  \frac{2\pi \nu D}{\tau }  P_2 \mathcal{K} \,,
\end{split}
\end{equation}
where $\nu\equiv {-1/\log\varepsilon}$.  The homogeneous system
\eqref{Global_System}, where ${\mathcal M}$ is a symmetric but
non-Hermitian matrix, is referred to as the globally coupled
eigenvalue problem (GCEP). The GCEP is a nonlinear matrix eigenvalue
problem for $\lambda$, and it has a nontrivial solution
$\pmb{c}\neq \pmb{0}$ if and only $\lambda$ satisfies
$\det \mathcal{M}(\lambda) = 0$. We label the set
$\Lambda({\mathcal M})$ as the union of all such roots, i.e.
\begin{equation}\label{TransDent}
\begin{split}
  \Lambda({\mathcal M}) \equiv \lbrace{ \lambda \, \vert \,
    \det \mathcal{M}(\lambda)=0 \rbrace} \,.
\end{split}
\end{equation}
\end{subequations}
The parameters in the GCEP \eqref{Global_System} are the bulk
diffusivity $D$, the reaction-timescale $\tau$ and the permeabilities
$d_{1j}$ and $d_{2j}$, for $j=1,\ldots,m$, which are encoded in the
matrices $P_1$ and $P_2$ given in \eqref{GreensMatrix}. Moreover, in
\eqref{Global_System}, the effect of the spatial configuration of the
centers $\{ \pmb{x}_1,\ldots, \pmb{x}_m \} \in \Omega$ of the cells
and the domain shape $\Omega$ arises from both the
eigenvalue-dependent reduced wave Green's interaction matrix
$\mathcal{G}_{\lambda}$ and the steady-state solution, which
determines ${\mathcal K}$ in \eqref{K_info}.

A recent survey of nonlinear matrix eigenvalue problems and available
solution strategies for certain classes of matrices is given in
\cite{guttel} and \cite{betcke}. A range of applications of such
problems, but in simpler contexts where ${\mathcal M}(\lambda)$ is
either a polynomial or rational function of $\lambda$, are discussed
in \cite{betcke2}.

Any element $\lambda\in \Lambda({\mathcal M})$ satisfying
$\mbox{Re}(\lambda)>0$ provides an approximation, valid as
$\varepsilon\to 0$, for an unstable discrete eigenvalue of the
linearized system \eqref{Linear_bulk}.  This leads to the following
criterion regarding the linear stability of the steady-state solution.

\vspace*{0.2cm}
\begin{prop}\label{prop:stab}  For $\varepsilon\to 0$, a steady-state
  solution to \eqref{DimLess_bulk} is linearly stable when there are
  no roots to $\mbox{det}{\mathcal M}(\lambda)=0$ in
  $\mbox{Re}(\lambda)>0$, i.e. for all
  $\lambda\in \Lambda({\mathcal M})$ we have
  $\mbox{Re}(\lambda)<0$. Moreover, if ${\mathcal A}_e$ and
  $\pmb{u}_{ej}$ for $j=1,\ldots,m$ is a non-degenerate solution to
  the NAS \eqref{Intra} and \eqref{Linear_System}, for which $J_j$ is
  non-singular, then $\lambda=0$ is not a root of
  $\mbox{det}{\mathcal M}(\lambda)=0$.
\end{prop}

\begin{proof} For $\varepsilon\to 0$, any discrete eigenvalue
  $\lambda$ of the linearization \eqref{stab:pert} corresponds to a
  non-trivial solution to \eqref{Linear_bulk}. Since, for
  $\varepsilon\to 0$, these discrete eigenvalues comprise the set
  $\Lambda({\mathcal M})$ in \eqref{TransDent}, the steady-state is
  linearly stable if all $\lambda\in \Lambda({\mathcal M})$ satisfy
  $\mbox{Re}(\lambda)<0$. Next, suppose that ${\mathcal A}_e$ and
  $\pmb{u}_{ej}$ for $j=1,\ldots,m$ is a non-degenerate solution to
  the NAS \eqref{Intra} and \eqref{Linear_System}.  Introducing the
  perturbation ${\mathcal A}={\mathcal A}_e + \pmb{\psi}$ and
  $\pmb{u}_{j}=\pmb{u}_{ej}+\pmb{v}_j$, where $|\pmb{v}_j|\ll 1$ and
  $|\pmb{\psi}|\ll 1$, we linearize \eqref{Intra} and
  \eqref{Linear_System} to obtain
  \begin{equation}\label{noz:one}
 \big{(} I + 2\pi \nu \mathcal{G} + \nu\,D P_1   \big{)}\, \pmb{\psi}=
 - \nu\,P_2\, \pmb{v}^1 \,, \qquad  J_j \pmb{v}_j = -\frac{2\pi D}{\tau} \,
 \psi_j \pmb{e}_1 \,. \qquad j=1,\ldots,m \,,
\end{equation}
where $\pmb{v}^1 \equiv (v_1^1, \ldots, v_m^1)^T$,
$\pmb{\psi}= (\psi_1,\ldots,\psi_m)^T$ and
$\pmb{v}_j=(v_{j}^1,\ldots,v_j^n)^T$. Assuming that $J_j$ is invertible,
\eqref{noz:one} yields that
\begin{equation}\label{noz:two}
\mathcal{J} \pmb{\psi} = \pmb{0}\,, \quad \mbox{where} \quad
    \mathcal{J} \equiv I + 2\pi \nu \mathcal{G}  + \nu\,D\,P_1 +
    \frac{2\pi \nu D}{\tau }  P_2 \mathcal{K}_0  \quad
    \mbox{and} \quad {\mathcal K}_0 \equiv -\mbox{diag}
    \left(\pmb{e}_1^T J_1^{-1} \pmb{e}_1,\ldots,\pmb{e}_1^T J_m^{-1}
      \pmb{e}_1\right) \,.
\end{equation}
We conclude that $\det{\mathcal{J}}\neq 0$ owing to the fact that
${\mathcal A}_e$ and $\pmb{u}_{ej}$ for $j=1,\ldots,m$ is assumed to
be a non-degenerate solution of the NAS \eqref{Intra} and
\eqref{Linear_System}. Finally, since it is readily verified that
$\mathcal{J}={\mathcal M}(0)$, where ${\mathcal M}(\lambda)$ is the
GCEP matrix in \eqref{Global_System}, we conclude that $\lambda=0$ is
not a root of \eqref{TransDent}.
\end{proof}

Proposition \ref{prop:stab} implies that branches of non-degenerate
solutions to the NAS \eqref{Intra} and \eqref{Linear_System} obtained
as a parameter is varied cannot lose stability through a
zero-eigenvalue crossing of the GCEP \eqref{Global_System}. This
simple observation is the motivation for analzying whether stability
can be lost through Hopf bifurcations associated with the
linearization.

In terms of the eigenvectors $\pmb{c}$ and eigenvalues
$\lambda\in \Lambda({\mathcal M})$ of the GCEP \eqref{full_gcep}, we
obtain from \eqref{stab:pert}, \eqref{Linear_OuterSol} and
\eqref{Jac_sys} that the linearization around the bulk and
intracellular steady-state solutions is, for $\varepsilon\to 0$, given
by the superposition
\begin{equation}\label{stabform:c}
  U(\pmb{x},t) \sim U_{e} - \sum_{\lambda \in \Lambda({\mathcal M})}
  e^{\lambda t}\left(\sum_{j=1}^{m} c_j G_{\lambda}(\pmb{x};
  \pmb{x}_j)\right) \,; \qquad
  \pmb{u}_j(t) \sim \pmb{u}_{ej} +  \frac{2\pi D}{\tau}
  \sum_{\lambda\in\Lambda({\mathcal M})} e^{\lambda t} c_j
  \left(\lambda I - J_j\right)^{-1} \pmb{e}_1\,, \quad j=1,\ldots,m
  \,,
\end{equation}
where each $\pmb{c}=(c_1,\ldots,c_m)^T$ depends on the particular
eigenvalue $\lambda$.  To relate the diffusive flux into the $j$-th
cell to the components of $\pmb{c}$ we use the the inner solutions
\eqref{Inner_Solution} and \eqref{leading_Inner_Sol}. To determine the
effect on the intracellular component that can be transported across
the membrane we calculate $\pmb{e}_1^T\pmb{u}_j\equiv u_j^{1}$ in
\eqref{stabform:c}. This yields that
\begin{equation}\label{nstabform:c}
  D \partial_{\rho} U \vert_{\rho=1} \sim D \left(A_j +
  \sum_{\lambda\in\Lambda({\mathcal M})} c_j  e^{\lambda t}\right) \,, 
  \qquad  u_{j}^{1} \sim u_{ej}^{1} + \frac{2\pi D}{\tau}
  \sum_{\lambda\in\Lambda({\mathcal M})} \left( {\mathcal K}\pmb{c}\right)_j
  e^{\lambda t} \,, \qquad j=1,\ldots,m \,,
\end{equation}
where
${\mathcal K}\equiv \mbox{diag}\left( \mathit{K}_1,\ldots,
  \mathit{K}_m\right)$, with $\mathit{K}_j$ defined in \eqref{K_info}.
As evident from \eqref{nstabform:c}, and discussed for various
examples in \S \ref{Subsec:SelKov_Exaample} and \S
\ref{subsec:LargeCellExample}, the modulus $|({\mathcal K}\pmb{c})_j|$
and argument $\mbox{arg}({\mathcal K}\pmb{c})_j$ of the components of
a complex-valued ${\mathcal K}\pmb{c}$, resulting from 
pure imaginary eigenvalues with $\mbox{Re}(\lambda)=0$ and
$\mbox{Im}(\lambda)\neq 0$, determines the relative amplitudes and
phase differences of small scale intracellular oscillations near a
Hopf bifurcation of the steady-state solution. 

Our linear stability theory for steady-state solutions of PDE-ODE
model \eqref{DimLess_bulk} can be applied to any configuration of
cells in an arbitrary 2-D bounded domain and for arbitrary local
reaction kinetics. However, in our illustrations of the theory below
in \cref{SpecConfigSec} and \cref{sec:LargePopulation}, we will
consider the two-component Sel'kov model, which is used in simple
models of glycolysis (cf.~\cite{nandy1998}, \cite{Selkov}).  From
\eqref{DimLess_Intra}, for an isolated cell with no influx from the
bulk, the intracellular dynamics within the $j^{\text{th}}$ cell that
accounts for efflux from the cell boundary is given by
${\text{d} \pmb{u}_j/\text{d}t} = \pmb{F}_j \left( \pmb{u}_j \right)-
{2\pi d_{2j} u_{j}^{1}\pmb{e}_1/\tau}$, where $\pmb{u}_j=(u_{j}^{1},u_{j}^{2})^T$
and the Sel'kov kinetics $\pmb{F}_j(v,w)=(F_{j1}(v,w),F_{j2}(v,w))^T$ are
defined by
\begin{equation}\label{isolated:selkov:ode}
\begin{split}
  F_{j1}(v,w) = \alpha_j w+wv^2-v\,, \qquad F_{j2}(v,w) =\zeta_j
  \left[\mu_j-\left(\alpha_j w+wv^2\right)\right] \,.
\end{split}
\end{equation}
The steady-state for this isolated cell is given by
$u_{j}^{1}={\mu_j/\chi_j}$ and
$u_{j}^{2} = {\mu_j/\left(\alpha_j + (u_{j}^{1})^2\right)}$, where
$\chi_j\equiv 1 + {2\pi d_{2j}/\tau}$. The
determinant and trace of the Jacobian $J_{ej}$ for this isolated cell
with boundary efflux is
\begin{equation}\label{isolated:selkov:jac}
  \mbox{det}(J_{ej})= \zeta_j \chi_j
  \left(\alpha_j + \frac{\mu_j^2}{\chi_j^2}\right) \,,
  \qquad \mbox{tr}(J_{ej}) = \frac{2 \mu_j^2}{\chi_j}
  \left(\alpha_j + \frac{\mu_j^2}{\chi_j^2}\right)^{-1} -
  \chi_j - \zeta_j\left(\alpha_j + \frac{\mu_j^2}{\chi_j^2}\right)\,.
\end{equation}
Since $\mbox{det}(J_{ej})>0$, the steady-state for this isolated cell
is linearly stable only if $\mbox{tr}(J_{ej})<0$.  We will choose
Sel'kov kinetic parameters $\alpha_j$, $\mu_j$ and $\zeta_j$ so that
an isolated cell with zero boundary efflux (i.e. $d_{2j}=0$) is
linearly stable, but with parameters rather close to the stability
threshold.  A set of such parameters is shown in
Fig.~\ref{fig:isolated:trace} where we verify that
$\mbox{tr}(J_{ej})<0$ on $0.7<\alpha_j<1.0$ when $\mu_j=2$ and
$\zeta_j=0.15$.  The Hopf bifurcation (HB) boundary in the $\alpha_j$
versus $\mu_j$ parameter plane, as obtained by setting
$\text{tr}(J_{ej})=0$ is
\begin{equation}\label{selkov:boundary}
  \alpha_j = -\frac{\mu_j^2}{\chi_j^2} +\frac{1}{2\zeta_j}
  \left[ -\chi_j + \sqrt{\chi_j^2+\frac{8\zeta_j\mu_j^2}{\chi_j}}\right]\,,
  \qquad \mbox{where} \quad \chi_j \equiv 1 + \frac{2\pi d_{2j}}{\tau} \,.
\end{equation}
For an isolated cell, a simple application of the Poincare-Bendixson
theorem shows that when the steady-state is unstable the cell will
have limit cycle oscillations. When there is no boundary efflux,
i.e.~$d_{2j}=0$, this parameter range of periodic solutions is given
by the green-shaded region in Fig.~\ref{fig:isolated:plane}. However,
no time-periodic oscillations with Sel'kov kinetics are possible when
the steady-state is linearly stable. In Fig.~\ref{fig:isolated:efflux}
we show how the HB boundary for an isolated cell depends on the
boundary efflux parameter $d_{2j}$. As expected, for the fixed value
$\mu_j=2$, the interval in $\alpha_j$ where oscillations are possible
is decreased when there is an efflux out of the cell boundary. The
shifting of the HB boundaries to the right in
Fig.~\ref{fig:isolated:efflux} indicates that a greater rate $\mu_j$
of production of $u_{j}^2$ is needed to ensure oscillations when there
is a boundary efflux. The interval in $\mu_j$ where oscillations are
possible, at least for some range of $\alpha_j>0$, is
$0<\mu_j< {\chi_j^{3/2}/\sqrt{\zeta_j}}$.

For the baseline parameter set $\mu_j=2$ and $\zeta_j=0.15$, in
\cref{SpecConfigSec} and \cref{sec:LargePopulation} we will show that
the inter-cell coupling via the bulk diffusion field can be sufficient
to trigger an oscillatory instability in the cells through a Hopf
bifurcation.

\begin{figure}[!h]
  \centering
    \begin{subfigure}[b]{0.32\textwidth}
      \includegraphics[width=\textwidth,height=4.4cm]{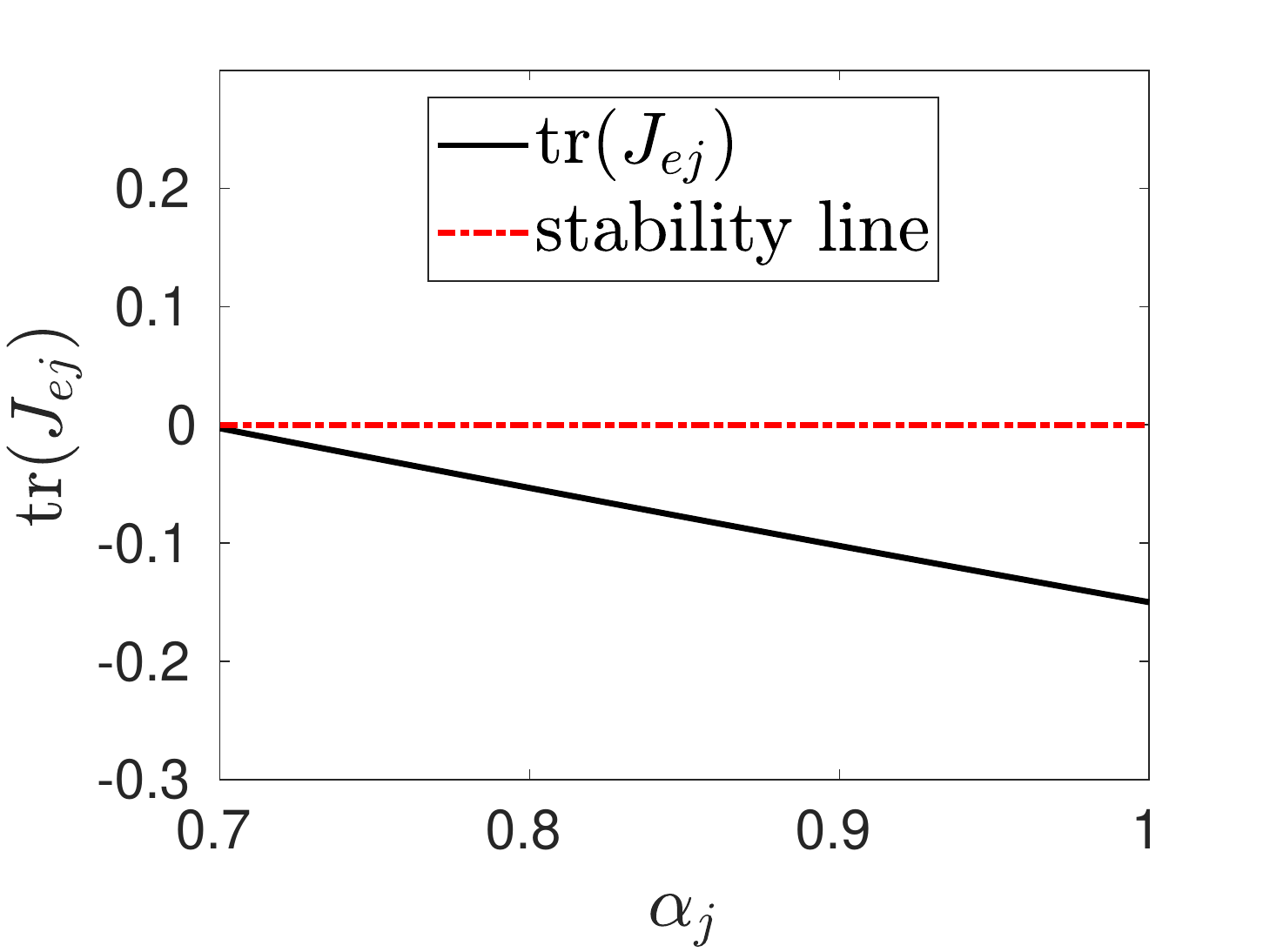}
        \caption{$\mbox{tr}(J_{ej})$ versus $\alpha_j$ for isolated cell} 
        \label{fig:isolated:trace}
    \end{subfigure}
    \begin{subfigure}[b]{0.32\textwidth}  
      \includegraphics[width=\textwidth,height=4.5cm]{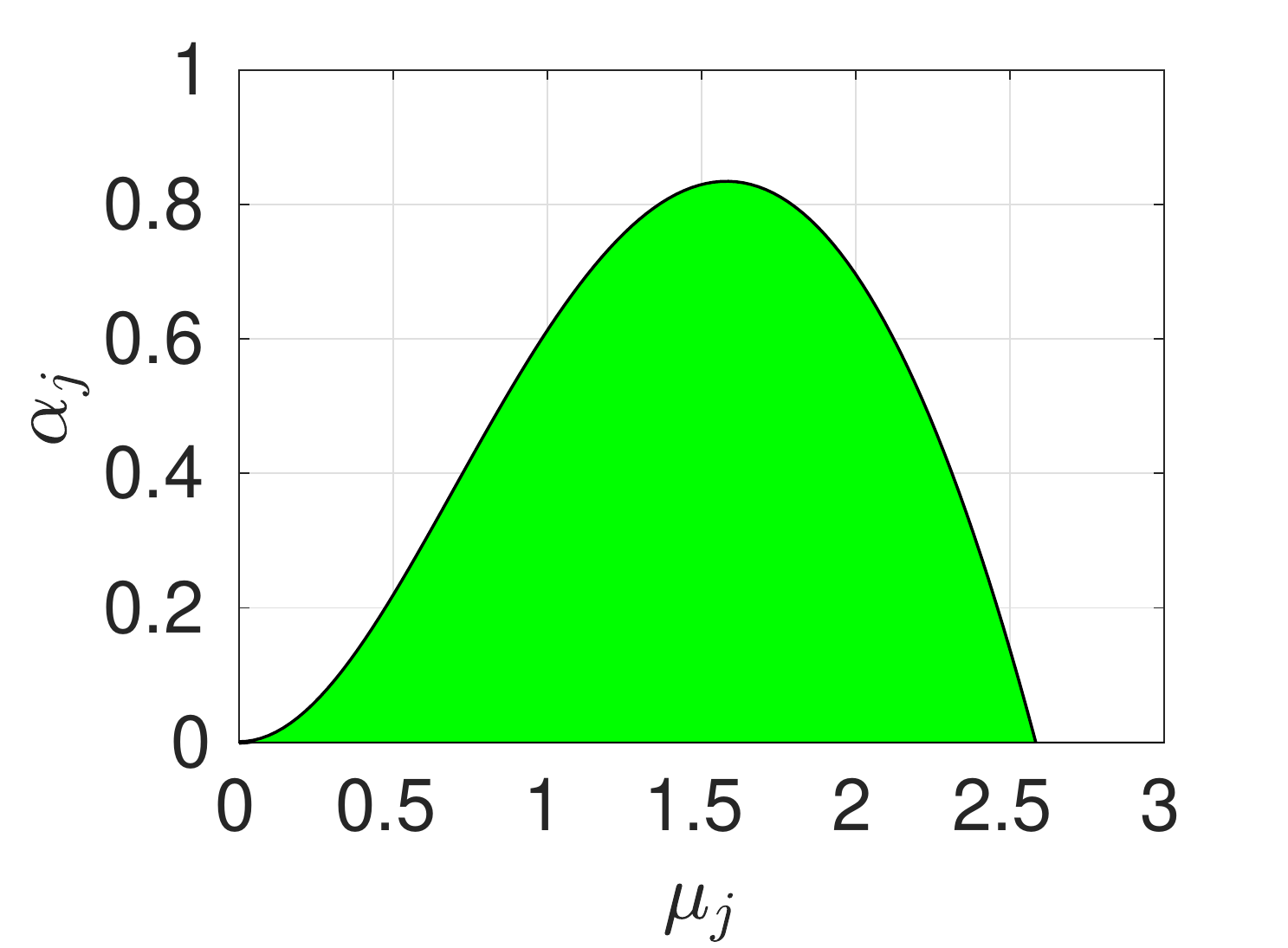}
        \caption{Instability region: isolated cell}
        \label{fig:isolated:plane}
      \end{subfigure}
    \begin{subfigure}[b]{0.32\textwidth}  
      \includegraphics[width=\textwidth,height=4.5cm]{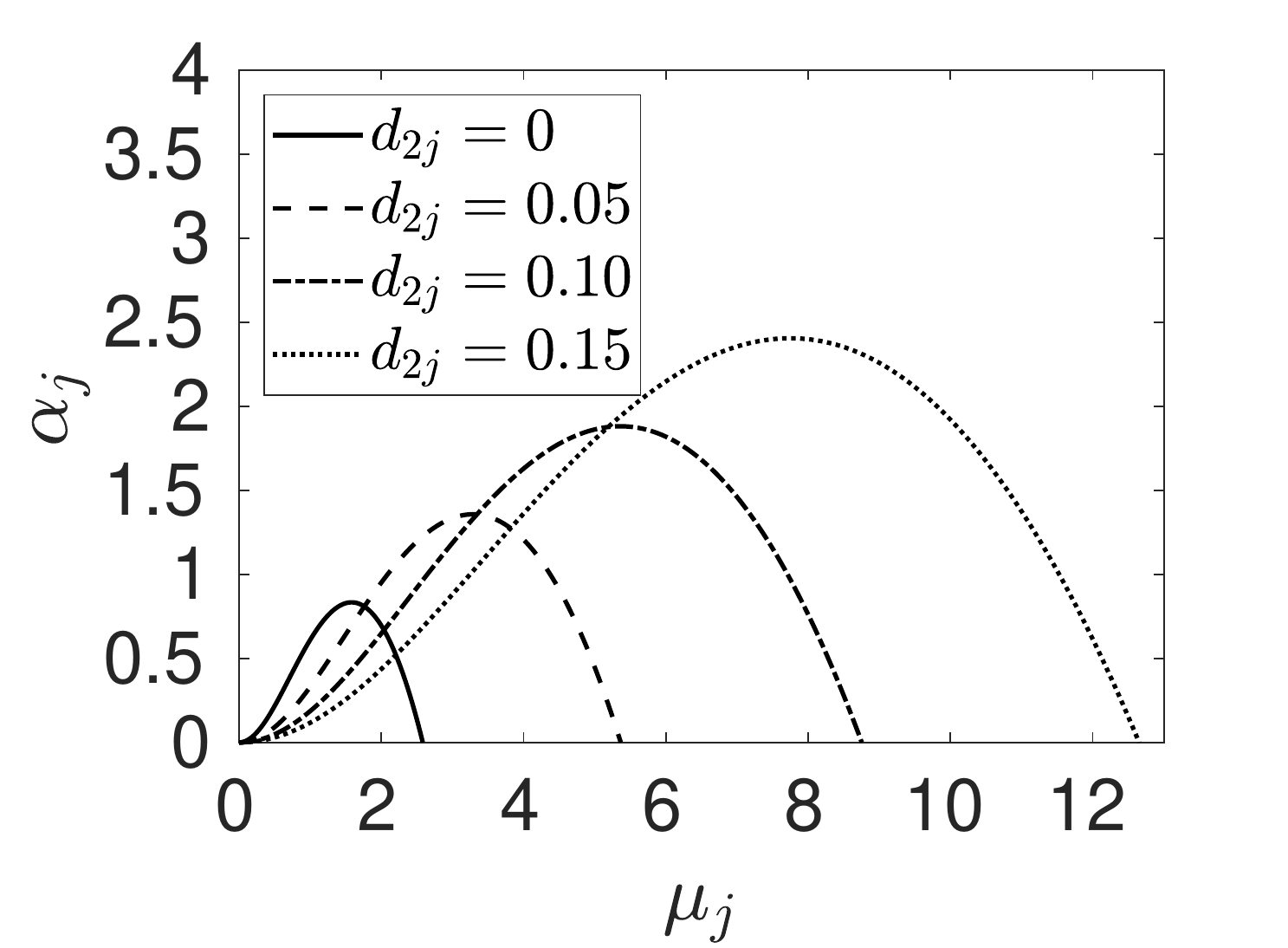}
        \caption{HB boundaries: isolated cell with efflux}
        \label{fig:isolated:efflux}
    \end{subfigure}
    \caption{Left panel: $\mbox{tr}(J_{ej})$, from
      \eqref{isolated:selkov:jac} versus $\alpha_j$, for the
      steady-state of the Sel'kov kinetics \eqref{isolated:selkov:ode}
      for an isolated cell, with $\mu_j=2$ and $\zeta_j=0.15$. This
      steady-state is linearly stable but the parameters are close to
      the stability threshold. Middle panel: Green-shaded region of
      instability where $\mbox{tr}(J_{ej})>0$ in the $\alpha_j$
      versus $\mu_j$ plane for the steady-state of an isolated cell
      with no boundary efflux and $\zeta_j=0.15$. Within this region,
      a time-periodic solution (limit cycle) occurs for an isolated
      cell. The HB boundary is given by \eqref{selkov:boundary} with
      $d_{2j}=0$. In the unshaded region the steady-state is linearly
      stable. Right panel: HB boundaries for an isolated cell (see
      \eqref{selkov:boundary}) with $\zeta_j=0.15$, $\tau=0.5$ and
      four boundary efflux parameters. A larger production rate
      $\mu_j$ is needed to support oscillations.}
\label{fig:selkov}
\end{figure}

For Sel'kov local reaction kinetics \eqref{isolated:selkov:ode} it is
readily shown that, with an arbitrary arrangement of cells, there is a
unique solution to the NAS \eqref{Intra} and \eqref{Linear_System}
given by
\begin{subequations}\label{ss:selkov}
\begin{equation}\label{ss:selkov_2}
  u_{ej}^{1}= \mu_j + \frac{2\pi D}{\tau} A_j \,, \qquad
  u_{ej}^{2}=\frac{\mu_j}{\alpha_j + \left( u_{ej}^{1}\right)^2} \,, \qquad
  j=1,\ldots,m \,,
\end{equation}
where ${\mathcal A}=(A_1,\ldots, A_m)^T$ satisfies the linear algebraic
system
\begin{equation}\label{ss:selkov_1}
\left(I + 2\pi \nu \mathcal{G}  + \nu\,D\,P_1 +
  \frac{2\pi \nu D}{\tau }  P_2 \right) {\mathcal A}= - \nu
P_2 \, \pmb{\mu} \,, \qquad \mbox{with} \qquad
 \pmb{\mu}\equiv (\mu_1,\ldots,\mu_m)^T \,.
\end{equation}
\end{subequations}
For $\nu\ll 1$ sufficiently small, the matrix in \eqref{ss:selkov_1}
is invertible, yielding a unique solution for ${\mathcal A}$.  For
Sel'kov kinetics we conclude that steady-state solutions of the
PDE-ODE model \eqref{DimLess_bulk} are always non-degenerate as
$\varepsilon\to 0$. As such, since by Proposition \ref{prop:stab}
stability cannot be lost via a zero-eigenvalue crossing, in
\cref{SpecConfigSec} and \cref{sec:LargePopulation} we will focus on
analyzing instabilities of the steady-state arising from Hopf
bifurcations.

\setcounter{equation}{0}
\setcounter{section}{2}
\section{A limiting ODE system with global coupling:
  \texorpdfstring{$D = \mathcal{O}(\nu^{-1})$}{Lg}}\label{largeD_ODE}

In this section, we use a singular perturbation approach to reduce the
dimensionless coupled PDE-ODE model \eqref{DimLess_bulk} to an ODE
system that is valid for the limiting regime
$D = \mathcal{O}(\nu^{-1}) \gg \mathcal{O}(1)$, where
$\nu \equiv {-1/\log\varepsilon}$ and $\varepsilon \ll 1$. This ODE
system depends weakly on the spatial configuration of the cells and on
the scaled diffusivity $D_0={\mathcal O}(1)$, defined by
$D={D_0/\nu}$.

Consider a collection of $m$ small cells centered at the points
$\pmb{x}_1, \ldots, \pmb{x}_m$ in a 2-D bounded domain $\Omega$. Define
$\Omega_p \equiv \cup_{j=1}^m \Omega_{\varepsilon_j}$ as the region
formed by the union of the cells, and the average bulk concentration 
$\overline{U} = \overline{U}(t)$ by
\begin{equation}\label{AveBulkConc}
  \overline{U} = \frac{1}{|\Omega\setminus\Omega_p|} \int_{\Omega\setminus\Omega_p}
  U \, \text{d}\pmb{x} \,.
\end{equation}
Our goal is to derive an ODE for
$\overline{U} \equiv \overline{U}(t;\nu)$, accurate to
${\mathcal O}(\nu)$, which is coupled to the intracellular dynamics of
the cells, as given in \eqref{DimLess_bulkb}. Upon multiplying
\eqref{DimLess_bulka} by $1/|\Omega\setminus\Omega_p|$ and using the
divergence theorem, we obtain
\begin{equation}\label{AveBulkConc_Diverg}
  \tau \overline{U}_t + \overline{U} =
  \frac{2 \pi}{|\Omega\setminus\Omega_p|} \sum_{j=1}^m \left( d_{2j} u_j^1  -
    \frac{ d_{1j} }{2 \pi \varepsilon} \int_{\partial \Omega_{\varepsilon_j}} U \,
    \text{d}s \right) \,,
\end{equation}
where we used $|\partial \Omega_{\varepsilon_j}| = 2 \pi \varepsilon$
for the perimeter of each cell. Since
$|\Omega_p| = m\pi \varepsilon^2$, we estimate
$|\Omega \setminus \Omega_{p}| = |\Omega|
-\mathcal{O}(\varepsilon^2)$, so that
$|\Omega \setminus \Omega_{p}| \to |\Omega| $ as $\varepsilon \to 0$
in \eqref{AveBulkConc_Diverg}. Next, by evaluating the integral in
\eqref{DimLess_bulkb}, and using
$|\partial \Omega_{\varepsilon_j}| = 2 \pi \varepsilon$, we obtain
\begin{equation}\label{Intra_Dynamic}
  \frac{\text{d} \pmb{u}_j}{\text{d} t}  = \pmb{F}_j( \pmb{u}_j) -
  \frac{2 \pi \pmb{e}_1 }{\,\tau }
  \left( d_{2j} u_j^1  - \frac{ d_{1j} }{2 \pi \varepsilon}
    \int_{\partial \Omega_{\varepsilon_j}} U \, \text{d}s \right)\,, \qquad
  j=1,\ldots,m \,.
\end{equation}
In \eqref{AveBulkConc_Diverg} and \eqref{Intra_Dynamic}, we must
estimate the bulk concentration $U \equiv U(\pmb{x},t)$ on the
$j^{\text{th}}$ cell boundary $\partial\Omega_{\varepsilon_j}$.

For $D = {\mathcal O}(\nu^{-1})$ we introduce the scaling
\begin{equation}\label{Eqn:Scaling}
  D = \frac{D_0}{\nu}, \quad \text{where} \quad \nu \equiv
  \frac{-1}{\log\varepsilon}\,, \quad  \varepsilon \ll 1 \quad
    \text{and} \quad D_0 = \mathcal{O}(1)\,.
\end{equation}
Upon substituting \eqref{Eqn:Scaling} into \eqref{DimLess_bulk}, we
obtain in the bulk region that
\begin{equation} \label{DimLess_Rescaled_bulk}
\begin{split}
  \tau U_t &= \frac{D_0}{\nu}\, \Delta U - \, U\,, \qquad t > 0\,,
  \quad \pmb{x} \in \Omega \setminus \cup_{j=1}^{m} \Omega_{\varepsilon_j}\,; \\
  \partial_n \, U =\, 0, \quad \pmb{x} \in \partial \Omega_{\varepsilon}\,; &
  \qquad \varepsilon \frac{D_0}{\nu} \, \partial_n U  =  d_{1j} \, U -
  d_{2j} u_j^1  \quad \text{on} \quad \partial \Omega_{\varepsilon_j}\,,
  \qquad j =1,\ldots, m\,.
\end{split}
\end{equation}
In the inner region at an $\mathcal{O}(\varepsilon)$ neighborhood of
the $j^{\text{th}}$ cell, we introduce the inner variables
$ \pmb{y} = \varepsilon^{-1} (\pmb{x} - \pmb{x}_j)$ and
$U(\pmb{x}) = V_j (\pmb{x}_j + \varepsilon \pmb{y}; \nu)$, with
$\rho = |\pmb{y}|$.  Writing \eqref{DimLess_Rescaled_bulk} in terms of
these inner variables, we obtain for $\varepsilon \to 0$ that
\begin{equation}\label{WM_dimlessInner}
\begin{split}
  \Delta V_j = 0,   \quad \rho \geq 1\,;  \qquad
  \, \partial_{\rho} \, V_j  =  \frac{\nu}{D_0} \Big( \, d_{1j} \, V_j -
  d_{2j} \, u_j^1 \Big)\,, \quad \text{on} \quad \rho = 1\,, \qquad
  j = 1, \ldots, m\,,
\end{split}
\end{equation}
which has the radially symmetric solution
\begin{equation}\label{WM_nextInnerSol}
\begin{split}
  V_{j} = \nu \,b_j \log\rho + U^0_{j}\,,\qquad \text{where} \qquad
  b_j \equiv \frac{1}{D_0} \Big( d_{1j}\, U^0_{j} - d_{2j}\,u_j^1\Big)\,,
  \qquad j=1,\ldots,m\,,
\end{split}
\end{equation}
where $V_j |_{\rho = 1} = U_j^0$ is to be determined. By writing
\eqref{WM_nextInnerSol} in the outer $\pmb{x}$ variable, and  using
$|\pmb{y}| = \varepsilon^{-1} |\pmb{x} - \pmb{x}_j|$, we obtain the
following asymptotic matching condition for the outer solution in the bulk
region:
\begin{equation}\label{WM_Match_outer}
\begin{split}
  U \sim \nu \, b_j \log|\pmb{x} - \pmb{x}_j| + \Big( \frac{d_{1j}}{D_0} +
  1 \Big) \, U^0_{j} - \frac{d_{2 j}}{D_0}\,u_j^1\,, \qquad \text{as} \qquad
  \pmb{x} \to \pmb{x}_j\,.
\end{split}
\end{equation}
Since $V_j(\rho) = U_j^0$ on $\rho = 1$ from \eqref{WM_nextInnerSol}, we
have $\int_{\partial \Omega_{\varepsilon j}} U \, \text{d}s = \varepsilon
\int_{0}^{2\pi} V_j |_{\rho = 1} \, \text{d}\theta = 2 \pi \varepsilon
\, U_j^0$. Then, from \eqref{AveBulkConc_Diverg} and
\eqref{Intra_Dynamic}, and recalling \eqref{WM_nextInnerSol} for
$b_j$, we obtain for $\varepsilon \to 0$ that
\begin{equation}\label{Intra_DynamInt}
  \begin{split}
    \tau \overline{U}_t  + \overline{U}  &= \frac{2 \pi}{|\Omega|}
    \sum_{j=1}^m \left( d_{2j} u_j^1  - d_{1j}  U_j^0 \right)=
    -\frac{2\pi D_0}{|\Omega|} \sum_{j=1}^{m} b_j \,, \\
  \frac{\text{d} \pmb{u}_j}{\text{d} t}  - \pmb{F}_j( \pmb{u}_j) &=
  - \frac{2 \pi \pmb{e}_1 }{\,\tau }  \left( d_{2j} u_j^1  - d_{1j}  U_j^0
  \right) = \frac{2\pi D_0}{\tau} \pmb{e}_1 b_j\,, \qquad j=1,\ldots,m \,.
 \end{split}
\end{equation}

To complete the derivation of the ODE system we must obtain an
algebraic system for $b_j$ for $j=1,\ldots,m$ from the analysis of the
outer solution.  From \eqref{DimLess_Rescaled_bulk} and
\eqref{WM_Match_outer}, and relating $U_j^0$ to $b_j$ using
\eqref{WM_nextInnerSol}, the outer problem for $U(\pmb{x},t)$ is
\begin{equation} \label{Compl_Outer_prob}
\begin{split}
  \tau U_t &= \frac{D_0}{\nu}\, \Delta U - \, U\,, \quad t > 0\,,
  \quad \pmb{x} \in \Omega \setminus \cup_{j=1}^{m} \Omega_{\varepsilon_j}\,;
  \qquad \partial_n \, U =\, 0\,, \quad \pmb{x}\in \Omega \,;\\
  \qquad U & \sim  \nu b_j \log|\pmb{x} - \pmb{x}_j| +
  b_j \left( 1 + \frac{D_0}{d_{1j}} \right) + \frac{d_{2j}}{d_{1j}} u_j^1\,,
  \qquad \text{as} \qquad \pmb{x} \to \pmb{x}_j \,.
\end{split}
\end{equation} 
We then expand $U(\pmb{x},t)$ as
\begin{equation}\label{WM_Uexpnd}
  U(\pmb{x},t) = \overline{U} + \frac{\nu}{D_0}\, U_1(\pmb{x},t) + \ldots\,,
  \qquad \mbox{where} \quad \int_{\Omega} U_1 \, d\pmb{x} = 0 \,.
\end{equation}
The zero average constraint on $U_1$ ensures that $\overline{U}$ is
the spatial average of $U$ to terms of order ${\mathcal O}(\nu)$. Upon
substituting \eqref{WM_Uexpnd} into \eqref{Compl_Outer_prob}, we
obtain in the sense of distributions that $U_1$ satisfies
\begin{equation}\label{WM_orderOne_Comp}
\begin{split}
  \Delta U_1 & = \tau \, \overline{U}_{t} + \overline{U} + 2\pi D_0
  \sum_{i=1}^{m} b_i \delta(\pmb{x}-\pmb{x}_i) \,, \quad
  \pmb{x} \in \Omega \,; \qquad \partial_n U_1 = 0\,, \quad
  \pmb{x} \in \partial \Omega\,; \\
  U_1  & \sim b_j D_0  \log|\pmb{x} - \pmb{x}_j| - \frac{D_0}{\nu}\overline{U}
  + \frac{D_0}{\nu }\left[ b_j \Big( 1+ \frac{D_0}{d_{1j}} \Big) +
    \frac{d_{2 j}}{d_{1j}} \,u_j^1\right]\,, \qquad \text{as} \qquad \pmb{x}
  \to \pmb{x}_j \,.
\end{split}
\end{equation}
The divergence theorem applied to \eqref{WM_orderOne_Comp} yields the
ODE given in \eqref{Intra_DynamInt} while the linear system for $b_j$,
for $j=1,\ldots,m$, is obtained from the constraints involved with
specifying the form of the regular part of the singularity behavior in
\eqref{WM_orderOne_Comp}.

The solution to \eqref{WM_orderOne_Comp}, with $\int_{\Omega}U_1 \, d\pmb{x}=0$
is written as
\begin{equation}\label{WM_U1_Sol}
U_1 = - 2\pi D_0 \sum_{i=1}^{m} b_i\,G_0(\pmb{x};\pmb{x}_i) \,,
\end{equation}
where $G_0(\pmb{x};\pmb{x}_j)$ is the unique Neumann Green's function
satisfying
\begin{subequations}\label{WM_NUEGREEN}
\begin{align}
  \Delta G_0 &= \frac{1}{|\Omega|} - \delta(\pmb{x} - \pmb{x}_j)\,, \quad
       \pmb{x} \in \Omega\,;\qquad \partial_n G_0  = 0\,, \quad
       \pmb{x} \in \partial \Omega\,;\label{WM_NUEGREENA}\\
  G_0(\pmb{x};\pmb{x}_j) & \sim  -\frac{1}{2\pi}
     \log|\pmb{x} - \pmb{x}_j| + R_{0j} + o(1)\,,
   \quad \text{as} \quad \pmb{x} \to \pmb{x}_j\,, \quad \text{and}
    \quad \int_{\Omega} \, G_0 \, \text{d}\pmb{x} =0 \,.\label{WM_NUEGREENB}
\end{align}
\end{subequations}
Here $R_{0j}$ is the regular part of $G_0$ at $\pmb{x} = \pmb{x}_j$.
By expanding \eqref{WM_U1_Sol} as $\pmb{x}\to\pmb{x}_j$, we enforce
that the nonsingular part of the resulting expression agrees with that
in \eqref{WM_orderOne_Comp}. This yields that 
\begin{equation}\label{ode:bj}
  b_j \left( 1 + \frac{D_0}{d_{1j}}\right) + 2\pi \nu \left(
    b_j R_{0j} + \sum_{i\neq j}^{m} b_i G_{0ji} \right) = \overline{U} -
  \frac{d_{2j}}{d_{1j}} u_{j}^{1} \,, \qquad j=1,\ldots, m\,,
\end{equation}
where $G_{0ji}=G_{0}(\pmb{x}_j;\pmb{x}_i)$. This linear system for
$b_1,\ldots,b_m$ is then coupled to the ODEs given in
\eqref{Intra_DynamInt}. Upon writing this ODE system in matrix form we
summarize the result as follows:

\vspace*{0.2cm}
\begin{prop}\label{prop:dyn}  Let $\varepsilon\to 0$ and assume that
  $D={D_{0}/\nu}\gg 1$ where $D_0={\mathcal O}(1)$ and
  $\nu={-1/\log\varepsilon}\ll 1$. Then, the PDE-ODE system
  \eqref{DimLess_bulk} reduces to the following $nm+1$ dimensional ODE
  system for
  $\overline{U}\approx |\Omega|^{-1}\int_{\Omega} U\,d\pmb{x}$ and the
  intracellular species:
\begin{subequations}\label{reducedODE}
\begin{equation}\label{reducedODE_1}
\begin{split}
  \frac{\text{d}}{\text{d}t} \overline{U}  = - \frac{1}{\tau} \overline{U}  -
  \frac{2 \pi D_0}{\tau |\Omega|} \pmb{e}^T \pmb{b}\,; \qquad 
  \frac{\text{d} \pmb{u}_j}{\text{d} t}  = \pmb{F}_j( \pmb{u}_j) +
  \frac{2 \pi D_0 \pmb{e}_1 }{\,\tau }  b_j\,, \qquad j=1,\ldots,m\,,
\end{split}
\end{equation}
where $\pmb{e}\equiv (1,\ldots,1)^T$, $\pmb{e}_1 \equiv (1,0,\ldots,0)^T$
and $\pmb{b} \equiv (b_1, \ldots, b_m)^T$. In \eqref{reducedODE_1},
$\pmb{b}$ is the solution to the linear system
\begin{equation}\label{reduced_bmat}
  \big( I + D_0 P_1 + 2 \pi \nu \, \mathcal{G}_0 \big) \pmb{b} =
  \overline{U}\,\pmb{e} -  P_2 \,\pmb{u}^1\,,
\end{equation}
\end{subequations}
where $\pmb{u}^1 \equiv (u^1_1, \ldots, u^1_m)^T$ and $P_1$ and $P_2$ are the
diagonal matrices defined in terms of the permeabilities by
\eqref{GreensMatrix}. In \eqref{reduced_bmat}, ${\mathcal G}_0$ is
the Neumann Green's matrix with matrix entries
\begin{equation}\label{WM_NeumannGmatrix}
  (\mathcal{G}_0)_{ij} = (\mathcal{G}_0)_{ji} = G_0(\pmb{x}_i,\pmb{x}_j)\,,
  \quad i \neq j \qquad \text{and} \quad (G_0)_{jj} = R_{0j}\,.
\end{equation}
For $\nu\ll 1$, this ODE system is equivalent up to ${\mathcal O}(\nu)$ terms
to
\begin{subequations}\label{reducedODE_approx}
\begin{equation}\label{reducedODE_approx_1}
\begin{split}
  \frac{\text{d}}{\text{d}t} \overline{U}  &= - \frac{1}{\tau} \overline{U}  -
\frac{2 \pi}{\tau |\Omega|} \left[\overline{U} \pmb{e}^T {\mathcal C} \pmb{e} -
    \pmb{e}^T {\mathcal C} P_2 \, \pmb{u}^1\right] + {\mathcal O}(\nu^2) \,,\\
  \frac{\text{d} \pmb{u}_j}{\text{d} t}  &= \pmb{F}_j( \pmb{u}_j) +
  \frac{2 \pi \pmb{e}_1 }{\,\tau }  \left[ \overline{U}\left({\mathcal C}
 \pmb{e}\right)_j - \left( {\mathcal C} P_2 \, \pmb{u}^1\right)_{j} \right] \,,
  \qquad j=1,\ldots,m\,,
\end{split}
\end{equation}
where the matrix ${\mathcal C}$ is defined in terms of ${\mathcal G}_0$ and
a diagonal matrix ${\mathcal P}$ by
\begin{equation}\label{reducedODE:cmat}
  {\mathcal C} \equiv {\mathcal P} - \frac{2\pi\nu}{D_0} {\mathcal P}
  {\mathcal G}_0 {\mathcal P} \,, \qquad
   {\mathcal P}\equiv \mbox{diag}\left(\frac{D_0 d_{11}}{d_{11}+D_0}\,,\ldots,
\frac{D_0 d_{1m}}{d_{1m}+D_0}\right)\,.
\end{equation}
\end{subequations}
In the well-mixed limit for which $D_0\to \infty$, \eqref{reduced_bmat} yields
$D_0 b_j \sim \overline{U} d_{1j} - d_{2j} u_{j}^{1}$, so that
\eqref{reducedODE_1} reduces to
\begin{equation}\label{ode:well_mixed}
  \overline{U}_t  = - \frac{1}{\tau} \overline{U}  -
  \frac{2 \pi }{\tau |\Omega|} \sum_{j=1}^{m}( \overline{U} d_{1j} -
  d_{2j} u_j^1)\,; \qquad \frac{\text{d} \pmb{u}_j}{\text{d} t}  =
  \pmb{F}_j( \pmb{u}_j) + \frac{2 \pi  \pmb{e}_1 }{\,\tau }
  (\overline{U} d_{1j} - d_{2j} u_j^1)\,, \qquad j=1,\ldots,m\,.
\end{equation}
\end{prop}

To derive \eqref{reducedODE_approx} from \eqref{reducedODE},
we approximate the solution $\pmb{b}$ to \eqref{reduced_bmat} up to
terms of order ${\mathcal O}(\nu)$. By inverting the diagonal matrix
$I+D_0P_1$, we obtain from \eqref{reduced_bmat} that
\begin{equation}\label{ode:b_red1}
    \pmb{b} = \frac{1}{D_0} \left( I + \frac{2\pi \nu}{D_0} {\mathcal
        P}{\mathcal G}_0\right)^{-1} \left( \overline{U} {\mathcal P}
      \pmb{e} - {\mathcal P} P_2 \pmb{u}^1\right) \sim \frac{1}{D_0}
    \left( I - \frac{2\pi \nu}{D_0} {\mathcal P}{\mathcal G}_0\right)
    \left( \overline{U} {\mathcal P} \pmb{e} - {\mathcal P} P_2
      \pmb{u}^1\right) =\frac{1}{D_0} \left( \overline{U} {\mathcal
        C}\pmb{e} - {\mathcal C} P_2 \pmb{u}^1\right) + {\mathcal
      O}(\nu^2)\,,
\end{equation}
where ${\mathcal C}$ and ${\mathcal P}$ are given in
\eqref{reducedODE:cmat}. The ODE system \eqref{reducedODE_approx}
results from substituting \eqref{ode:b_red1} into \eqref{reducedODE}.

The ODE systems \eqref{reducedODE}, or alternatively
\eqref{reducedODE_approx}, for the regime $D={\mathcal O}(\nu^{-1})$
are accurate up to and including terms of order ${\mathcal O}(\nu)$
and show how the intracellular species are globally coupled through
the spatial average of the bulk field. Since these ODE systems depend
on the scaled diffusivity parameter $D_0$ and include the effect of
the spatial configuration $\pmb{x}_1\,,\ldots,\pmb{x}_m$ of the cells
through the Neumann Green's matrix, these ODE systems can account for
both diffusion-sensing and quorum-sensing behavior (see
\cref{SpecConfigSec} and \cref{sec:LargePopulation}). In contrast, the
limiting well-mixed ODE system \eqref{ode:well_mixed}, originally
derived in \cite{jia2016} for the simpler case of identical cells,
depends only on the number $m$ of cells. As a result, the well-mixed
ODE dynamics is independent of the diffusivity and the spatial
configuration of the cells.

In our numerical experiments in \cref{SpecConfigSec} and
\cref{sec:LargePopulation} using the ODE system \eqref{reducedODE} the
domain $\Omega$ is the unit disk. For the unit disk, the Neumann Green's
function $G_0(\pmb{x};\pmb{x}_j)$ and its regular part $R_{0j}$,
satisfying \eqref{WM_NUEGREEN}, are (see equation (4.3) of
\cite{KTW2005})
\begin{equation}\label{gr:gmrm}
\begin{split}
  G_0(\pmb{x};\pmb{x}_j) &= -\frac{1}{2\pi}\log|\pmb{x}-\pmb{x}_j| -
  \frac{1}{4\pi}\log\left( |\pmb{x}|^2|\pmb{x}_j|^2 + 1 - 2 \pmb{x}
    \cdot \pmb{x}_j\right) + \frac{ (|\pmb{x}|^2 + |\pmb{x}_j|^2
    )}{4\pi} - \frac{3}{8\pi} ,
  \\
  R_{0j} &= -\frac{1}{2\pi}\log\left(1 - |\pmb{x}_j|^2\right) +
  \frac{|\pmb{x}_j|^2}{2\pi} - \frac{3}{8\pi} \,.
\end{split}
\end{equation}
For an arbitrary cell pattern
$\lbrace{\pmb{x}_1,\ldots,\pmb{x}_m \rbrace}$, \eqref{gr:gmrm} is used
to evaluate the Neumann Green's matrix ${\mathcal G}_0$ as needed in
\eqref{reducedODE}.

\setcounter{equation}{0}
\setcounter{section}{3}
\section{A ring and center cell pattern}\label{SpecConfigSec}

With Sel'kov reaction kinetics, we apply the theory developed in
\cref{Analysis} to a ring and center cell configuration in the unit
disk. This pattern is characterized by $m-1\geq 2$ equally spaced
cells on a concentric ring within the unit disk, and with one at the
center of the disk (see Fig.~\ref{def_Mcells}).  For this pattern, the
GCEP \eqref{full_gcep} will be used to obtain tractable nonlinear
algebraic equations that can be solved numerically to compute HB
boundaries in the $\tau$ versus $D$ parameter plane. In addition, a
winding number criterion is developed to count the number of unstable
eigenvalues in open regions of this parameter plane.  Some of our
examples will show that rather small changes in either the
permeabilities or reaction kinetic parameters of the center cell can
significantly alter the region in parameter space where oscillations
occurs.

\begin{figure}[!ht]
  \centering
    \begin{subfigure}[b]{0.32\textwidth}
      \includegraphics[width=\textwidth,height=4.4cm]{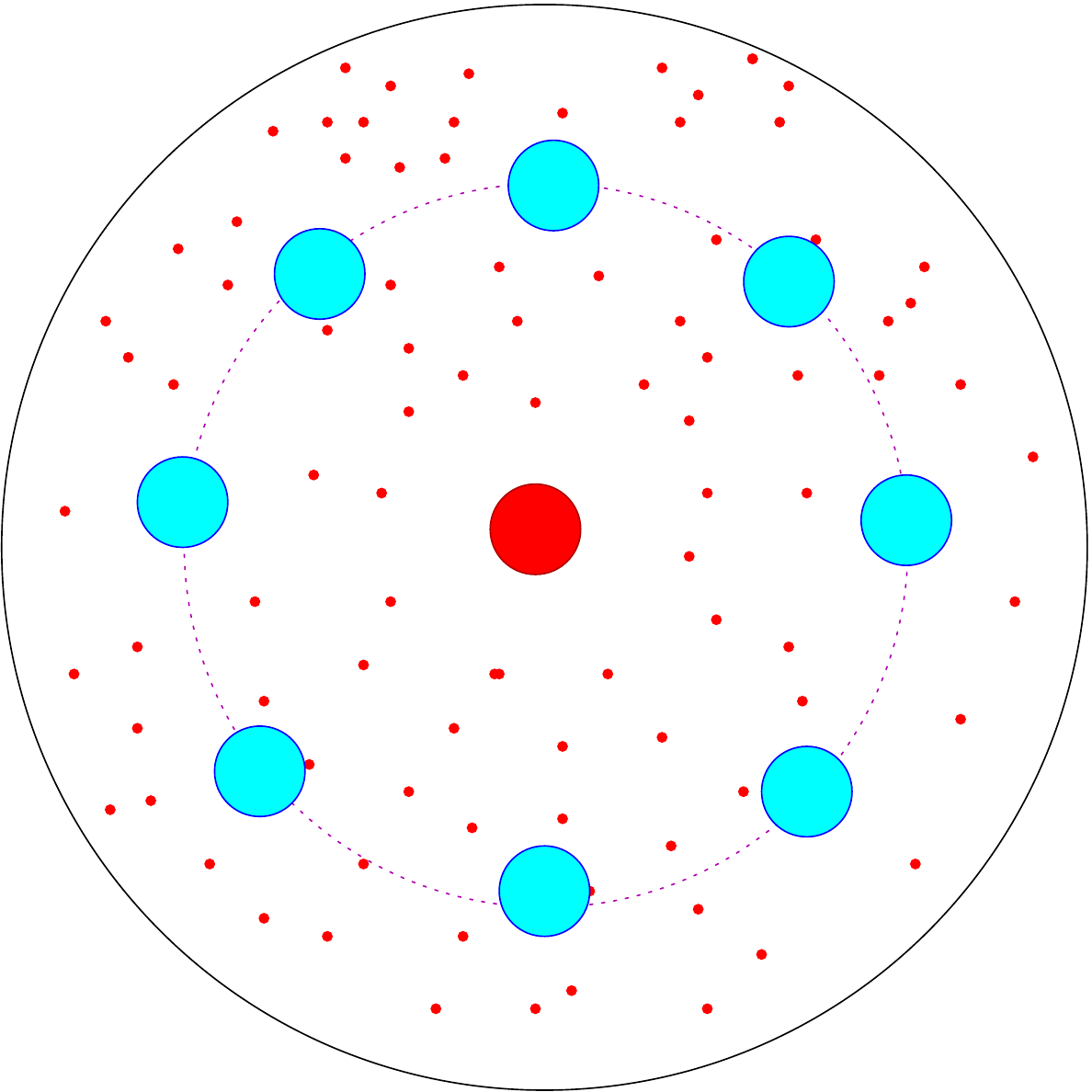}
    \end{subfigure}
    \qquad\qquad
    \begin{subfigure}[b]{0.32\textwidth}  
      \includegraphics[width=\textwidth,height=4.4cm]{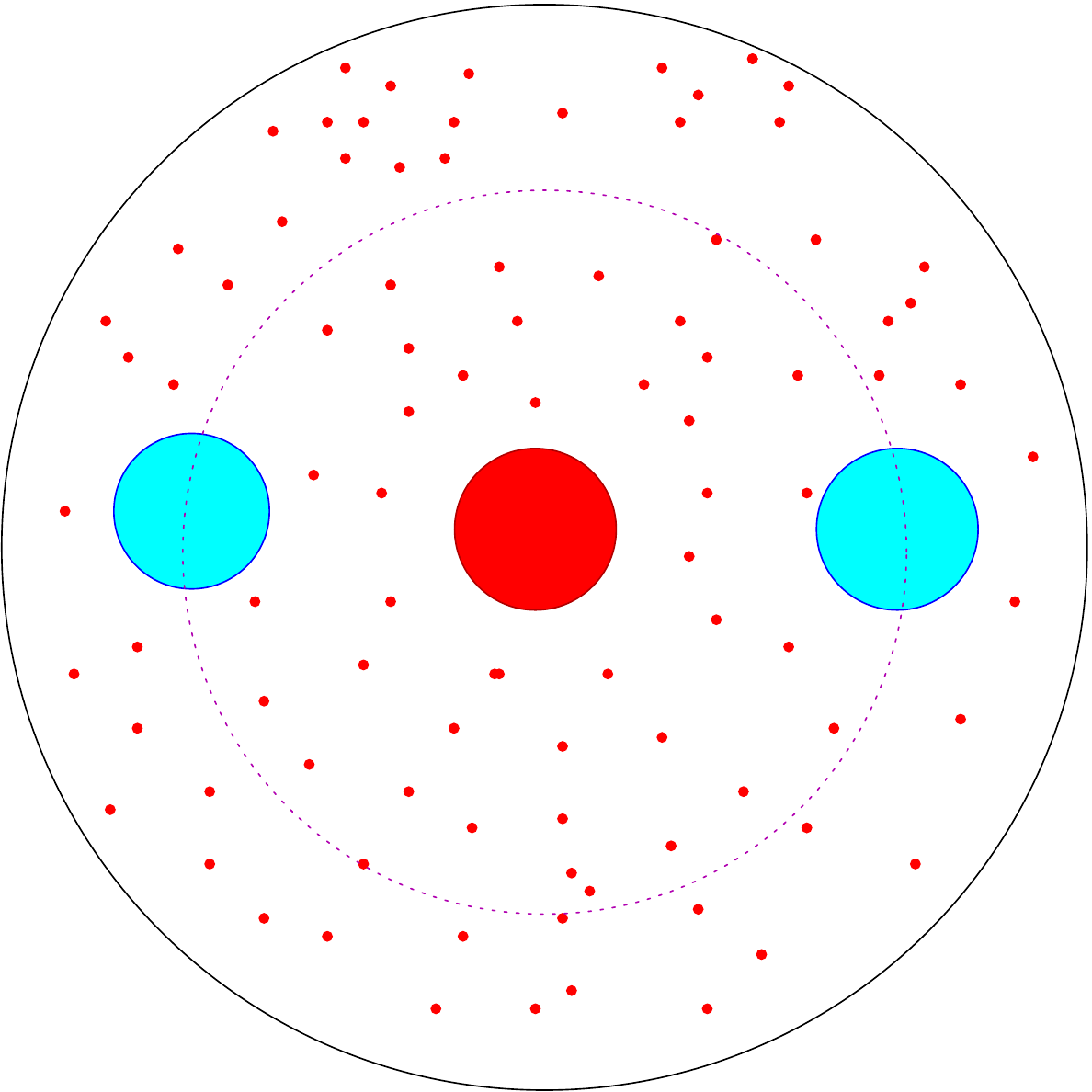}
    \end{subfigure}
\caption{Schematic showing a ring and center cell configuration of
  $m=9$ (left) and $m=3$ (right) cells in the unit disk. The ring
  cells are identical and equally spaced on a concentric ring within
  the disk (in cyan). The center cell (in red), possibly has different
  parameters. The red dots represent the signaling molecules secreted
  in the bulk region by the cells.}
\label{def_Mcells}
\end{figure}

\subsection{Analysis of the GCEP}\label{RingCenterHole}
 
Consider a ring and center cell pattern of $m$ cells where the cells
on the ring of radius $r_0$ have identical parameters, but where the
center cell has possibly different permeabilities or Sel'kov kinetic
parameters (see Fig.~\ref{def_Mcells}). The cell centers are at
\begin{equation}\label{cent:xj}
  \pmb{x}_j=r_0\left(\cos\left(\frac{2\pi (j-1)}{(m-1)}\right),\,\,
    \sin\left(\frac{2\pi (j-1)}{(m-1)}\right)\right) \,, \quad j=1\,,\ldots,
m-1  \,; \qquad \pmb{x}_m={\bf 0} \,,
\end{equation}
where $0<r_0<1$. For this pattern, the reduced-wave Green's matrix
$\mathcal{G}$ in \eqref{GreensMatrix} can be partitioned as
\begin{equation}\label{ring:gblock}
\mathcal{G} = 
\left(\begin{array}{ccc:c}
            &  &  & g_{m}\\ 
		    & \mathcal{G}_{m-1}  &  & \vdots\\
		    &  &  &  g_{m} \\ \hdashline[2pt/2pt]
		g_{m} & \dots & g_{m} & R_{m}  
		\end{array}
              \right)\,; \qquad 
              g_{m} \equiv G(\pmb{x}_j,\pmb{x}_m) = G(\pmb{x}_m, \pmb{x}_j)\,,
              \quad j=1\,\ldots,m-1 \,, \quad R_m = R(\pmb{x}_m) \,.
\end{equation}
Here ${\mathcal G}_{m-1}$ is the $(m-1) \times (m-1)$ symmetric matrix
block representing the interaction between the cells on the
ring. Since this block is also cyclic it has the eigenpair
\begin{equation}\label{cent:g_eig}
  {\mathcal G}_{m-1} \pmb{e} = \omega_1 \pmb{e} \,, \qquad
  \mbox{with} \qquad \pmb{e} = (1,\dots,1)^T\in \R^{m-1} \,,
  \quad \mbox{and} \quad \omega_1\equiv
  R_1 + \sum_{j=2}^{m-1} G(\pmb{x}_1,\pmb{x}_j) \,.
\end{equation}
In \eqref{ring:gblock} there is a common interaction, represented
by $g_m$, between each ring cell and the center cell owing to the
rotational symmetry and the fact that the ring cells are all
equidistant from the center cell.
 
For the identical ring cells, we label their permeabilities as
$d_1 = d_{1j}$ and $d_2 = d_{2j}$ for $j = 1, \dots, m-1$ and their
common Sel'kov kinetic parameters as $\mu_1=\mu_{j}$,
$\alpha_1=\alpha_j$ and $\zeta_1=\zeta_j$ for $j=1,\ldots,m-1$.  Since
the unique steady-state solution to \eqref{ss:selkov_1} has the form
${\mathcal A}=(A_c,\ldots,A_c,A_m)^T$, we readily find that $A_c$ and
$A_m$ satisfies the $2\times 2$ linear system
\begin{subequations}\label{Linear_reduceSys}
\begin{align}
  \left( 1 + \frac{\nu D}{d_1} + 2\pi\nu \omega_1  + \frac{2\pi \nu D}{\tau}
  \frac{d_2}{d_1} \right) & A_c +2 \pi \nu g_m \, A_m = -
   \nu \mu_1 \frac{d_2}{d_1}\,; \label{Linear_reduceSysA}\\
  \left( 1 + \frac{\nu D}{d_{1m}} + 2\pi\nu R_m  +
  \frac{2\pi \nu D}{\tau} \frac{d_{2m}}{d_{1m}} \right) & A_m  +
  2 \pi \nu \, g_m \, (m - 1) A_c   = - \nu \mu_m \frac{d_{2m}}{d_{1 m}} \,,
    \label{Linear_reduceSysB}
\end{align}
\end{subequations}
where $\omega_1$ is the eigenvalue of ${\mathcal G}_{m-1}$ in
\eqref{cent:g_eig}. In terms of ${\mathcal A}=(A_c,\ldots,A_c,A_m)^T$, the
steady-state for the intracellular species as obtained from
\eqref{ss:selkov_2} is
\begin{equation}\label{Stdy_State_U}
  \begin{split}
   u_{e j}^{1} = \begin{cases} 
     u_{e 1}^{1} \equiv \mu_1 + \frac{2 \pi D}{\tau} \, A_c \,, & 
     j=1 \,,\ldots, m-1 \,, \\
     u_{e m}^{1} \equiv \mu_m + \frac{2 \pi D}{\tau} \, A_m \,, & 
     j=m \,,
   \end{cases} \,; \qquad
    u_{e j}^{2} = \begin{cases} 
      u_{e 1}^{2} \equiv  \frac{\mu_1}{\alpha_1 \,\, + \,\, (u_{e1}^1)^2}\,,
         & j=1 \,,\ldots, m-1 \,, \\
         u_{e m}^{2} \equiv \frac{\mu_m}{\alpha_m \,\, + \,\, (u_{e m}^1)^2}\,,
         & j=m \,.
         \end{cases}
\end{split}
\end{equation}

Next, we determine the GCEP for the ring and center cell pattern using
\eqref{Global_System}. For this pattern, the GCEP matrix
$\mathcal{M}(\lambda)$ in \eqref{Global_System} is written as
\begin{subequations}\label{M_M0}
\begin{equation}\label{M_M0_1}
\mathcal{M}(\lambda) = 2\pi \nu \, \mathcal{G}_{\lambda} +  \mathcal{M}_0\,,
\end{equation}
where $\mathcal{G}_{\lambda}$ is the eigenvalue-dependent Green's
matrix, as defined in \eqref{GreenMat}, and where the diagonal
$\mathcal{M}_0$ is defined by
\begin{equation}\label{M_matrix1}
\begin{split}
\mathcal{M}_0 = 
\begin{pmatrix}
M_0 & 	& 	& 	&  \\
 & \ddots & 	&	 &   \\
	 & 	& 	& M_0	&  \\
	  & 	& 	& 	& M_m
\end{pmatrix}\,, \qquad \mbox{where} \qquad \begin{cases} 
  M_0 \equiv 1 + \frac{\nu D}{d_1} + \frac{2\pi\nu D}{\tau} \frac{d_2}{d_1}
  \mathit{K}_c\,, \\
    \\
    M_m \equiv 1 + \frac{\nu D}{d_{1m}} + \frac{2\pi\nu D}{\tau}
    \frac{d_{2m}}{d_{1m}} \mathit{K}_m\,.
  \end{cases}
\end{split}
\end{equation}
\end{subequations}
Here $\mathit{K}_c=\mathit{K}_c(\lambda)$ and
$\mathit{K}_m=\mathit{K}_m(\lambda)$ are the entries of the
$m \times m$ diagonal matrix
$\mathcal{K} =
\text{diag}(\mathit{K}_c,\dots,\mathit{K}_c,\mathit{K}_m)$ defined in
\eqref{K_mat_entry}. For the Sel'kov kinetics given in
\eqref{isolated:selkov:ode}, \eqref{K_mat_entry} yields
\begin{subequations}\label{cent:KcKm}
\begin{equation}\label{cent:KcKm_def}
  \mathit{K}_c \equiv \frac{\lambda + \mbox{det}(J_1)}{\lambda^2 -
    \mbox{tr}(J_1) \lambda + \mbox{det}(J_1)} \,, \qquad
  \mathit{K}_m \equiv \frac{\lambda + \mbox{det}(J_m)}{\lambda^2 -
    \mbox{tr}(J_m) \lambda + \mbox{det}(J_m)} \,,
\end{equation}
where the trace and determinant of the Jacobians of the intracellular
dynamics for the identical ring cells and the center cell are given in
terms of the steady-state values in \eqref{Stdy_State_U} by
\begin{equation}\label{cent:KcKm_jac}
  \begin{split}
  \mbox{det}(J_1) = \zeta_1 \left(\alpha_1 + (u_{e 1}^{1})^2\right)\,, \qquad
  \mbox{det}(J_m) &= \zeta_m \left(\alpha_m + (u_{e m}^{1})^2\right)\,, \\
    \mbox{tr}(J_1) = \frac{ \left[2 (u_{e 1}^{1}) \mu_1 -
      \left(\alpha_1 + (u_{e 1}^{1})^2\right) - \zeta_1
      \left(\alpha_1 + (u_{e 1}^{1})^2\right)^2\right]}{\alpha_1+
    (u_{e 1}^{1})^2} \,, \quad
  \mbox{tr}(J_m) &= \frac{ \left[2 (u_{e m}^{1}) \mu_m -
      \left(\alpha_m + (u_{e m}^{1})^2\right) - \zeta_m
      \left(\alpha_m + (u_{e m}^{1})^2\right)^2\right]}{\alpha_m+
    (u_{e m}^{1})^2} \,.
  \end{split}
\end{equation}
\end{subequations}

From Proposition \ref{prop:stab}, discrete eigenvalues $\lambda$
associated with the ring and center cell pattern are roots of
$\mbox{det} {\mathcal M}(\lambda)=0$. A convenient way to implement this
determinant root-finding problem numerically is to use the special
structure of ${\mathcal M}(\lambda)$ in order to determine
explicit formulae for its matrix spectrum
${\mathcal M}(\lambda)\pmb{c}_j =\sigma_j \pmb{c}_j$, for $j=1,\ldots,m$, where
$\sigma_j=\sigma_j(\lambda)$.  Then, we need only numerically solve
the scalar root-finding problems $\sigma_j(\lambda)=0$ for $\lambda$
for each $j=1,\ldots,m$.

To do so, we use the convenient fact that ${\mathcal G}_{\lambda}$ can be
partitioned, similar to that in \eqref{ring:gblock}, as
\begin{equation}\label{ring:eig_gblock}
\mathcal{G}_{\lambda} = 
\left(\begin{array}{ccc:c}
            &  &  & g_{\lambda m}\\ 
		    & \mathcal{G}_{\lambda(m-1)}  &  & \vdots\\
		    &  &  &  g_{\lambda m} \\ \hdashline[2pt/2pt]
		g_{\lambda m} & \dots & g_{\lambda m} & R_{\lambda m}  
		\end{array}
              \right)\,; \qquad 
              g_{\lambda m} \equiv G_{\lambda}(\pmb{x}_j,\pmb{x}_m) =
              G_{\lambda}(\pmb{x}_m, \pmb{x}_j)\,,
              \quad j=1\,\ldots,m-1 \,, \quad R_{\lambda m}\equiv
              R_{\lambda}(\pmb{x}_m) \,,
\end{equation}
where $G_{\lambda}$ is the eigenvalue-dependent reduced-wave Green's
function with regular part $R_{\lambda}$ satisfying
\eqref{EigGreen}. In \eqref{ring:eig_gblock}, the $(m-1)\times (m-1)$
matrix block $\mathcal{G}_{\lambda(m-1)}$, representing cell
interactions on the ring, is symmetric and cyclic. As a result, it has
the well-defined eigenspace
\begin{equation}\label{cen:glam_eig}
  \mathcal{G}_{\lambda(m-1) } \, \pmb{v}_j = \omega_{\lambda j} \,  \pmb{v}_j, \quad
  j =1, \dots,m-1\,; \qquad \pmb{e}^T\pmb{v}_j=0 \,, \quad j=1,\ldots,m-2\,,
  \quad \pmb{v}_{m-1}=\pmb{e}\equiv (1,\ldots,1)^T\in \R^{m-1} \,.
\end{equation}

By using this special matrix structure, it readily follows that the
GCEP matrix ${\mathcal M}(\lambda)$ in \eqref{M_M0} admits $m-2$ {\em
  anti-phase modes}, characterized by eigenvectors of the form
$\pmb{c}_j=(\pmb{v}_j,0)^T\in \R^{m}$ where $\pmb{v}_j\in \R^{m-1}$
are those eigenvectors of ${\mathcal G}_{\lambda(m-1)}$ in
\eqref{cen:glam_eig}, which satisfy $\pmb{e}^T\pmb{v}_j=0$ for
$j=1,\ldots,m-2$. With this choice, \eqref{M_M0} becomes
\begin{equation}\label{AsynEigProb2}
\begin{pmatrix}
2\pi \nu \mathcal{G}_{\lambda(m-1) }\pmb{v}_j + M_0 \pmb{v}_j \\ \\
2\pi \nu g_{\lambda m}\, \pmb{e}^T \pmb{v}_j
\end{pmatrix}  = \begin{pmatrix}
\sigma_j \pmb{v}_j  \\ \\ 0
\end{pmatrix} \,, \quad \mbox{for} \quad j=1,\ldots,m-2 \,.
\end{equation}
Since $\pmb{e}^T\pmb{v}_j=0$ for $j=1,\ldots,m-2$, we obtain from
\eqref{AsynEigProb2}, \eqref{cen:glam_eig} and \eqref{M_matrix1}
that $m-2$ eigenpairs of ${\mathcal M}(\lambda)$ are
\begin{equation}\label{cent:root_asynch}
  \sigma_j = 2 \pi \nu \omega_{\lambda j} + M_0 = 2 \pi \nu \omega_{\lambda j} +
1 + \frac{\nu D}{d_1} + \frac{2\pi\nu D}{\tau} \frac{d_2}{d_1}
\mathit{K}_c\,,  \quad \pmb{c}_j=(\pmb{v}_j,0)^T \,, \quad
\mbox{for} \quad j =1, \dots,m-2\,,
\end{equation}
where $\mathit{K}_c=\mathit{K}_c(\lambda)$ is defined in
\eqref{cent:KcKm_def}. We remark that since $\mathit{K}_c$ depends on
the steady-state values $A_c$ and $A_m$, as obtained from the linear
system \eqref{Linear_reduceSys}, this term depends on the
permeabilities and local kinetics of the center cell. Discrete
eigenvalues $\lambda$ for the anti-phase modes are union of the zeroes
of $\sigma_j(\lambda)=0$ for $j=1,\ldots,m-2$.

For the remaining two eigenpairs of the GCEP matrix
${\mathcal M}(\lambda)$ the associated eigenvector $\pmb{c}$ has the
form $\pmb{c} = (\pmb{e},\gamma)^T$ for some scalar $\gamma$ to be
determined and $\pmb{e}=(1,\ldots,1)^T\in \R^{m-1}$. This eigenvector
is referred to as the {\em in-phase mode} since any instability
associated with this mode has the same phase for the cells on the ring.
With this choice, \eqref{M_M0} reduces to
\begin{equation}\label{SyncEigen}
\begin{pmatrix}
  \omega_{\lambda(m-1)} + {M_0/(2\pi \nu)} & g_{\lambda m} \\
  (m-1) g_{\lambda m} & R_{\lambda m} + {M_m/(2\pi \nu)} 
\end{pmatrix}
\begin{pmatrix}
  1\\
  \gamma
\end{pmatrix} = \frac{\sigma}{2\pi\nu}
\begin{pmatrix}
  1 \\
  \gamma
\end{pmatrix}\,,
\end{equation}
where
${\mathcal G}_{\lambda(m-1)}\pmb{e}=\omega_{\lambda(m-1)}\pmb{e}$ from
\eqref{cen:glam_eig}.

Upon eliminating $\sigma$ from the $2\times 2$ matrix problem
\eqref{SyncEigen} we obtain that $\gamma_{\pm}$ are the roots of
the quadratic equation
\begin{equation}\label{SigmaSyncGamma}
  \gamma^2 + \frac{1}{2\pi\nu g_{\lambda m}} \Big{[} (M_0 - M_m) +
  2\pi \nu (\omega_{\lambda (m-1)}  - R_{\lambda m})  \Big{]} \gamma  - (m-1) = 0\,,
\end{equation}
given by
\begin{equation}\label{Gamma}
\begin{split}
  \gamma_{\pm} = - \frac{\beta_{\lambda}}{2} \pm \frac{1}{2}
  \sqrt{\beta^2_{\lambda} + 4(m-1)} \,, \qquad \mbox{where} \qquad
  \beta_{\lambda} \equiv \frac{1}{2\pi \nu g_{\lambda m}} \Big{[} (M_0 -
  M_m) + 2\pi \nu (\omega_{\lambda (m-1)} - R_{\lambda m}) \Big{]} \,.
\end{split}
\end{equation}
Since $\gamma_{+} \gamma_{-} = -(m-1)>0$, but with $\gamma_{\pm}$
possibly complex-valued, we confirm that the two possible in-phase
modes $\pmb{c}_{\pm} = (1,\dots,1,\gamma_{\pm})^T$ are orthogonal. The
two eigenvalues $\sigma=\sigma_{\pm}(\lambda)$, given by
$\sigma_{\pm} \equiv 2 \pi \nu \big{(} \omega_{\lambda (m-1)} +
\gamma_{\pm} \,g_{\lambda m} \big{)} + M_0$, can be written as
\begin{equation}\label{cent:sync}
  \sigma_{\pm}(\lambda) = \frac{(h_1 + h_2)}{2} \pm \frac{1}{2}
  \sqrt{\left(h_1 - h_2 \right)^2 + 16 \pi^2 \nu^2(m-1) g_{\lambda m}^2
  } \,, \quad \mbox{where} \quad
  h_1 \equiv 2 \pi \nu\, \omega_{\lambda (m-1)} + M_0\,, \quad
  h_2 \equiv 2 \pi \nu\, R_{\lambda m} + M_m\,.
\end{equation}
Here $M_0=M_0(\lambda)$ and $M_m=M_m(\lambda)$ are given in
\eqref{M_matrix1} and $\omega_{\lambda (m-1)} $ is defined by
$\mathcal{G}_{\lambda (m-1)}\pmb{e}=\omega_{\lambda (m-1)} \pmb{e}$
from \eqref{cen:glam_eig}.

Alternatively, rather than solving \eqref{SyncEigen} for
$\sigma_{\pm}$, and then setting $\sigma_{\pm}(\lambda)=0$ by using a
root-finder for $\lambda$, we can more directly conclude that
${\mathcal M}(\lambda)\pmb{c}=\bf{0}$ for $\pmb{c}=(\pmb{e},\gamma)^T$
if and only if the determinant of the $2\times 2$ matrix in
\eqref{SyncEigen} vanishes. In this way, a discrete eigenvalue
$\lambda$ of the GCEP \eqref{full_gcep} for the in-phase mode,
which satisfies $\mbox{det}({\mathcal M}(\lambda))=0$, is a root of
${\mathcal H}_{c}(\lambda)=0$ defined by
\begin{equation}\label{SyncEigen_root}
  {\mathcal H}_{c}(\lambda)= \left( \omega_{\lambda (m-1)} + \frac{M_0}{2\pi\nu}
\right)\left( R_{\lambda m} + \frac{M_m}{2\pi\nu} \right)- (m-1)
g_{\lambda m}^2\,.
\end{equation}
For any root of \eqref{SyncEigen_root}, the
corresponding eigenvector $\pmb{c}$ of ${\mathcal M}(\lambda)$ is
\begin{equation}\label{SyncEigen_vec}
  \pmb{c}=(1,\ldots,1,\gamma)^T \,, \qquad \gamma
  =  - \frac{1}{g_{\lambda m}} \left( \omega_{\lambda(m-1)} +
    \frac{M_0}{2\pi \nu}\right) \,.
\end{equation}
It is readily verified using \eqref{cent:sync} that if
$\lambda^{\star}$ satisfies either $\sigma_{+}(\lambda^{\star})=0$ or
$\sigma_{-}(\lambda^{\star})=0$, then we must have
${\mathcal H}_c(\lambda^{\star})=0$. In contrast, if $\lambda^{\star}$
satisfies ${\mathcal H}_c(\lambda^{\star})=0$, then we can only
conclude that {\em either} $\sigma_{+}(\lambda^{\star})=0$ or
$\sigma_{-}(\lambda^{\star})=0$.  Therefore, in implementing a
root-finding strategy based on the single scalar equation
\eqref{SyncEigen_root} instead of the two scalar equations
\eqref{cent:sync} care must be taken to identify all possible
roots of ${\mathcal H}_{c}(\lambda)=0$ for the same parameter set.

We summarize our result for eigenvalues $\lambda$ of the GCEP
\eqref{full_gcep} for a ring and center cell pattern as follows:

\vspace*{0.2cm}
\begin{prop}\label{prop:ring+center} Consider a ring and center hole pattern
  of $m\geq 3$ cells in the unit disk with cell centers at
  \eqref{cent:xj}.  The set $\Lambda({\mathcal M})$ as obtained from
  the GCEP \eqref{full_gcep}, and which approximates as $\varepsilon\to 0$
  all the discrete eigenvalues of the linearization of the PDE-ODE
  system \eqref{DimLess_bulk} around the steady-state solution, is
\begin{equation}\label{cent:spec_all}
  \Lambda({\mathcal M}) \equiv \Big{\lbrace}\lambda \,\, \big{\vert} \,\,
  \bigcup\limits_{j=1}^{m-2} \lbrace{\, \sigma_{j}(\lambda)=0\, \rbrace}\,,
  \,\, \bigcup \,
  \lbrace{\sigma_{\pm}(\lambda)=0\rbrace} \Big{\rbrace}\,.
\end{equation}
Here $\sigma_j(\lambda)$, for $j=1,\ldots,m-2$, for the anti-phase
modes are defined in \eqref{cent:root_asynch}, while
$\sigma_{\pm}(\lambda)$ for the in-phase modes are defined in
\eqref{cent:sync}. As shown in Remark \ref{app:degen} of Appendix
\ref{app:green}, due to mode degeneracy of
${\mathcal G}_{\lambda (m-1)}$, there are ${(m-1)/2}$ distinct
anti-phase modes if $m$ is odd and ${(m-2)/2}$ distinct anti-phase
modes if $m$ is even.
\end{prop}  

For the unit disk, where an infinite series representation of the
solution to \eqref{EigGreen} is available, explicit formulae for the
eigenvalues $\omega_{\lambda j}$ of ${\mathcal G}_{\lambda(m-1)}$, as
needed in \eqref{cent:root_asynch} and \eqref{cent:sync}, are given in
Appendix \ref{app:green}. The result in Remark \ref{app:degen} of
Appendix \ref{app:green} regarding mode degeneracy results from the
fact that ${\mathcal G}_{\lambda(m-1)}$ is both symmetric and
cyclic. In Appendix \ref{app:green} we also show how to readily
calculate the quantities in \eqref{ring:gblock} and
\eqref{cent:g_eig}, which are needed in \eqref{Linear_reduceSys} for
determining the steady-state.

As indicated by Proposition \ref{prop:stab}, stability boundaries in
the $\tau$ versus $D$ parameter space for the steady-state under
Sel'kov kinetics are determined by HB boundaries where
$\lambda=i\lambda_{I}\in \Lambda({\mathcal M})$. To determine Hopf
bifurcation boundaries for the anti-phase modes we set
$\mbox{Re}(\sigma_{j}(i\lambda_I))=0$ and
$\mbox{Im}(\sigma_{j}(i\lambda_I))=0$ for $j=1,\ldots,m-2$ in
\eqref{cent:root_asynch}. Upon separating real and imaginary parts in
\eqref{cent:KcKm_def} we obtain the following nonlinear algebraic
system for each $j=1,\ldots,m-2$:
\begin{subequations}\label{cent:HB_async}
\begin{equation}\label{SigmaAsync2_1} 
  2 \pi  \nu  \, \mbox{Re}(\omega_{\lambda j}) +
   \left( 1 + \nu \frac{D}{d_1} \right) + \frac{2\pi \nu D}{\tau}
   \frac{d_2}{d_1} \, \mbox{Re}(\mathit{K}_{c}(i\lambda_I)) =0 \,, \qquad
   \mbox{Im}(\omega_{\lambda j}) + \frac{D}{\tau} \frac{d_2}{d_1} \,
   \mbox{Im}(\mathit{K}_{c}(i\lambda_I)) = 0 \,,
\end{equation}
where
\begin{equation}\label{kcrealimag}
  \mbox{Im}(\mathit{K}_{c}(i\lambda_I)) = \frac{\lambda_I (\mbox{det}(J_1)-\lambda_I^2)
    \, + \lambda_I  \mbox{det}(J_1)\mbox{tr}(J_1)}
  {(\mbox{det}(J_1) - \lambda_I^2)^2 + (\lambda_I\,
    \mbox{tr}(J_1))^2} \,, \quad
  \mbox{Re}(\mathit{K}_{c}(i\lambda_I)) = \frac{\mbox{det}(J_1)(\mbox{det}(J_1) -
    \lambda_I^2) -\lambda_I^2 \, \mbox{tr}(J_1)}
    {(\mbox{det}(J_1) - \lambda_I^2)^2 + (\lambda_I\, \mbox{tr}(J_1))^2}\,.
\end{equation}
\end{subequations}
Here $J_1$ is the Jacobian of the Sel'kov kinetics for the ring cells with
determinant and trace given in \eqref{cent:KcKm_jac}. Similarly, the Hopf
bifurcation boundaries for the in-phase modes are obtained by
setting $\sigma_{\pm}(i\lambda_I)=0$ in \eqref{cent:sync}, which yields
the nonlinear algebraic system
\begin{equation}\label{cent:HB_sync}
  \mbox{Re}(\sigma_{\pm}(i\lambda_I)) = 0 \,, \qquad
  \mbox{Im}(\sigma_{\pm}(i\lambda_I)) = 0 \,,
\end{equation}
or equivalently ${\mathcal H}_c(i\lambda_I)=0$ from \eqref{SyncEigen_root}.

\subsection{Example: Two cells on a ring with a center cell}\label{Subsec:SelKov_Exaample}

We now apply the theory developed in \cref{RingCenterHole} to a
population of $m=3$ cells, where two of the cells are equally spaced
on a concentric ring of radius $r_0$ within the unit disk, with the
remaining one centered at the origin (see Fig.~\ref{def_Mcells}). For this
configuration, the eigenvalues of the $3\times 3$ GCEP matrix
$\mathcal{M}(\lambda)$ are given in \eqref{cent:root_asynch} and
\eqref{cent:sync} for a single anti-phase mode and the two in-phase
modes, respectively. To compute the HB boundaries in the $\tau$ versus
$D$ parameter plane for these modes, we solve \eqref{cent:HB_async} and
\eqref{cent:HB_sync} numerically by implementing the psuedo-arclength
continuation algorithm TEST\_CON (cf.~\cite{TestCon}) with respect to
$D$, while using Newton's method to compute $\tau$ and $\lambda_I$
at each point on the solution path.  Such a continuation scheme in $D$
is needed owing to the possibility of fold points along the HB boundary.

To determine regions of instability in open sets of the $\tau$ versus
$D$ parameter plane we use \eqref{cent:spec_all} of Proposition
\ref{prop:ring+center}, together with the winding number criterion of
complex analysis, to identify the number ${\mathcal N}$ of eigenvalues
$\lambda\in \Lambda({\mathcal M})$ with $\mbox{Re}(\lambda)>0$. To do
so, we first define
$\mathcal{F}(\lambda) \equiv \det(\mathcal{M}(\lambda))$, where
${\mathcal M}(\lambda)$ is the GCEP matrix in \eqref{M_M0}.  Provided
that there are no zeroes or poles on the imaginary axis, $\mathcal{N}$
is the number of zeroes of $\mathcal{F}(\lambda) = 0$ in
$\mbox{Re}(\lambda)> 0$, which from the argument principle is
\begin{equation}\label{ArgMent_Principle}
  \mathcal{N} = \frac{1}{2\pi} \big[ \text{arg} \,
  \mathcal{F}(\lambda) \big]_{\Gamma} + {\mathcal P}\,.
\end{equation}
Here ${\mathcal P}$ is the number of poles of $\mathcal{F}(\lambda)$
in $\mbox{Re}(\lambda) > 0$, while
$\big[\text{arg}\,\mathcal{F}(\lambda)\big]_{\Gamma}$ denotes the
change in the argument of $\mathcal{F}(\lambda)$ over the closed,
counter-clockwise oriented contour $\Gamma$. This contour $\Gamma$ is
the limit as $\mathcal{R} \to \infty$ of the union of the imaginary
axis $\Gamma_{I} = i\lambda_I$, for $|\lambda_I|\leq \mathcal{R}$, and
the semi-circle $\Gamma_{\mathcal{R}}$, defined by
$|\lambda| = \mathcal{R}$ with $|\text{arg}(\lambda)| \leq \pi/2$. To
count the number of poles of $\mathcal{F}(\lambda)$ in
$\mbox{Re}(\lambda) > 0$, we must examine the analyticity properties
of the GCEP matrix $\mathcal{M}(\lambda)$ in \eqref{M_M0}. Since the
entries of the Green's matrix $\mathcal{G}_{\lambda}$ are analytic in
$\mbox{Re}(\lambda) > 0$, any singularity of $\mathcal{F}(\lambda)$
must arise from the diagonal matrix
${\mathcal K}(\lambda)\equiv
\mbox{diag}(\mathit{K}_c,\mathit{K}_c,\mathit{K}_m)$ of
\eqref{M_matrix1}, which is given explicitly in \eqref{cent:KcKm} in
terms of the Jacobians $J_1$ and $J_m$ of the Sel'kov kinetics for the
identical ring cells and the center cell, respectively. Since
$\mbox{det}(J_1)>0$, \eqref{cent:KcKm_def} yields that $\mathit{K}_c$ has a
complex conjugate pair of poles in $\mbox{Re}(\lambda)>0$ only if
$\mbox{tr}(J_1)>0$. Since this term involves two rows of
${\mathcal M}(\lambda)$, ${\mathcal P}$ must be incremented by four
whenever $\mbox{tr}(J_1)>0$. Similarly, since $\mbox{det}(J_m)>0$,
${\mathcal P}$ is increased by two when $\mbox{tr}(J_m)>0$ for the
center cell.

With ${\mathcal P}$ determined in this way, we numerically compute the
number of unstable eigenvalues $\mathcal{N}$ of the linearization of
the steady-state by evaluating \eqref{ArgMent_Principle} for each
point $(D,\tau)$ in the $\tau$ versus $D$ plane. For each such point,
we numerically construct the closed contour $\Gamma$ for some value
$\mathcal{R}$ (chosen so that $\Gamma$ encloses all the poles of
$\mathcal{F}$) in the complex $\lambda$-plane.  As the closed curve
$\Gamma$ is traversed in a counter-clockwise direction, the closed
image curve $\mathcal{F}(\lambda) = \mathcal{F}_R + i\mathcal{F}_I $
is evaluated numerically in the complex $\mathcal{F}$-plane.  The
winding number, denoting the number of times $\mathcal{F}$
encloses/winds around the origin, is computed numerically from the
algorithm of \cite{alciatore1995winding}, and this is used to
calculate $\big[ \text{arg} \, \mathcal{F}(\lambda) \big]_{\Gamma}$.
If the orientation of
${\mathcal F}(\lambda)={\mathcal F}_R+i{\mathcal F}_I$ around the
origin is in the counter-clockwise direction, then
$\big[ \text{arg} \, \mathcal{F}(\lambda) \big]_{\Gamma}$ is positive;
otherwise, it is negative. In our computations, we chose
$\mathcal{R} = 1.5$ (since any pole of $\mathcal{F}$ is close to the
imaginary axis of the $\lambda$-plane), and we discretized the closed
contour $\Gamma = \Gamma_{I} \cup \Gamma_{\mathcal{R}} $ into
subintervals, with $\Gamma_{I}$ having 800 subintervals while the
semi-circle $\Gamma_{{\mathcal R}}$ had $50$ subintervals. In the
function evaluation, the identity
${\mathcal F}(\overline{\lambda})=\overline{{\mathcal F}(\lambda)}$
was used to halve the computational effort.

\begin{figure}[!ht]
  \centering
    \centering
    \begin{subfigure}[b]{0.40\textwidth}
      \includegraphics[width=\textwidth,height=4.8cm]{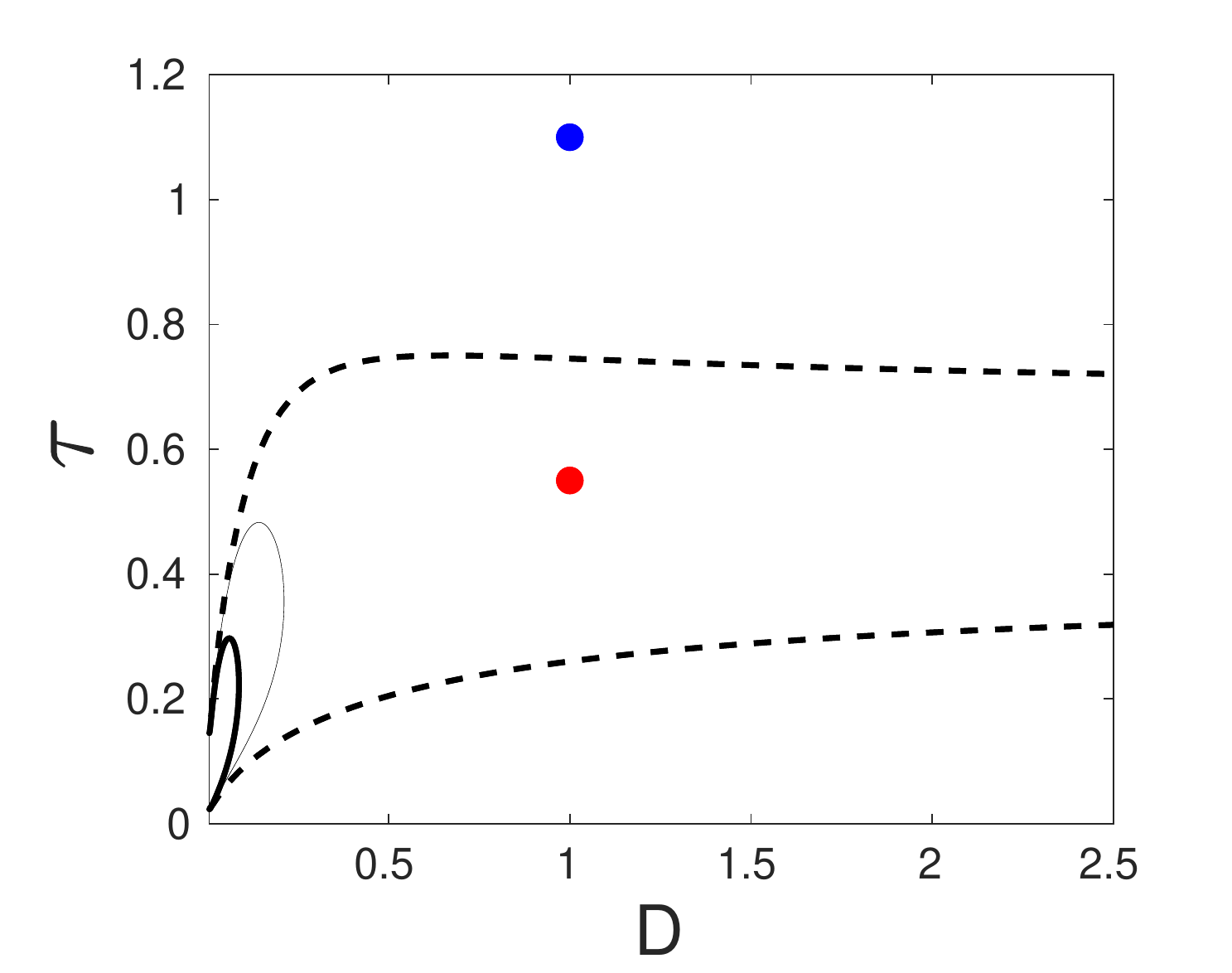}
        \caption{HB boundaries: Identical cells}
          \label{Bifur_3cells_ID}
    \end{subfigure}
    \begin{subfigure}[b]{0.40\textwidth}  
      \includegraphics[width=\textwidth,height=4.8cm]{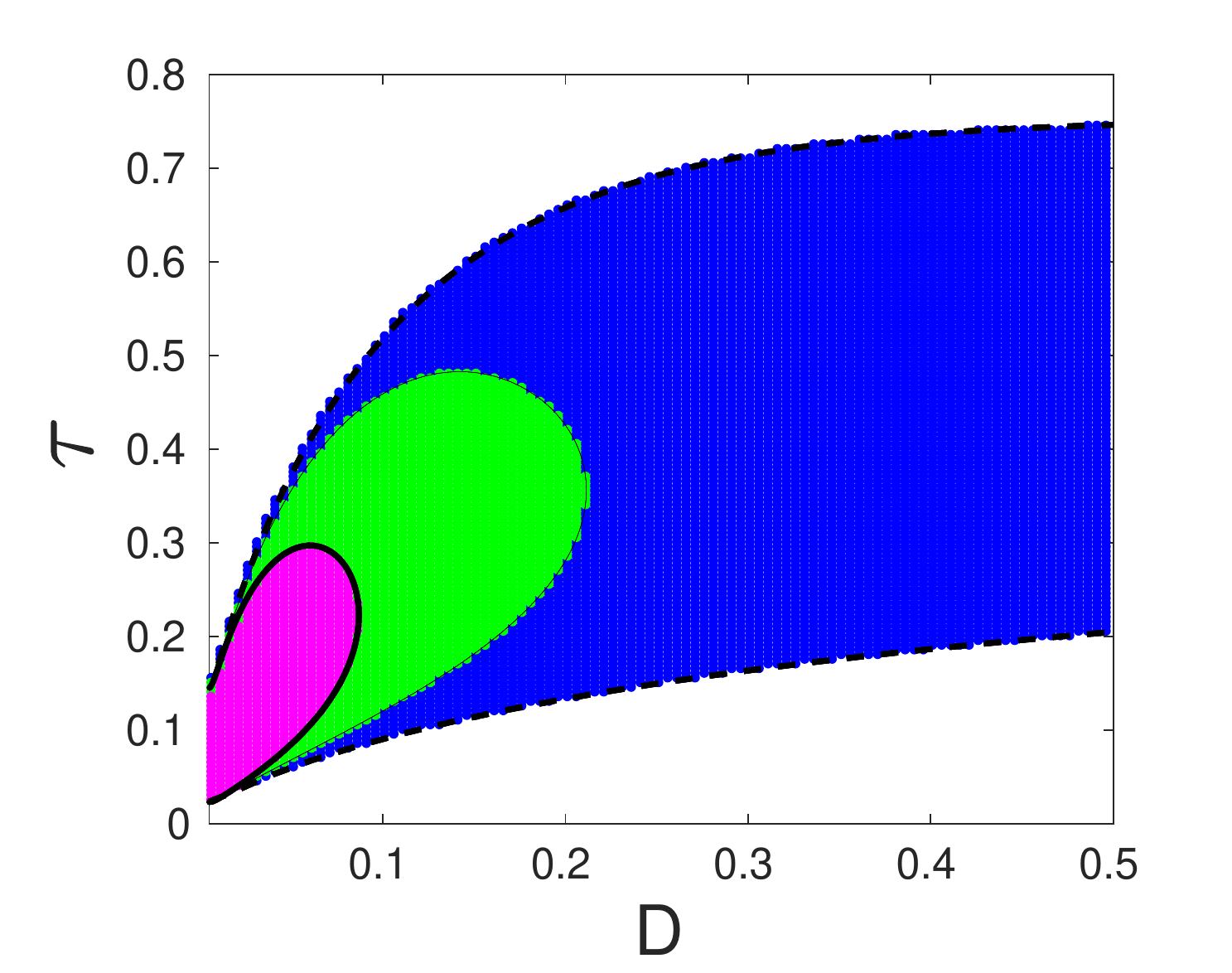}
      \caption{Number of unstable eigenvalues of the GCEP}
        \label{Bifur_Wind_RCH_ID_Zoomed}
    \end{subfigure}
  \vspace*{-1ex}
  \caption{Left panel: HB boundaries in the $\tau$ versus $D$ plane
    for a ring and center hole pattern of $m=3$ identical cells with
    ring radius $r_0=0.5$, parameters as in \eqref{Selkov_para} and
    permeabilities $d_1=0.8$ and $d_2=0.2$. The dashed and heavy solid
    curves are for the in-phase modes computed from
    \eqref{cent:HB_sync} with $(+)$ and $(-)$, respectively. The thin
    solid curve is for the anti-phase mode computed from
    \eqref{cent:HB_async}. Each mode is linearly unstable within its
    respective lobe. Linearly stable steady-state solutions exist
    outside the union of the lobes. Full PDE simulations of
    \eqref{DimLess_bulk} are shown in Figs.~\ref{FlexPDE_ID3cells} and
    \ref{FlexPDE_ID3cellsStable} at the red and blue dots,
    respectively.  Right panel: The regions of instability computed
    from the winding number \eqref{ArgMent_Principle}. Blue region:
    in-phase `+' mode is unstable with 2 roots of
    $\mathcal{F}(\lambda) = \det(\mathcal{M}(\lambda)) = 0$ in
    $\mbox{Re}(\lambda) > 0$. Green region: anti-phase mode is also
    unstable, yielding 4 roots. Magenta region: all modes are
    unstable, and there are 6 roots. Plot on right is a zoom of the
    one on the left. The HB boundaries in the left panel are
    superimposed in this figure.}
  \label{Bifur_3cells_first}
\end{figure}

\subsection{A defective center cell: different permeabilities}\label{sec:def_center_permea}

In our results below, except when otherwise stated, the Sel'kov
parameters $\alpha$, $\mu$ and $\zeta$, and permeabilities $d_1$ and
$d_2$ for the identical cells on the ring, and the common cell radius
$\varepsilon$ are
\begin{equation}\label{Selkov_para}
  \alpha = 0.9\,, \qquad \mu = 2\,, \qquad \zeta = 0.15\,,\qquad
  d_1 = 0.8\,, \qquad d_2=0.2\,,\qquad \varepsilon = 0.05\,.
\end{equation}
From Fig.~\ref{fig:selkov}, we conclude that each cell, when isolated,
has no intracellular oscillations. The permeabilities $d_{13}$ and
$d_{23}$ for the center cell will be stated in the figure captions
below.

Fig.~\ref{Bifur_3cells_ID} shows the computed HB boundaries in the
$\tau$ versus $D$ plane for a ring of radius $r_0 = 0.5$ when the
cells are all identical. We observe that one of the in-phase lobes is
open/unbounded, which predicts the existence of intracellular
oscillations even for large $D$. In
Fig.~\ref{Bifur_Wind_RCH_ID_Zoomed}, we show the corresponding regions
of instability in the $\tau$ versus $D$ parameter plane. In generating
this figure, we pixelated the $\tau$ versus $D$ plane with the
uniform spacing $\Delta \tau = \Delta D = 0.005$, and at each discrete
point $(D, \tau)$ used our winding number algorithm to count the
number of roots $\mathcal{N}$ of $\mbox{det}{\mathcal M}(\lambda)=0$
in $\mbox{Re}(\lambda)>0$. In Fig.~\ref{Bifur_Wind_RCH_ID_Zoomed},
each point in the blue-shaded region has two unstable eigenvalues for
the GCEP, and they correspond to the in-phase `+' mode. The
green-shaded region contains four unstable eigenvalues, two of which
are for the anti-phase mode while the other two for the in-phase `+'
mode. Finally, in the magenta-shaded region there are six unstable
eigenvalues of the linearization, with two such eigenvalues associated
with each of the three possible modes of instability (in-phase $\pm$
and anti-phase).  The HB boundaries in the left panel of
Fig.~\ref{Bifur_3cells_first} are superimposed on these instability
regions.

The $\mathcal{F}(\lambda)={\mathcal F}_R+i{\mathcal F}_I$ curve in the
complex $\mathcal{F}$-plane is shown in Fig.~\ref{WindingF} for a
specific point in each of the three instability regions in
Fig.~\ref{Bifur_Wind_RCH_ID_Zoomed}. These plots show how
${\mathcal F}(\lambda)$ winds around the origin
$({\mathcal F}_R,{\mathcal F}_I)=(0,0)$ (shown with a green dot) as
$\lambda$ traverses $\Gamma$ in the counterclockwise direction. For
the point $(D,\tau) = (0.05,0.15)$ in the magenta-shaded region, we
observe from the left panel of Fig.~\ref{WindingF} that
$\big[ \text{arg} \, \mathcal{F}(\lambda) \big]_{\Gamma} = 0$. At this
point, ${\mathcal F}(\lambda)$ has two poles, one of order four and
the other of order 2, so that ${\mathcal P}=6$. As such,
\eqref{ArgMent_Principle} yields that there are 6 roots (counting
multiplicity) to $\mathcal{F}(\lambda) = 0$ in
$\mbox{Re}(\lambda) > 0$. At the point $(D,\tau) = (0.15,0.35)$ in the
anti-phase mode instability region (green-shaded region in
Fig.~\ref{Bifur_Wind_RCH_ID_Zoomed}), ${\mathcal F}(\lambda)$ winds
round the origin twice in the clockwise direction as shown in the
middle panel of Fig.~\ref{WindingF}, which yields
$\big[ \text{arg} \, \mathcal{F}(\lambda) \big]_{\Gamma} =
-4\pi$. Since ${\mathcal P}=6$ at this point,
\eqref{ArgMent_Principle} yields that $\mathcal{F}(\lambda)=0$ has
four roots (counting multiplicity) in $\mbox{Re}(\lambda) > 0$.  In
the right panel of Fig.~\ref{WindingF}, we present a similar result
for the point $(D,\tau) = (0.4,0.5)$ in the blue-shaded region in
Fig.~\ref{Bifur_Wind_RCH_ID_Zoomed}. At this point, we calculate
$\big[ \text{arg} \, \mathcal{F}(\lambda) \big]_{\Gamma} = -8\pi$ and
that ${\mathcal F}$ has a pole of order four and a pole of order two
in $\mbox{Re}(\lambda)>0$. As such, \eqref{ArgMent_Principle} yields
that $\mathcal{F}(\lambda)=0$ has two roots in
$\mbox{Re}(\lambda) > 0$.

\begin{figure}[!ht]
  \centering
  \makebox{
    \raisebox{0.5ex}{}
 \includegraphics[width=0.30 \textwidth,height=4.5cm]{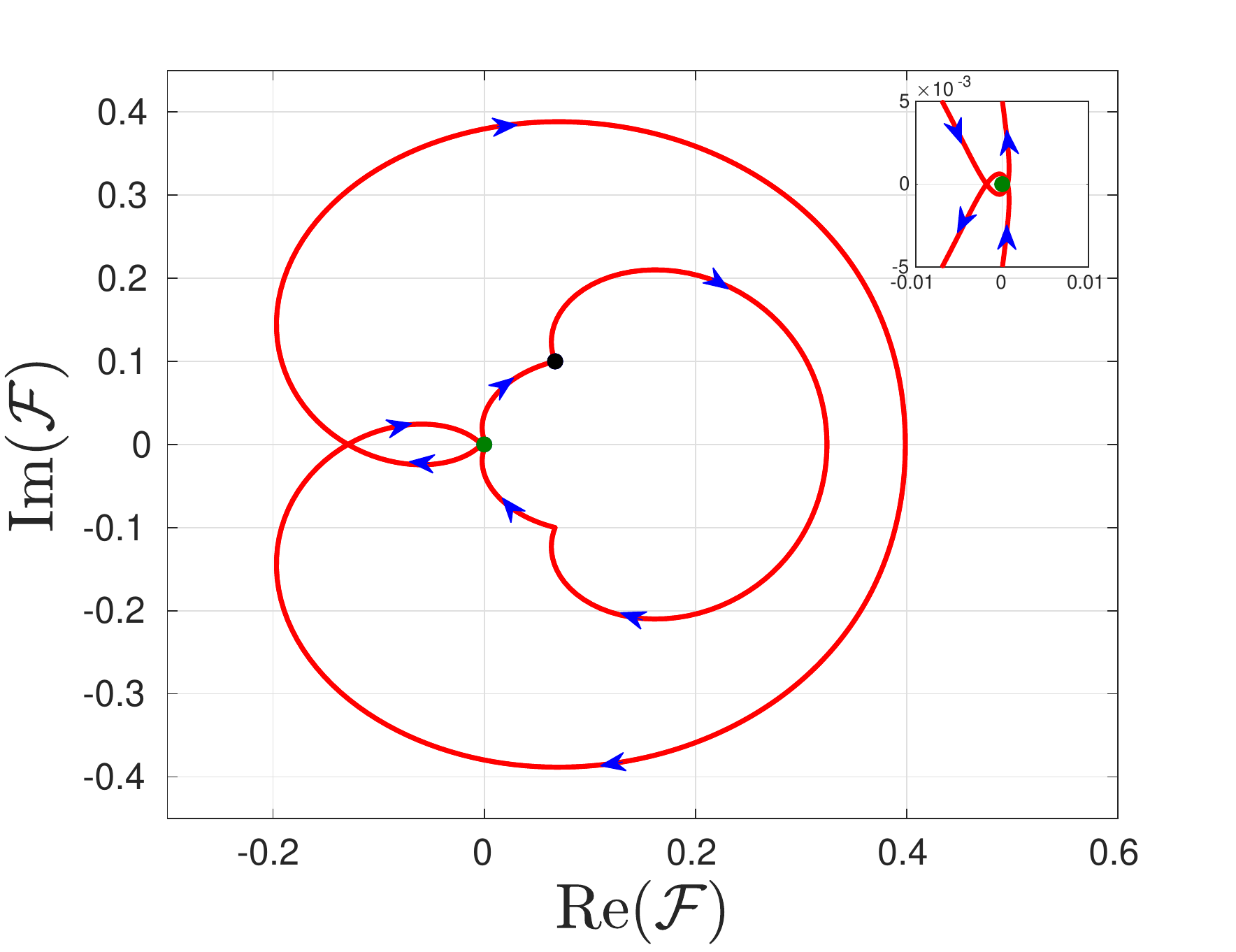}
    \phantomsubcaption
    \label{WindingFA}
  }  
  \makebox{
    \raisebox{0.5ex}{}
\includegraphics[width=0.30 \textwidth,height=4.5cm]{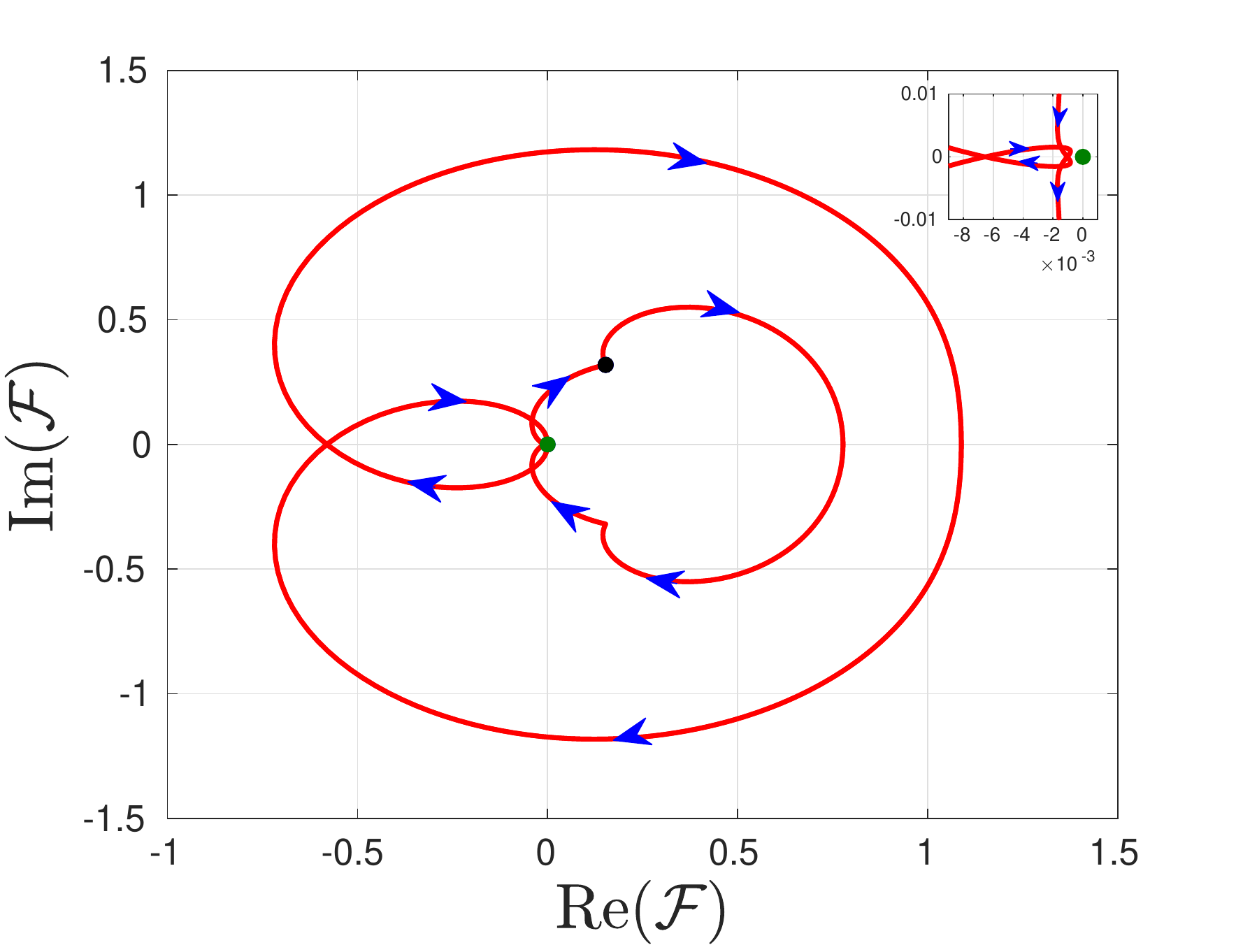}
    \phantomsubcaption
    \label{WindingFB}
  }
  \makebox{
    \raisebox{0.5ex}{}
    \includegraphics[width=0.30\textwidth,height=4.5cm]{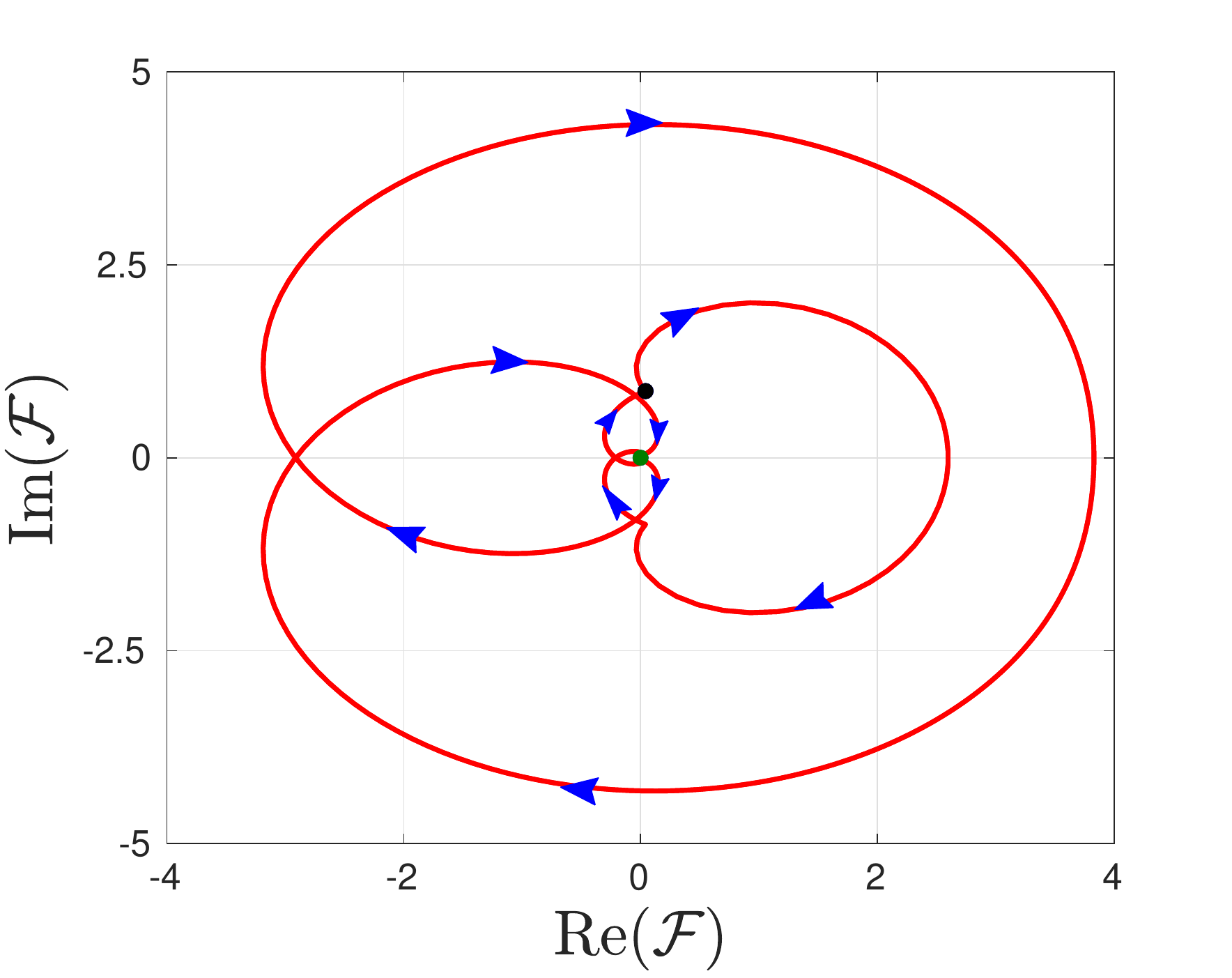}
    \phantomsubcaption
    \label{WindingFC}
  }
  \vspace*{-1ex}
  \caption{The closed curve
      $\mathcal{F(\lambda)} = \det(\mathcal{M}(\lambda))$ plotted for
      a specific points in each of the three regions of instability in
      Fig.~\ref{Bifur_Wind_RCH_ID_Zoomed}. The closed contour
      $\Gamma = \Gamma_{I} \cup \Gamma_{{\mathcal R}} $ in the
      $\lambda$-plane was constructed with $\mathcal{R} = 1.5$, and
      discretized with $800$ subintervals on $\Gamma_{I}$ and $50$
      subintervals on $\Gamma_{{\mathcal R}}$. The green dot
      represents the origin $({\mathcal F}_R,{\mathcal F}_I)=(0,0)$,
      the black dot indicates the starting point of the curve, and the
      blue arrows show the direction of the curve. The inserts detail
      the behavior near the origin.  Left panel: for
      $(D,\tau) = (0.05,0.15)$ in the magenta-shaded region of
      Fig.~\ref{Bifur_Wind_RCH_ID_Zoomed}, we have
      $\big[ \text{arg} \, \mathcal{F}(\lambda) \big]_{\Gamma} = 0$.
      Middle panel: for $(D,\tau) = (0.15,0.35)$ in the green-shaded
      region in Fig.~\ref{Bifur_Wind_RCH_ID_Zoomed}, we have
      $\big[ \text{arg} \, \mathcal{F}(\lambda) \big]_{\Gamma} =
      -4\pi$. Right panel: for $(D,\tau) = (0.4,0.5)$ in the
      blue-shaded region of Fig.~\ref{Bifur_Wind_RCH_ID_Zoomed}, we
      have
      $\big[ \text{arg} \, \mathcal{F}(\lambda) \big]_{\Gamma} =
      -8\pi$.}
  \label{WindingF}
\end{figure}

The real and imaginary parts of the normalized eigenvector $\pmb{c}$
of the GCEP matrix ${\mathcal M}$ in \eqref{M_M0} is given in
Table~\ref{Table:ID_RCH_3cells} for selected points on the HB
boundaries in Fig.~\ref{Bifur_3cells_ID}. Recall from
\eqref{nstabform:c} that the magnitude of the components of the
eigenvector $\pmb{c}$ measure the diffusive flux at the boundary of
each cell, while $\tilde{\pmb{c}}\equiv {\mathcal K}\pmb{c}$ predicts
the relative amplitude and phase shifts of the intracellular
oscillations within the cells at the Hopf bifurcation point. For our
ring and center-cell pattern
${\mathcal K}=\mbox{diag}\left(\mathit{K}_c,\mathit{K}_c,\mathit{K}_m
\right)$, where $\mathit{K}_c$ and $\mathit{K}_m$ are given in
\eqref{cent:KcKm_def}.

\begin{table}[!httbp]
\centering
  \begin{tabular}{ c | c | c | c | c | c } \hline 
\rowcolor{LightCyan}
mode &$(D,\tau)$  &  j  &  $\Big( \mbox{Re}(c_j),\mbox{Im}(c_j) \Big)$ & $\theta_j\, (\text{rad})$ & $ \Big( \mbox{Re}(\tilde{c}_j),\mbox{Im}(\tilde{c}_j) \Big)$\\  \hline \hline    \rowcolor{Cyan}
    &         & 1 & $(0.579,0)$  & $0$  & $(-0.569,0.067)$ \\ \rowcolor{Cyan}
    In-phase ($+$) & $ (1.021,0.262) $ & 2 & $(0.579,0)$ & $0$  & $(-0.569,0.067)$ \\
    \rowcolor{Cyan}
 (dashed curve)   &   & 3 & $(0.575,0.0144)$  & $0.0251$ & $(-0.582,0.0625)$ \\ 
\hline  \hline \rowcolor{Gray}
     &        & 1 & $(-0.412,-0.004)$  & $3.15$ & $(0.395,-0.00429)$ \\ \rowcolor{Gray}
    In-phase ($ - $) & $ (0.0857,0.199) $ & 2 & $(-0.412,-0.004)$  & $3.15$ &
    $(0.395,-0.00429)$ \\ \rowcolor{Gray}
 (heavy solid)   &      & 3 & $(0.813,0)$   & $0$  & $(-0.829,0.0239)$ \\ 
\hline \hline  \rowcolor{Cyan}
     &        & 1 & $(0.707,0)$    & $0$  & $(-0.706,0.042)$ \\ \rowcolor{Cyan}
Anti-phase  &$ (0.211,0.365) $ & 2 & $(-0.707,0)$   & $\pi$ & $(0.706,-0.042)$ \\ \rowcolor{Cyan}
 (thin solid)   &         & 3 & $(0,0)$   & $0$  & $(0,0)$ \\ 
  \hline
\end{tabular}
\caption{Real and imaginary parts of the eigenvector $\pmb{c}$ of the
  GCEP matrix ${\mathcal M}(\lambda)$ in \eqref{M_M0}, together with
  $\tilde{\pmb{c}}\equiv {\mathcal K}\pmb{c}$, as computed for a few points on
  the HB boundaries in Fig.~\ref{Bifur_3cells_ID} for three identical
  cells.  The second to last column shows the phase shifts measured in
  terms of the angle each component of the vector $\pmb{c}$ makes with
  the positive real axis in anticlockwise direction.}
\label{Table:ID_RCH_3cells}
\end{table}

In Fig.~\ref{FlexPDE_ID3cells} we show full numerical simulations of
the coupled PDE-ODE model \eqref{DimLess_bulk} obtained using the
commercial PDE software package FlexPDE \cite{flexpde2015solutions}
for $\tau = 0.55$ and $D=1$, which corresponds to the red dot in the
phase diagram of Fig.~\ref{Bifur_3cells_ID}. We observe from the
results in this figure that the intracellular dynamics of the cells
are synchronized with a very slight phase shift, which agrees
  with the prediction by the eigenvector ${\mathcal K}\pmb{c}$ in the
  first three rows of Table~\ref{Table:ID_RCH_3cells} from the
  linearized theory.  Although the cells have identical parameters,
the center and ring cells have slightly different dynamics owing to
the fact that the full Green's matrix is not cyclic for a ring and
center cell pattern. Our numerical computations of
$\det{\mathcal M}(\lambda)=0$ for $(D,\tau)=(1,0.55)$ using the GCEP
matrix in \eqref{M_M0} yields that $\mbox{Re}(\lambda)\approx 0.0143$,
$\mbox{Im}(\lambda)\approx 0.762$,
$\mbox{Re}(\pmb{c})\approx (0.588,0.588,0.556)$ and
$\mbox{Im}(\pmb{c})\approx (0,0,0.0178)$. Observe that the eigenvector
is rather close to that on the nearby point on the HB boundary, as
given in the first three rows of Table \ref{Table:ID_RCH_3cells}. The
prediction from linearized theory is that the period of oscillations
is approximately ${2\pi/\mbox{Im}(\lambda)}\approx 8.25$, which is
rather close to the period observed in the full PDE simulations of
Fig.~\ref{FlexPDE_ID3cells}. A similar full numerical result is
presented in Fig.~\ref{FlexPDE_ID3cellsStable} for $\tau = 1.1$ and
$D = 1$, corresponding to the blue dot in
Fig.~\ref{Bifur_3cells_ID}. At this pair $(D,\tau)$, our phase diagram
predicts no intracellular oscillations. This is confirmed from the
full numerical results shown in Fig.~\ref{FlexPDE_ID3cellsStable}.

\begin{figure}[htbp]
  \centering
	\makebox{
    \raisebox{0.5ex}{}
    \includegraphics[width=0.30\textwidth,height=4.5cm]{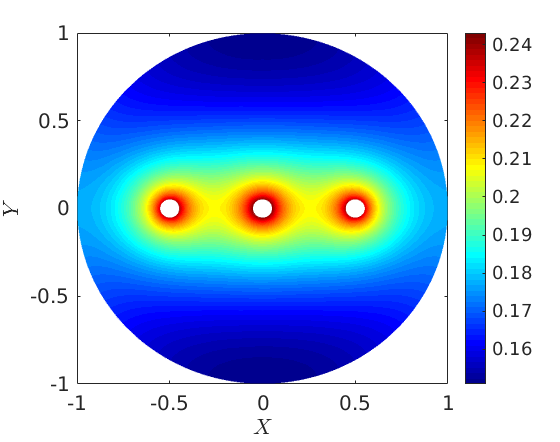}
    \phantomsubcaption
    \label{FlexPDE_ID3cells_surf}
  }  
  \makebox{
    \raisebox{0.5ex}{}
    \includegraphics[width=0.30\textwidth,height=4.5cm]{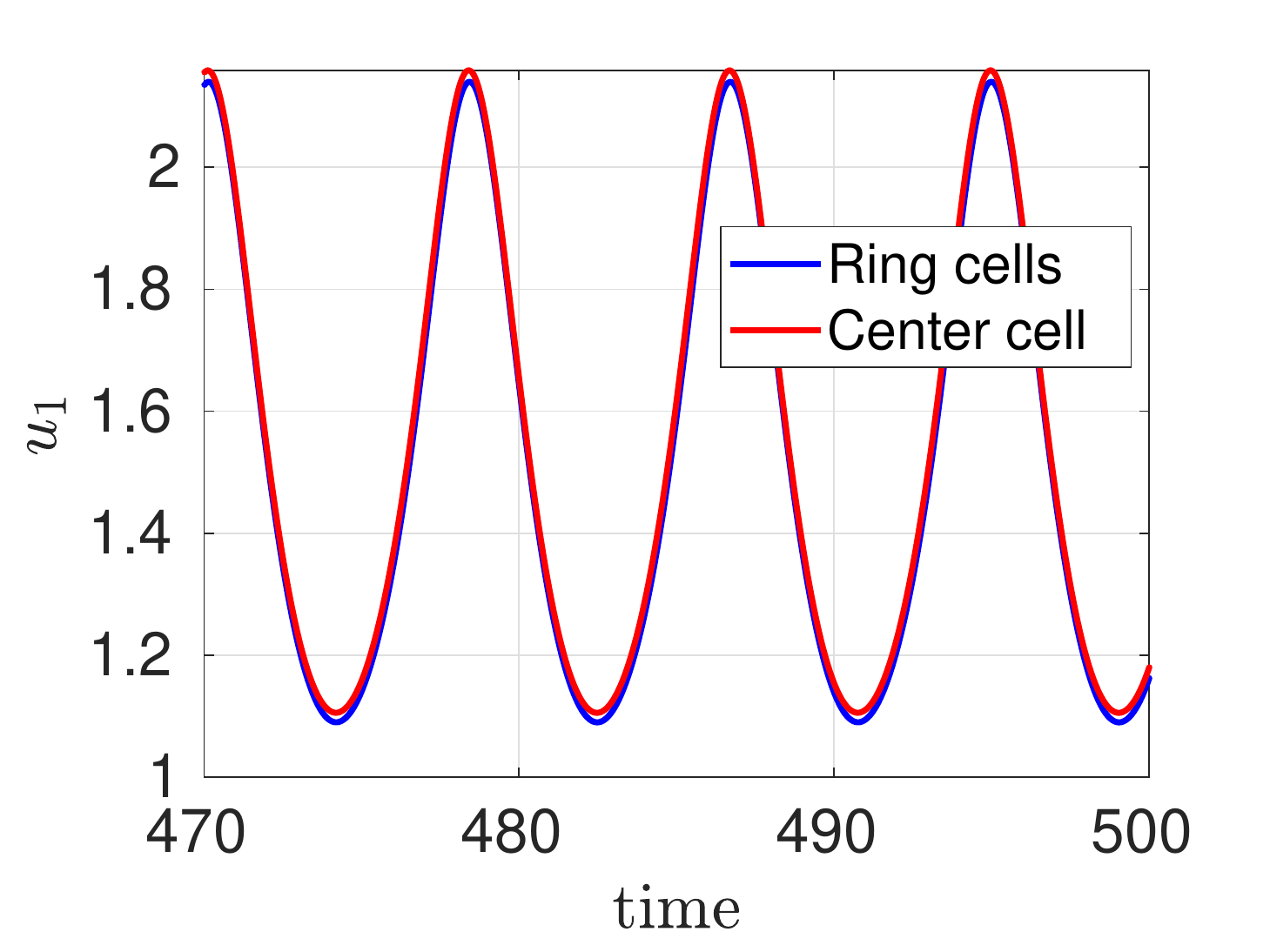}
    \phantomsubcaption
    \label{FlexPDE_ID3cells_bulk}
  }
  \makebox{
    \raisebox{0.5ex}{}
    \includegraphics[width=0.30\textwidth,height=4.5cm]{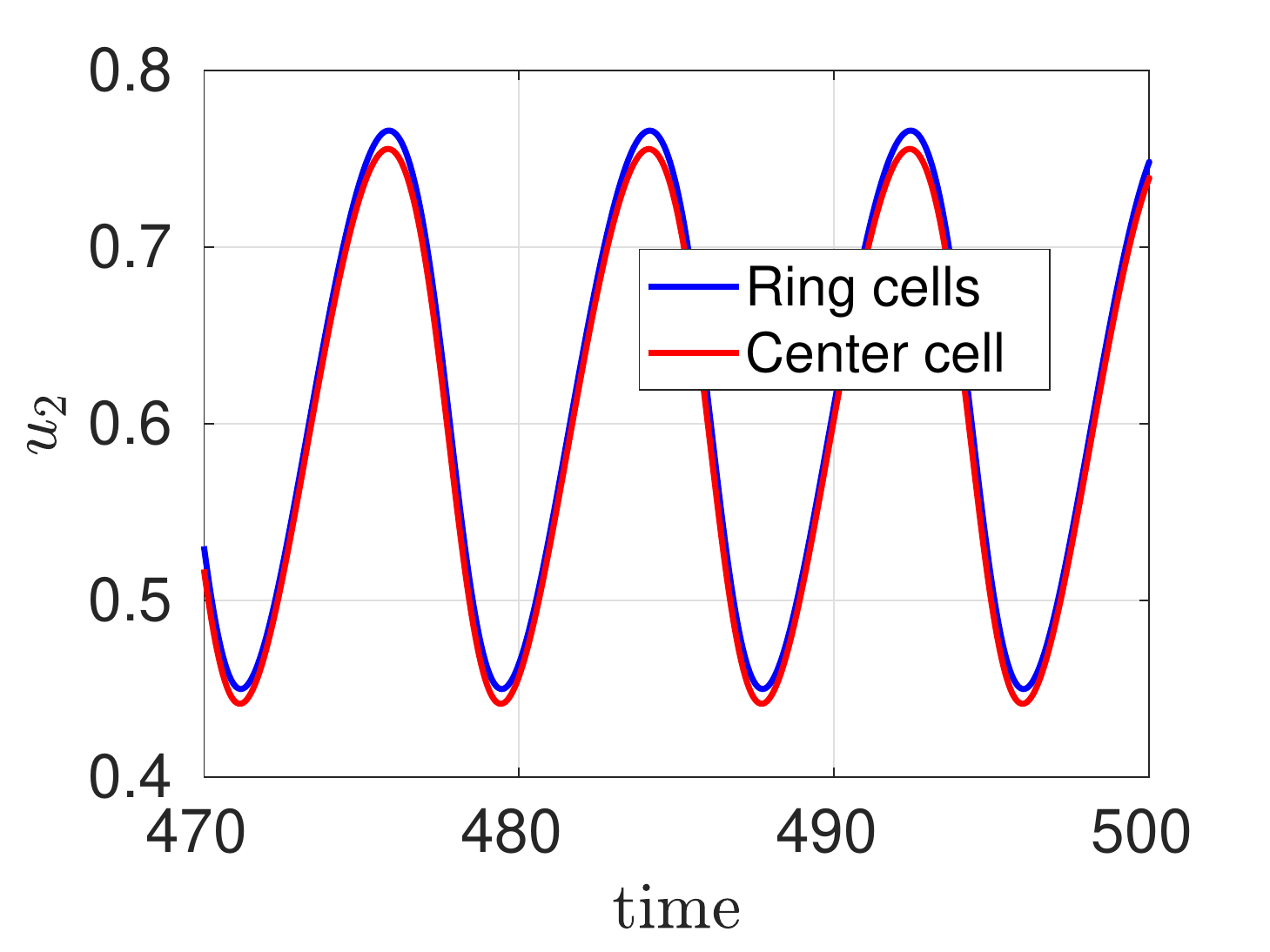}
    \phantomsubcaption
    \label{FlexPDE_ID3cells_u1u2}
  }
  \vspace*{-1ex}
  \caption{Full PDE simulations of \eqref{DimLess_bulk}, computed
    with FlexPDE \cite{flexpde2015solutions}, for $\tau =0.55 $ and
    $D = 1$ for three identical cells corresponding to the red dot in
    Fig.~\ref{Bifur_3cells_ID}. Left panel: surface plot at time
    $t = 400$. Middle panel: intracellular species $u_1$ versus $t$.
    Right panel: intracellular species $u_2$ versus $t$. The blue and
    red curve is for the cells on the ring and the center cell,
    respectively.}
  \label{FlexPDE_ID3cells}
\end{figure}

\begin{figure}[htbp]
  \centering
	\makebox{
    \raisebox{0.5ex}{}
    \includegraphics[width=0.30\textwidth,height=4.5cm]{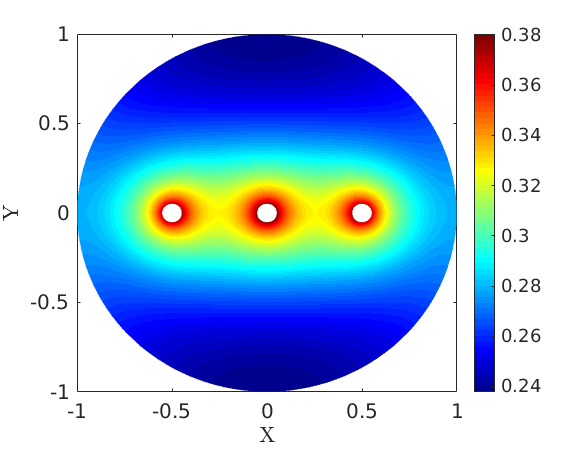}
    \phantomsubcaption
    \label{FlexPDE_3ID_surfStable}
  }  
  \makebox{
    \raisebox{0.5ex}{}
    \includegraphics[width=0.30\textwidth,height=4.5cm]{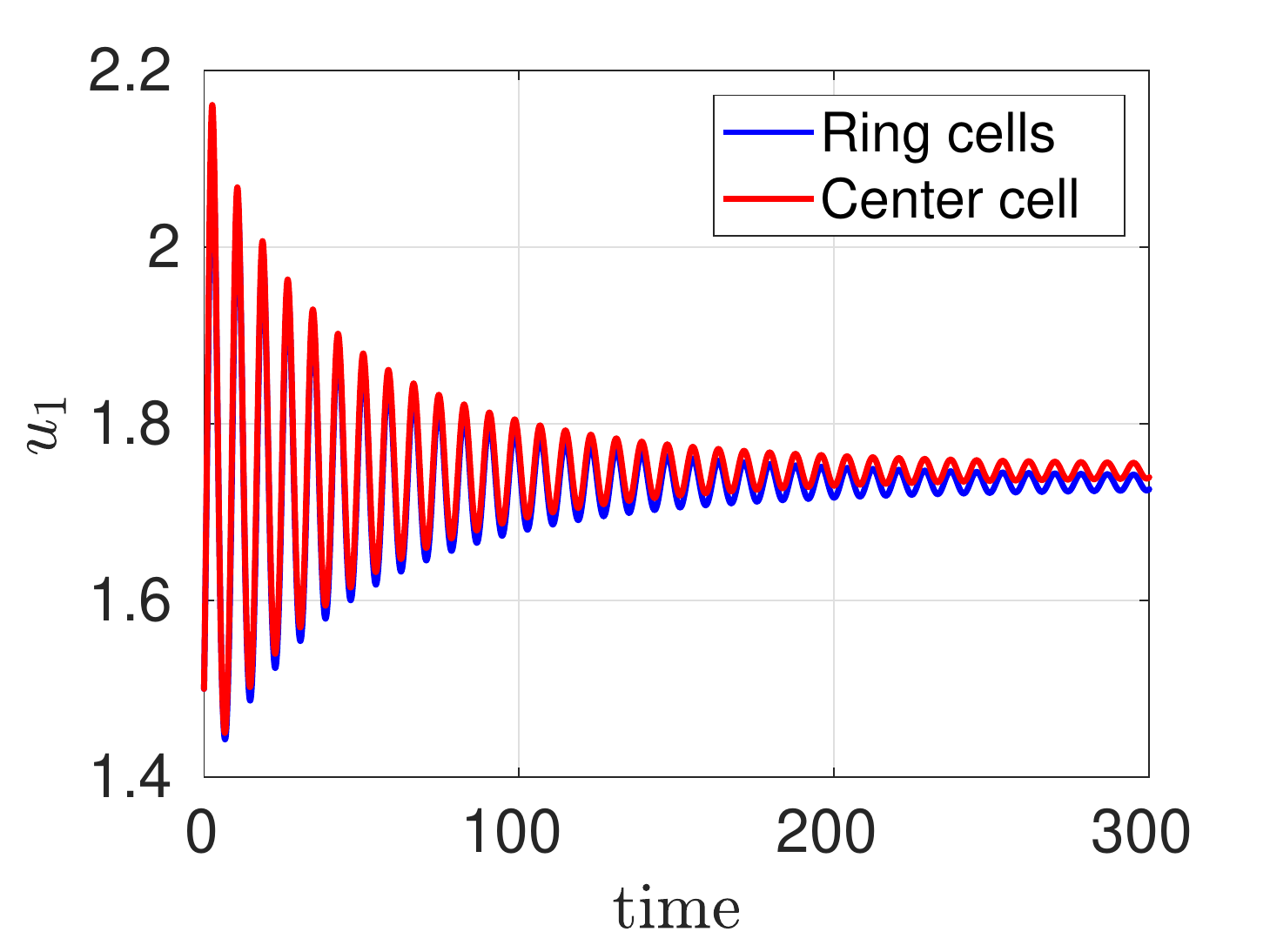}
    \phantomsubcaption
    \label{FlexPDE_3ID_bulkStable}
  }
  \makebox{
    \raisebox{0.5ex}{}
    \includegraphics[width=0.30\textwidth,height=4.5cm]{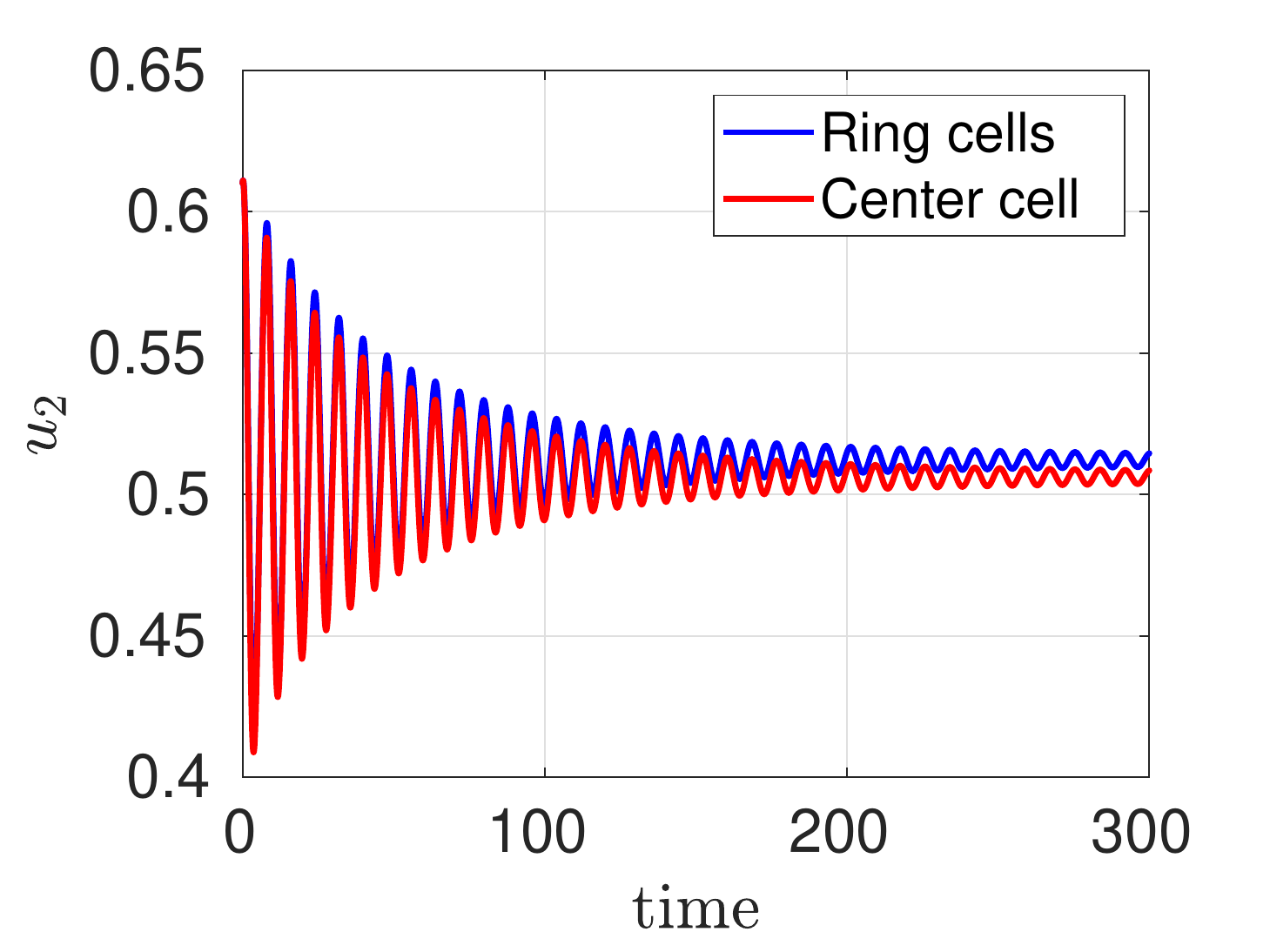}
    \phantomsubcaption
    \label{FlexPDE_ID3cells_u1u2Stable}
  }
  \vspace*{-1ex}
  \caption{Same caption as in Fig.~\ref{FlexPDE_ID3cells} except that
    now $\tau =1.1$ and $D = 1$, corresponding to the blue dot in
    Fig.~\ref{Bifur_3cells_ID}. There are no intracellular
    oscillations and the steady-state is linearly stable.}
  \label{FlexPDE_ID3cellsStable}
\end{figure}

\begin{figure}[htbp]
  \begin{minipage}[c]{0.38\textwidth}
    \centering
    \includegraphics[width=\textwidth,height=5.4cm]{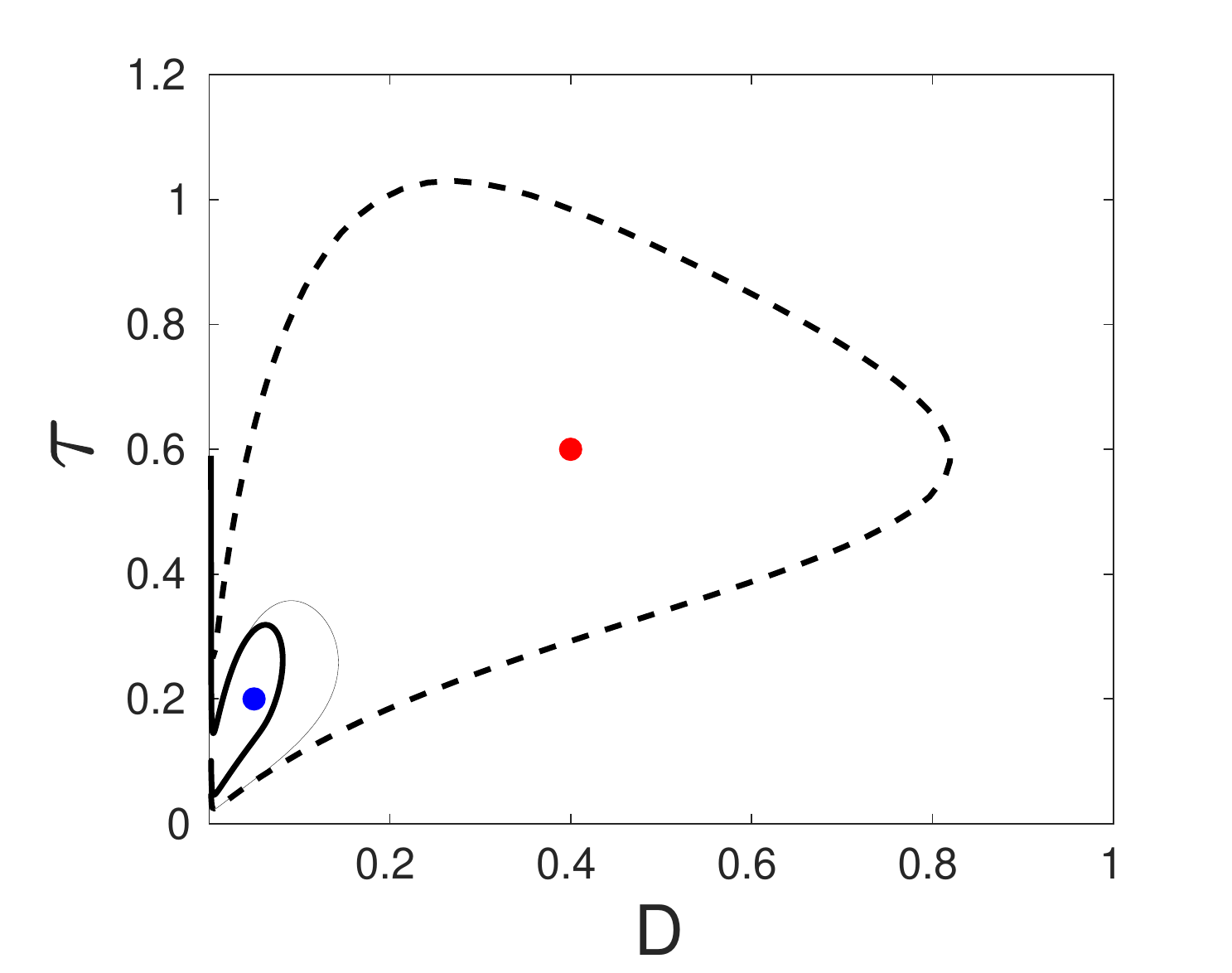}
  \end{minipage}%
  \begin{minipage}[c]{0.62\textwidth}
    \centering
 \begin{tabular}[c]{ c | c | c | c | c | c } \hline  
\rowcolor{LightCyan}
 mode &$(D,\tau)$  &  j  &  Re(${c}_j$)   & Im(${c}_j$)  & $\theta_j\, (\text{rad})$ \\  \hline \hline    \rowcolor{Cyan}
    &         & 1 & $0.345$    & $0.171$  & $0.460$ \\ \rowcolor{Cyan}
In-phase ($+$) & $ (0.820,0.584) $ & 2 & $0.345$    & $0.171$  & $0.460$ \\ \rowcolor{Cyan}
 (dashed curve)  &        & 3 & $0.839$   & $0$  & $0$ \\ 
\hline  \hline \rowcolor{Gray}
     &        & 1 & $-0.446$     & $0.256$ & $2.62$ \\ \rowcolor{Gray}
In-phase ($-$) & $ (0.0775,0.216) $ & 2 & $-0.446$    & $0.256$  & $2.62$ \\ \rowcolor{Gray}
 (heavy solid)  &      & 3 & $0.686$   & $0$  & $0$ \\ 
\hline \hline  \rowcolor{Cyan}
     &        & 1 & $0.707$    & $0$  & $0$ \\ \rowcolor{Cyan}
Anti-phase  &$ (0.129,0.325) $ & 2 & $-0.707$    & $0$  & $\pi$ \\ \rowcolor{Cyan}
(thin solid)  &         & 3 & $0$   & $0$  & $0$ \\ 
  \hline
  \end{tabular}
\end{minipage}
\caption{Left panel: same caption as in Fig.~\ref{Bifur_3cells_ID}
  except that now the center cell is a defector with permeabilities
  $d_{13} = 0.4$ and $d_{23} = 0.2$, corresponding to a reduced influx
  into the center cell. Full PDE simulations are shown in
  Figs.~\ref{FlexPDE_Def_3cells} and \ref{FlexPDE_Def3cellsAsync} at
  the red and blue dots, respectively. Right panel: same caption as in
  Table~\ref{Table:ID_RCH_3cells} except that now $d_{13}=0.4$ and
  $d_{23}=0.2$. The real and imaginary parts of the eigenvector
  $\pmb{c}$ of the GCEP matrix ${\mathcal M}(\lambda)$ in
   \eqref{M_M0} are computed for a few points on the HB
  boundaries shown in the left panel.}
\label{fig:defective_middle}
\end{figure}

\begin{figure}[!httbp]
  \centering
	\makebox{
    \raisebox{0.5ex}{}
    \includegraphics[width=0.30\textwidth,height=4.5cm]{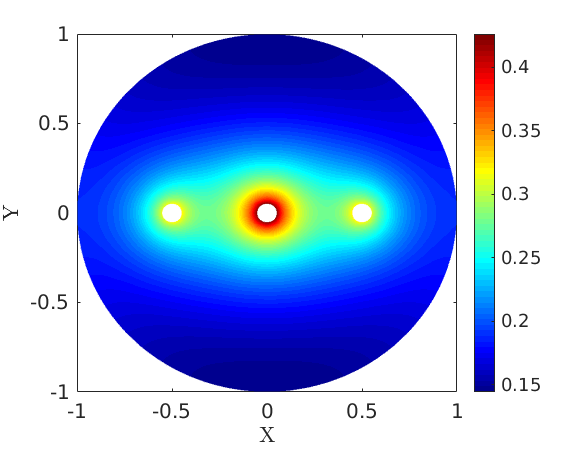}
    \phantomsubcaption
    \label{FlexPDE_Def3cells_surfSync}
  }  
  \makebox{
    \raisebox{0.5ex}{}
    \includegraphics[width=0.30\textwidth,height=4.5cm]{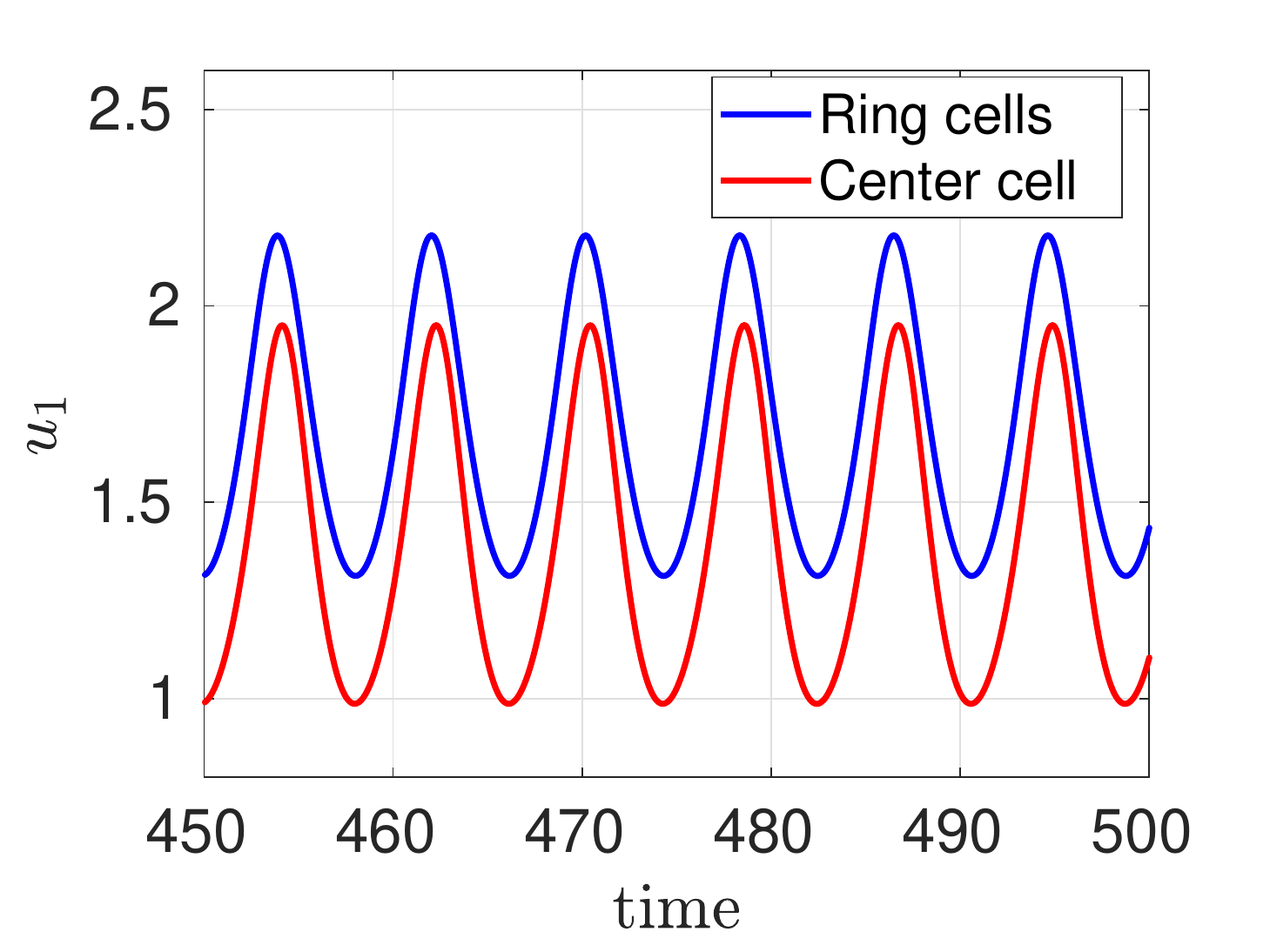}
    \phantomsubcaption
    \label{FlexPDE_Def3cells_RingSync}
  }
  \makebox{
    \raisebox{0.5ex}{}
    \includegraphics[width=0.30\textwidth,height=4.5cm]{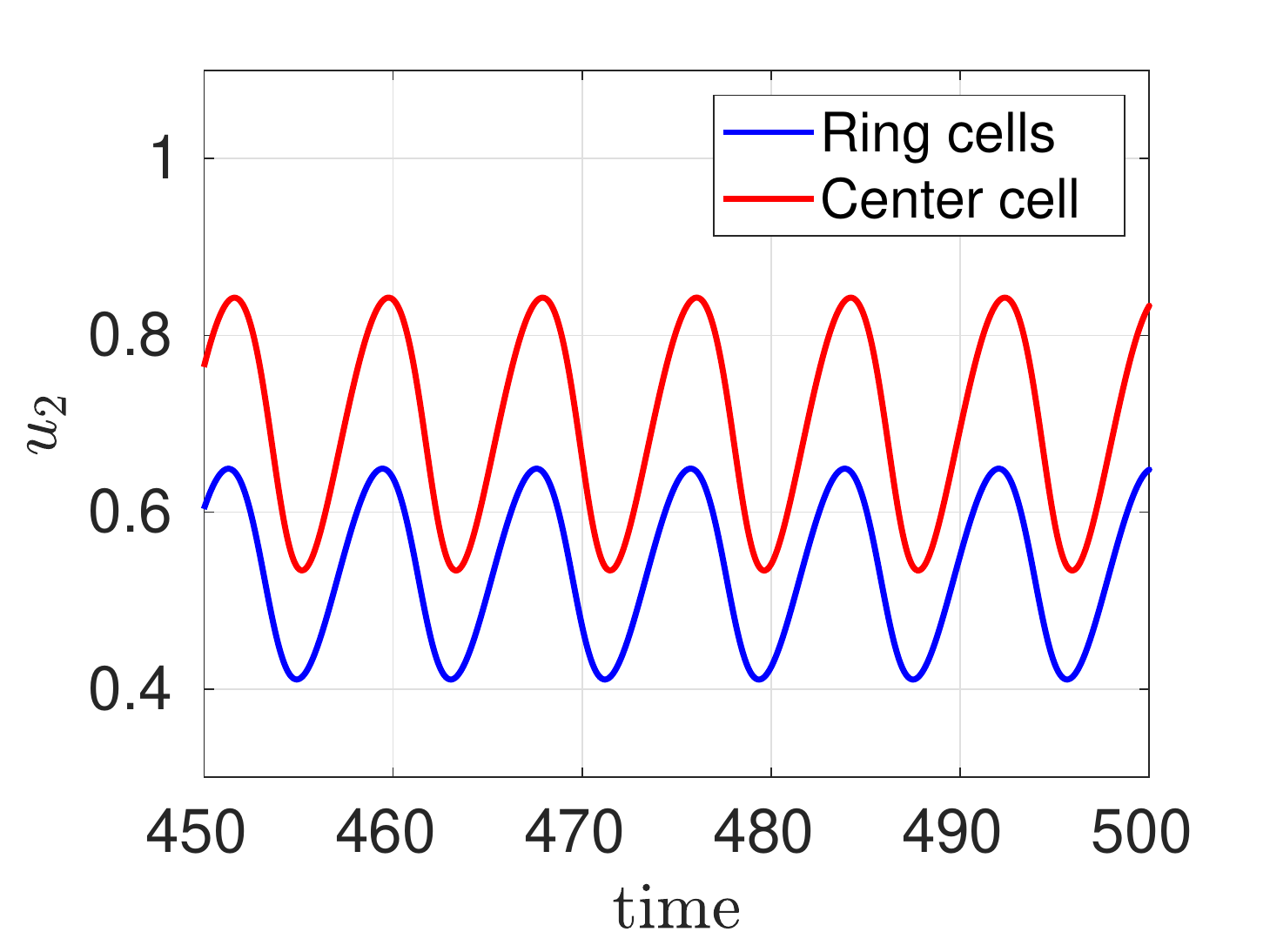}
    \phantomsubcaption
    \label{FlexPDE_Def3cells_CentreSync}
  }
  \vspace*{-1ex}
  \caption{Full PDE simulations of \eqref{DimLess_bulk}, computed
    with FlexPDE \cite{flexpde2015solutions}, for $\tau =0.6$ and
    $D = 0.4$, corresponding to the red dot in the left panel of
    Fig.~\ref{fig:defective_middle}. The center cell is a defector with
    permeabilites $d_{13}=0.4$ and $d_{23}=0.2$. Left panel: surface plot at
    time $t = 400$. Middle panel: intracellular species $u_1$ versus $t$.
    Right panel: intracellular species $u_2$ versus $t$.}
  \label{FlexPDE_Def_3cells}
\end{figure}

\begin{figure}[!httbp]
  \centering
	\makebox{
    \raisebox{0.5ex}{}
    \includegraphics[width=0.30\textwidth,height=4.5cm]{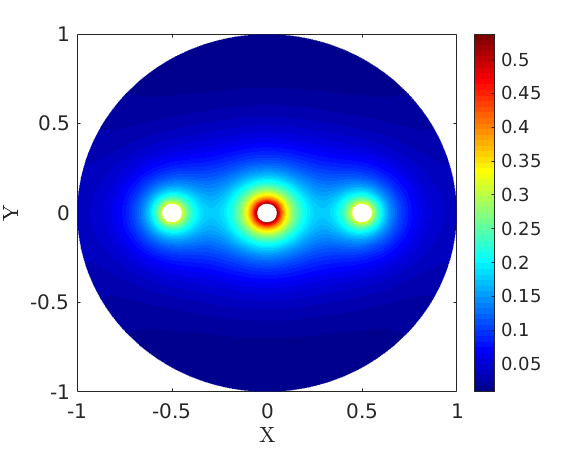}
    \phantomsubcaption
    \label{FlexPDE_Def3cells_surfAsync}
  }  
  \makebox{
    \raisebox{0.5ex}{}
    \includegraphics[width=0.30\textwidth,height=4.5cm]{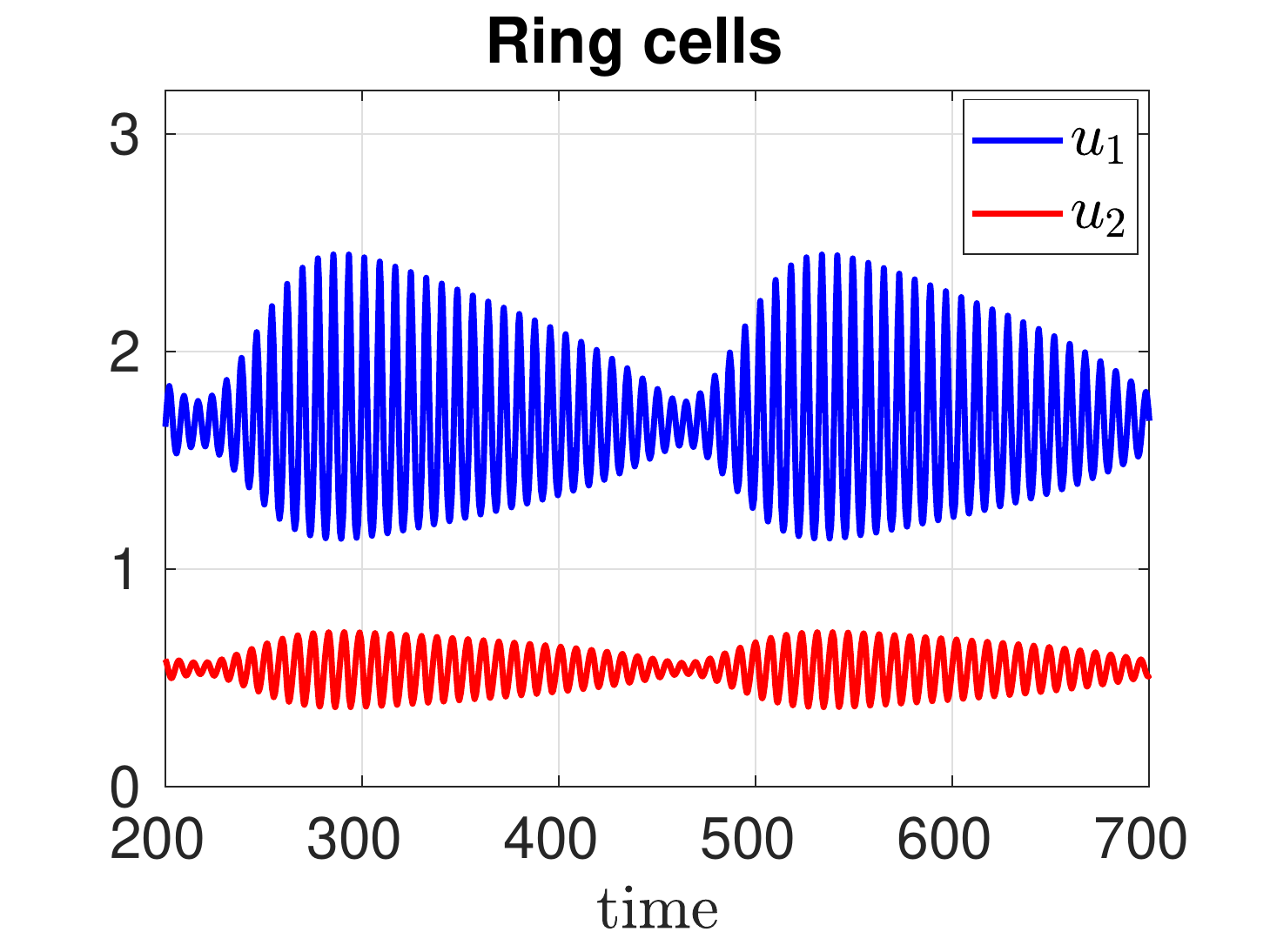}
    \phantomsubcaption
    \label{FlexPDE_Def3cells_ringcellsAsync}
  }
  \makebox{
    \raisebox{0.5ex}{}
    \includegraphics[width=0.30\textwidth,height=4.5cm]{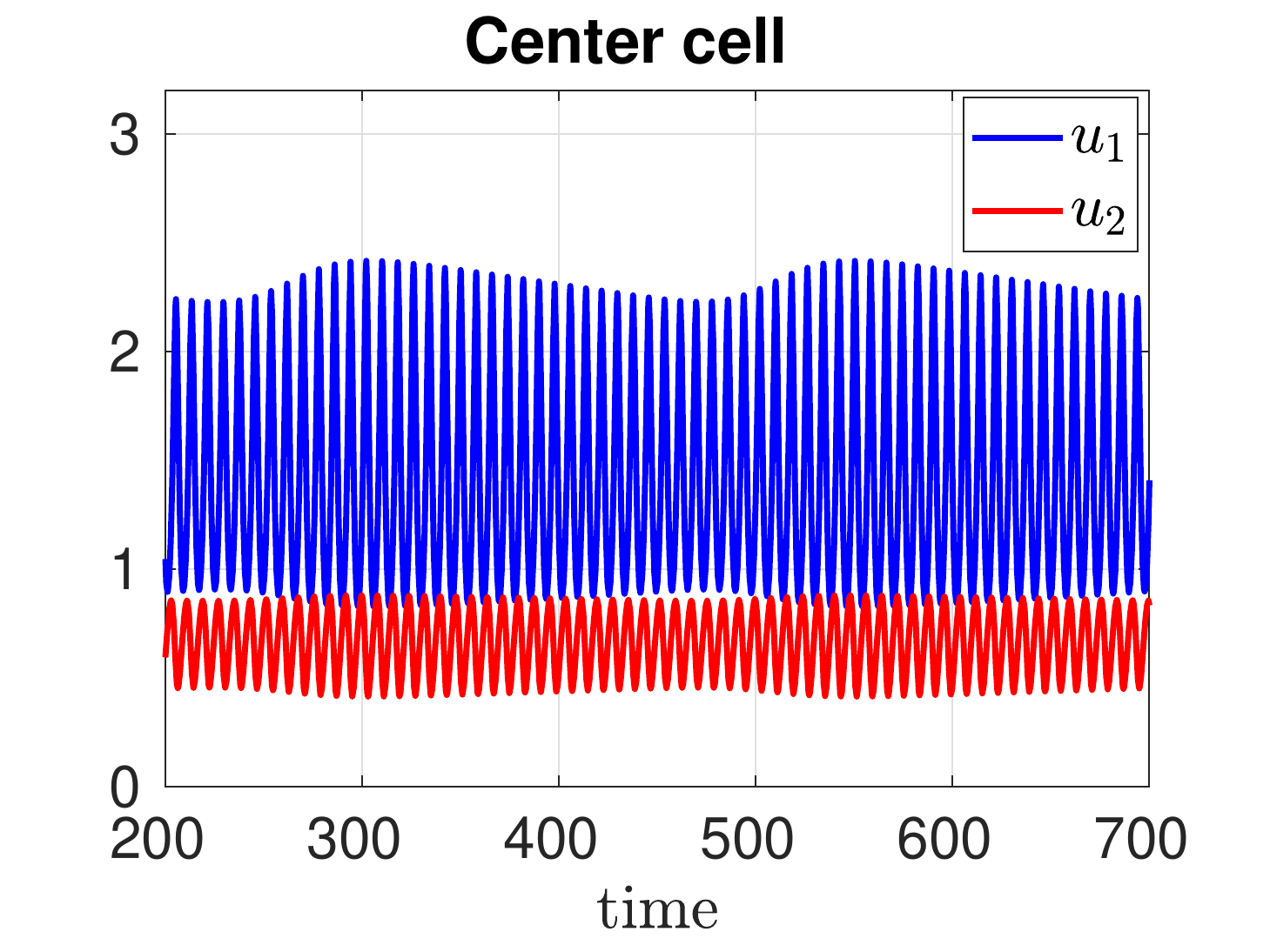}
    \phantomsubcaption 
    \label{FlexPDE_Def3cells_centreCellAsync}
  }
  \vspace*{-1ex}
  \caption{Same caption as in Fig.~\ref{FlexPDE_Def_3cells},
      except that the FlexPDE \cite{flexpde2015solutions} simulation
      of \eqref{DimLess_bulk} is done at $\tau =0.2$ and $D = 0.05$,
      corresponding to the blue dot in the left panel of
      Fig.~\ref{fig:defective_middle}. At this point, both in-phase
      modes and the anti-phase mode are unstable. The beating-behavior
      observed arises from the fact that these three modes all have
      comparable frequencies.}
  \label{FlexPDE_Def3cellsAsync}
\end{figure}

In the left panel of Fig.~\ref{fig:defective_middle} we show the
computed HB boundaries in the $\tau$ versus $D$ plane for the same
parameters as in Fig.~\ref{Bifur_3cells_ID}, except that the influx
rate into the center cell is reduced to $d_{13}=0.4$ (keeping
$d_{23}=0.2$). With this lower influx rate,
Fig.~\ref{fig:defective_middle} shows that the in-phase lobes are now
closed, so that there are no longer intracellular oscillations when
$D$ is large. Qualitatively, when $D$ is large, the bulk chemical
diffuses quickly in the entire disk and there is insufficient feedback
of it into the center cell, owing to the smaller value of $d_{13}$, to
sustain intracellular oscillations. The winding number algorithm of
\eqref{ArgMent_Principle} can be used, with the same result as in
Fig.~\ref{Bifur_3cells_first}, to determine the number of unstable
eigenvalues of the GCEP within the lobes (not shown). In the right
panel of Fig.~\ref{fig:defective_middle} we give the real and
imaginary parts of the normalized eigenvector $\pmb{c}$ of the GCEP
matrix in \eqref{M_M0} at a few selected points on the HB boundaries
shown in the left panel of Fig.~\ref{fig:defective_middle}.

In Fig.~\ref{FlexPDE_Def_3cells}, we show full FlexPDE
\cite{flexpde2015solutions} numerical simulations of
\eqref{DimLess_bulk} for $\tau = 0.6$ and $D=0.4$, which corresponds
to the red dot in the left panel of
Fig.~\ref{fig:defective_middle}. Our numerical computations
  of $\det{\mathcal M}(\lambda)=0$ for $(D,\tau)=(0.4,0.6)$ using the
  GCEP matrix in \eqref{M_M0} yields that
  $\mbox{Re}(\lambda)\approx 0.0115$,
  $\mbox{Im}(\lambda)\approx 0.778$,
  $\mbox{Re}({\mathcal K}\pmb{c})\approx (-0.551,-0.551,-0.624)$ and
  $\mbox{Im}({\mathcal K}\pmb{c})\approx (0.0197,0.0197,0.0569)$. As
  such, our linearized theory predicts that the center cell will have
  larger amplitude oscillations near onset than the ring cells, and
  there will be a $\approx 22^{\circ}$ phase shift between the
  oscillations. Our FlexPDE numerical results in the middle and right
  panels of Fig.~\ref{FlexPDE_Def_3cells} show that the prediction of
  the linearized theory does extend to the fully nonlinear regime in
  that the defective center cell has larger amplitude oscillations
  than do the identical ring cells, and there is a slight phase shift
  in the oscillations. From the surface plot shown in the left panel
  of Fig.~\ref{FlexPDE_Def_3cells}, we observe that the coupling
  between the cells mediated by the bulk medium is rather weak.
  Moreover, the rather large concentration of the bulk chemical close
  to the center cell, with flux measured by the modulus of the third
  component of $\mbox{Re}(\pmb{c})\approx (0.267,0.267,0.914)$ and
  $\mbox{Im}(\pmb{c})\approx (0.107,0.107,0.0)$, is due to its smaller
  rate $d_{13}=0.4$ of influx into the center cell than for the
  identical ring cells centered at $(\pm 0.5,0)$. Paradoxically,
  however, this large buildup of the bulk signal near the center cell
  counteracts the relatively smaller rate of influx into the center
  cell, and has the effect of triggering a larger amplitude
  oscillation in the center cell than for the ring cells.

\begin{figure}
  \begin{minipage}[c]{0.38\textwidth}
    \centering
    \includegraphics[width=\textwidth,height=5.4cm]{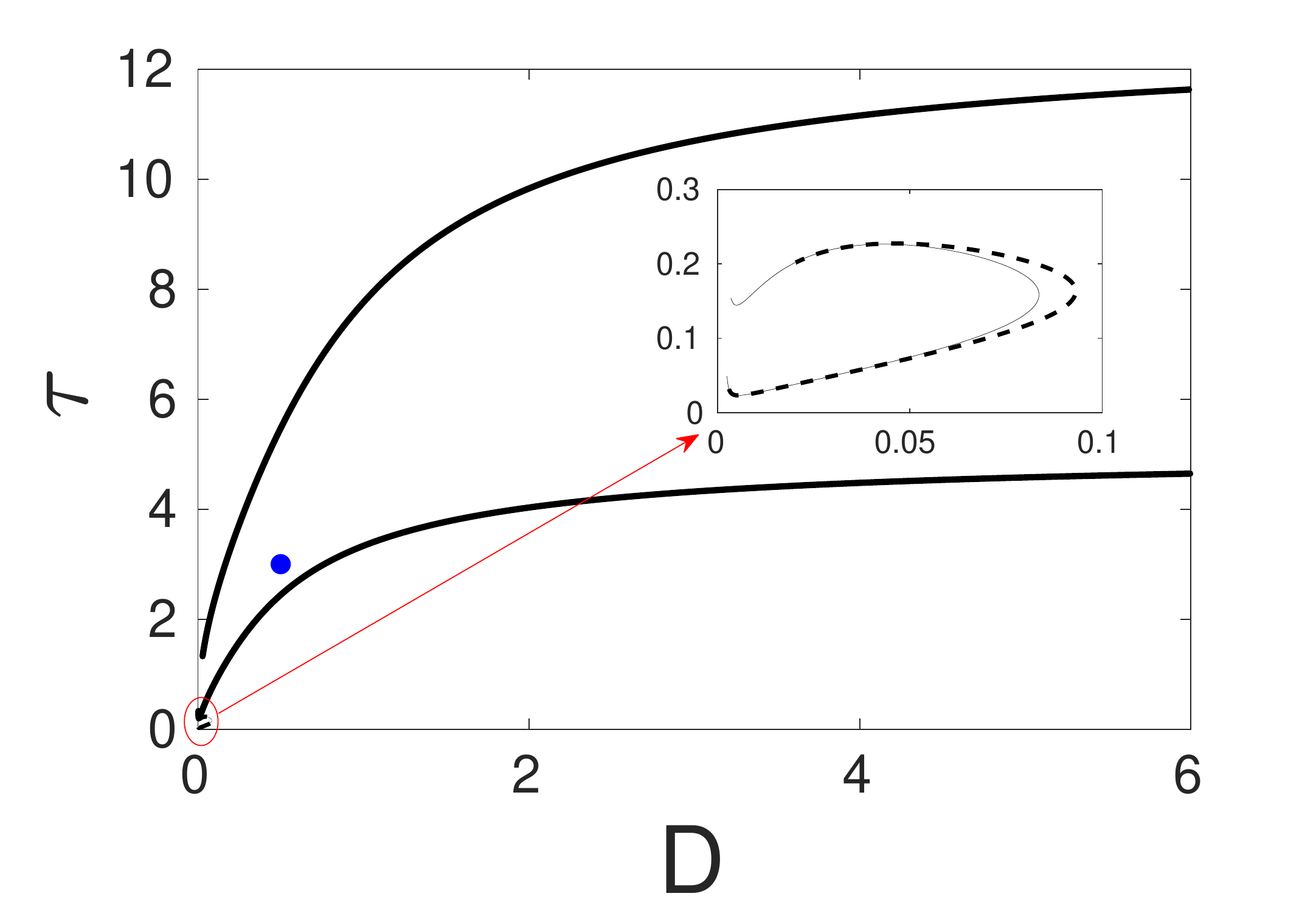}
  \end{minipage}
  \begin{minipage}[c]{0.62\textwidth}
    \centering
   \begin{tabular}[c]{ c | c | c | c | c | c } \hline  
\rowcolor{LightCyan}
mode &$(D,\tau)$  &  j  &  Re(${c}_j$)   & Im(${c}_j$)  & $\theta_j \,(\text{rad})$ \\  \hline \hline    \rowcolor{Cyan}
    &         & 1 & $0.707$    & $0$  & $0$ \\ \rowcolor{Cyan}
In-phase ($+$) & $ (0.0931,0.165) $ & 2 & $0.707$    & $0$  & $0$ \\ \rowcolor{Cyan}
(dashed curve)   &        & 3 & $0.0245$   & $-0.0314$  & $5.375$ \\ 
\hline  \hline \rowcolor{Gray}
     &        & 1 & $-0.0638$    & $0.0985$  & $2.146$ \\ \rowcolor{Gray}
In-phase ($-$) & $ (1.241,3.577) $ & 2 & $-0.0638$    & $0.0985$  & $2.146$ \\ \rowcolor{Gray}
 (heavy solid)   &      & 3 & $0.986$   & $0$  & $0$ \\ 
\hline \hline  \rowcolor{Cyan}
     &        & 1 & $0.707$    & $0$  & $0$ \\ \rowcolor{Cyan}
Anti-phase  &$ (0.0836,0.159) $ & 2 & $-0.707$    & $0$  & $\pi$ \\ \rowcolor{Cyan}
(thin solid)  &         & 3 & $0$   & $0$  & $0$ \\ 
  \hline
   \end{tabular}
\end{minipage}
\caption{Same caption as in Fig.~\ref{fig:defective_middle}
    except that now the center cell is a defector with permeabilities
    $d_{13} = 0.2$ and $d_{23} = 0.4$, corresponding to a reduced
    influx and a larger efflux out of the center cell. Full PDE
    simulations are shown in Fig.~\ref{FlexPDE_ID3cells_d10p2} at the
    blue dot in the left panel.}
\label{fig:defective_right}
\end{figure}

In Fig.~\ref{FlexPDE_Def3cellsAsync}, we show full FlexPDE
\cite{flexpde2015solutions} numerical simulations of
\eqref{DimLess_bulk} for $\tau = 0.2$ and $D=0.05$, corresponding to
the blue dot in the left panel of Fig.~\ref{fig:defective_middle}.  As
seen from Fig.~\ref{fig:defective_middle}, this point is located
within the region of instability that is common to all three modes of
instability.  By solving $\det{\mathcal M}(\lambda)=0$ for
$(D,\tau)=(0.05,0.2)$ numerically, the eigenvalues $\lambda$ and
eigenvectors $\pmb{c}$ for the three modes are:
\begin{equation}\label{beats:spectra}
    \begin{split}
  \mbox{in-phase (+):}\quad \lambda &\approx 0.0350 + 0.779\,i\,,\quad
  \mbox{Re}(c)\approx (   0.152, 0.152, 0.829) \,, \quad
  \mbox{Im}(c) \approx  (0,0,-0.517)\,, \\
  \mbox{in-phase (-):}\quad \lambda &\approx 0.0144 + 0.813\, i\,,\quad
  \mbox{Re}(c)\approx (   0.624, 0.624, -0.434) \,, \quad
  \mbox{Im}(c) \approx (0,0,-0.180)\,,\\
  \mbox{anti-phase:}\quad \lambda &\approx 0.0189 + 0.812\, i\,,\quad
  \mbox{Re}(c)\approx (   0.707, -0.707, 0) \,, \quad
  \mbox{Im}(c) =  (0,0,0)\,.\\
  \end{split}
\end{equation}
We observe that this linearized theory predicts three distinct
unstable oscillatory modes with roughly similar frequencies
$\mbox{Im}(\lambda)$ and growth rates $\mbox{Re}(\lambda)$.  The
beating-type intracellular oscillations observed in
Fig.~\ref{FlexPDE_Def3cellsAsync} is likely related to the well-known
linear phenomenon of superimposing two or more single-mode
oscillations with comparable frequencies.

\begin{figure}[htbp]
  \centering
	\makebox{
    \raisebox{0.5ex}{}
    \includegraphics[width=0.30\textwidth,height=4.5cm]{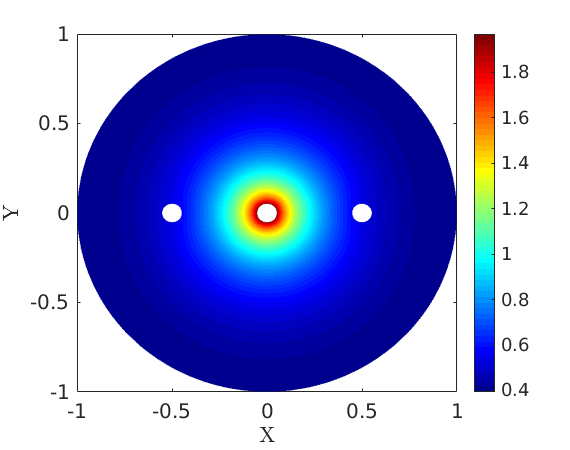}
    \phantomsubcaption
    \label{FlexPDE_Def3cells_surfSync_d10p2}
  }  
  \makebox{
    \raisebox{0.5ex}{}
    \includegraphics[width=0.30\textwidth,height=4.5cm]{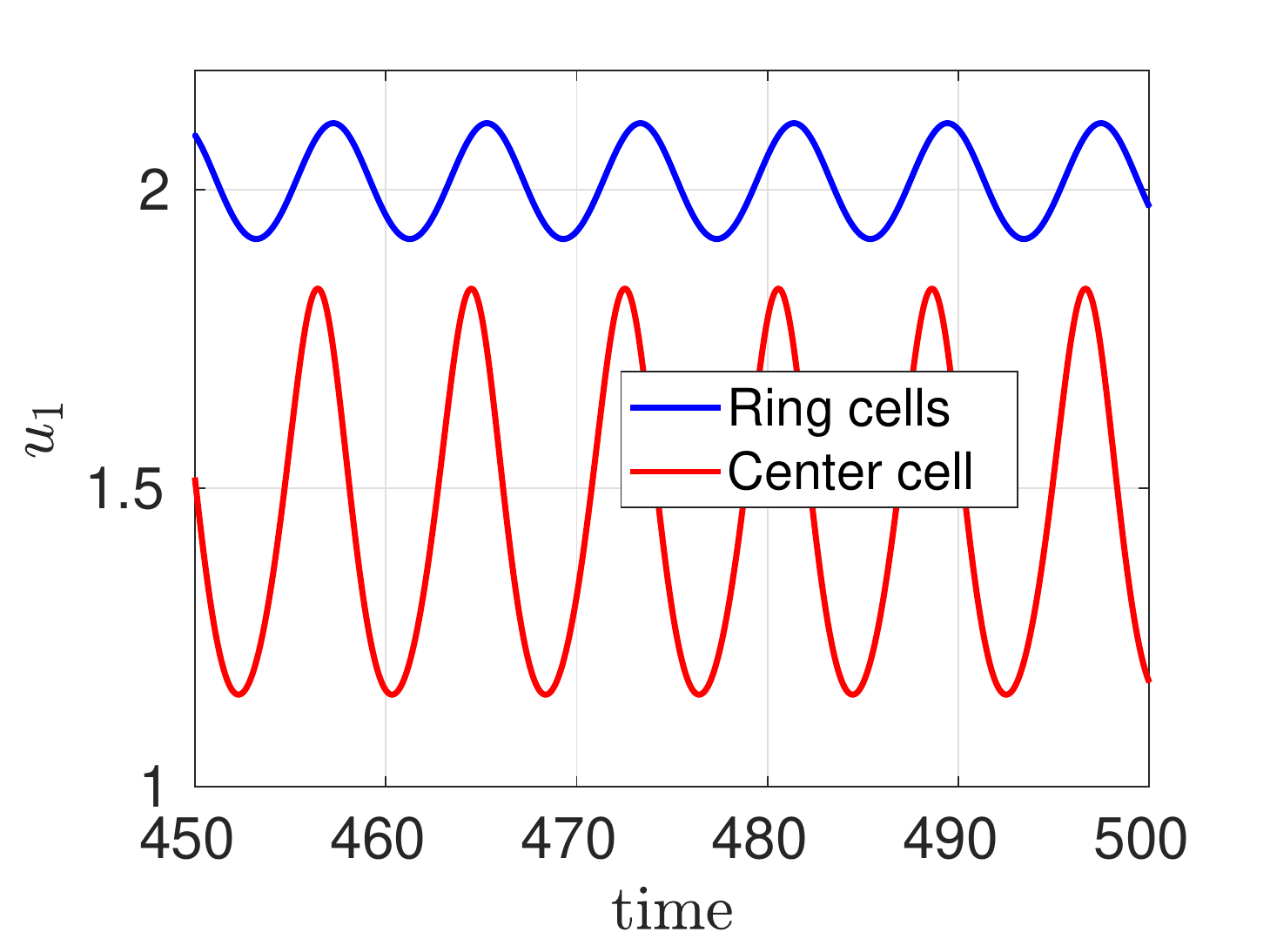}
    \phantomsubcaption
    \label{FlexPDE_Def3cells_RingSync_d10p2}
  }
  \makebox{
    \raisebox{0.5ex}{}
    \includegraphics[width=0.30\textwidth,height=4.5cm]{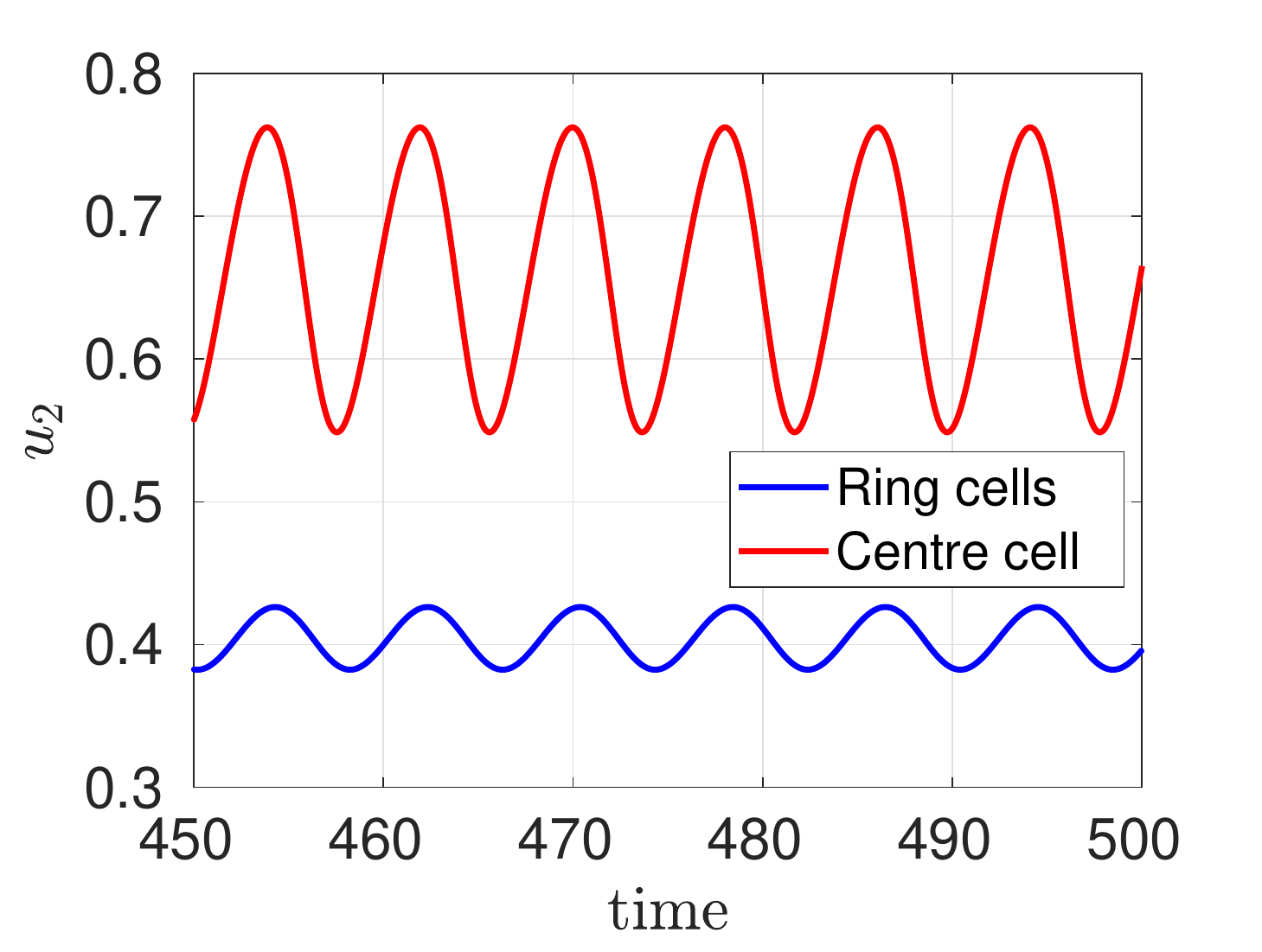}
    \phantomsubcaption
    \label{FlexPDE_Def3cells_CentreSync_d10p2}
  }
  \vspace*{-1ex}
  \caption{Full PDE simulations of \eqref{DimLess_bulk},
      computed with FlexPDE \cite{flexpde2015solutions}, for $\tau =3$
      and $D = 0.5$, corresponding to the blue dot in the left panel
      of Fig.~\ref{fig:defective_right}. The center cell is a defector
      with permeabilites $d_{13}=0.2$ and $d_{23}=0.4$. Left panel:
      surface plot at time $t = 400$. Observe the buildup of the bulk
      chemical near the center cell in comparison to the two ring
      cells. Middle panel: intracellular species $u_1$ versus $t$.
      Right panel: intracellular species $u_2$ versus $t$. As
      predicted by the second row in the table in the right panel of
      Fig.~\ref{fig:defective_right}, we confirm that the center cell
      has larger amplitude oscillations than do the ring cells.}
  \label{FlexPDE_ID3cells_d10p2}
\end{figure}

In the left panel of Fig.~\ref{fig:defective_right} we plot the HB
boundaries for the case where the center cell has permeabilities
$d_{13} = 0.2$ and $d_{23}=0.4$. This corresponds to a larger
secretion or efflux rate for the center cell, while the feedback it
receives from the ring cells is reduced. Unlike the HB boundaries
shown in Figs.~\ref{Bifur_3cells_first} and \ref{fig:defective_middle}
where the instability regions for the modes are nested within each
other, we observe from the insert in the left panel of
Fig.~\ref{fig:defective_right} that only the in-phase $(+)$ mode and
the anti-phase mode overlap. Moreover, since the other in-phase $(-)$
lobe is unbounded in $D$, intracellular oscillations always occur in
some range of $\tau$ as $D$ increases. In the right panel of
  Fig.~\ref{fig:defective_right}, we give the real and imaginary parts
  of the normalized eigenvector $\pmb{c}$ of the GCEP matrix
  \eqref{M_M0} at a few points on the HB boundaries. From the second
  row of this table, the linearized theory suggests that the amplitude
  of intracellular oscillations associated with the dominant in-phase
  $(-)$ instability lobe will be much larger in the center cell than
  in the identical ring cells at the HB point, and that there will be
  a significant phase shift in the oscillations between the center
  cell and the ring cells. More precisely, at the point
  $(D,\tau)=(0.5,3)$ interior to the instability lobe, we solve
  $\mbox{det}{\mathcal M}(\lambda)=0$ numerically to obtain that
  $\mbox{Re}(\lambda)\approx 0.00665$,
  $\mbox{Im}(\lambda)\approx 0.781$,
  $\mbox{Re}(\pmb{c})\approx (-0.0534,-0.0534,0.9908)$,
  $\mbox{Im}(\pmb{c})\approx (0.0792,0.0792,0)$,
  $\mbox{Re}({\mathcal K}\pmb{c})\approx (-0.177,-0.177,-0.926)$, and
  $\mbox{Im}({\mathcal K}\pmb{c})\approx (0.187,0.187,0.106)$. From
  ${\mathcal K}\pmb{c}$ and $\mbox{Im}(\lambda)$ this indicates that
  the center cell will have much larger oscillations near onset than
  the ring cells, with a period of oscillations of
  $\approx 8$ and with a phase shift of $\approx 40^{\circ}$ between the
  ring and center cell oscillations. From the FlexPDE simulations of
  \eqref{DimLess_bulk} when $\tau = 3$ and $D = 0.5$ shown in
  Fig.~\ref{FlexPDE_ID3cells_d10p2}, corresponding to the blue dot in
  the left panel of Fig.~\ref{fig:defective_right}, we observe that
  these predictions of the linearized theory are roughly
  satisfied. Moreover, since the efflux rate out of the center cell is
  larger, while the influx rate is smaller, the rather small bulk
  diffusivity $D=0.5$ should qualitatively lead to a buildup of the
  bulk chemical near the center cell at certain times in the
  oscillation. This feature is observed in the surface plot in the
  left panel of Fig.~\ref{FlexPDE_ID3cells_d10p2}, and provides a
  clear example of the diffusion-sensing behavior, as regulated by the
  permeability parameters.

\begin{figure}[!h]
  \centering
    \begin{subfigure}[b]{0.40\textwidth}
      \includegraphics[width=\textwidth,height=4.3cm]{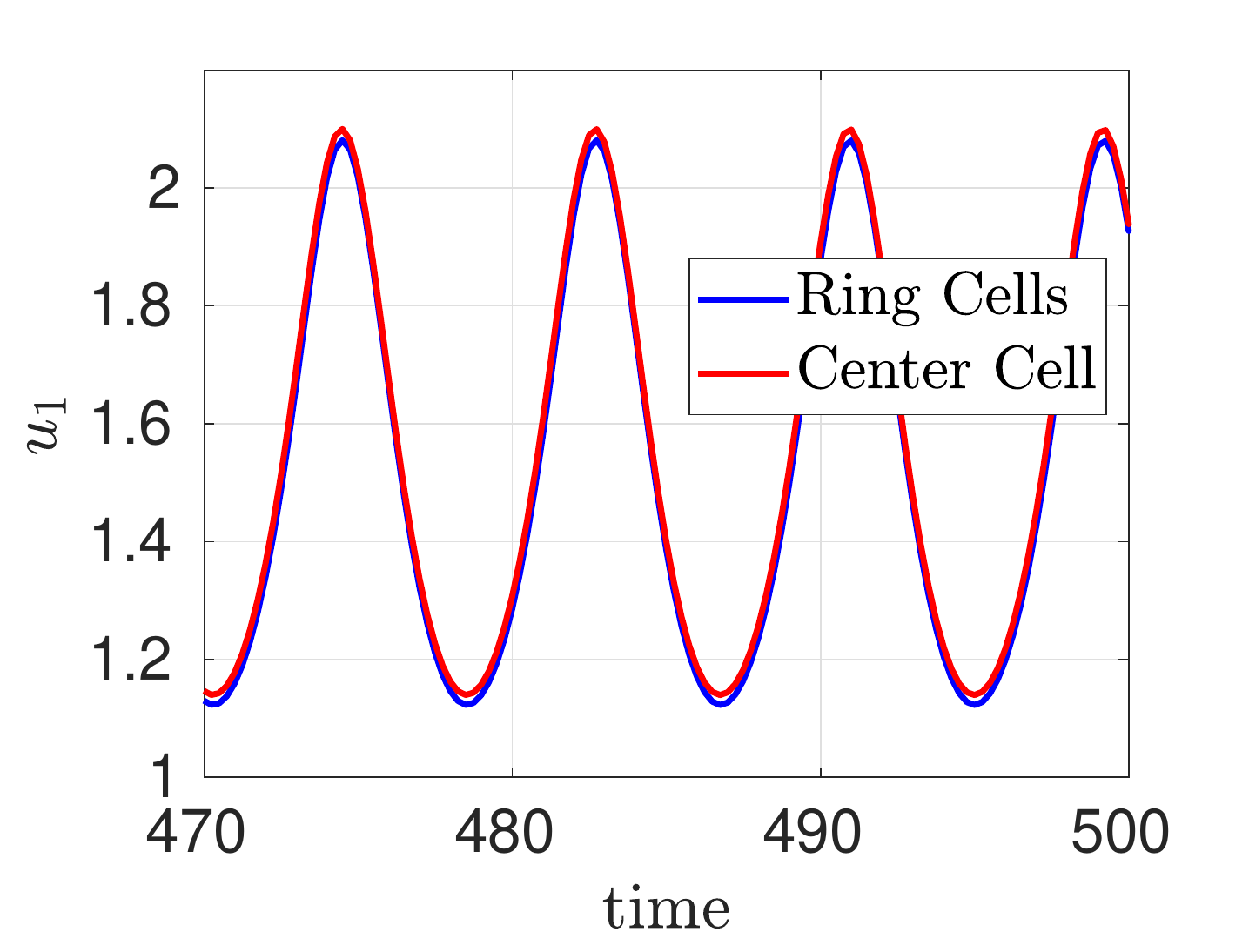}
        \phantomsubcaption
        \label{fig:fig6_odes_1}
    \end{subfigure}
    \begin{subfigure}[b]{0.40\textwidth}  
      \includegraphics[width=\textwidth,height=4.3cm]{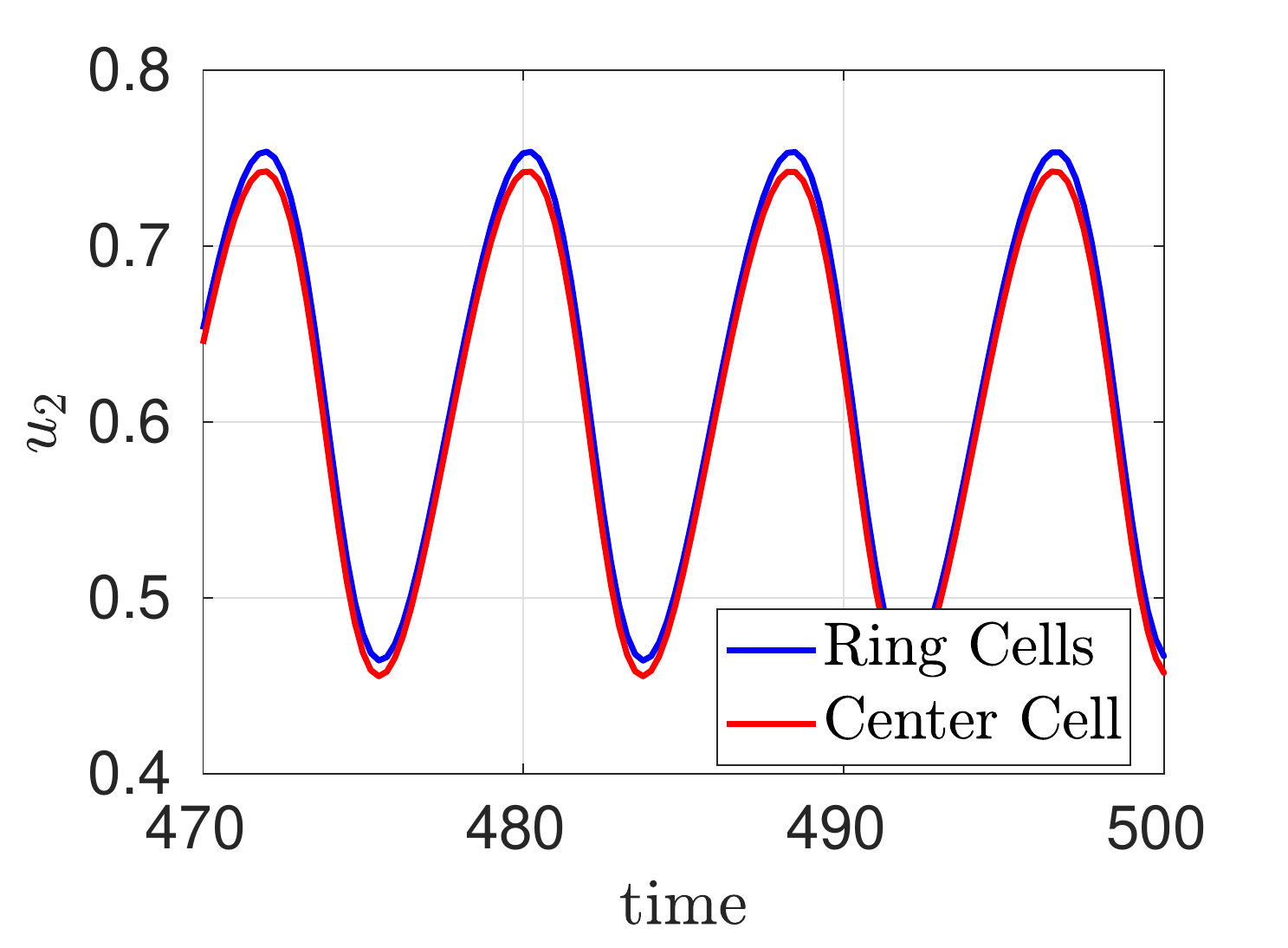}
        \phantomsubcaption
        \label{fig:fig6_odes_3}
    \end{subfigure}
    \begin{subfigure}[b]{0.40\textwidth}
      \includegraphics[width=\textwidth,height=4.3cm]{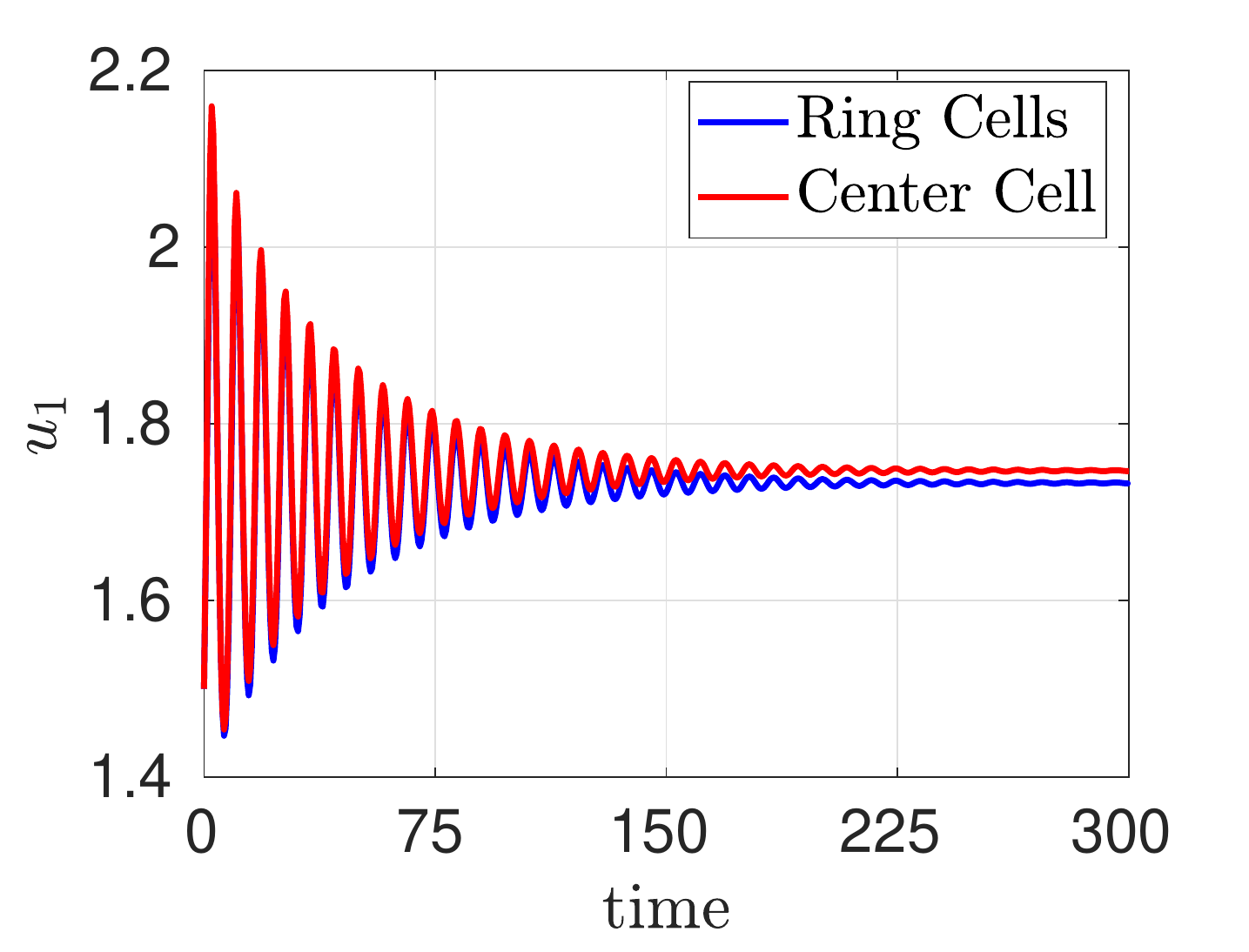}
                \phantomsubcaption
        \label{fig:fig7_odes_1}
    \end{subfigure}
    \begin{subfigure}[b]{0.40\textwidth}  
      \includegraphics[width=\textwidth,height=4.3cm]{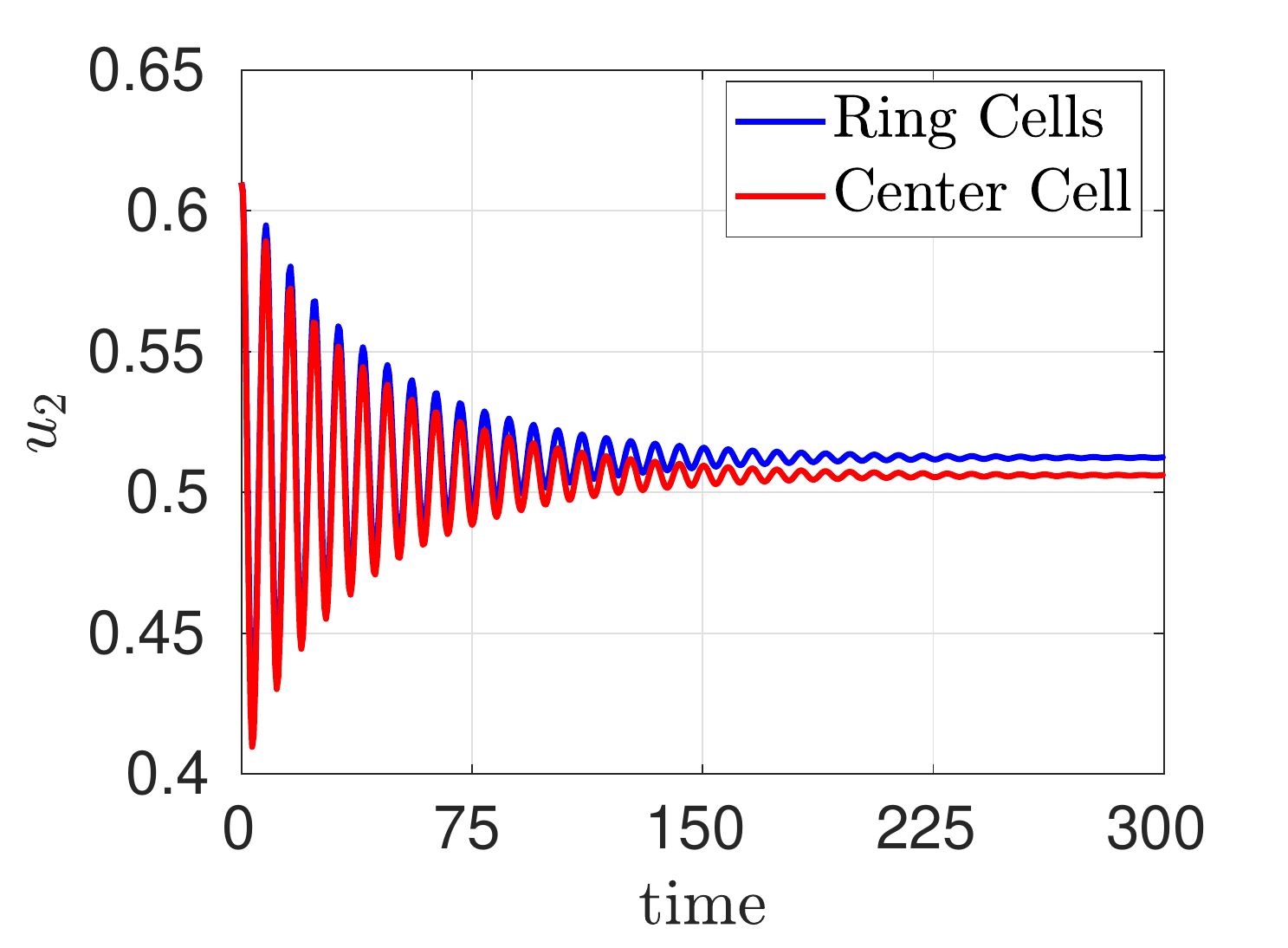}
               \phantomsubcaption
        \label{fig:fig7_odes_3}
    \end{subfigure}
    \begin{subfigure}[b]{0.40\textwidth}
      \includegraphics[width=\textwidth,height=4.3cm]{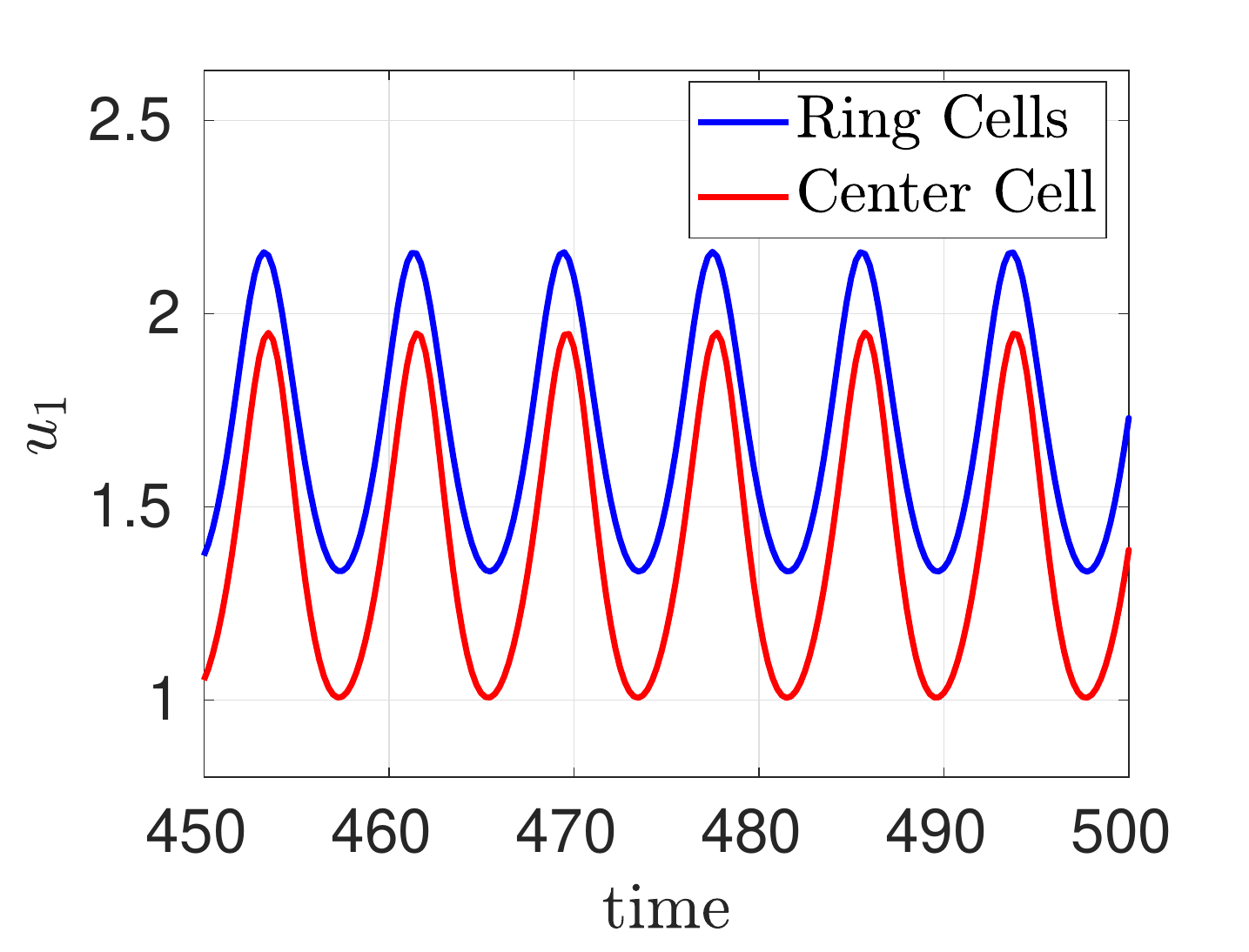}
                \phantomsubcaption
        \label{fig:fig9_odes_1}
    \end{subfigure}
    \begin{subfigure}[b]{0.40\textwidth}  
      \includegraphics[width=\textwidth,height=4.3cm]{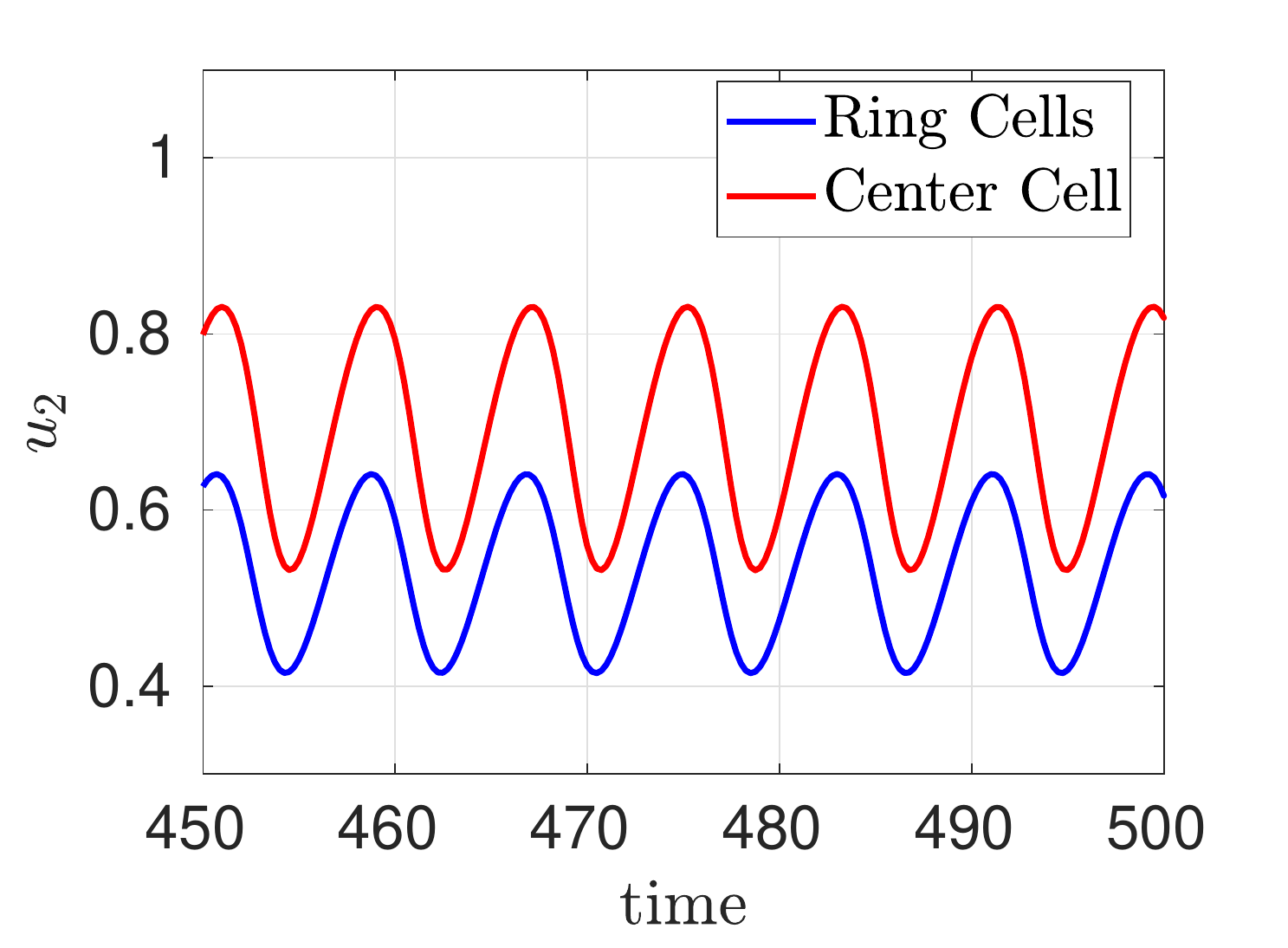}
                \phantomsubcaption
        \label{fig:fig9_odes_3}
    \end{subfigure}
    \begin{subfigure}[b]{0.40\textwidth}
      \includegraphics[width=\textwidth,height=4.3cm]{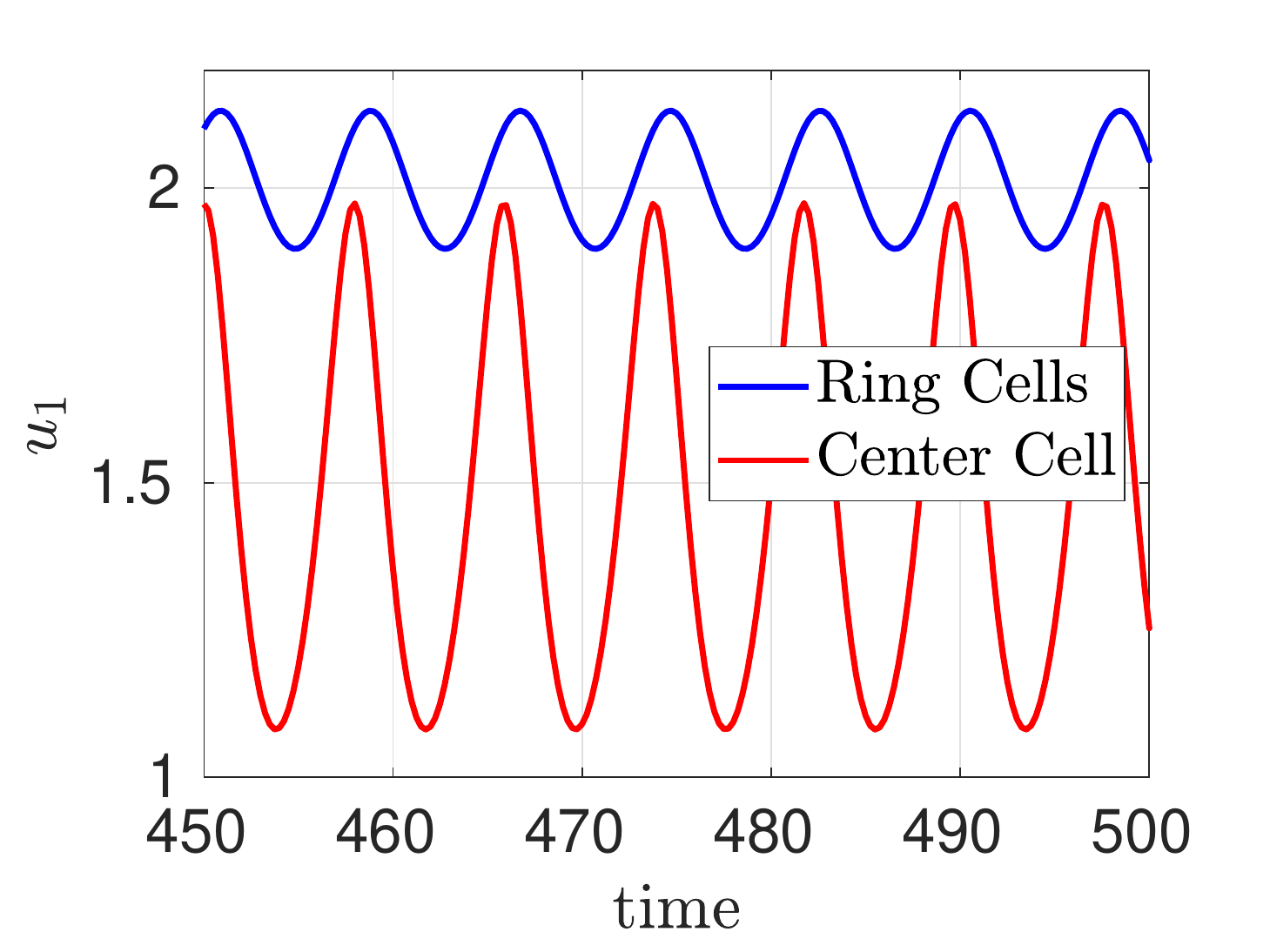}
                \phantomsubcaption
        \label{fig:fig10_odes_1}
    \end{subfigure}
    \begin{subfigure}[b]{0.40\textwidth}  
      \includegraphics[width=\textwidth,height=4.3cm]{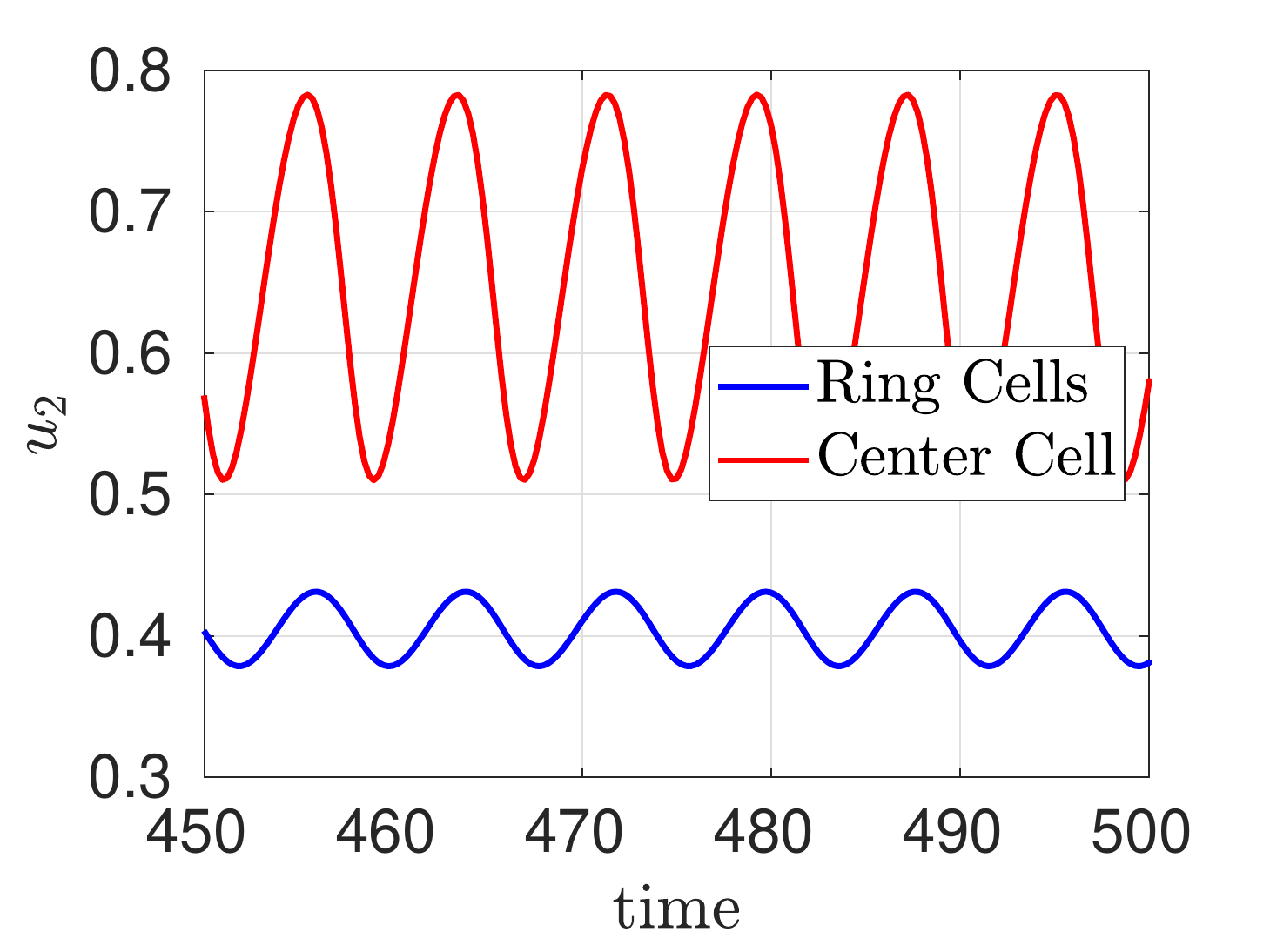}
                \phantomsubcaption
        \label{fig:fig10_odes_3}
    \end{subfigure}
    \caption{Numerical results for intracellular dynamics
        versus time computed from the ODE system \eqref{reducedODE} of
        Proposition \ref{prop:dyn} corresponding to the PDE
        simulations shown in Fig.~\ref{FlexPDE_ID3cells} (first row),
        Fig.~\ref{FlexPDE_ID3cellsStable}, (second row),
        Fig.~\ref{FlexPDE_Def_3cells} (third row) and
        Fig.~\ref{FlexPDE_ID3cells_d10p2} (fourth row), respectively.
        In each case, although $D$ is not large, results from the ODE
        system \eqref{reducedODE} are seen to compare surprisingly
        well with the full FlexPDE simulations of the PDE-ODE system
        \eqref{DimLess_bulk}. First row: $(D,\tau)=(1.0,0.55)$ (red
        dot in Fig.~\ref{Bifur_3cells_ID}). Compare with
        Fig.~\ref{FlexPDE_ID3cells}.  Second row: $(D,\tau)=(1.0,1.1)$
        (blue dot in Fig.~\ref{Bifur_3cells_ID}). Compare with
        Fig.~\ref{FlexPDE_ID3cellsStable}. Third row:
        $(D,\tau)=(0.4,0.6)$ (red dot in the left panel of
        Fig.~\ref{fig:defective_middle}). Compare with
        Fig.~\ref{FlexPDE_Def_3cells}.  Fourth row: $(D,\tau)=(0.5,3)$
        (blue dot in the left panel of
        Fig.~\ref{fig:defective_right}).  Compare with
        Fig.~\ref{FlexPDE_ID3cells_d10p2}. Although there is a phase
        shift between the ODE and full PDE results for intracellular
        oscillations due to different initial conditions used, the ODE
        system \eqref{reducedODE} captures well the amplitude and
        period of intracellular oscillations observed in the full PDE
        simulations.}
\label{fig:odes_ring-center}
\end{figure}

Next, we examine whether the ODE system \eqref{reducedODE} of
Proposition \ref{prop:dyn}, derived under the assumption of large bulk
diffusivity $D={\mathcal O}(\nu^{-1})\gg 1$, can still be used to
reliably approximate the intracellular dynamics observed in the full
FlexPDE simulation results of \eqref{DimLess_bulk}, performed for
${\mathcal O}(1)$ values of $D$, in Figs.~\ref{FlexPDE_ID3cells},
\ref{FlexPDE_ID3cellsStable}, \ref{FlexPDE_Def_3cells} and
\ref{FlexPDE_ID3cells_d10p2}.  In Fig.~\ref{fig:odes_ring-center} we
show that the numerical results computed from the ODE system
\eqref{reducedODE} compare surprisingly well with the full PDE
simulations with respect to the amplitude and period of intracellular
oscillations (first, third and fourth rows of
Fig.~\ref{fig:odes_ring-center}) and the prediction of a linearly
stable steady-state (second row of
Fig.~\ref{fig:odes_ring-center}). In using the ODE system
\eqref{reducedODE} we calculated $D_0$ as $D_0=D\nu$, where
$\nu={-1/\log\varepsilon}$ with $\varepsilon=0.05$. We remark that the
beating-type oscillations observed in
Fig.~\ref{FlexPDE_Def3cellsAsync} for the very small value $D=0.05$
are not captured by the ODE system \eqref{reducedODE}.  Moreover, we
emphasize that the simpler ODE system corresponding to the well-mixed
limit $D\to \infty$ as given in \eqref{ode:well_mixed}, and which was
used in \cite{jia2016} and \cite{smjw_quorum} for studying
quorum-sensing behavior, does not reliably approximate the
intracellular oscillations for the values of bulk diffusivity given in
Figs.~\ref{FlexPDE_ID3cells}, \ref{FlexPDE_ID3cellsStable},
\ref{FlexPDE_Def_3cells}, \ref{FlexPDE_Def_3cells}, and
\ref{FlexPDE_ID3cells_d10p2}.

\subsection{A defective center cell: different Sel'kov kinetics}
\label{sec:alpha}

Next, we study how the HB boundaries in the $(D,\tau)$ plane are
altered by varying a Sel'kov kinetic parameter of the center cell. In
Fig.~\ref{Bifur_3cells_alpha} we plot the HB boundaries for a ring and
center cell pattern of $m=3$ cells, where the identical cells on the
ring have parameters as in \eqref{Selkov_para}, while the Sel'kov
parameter $\alpha_3$ for the defective center cell is either
$\alpha_3=0.86$ (red curves), $\alpha_3=0.96$ (blue curves), or
$\alpha_3=0.9$ (same as in the left panel of
Fig.~\ref{Bifur_3cells_first}). The HB boundaries in the right panel
of Fig.~\ref{Bifur_3cells_alpha} show a zoom for smaller values of $D$
than the figure in the left panel. In these figures, the dashed and
the heavy solid curves are for the in-phase modes computed from
\eqref{cent:HB_sync} with $(+)$ and $(-)$, respectively, while the
thin solid curve is for the anti-phase mode computed from
\eqref{cent:HB_async}. We observe that the HB boundary for the
anti-phase mode is independent of the Sel'kov kinetic parameter
$\alpha_3$ for the center cell. This follows from the facts that the
steady-state solution to \eqref{Linear_reduceSys} for the ring cells
and its Jacobian $J_1$ in \eqref{cent:KcKm_jac} do not depend on
$\alpha_3$. As a result, the anti-phase HB boundary, computed from
\eqref{cent:HB_async}, is independent of $\alpha_3$.

\begin{figure}[htbp]
  \centering
	\makebox{
    \raisebox{0.5ex}{}
    \includegraphics[width=0.40\textwidth,height=4.5cm]{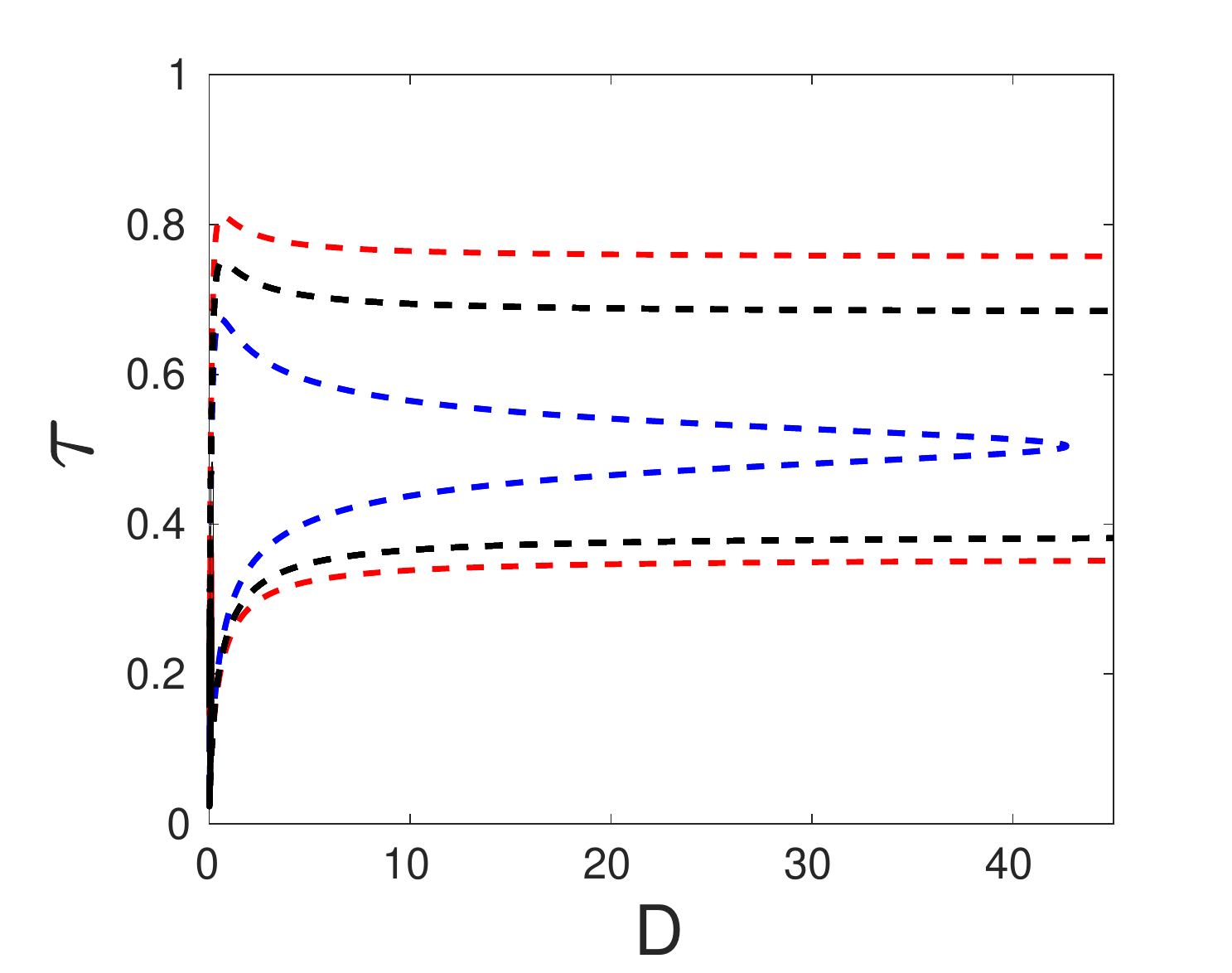}
    \phantomsubcaption
    \label{Bifur_3cells_alphaA}
  }  \qquad
  \makebox{
    \raisebox{0.5ex}{}
    \includegraphics[width=0.40\textwidth,height=4.5cm]{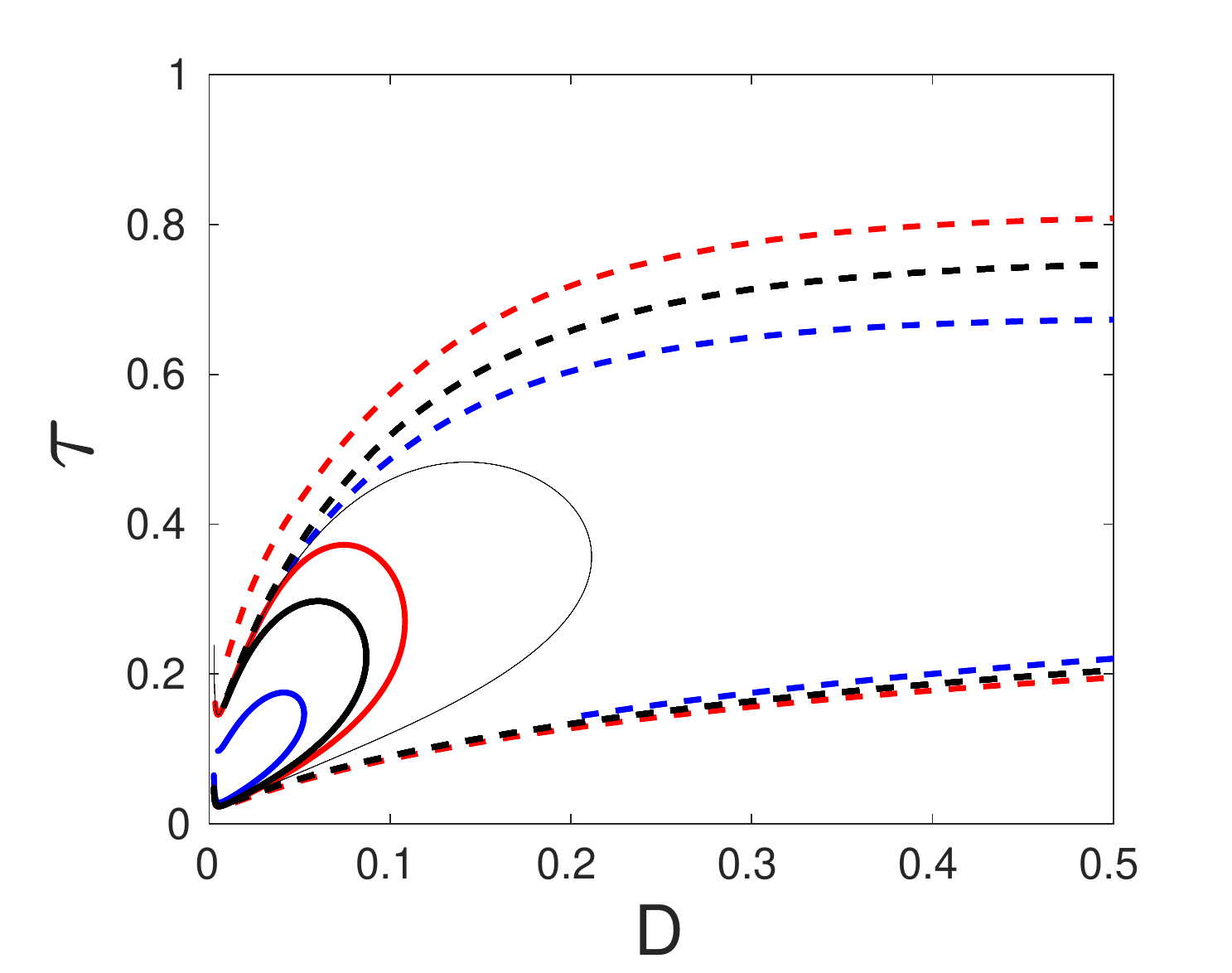}
    \phantomsubcaption
    \label{Bifur_3cells_alphaB}
  }
  \vspace*{-1ex}
  \caption{HB boundaries in the $\tau$ versus $D$ plane for
      a ring and center cell pattern of $m=3$ cells, where the Sel'kov
      kinetic parameter $\alpha_3$ for the center cell is varied.  The
      ring cells are centered at $(\pm 5,0)$ and, with the exception
      of the Sel'kov kinetic parameter for the center cell, all three
      cells have parameters as in \eqref{Selkov_para}. Left panel:
      $\alpha_3 = 0.86$ (red curves), $\alpha_3 = 0.9$ (black curves),
      $\alpha_3 = 0.96$ (blue curves). Right panel: a zoomed-in
      version of the left panel for smaller values of $D$. In both
      panels, the dashed and heavy solid curves are for the in-phase
      modes computed from \eqref{cent:HB_sync}, while the thin solid
      curve is for the anti-phase mode computed from
      \eqref{cent:HB_async}. Each mode is linearly unstable in their
      respective lobes, while linearly stable steady-state solutions
      exists outside the union of the lobes. The anti-phase HB
      boundaries are independent of $\alpha_3$, and so are plotted on top
      of each other.}
  \label{Bifur_3cells_alpha}
\end{figure}

From the left panel of Fig.~\ref{Bifur_3cells_alpha} we observe that
as $\alpha_3$ increases the parameter regions in the $(D,\tau)$ plane
where intracellular oscillations occur decreases. In particular, the
in-phase $(+)$ instability lobe that was open for $\alpha_3 = 0.86$
becomes closed when $\alpha_3 = 0.96$, thereby precluding the
possibility of intracellular oscillations when $D$ is sufficiently
large. To interpret this result, we observe from the middle panel of
Fig.~\ref{fig:selkov} that, as $\alpha_3$ increases, the center cell
becomes less activated, with the parameters drifting further from the
HB boundary for the uncoupled cell. As a result, it becomes more
difficult to trigger in-phase intracellular oscillations for a group
of coupled cells as $\alpha_3$ increases. Overall,
Fig.~\ref{Bifur_3cells_alpha} does show that rather small increases or
decreases in a parameter value of the nonlinear Sel'kov kinetics for
one specific cell can either extinguish or trigger intracellular
oscillations for an entire group of cells.  The corresponding
eigenvectors for selected points on the bifurcation diagrams in
Fig.~\ref{Bifur_3cells_alpha} are shown in Table~\ref{Table:Def_alpha}
for $\alpha_3=0.86$ and $\alpha_3=0.96$.

\begin{table}[!httbp]
\centering
\begin{tabular}{ c | c || c |  c | c || c | c | c|} \hline
  \rowcolor{LightCyan}\hline   \rowcolor{LightCyan}\hline
  &  &
   \multicolumn{3}{c||}{$\qquad\alpha_3=0.86$} &
    \multicolumn{3}{c|}{$\qquad\alpha_3=0.96$}  \\ 
  \rowcolor{LightCyan} 
\multirow{-2}*{mode} & \multirow{-2}*{j} &    $(D,\tau)$ &  $(\mbox{Re}c_j,\mbox{Im}c_j)$  & $\theta_j \,(\text{rad})$ &  $(D,\tau)$ &  $(\mbox{Re}c_j,\mbox{Im}c_j)$  & $\theta_j \,(\text{rad})$ \\ \hline \hline \rowcolor{Cyan}
    &   1      &  & $(0.539,-0.080)$  & $6.14$ &    & $(0.639,0)$ & 0   \\ \rowcolor{Cyan}
In-phase ($+$) & 2 & $ (0.233,0.744) $ & $(0.539,-0.080)$  & $6.14$ & $(0.475,0.672)$  & $(0.639,0)$   &  0 \\ \rowcolor{Cyan}
(heavy solid)     &  3    &  &$(0.638,0)$  & $0$ &    & $(0.427,-0.0118)$   & $6.256$ \\ 
\hline   \rowcolor{Gray}
     &   1    &   & $-(0.461,0.0415)$ & $3.23$ &   & $(-0.254,0.0428)$ & $2.974$ \\ \rowcolor{Gray}
    In-phase ($-$) & 2 & $ ( 0.108,0.284) $ & $-(0.461,0.0415)$  & $3.23$ & $( 0.0492,0.168)$   & $(-0.254,0.0428)$
    &   $2.974$   \\ \rowcolor{Gray}
 (dashed curve)      &  3   &  & $(0.756,0)$ & $0$ &    & $(0.932,0)$    & $0$  \\ 
\hline  \rowcolor{Cyan}
     &   1    &  & $(0.707,0)$ & $0$ &    & $(0.707,0)$   & 0 \\ \rowcolor{Cyan}
Anti-phase  & 2 & $ (0.211,0.365) $ & $-(0.707,0)$ & $\pi$ & $ (0.211,0.365) $   &  $-(.707,0)$  & $\pi$ \\ \rowcolor{Cyan}
 (thin solid)   &  3   &  & $(0,0)$  & $0$ &    & $(0,0)$  & $0$  \\ 
  \hline
\end{tabular}
\caption{Real and imaginary parts of the eigenvector
    $\pmb{c}$ of the GCEP matrix ${\mathcal M}(\lambda)$ in
    \eqref{M_M0}, computed for a few points on the HB boundaries in
    Fig.~\ref{Bifur_3cells_alpha}. The cells on the ring are identical
    and located at $(\pm 0.5, 0)$ with parameters as given in
    \eqref{Selkov_para}. The center cell has the same parameters, with
    the exception that it has a Sel'kov kinetic parameter $\alpha_3$
    different than those on the ring.  Middle three columns:
    $\alpha_3=0.86$. Last three columns: $\alpha_3=0.96$.}
 \label{Table:Def_alpha}
\end{table}

\subsection{Diffusion-sensing: effect of cell locations}
\label{sec:ring}

For a ring and center cell pattern, with cells centered at
$(\pm r_0,0)$ and $(0,0)$, we study how the HB boundaries in the
$(D,\tau)$ plane depend on the ring radius $r_0$. In the first row of
Fig.~\ref{fig:ring_radius} we show the HB boundaries for the in-phase
and anti-phase modes when $r_0=0.25$, $r_0=0.5$ (same as
Fig.~\ref{Bifur_3cells_ID}), and $r_0=0.75$ for the case where the
cells are all identical with permeabilities $d_1=0.8$ and
$d_2=0.4$. Similar plots for the three values of $r_0$ are shown in
the second row of Fig.~\ref{fig:ring_radius} for the case where the
center cell is now defective with $d_{13} = 0.4$ and $d_{23} =0.2$
(Fig.~\ref{Bifur_3cells_Def_r0p50} for $r_0=0.5$ is the same as in
Fig.~\ref{fig:defective_middle}).

When the cells are all identical, we observe from the top row in
Fig.~\ref{fig:ring_radius} that the regions of instability for the
anti-phase mode and the in-phase $(-)$ mode shrink noticeably as $r_0$
decreases. In contrast, the instability region for the in-phase $(+)$
mode is relatively insensitive to changes in $r_0$.  To qualitatively
explain this observation, we examine the eigenvector of the GCEP
matrix given in Table~\ref{Table:ID_RCH_3cells} at points on the HB
boundaries for the three modes. From this table, for the in-phase
$(-)$ mode only the two ring cells oscillate in phase while the center
cell has larger amplitude oscillations that are roughly $180^{\circ}$
out of phase. For the anti-phase $(-)$ mode, the center cell is
quiescent while the ring cells oscillate $180^{\circ}$ out of
phase. Finally, for the dominant in-phase $(+)$ mode, the three
identical cells are synchronized with very similar amplitudes and
phases. Intuitively, as the ring radius $r_0$ decreases the cells
become more clustered and, as a result, intracellular dynamics can
still be synchronized even for smaller values of the bulk diffusivity
$D$. As a result, when $r_0$ is small, the anti-phase and in-phase
$(-)$ instability lobes, in which the center and ring cells are not
synchronized, should exist only for very small values of $D$ (see
Fig.~\ref{Bifur_3cells_ID_r0p25}). When $D$ is very small,
communication between cells that are close can still be rather
weak. Notice that for $r_0=0.75$, where the cells are farther apart,
the anti-phase lobe in Fig.~\ref{Bifur_3cells_ID_r0p75} exists for
larger values of $D$ than when $r_0=0.25$ or $r_0=0.5$. For
$r_0=0.75$, the ring cells can oscillate, maintaining a quiescent
center cell, provided that the bulk chemical is not washed away,
i.e.~$D$ is not too large, owing to the fact that each of the two ring
cells are now relatively close to their images across the domain
boundary due to the reflecting boundary condition imposed. However,
when $r_0$ is large, the ring cells are far from each other and so,
within the anti-phase instability lobe, their oscillations are
$180^{\circ}$ out of phase.

\begin{figure}[!h]
  \centering
    \begin{subfigure}[b]{0.31\textwidth}
  \includegraphics[width=\textwidth,height=4.3cm]{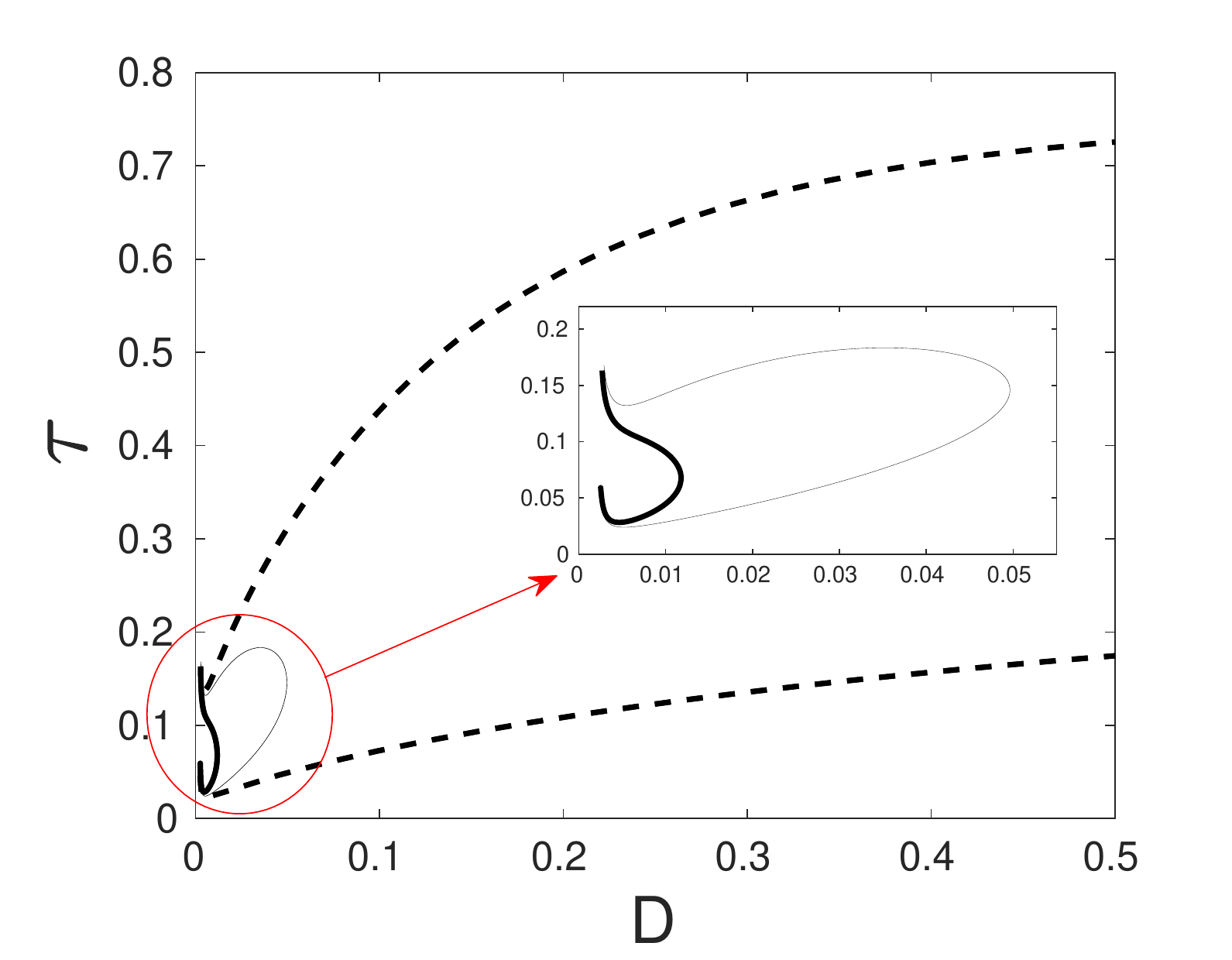}
        \subcaption{$r_0=0.25$}
            \label{Bifur_3cells_ID_r0p25}
    \end{subfigure}
    \begin{subfigure}[b]{0.31\textwidth}  
      \includegraphics[width=\textwidth,height=4.3cm]{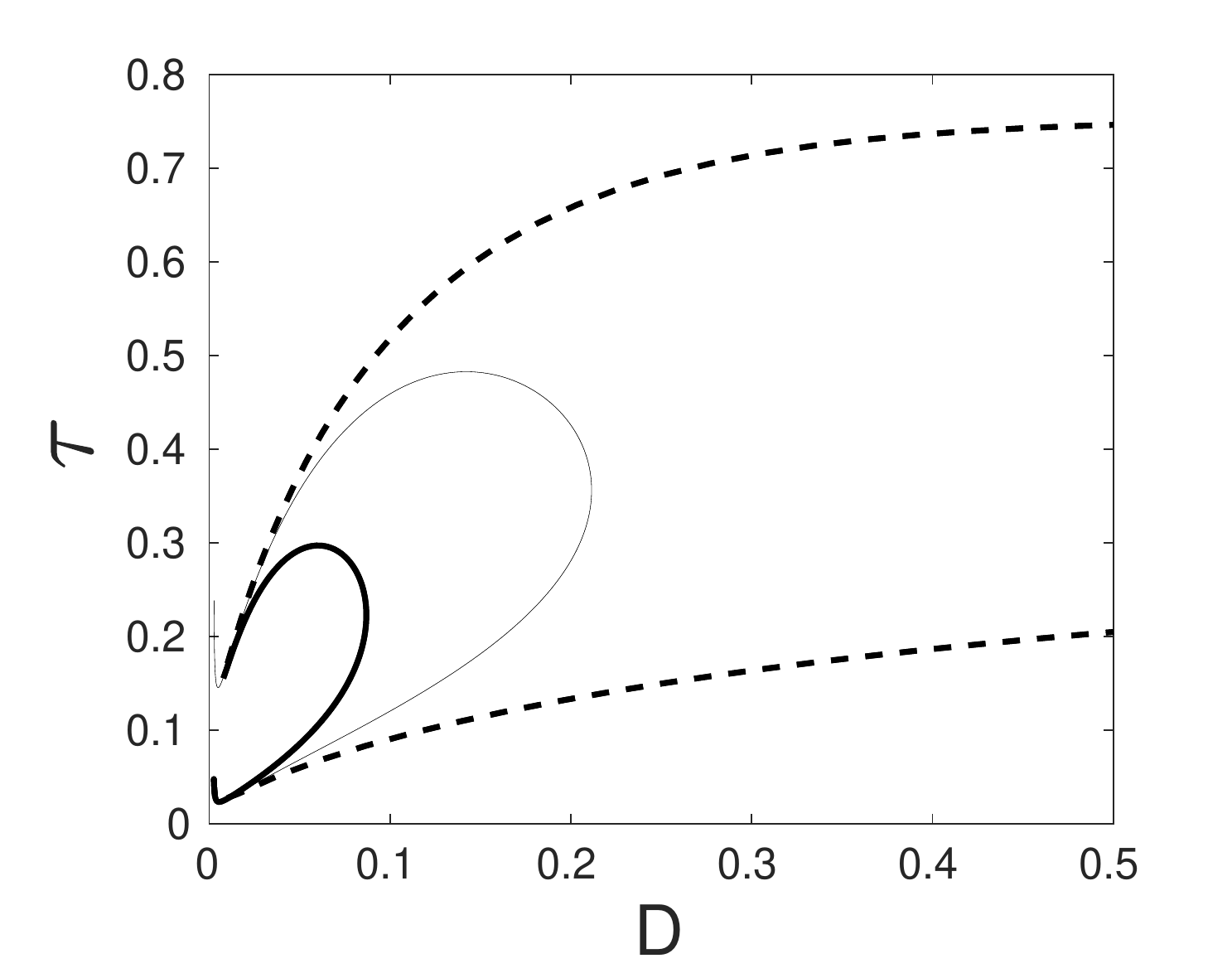}
         \subcaption{$r_0=0.5$}
    \label{Bifur_3cells_ID_r0p5}
    \end{subfigure}
    \begin{subfigure}[b]{0.31\textwidth}
      \includegraphics[width=\textwidth,height=4.3cm]{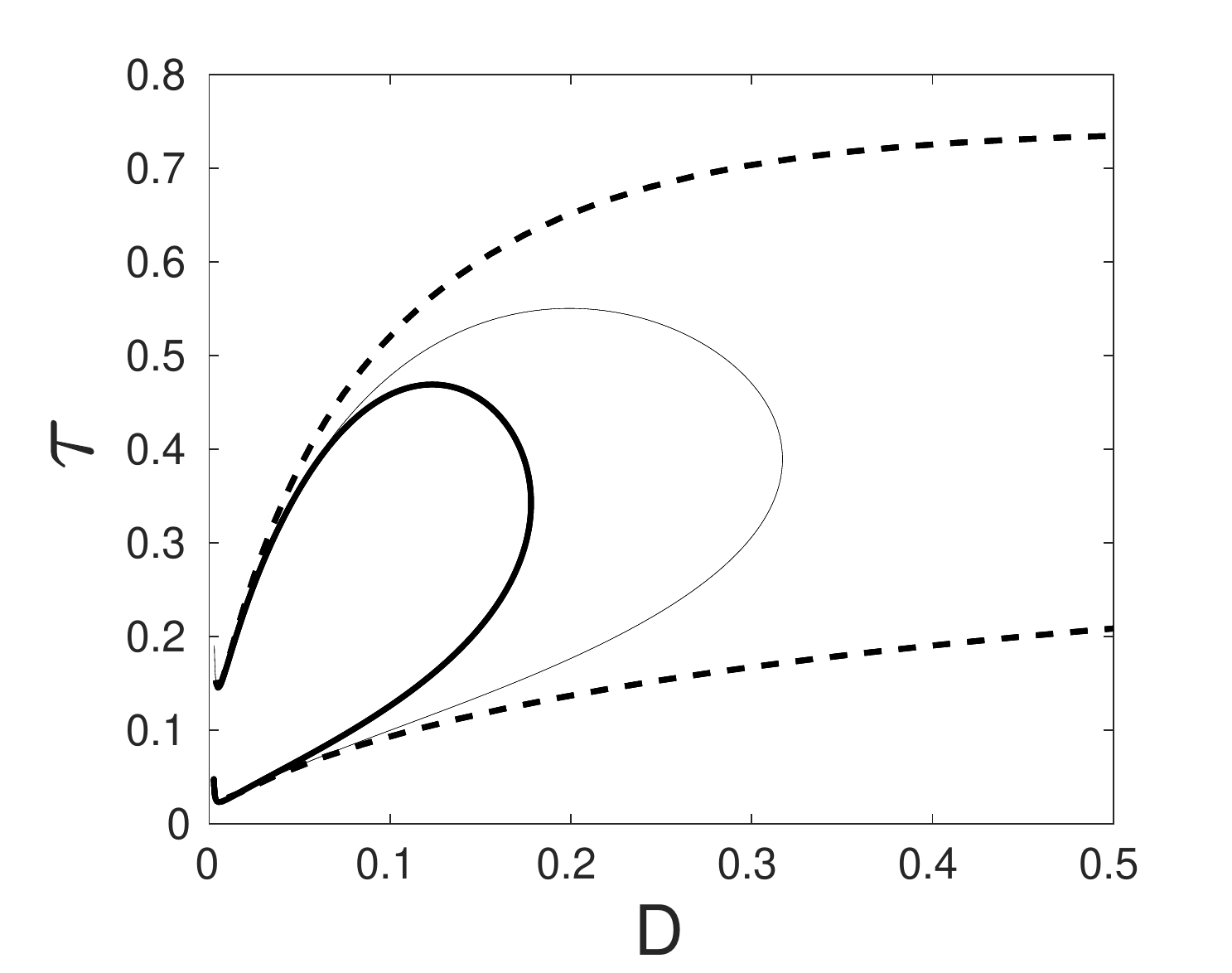}
      \subcaption{$r_0=0.75$}
      \label{Bifur_3cells_ID_r0p75}
    \end{subfigure}
    \begin{subfigure}[b]{0.31\textwidth}  
 \includegraphics[width=\textwidth,height=4.3cm]{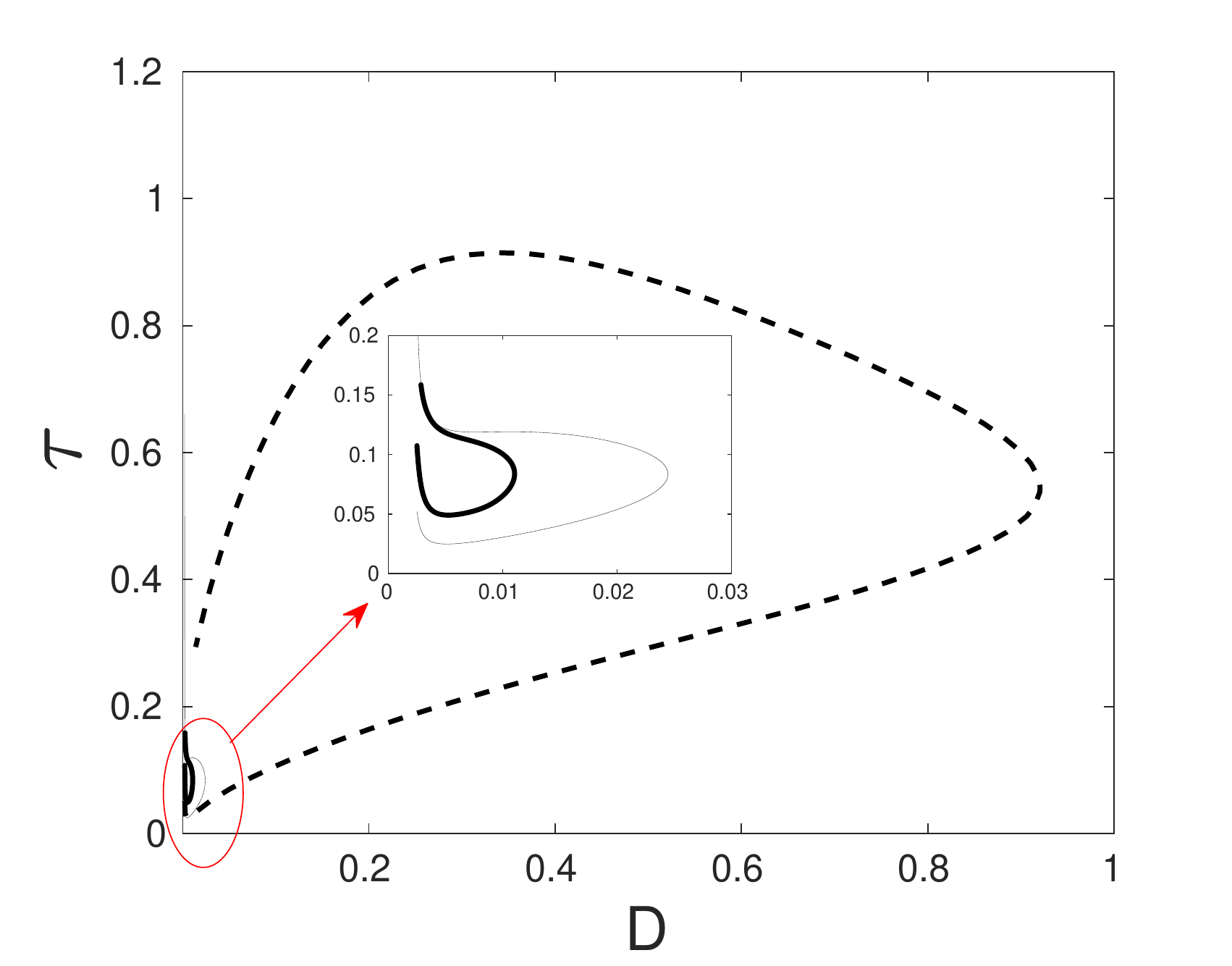} 
   \subcaption{$r_0=0.25$}
    \label{Bifur_3cells_Def_r0p25}
    \end{subfigure}
    \begin{subfigure}[b]{0.31\textwidth}
      \includegraphics[width=\textwidth,height=4.3cm]{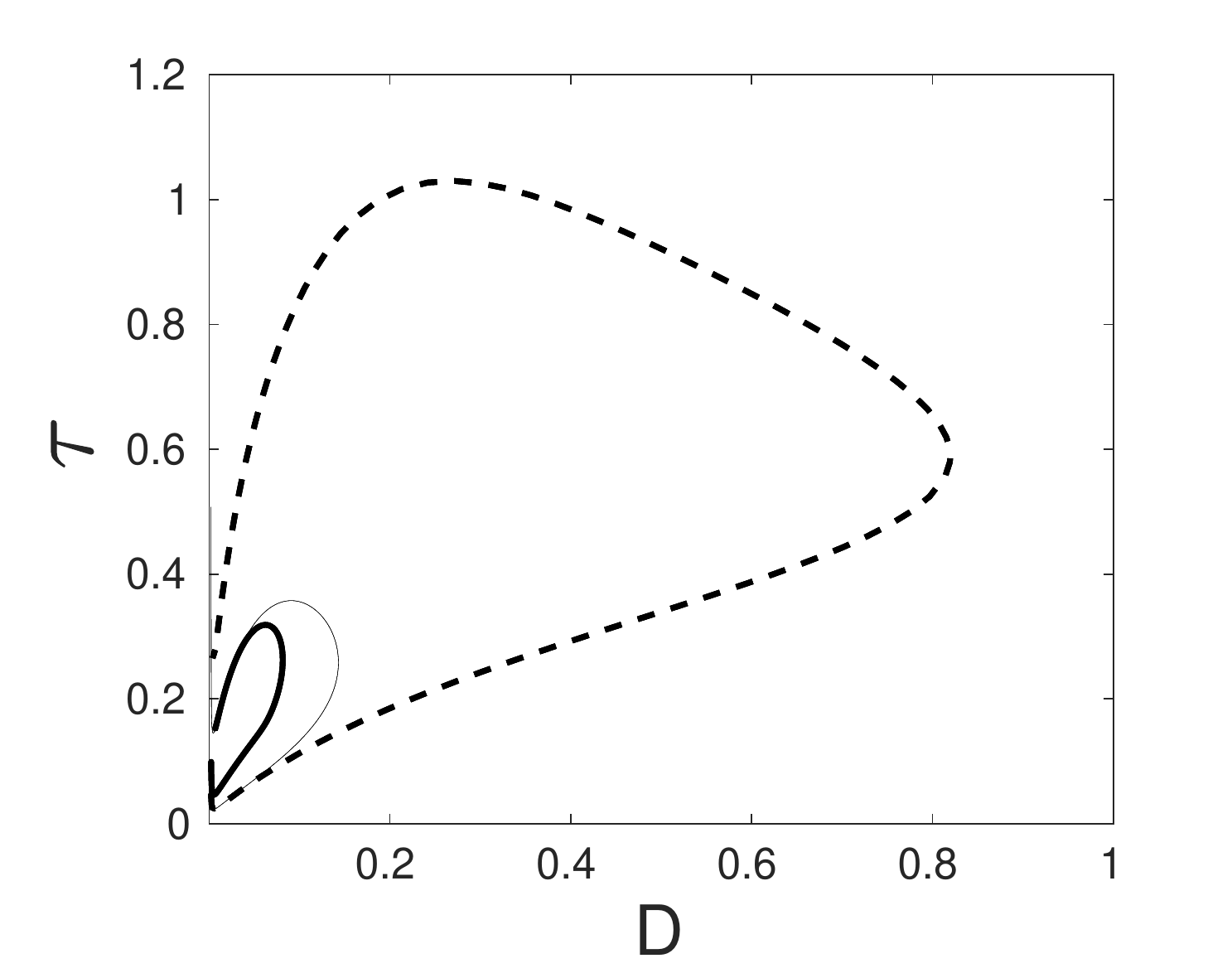}
             \subcaption{$r_0=0.5$}
        \label{Bifur_3cells_Def_r0p50}
    \end{subfigure}
    \begin{subfigure}[b]{0.31\textwidth}  
      \includegraphics[width=\textwidth,height=4.3cm]{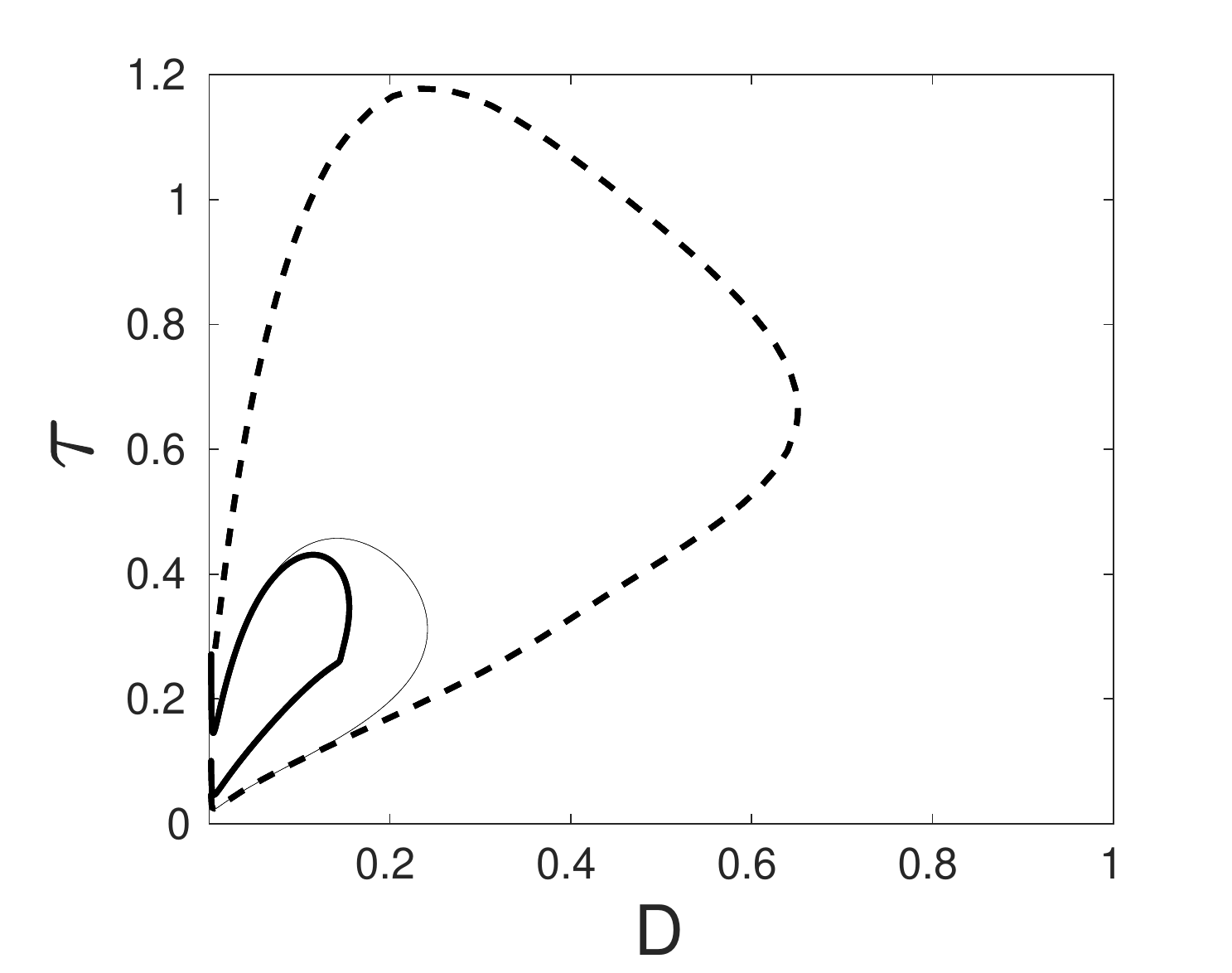}
                             \subcaption{$r_0=0.75$}
                \label{Bifur_3cells_Def_r0p75}
    \end{subfigure}
    \caption{HB boundaries in the $\tau$ versus $D$ plane
        for a ring and center hole configuration of $m=3$ cells, for
        three different ring radii $r_0$ as indicated. The Sel'kov
        parameters and cell radius are as in \eqref{Selkov_para}. Top
        row: identical cells with permeabilities $d_1=0.8$ and
        $d_2=0.2$. Bottom row: center cell is defective with
        $d_{13}=0.4$ and $d_{23}=0.2$. The dashed and heavy solid
        curves are the HB boundaries for the in-phase $(+)$ and $(-)$
        modes, respectively, as computed from
        \eqref{cent:HB_sync}. The thin solid curve is the HB boundary
        for the anti-phase mode computed from
        \eqref{cent:HB_async}. Each mode is unstable in their
        respective lobes, while linearly stable steady-state solutions
        occur outside the union of the lobes.}
\label{fig:ring_radius}
\end{figure}

The second row of Fig.~\ref{fig:ring_radius} shows similar results for
the case where the center cell is defective, with permeabilities
$d_{13} = 0.4$ and $d_{23} = 0.2$, corresponding to a reduced rate of
influx into the center cell. Similar to the results for identical
cells presented in the first row of Fig.~\ref{fig:ring_radius}, as
$r_0$ decreases the region of instability for the in-phase $(-)$ and
the anti-phase modes shrink. Moreover, we observe that the maximum
extent in $D$ of the in-phase $(+)$ lobes, in which all the cells are
essentially synchronized in amplitude and phase, decreases as $r_0$
increases. This is because when the cells are farther apart, a
relatively smaller value of the bulk diffusivity can lead to a washing
out of the bulk signal, which thereby weakens the communication
between the three cells and precludes synchronization.

\setcounter{equation}{0}
\setcounter{section}{4}
\section{Asymptotic analysis for a large population of cells}\label{sec:LargePopulation}

As the number of cells increases it becomes increasingly more
challenging numerically to determine the stability boundaries in
parameter space from the root-finding condition
$\mbox{det}{\mathcal M}(\lambda)=0$ given an arbitrary cell
configuration $\pmb{x}_j$ with arbitrary permeability parameters
$d_{1j}$ and $d_{2j}$ for $j=1,\ldots,m$. In this section we provide a
simpler approach to determine the stability boundaries from the GCEP
\eqref{full_gcep} for the large bulk diffusion parameter regime where
$D={\mathcal O}(\nu^{-1})$. To isolate only the effect of different
cell configurations, such as shown in Fig.~\ref{CellConfig}, as well
as the effect of different cell permeabilities, in our analysis below
we will assume that the Sel'kov kinetic parameters are all
cell-independent, i.e.~that $\alpha_j=\alpha$, $\mu_j=\mu$ and
$\zeta_j=\zeta$ for $j=1,\ldots,m$. Our analysis is easily extended to
remove this assumption.

\begin{figure}[!ht]
    \centering
    \begin{subfigure}[b]{0.2\textwidth}
        \includegraphics[width=\textwidth]{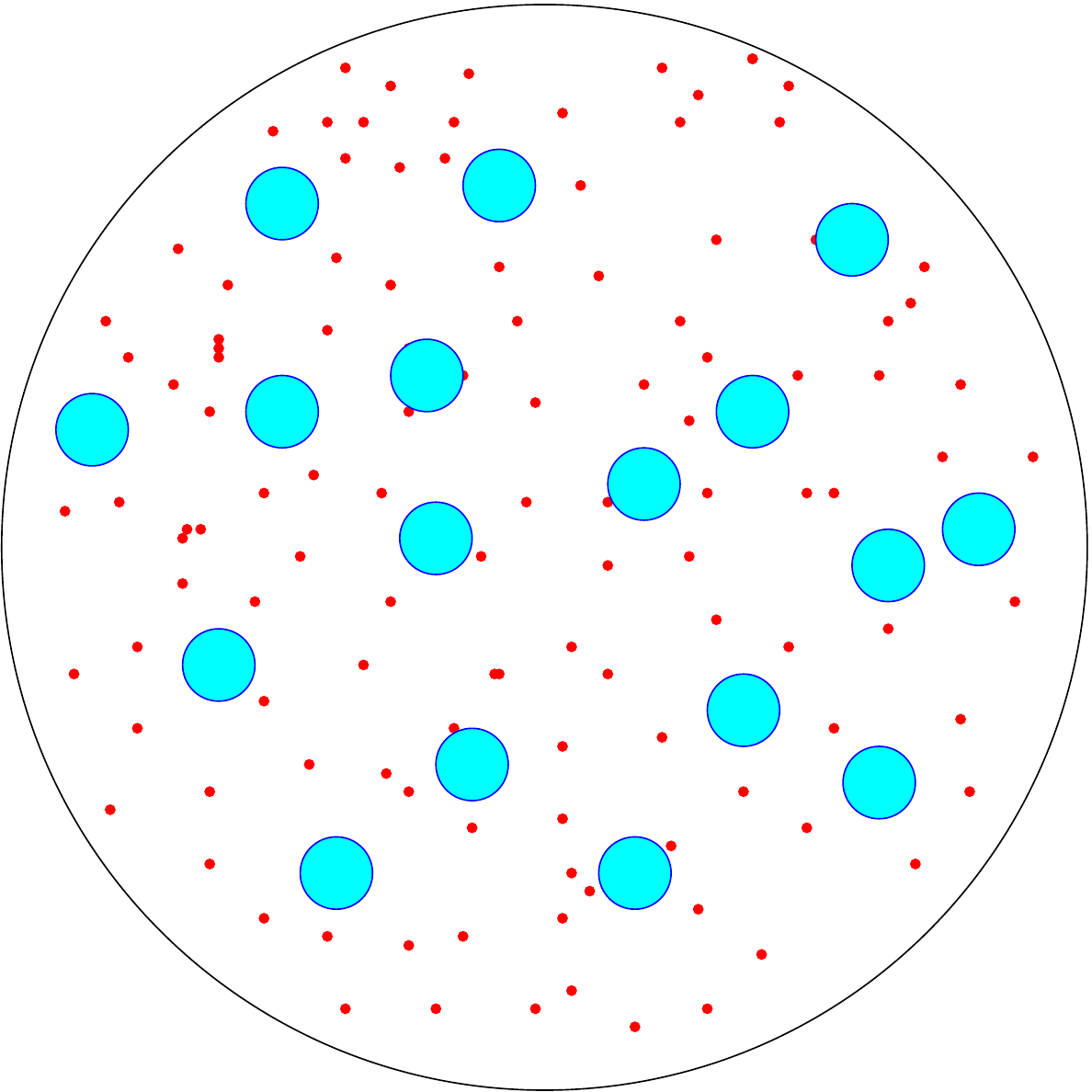}
        \caption{Randomly placed cells}
        \label{Ringcof}
    \end{subfigure}
    ~\qquad \qquad \qquad
    \begin{subfigure}[b]{0.2\textwidth}  
        \includegraphics[width=\textwidth]{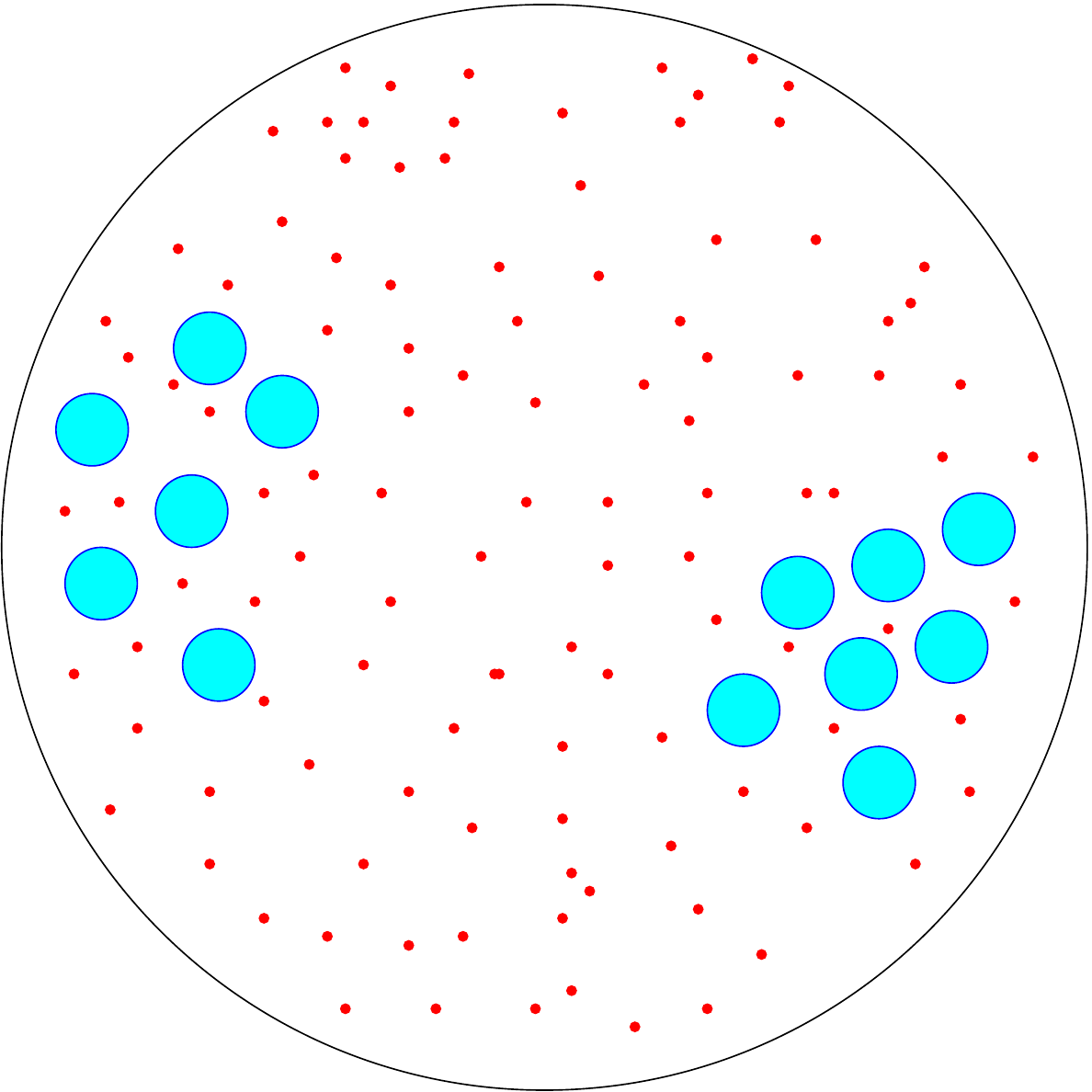}
        \caption{Two clusters of cells}
        \label{Chimeraconf}
    \end{subfigure}
    \caption{Schematic showing different cell configurations
      in the unit disk. The cells are represented by smaller disks (in
      cyan) and the diffusing bulk species corresponds to the red
      dots. Left panel: randomly placed cells. Right panel: two spatially
      clustered groups of cells.} \label{CellConfig}
\end{figure}

We first recall from \eqref{ss:selkov_1} of
\cref{LinearAnalysisDefect} that, for $\varepsilon\to 0$, the
steady-state solution with Sel'kov kinetics is determined in terms of
the solution $\pmb{\mathcal{A}} = (A_1, \dots, A_m)^T$ to the linear
system
\begin{subequations}\label{large:steady-state}
\begin{equation}\label{large:selkov_1}
\left(I + 2\pi \nu \mathcal{G}  + \nu\,D\,P_1 +
  \frac{2\pi \nu D}{\tau }  P_2 \right) {\mathcal A}= - \mu \nu
P_2 \, \pmb{e} \,; \qquad
P_1 = \mbox{diag}\Big{(} \frac{1}{d_{11}}, \ldots, \frac{1}{d_{1m}} \Big{)}
\quad  P_2 = \mbox{diag}\Big{(}\frac{d_{21}}{d_{11}}, \ldots,
\frac{d_{2m}}{d_{1m}}\Big{)} \,,
\end{equation}
where $\pmb{e} = (1,\dots,1)^T$ and $\mathcal{G}$ is the Green's
interaction matrix defined in \eqref{GreensMatrix}. In terms of the
solution $\pmb{\mathcal{A}}$ to \eqref{large:selkov_1}, we obtain from
\eqref{ss:selkov_2} that the steady-state for the intracellular species
$\pmb{u}_{e} = (u^1_{ej}, u^2_{ej} )^T$ in each cell is 
\begin{equation}\label{large:selkov_2}
  u^1_{ej} = \mu  + \frac{2 \pi D}{\tau} A_j \,, \qquad \mbox{and} \qquad
  u^2_{ej} = \frac{\mu}{\alpha + (u^1_{ej})^2}\,.
\end{equation}
\end{subequations}

Since this simple steady-state construction is accurate to all orders
in $\nu$ for any $D>0$, we will not seek an approximation to it
valid for the regime $D={\mathcal O}(\nu^{-1})$. Therefore, in the
GCEP matrix ${\mathcal M}(\lambda)$ given in \eqref{Global_System},
the diagonal matrix ${\mathcal K}(\lambda)$, as defined in
\eqref{K_mat_entry}, will be evaluated at the solution to
\eqref{large:steady-state}. With Sel'kov kinetics, this yields
\begin{subequations}\label{large:Kmatrix}
\begin{equation}\label{large:Kmatrix_1}
  \mathcal{K} = \mathcal{K}(\lambda) \equiv
  \mbox{diag} \Big{(}\mathit{K}_1,\dots, \mathit{K}_m  \Big{)}\,, \qquad
  \mbox{where} \qquad \mathit{K}_j = \frac{\lambda + \mbox{det}(J_j)}
  {\lambda^2 - \lambda \, \mbox{tr}(J_j) + \mbox{det}(J_j) } \,,
\end{equation}
where $J_j$ is the Jacobian matrix for the local Sel-kov kinetics
at the $j^{\text{th}}$ cell, for which
\begin{equation}\label{large:Kmatrix_jac}
  \mbox{det}(J_j) = \zeta \left(\alpha + (u_{e j}^{1})^2\right)\,, \qquad
   \mbox{tr}(J_j) = \frac{ \left[2 \mu u_{e j}^{1} -
      \left(\alpha + (u_{e j}^{1})^2\right) - \zeta
      \left(\alpha + (u_{e j}^{1})^2\right)^2\right]}{\alpha +
    (u_{e j}^{1})^2}\,.
\end{equation}
\end{subequations}

Our goal in this subsection of determining a more tractable formula
for determining the roots of $\det\mathcal{M}(\lambda)=0$ when $m$ is
large and in the limit $D = {D_0/\nu}$, where
$\nu = {-1/\log\varepsilon}$ and $D_0={\mathcal O}(1)$, is based on
approximating the eigenvalue-dependent Green's matrix
${\mathcal G}_{\lambda}$ in \eqref{Global_System}, while retaining the
full $\mathcal{K}$ matrix, as defined in \eqref{large:Kmatrix}.

To do so, we readily calculate from  \eqref{EigGreen} that when
$D = {D_0/\nu}\gg 1$ we have (see \cite{jia2016})
\begin{equation}\label{large:G-matrix}
  \mathcal{G}_{\lambda} =  \frac{ D_0}{ \nu |\Omega| (1 + \tau \lambda)} +
  \mathcal{G}_{0} + \mathcal{O}(\nu)\,.
\end{equation}
Here $\mathcal{G}_{0}$ is the Neumann Green's matrix of
\eqref{WM_NeumannGmatrix} given in terms of the Neumann Green's
function $G_0$ of \eqref{WM_NUEGREEN}, which is known analytically for
the the unit disk in \eqref{gr:gmrm}. By using
\eqref{large:G-matrix} in \eqref{Global_System}, we find that the GCEP
matrix reduces to
\begin{subequations} \label{large:GCEP}
\begin{equation}\label{large:M}
\mathcal{M}(\lambda) = {\mathcal M}_{\infty}(\lambda)   +
\mathcal{O}(\nu^2)\,,
\qquad \mbox{where} \qquad {\mathcal M}_{\infty} \equiv
{\mathcal B} + \gamma E + 2\pi \nu {\mathcal G}_0 \,.
\end{equation}
In \eqref{large:M} the scalar $\gamma$, the rank-one matrix $E$ and the
diagonal matrix ${\mathcal B}$ are defined by
\begin{equation}\label{large:M_sub}
  \gamma =\gamma(\lambda) \equiv\frac{2 \pi m D_0 }
  {(1 + \tau \lambda)|\Omega|} \,, \qquad
  E \equiv \frac{1}{m} \pmb{e}\,\pmb{e}^T \,, \qquad
  {\mathcal B} = {\mathcal B}
  (\lambda) \equiv I + D_0P_1 + \frac{2\pi D_0}{\tau} P_2\,
  \mathcal{K}(\lambda) \,.
\end{equation}
Neglecting terms of order ${\mathcal O}(\nu^2)$, in the limit
$D={D_0/\nu}\gg 1$ we conclude that the GCEP \eqref{full_gcep} reduces
to finding values of $\lambda$ for which there is a nontrivial solution
$\pmb{c}$ to
\begin{equation}\label{large:gcep_mat}
  \mathcal{M}_{\infty}(\lambda) \pmb{c} = \pmb{0}\,, \qquad
  \mbox{which occurs iff} \qquad \det \mathcal{M}_{\infty}(\lambda)=0 \,.
\end{equation}
\end{subequations}
We remark that in our asymptotic reduction \eqref{large:GCEP} of the
GCEP \eqref{full_gcep}, the cell configuration enters through the
diagonal matrix ${\mathcal K}$, as defined in \eqref{large:Kmatrix},
and at order $\nu$ from the Neumann Green's matrix ${\mathcal G}_0$ in
\eqref{large:M}.

To analyze  \eqref{large:GCEP} we first observe from
\eqref{large:M_sub} and \eqref{large:Kmatrix_1} that
\begin{equation}\label{large:bdef}
  {\mathcal B} = \mbox{diag}(b_1,\ldots,b_m)\,, \qquad \mbox{where} \qquad
  b_j = \frac{(d_{1j}+D_0)}{d_{1j}} \left[ 1 + \frac{\eta_j}{\tau}
    \mathit{K}_j \right] \,, \qquad \eta_j \equiv \frac{2\pi d_{2j} D_0}{d_{1j}+D_0}\,,
  \quad j=1,\ldots, m \,.
\end{equation}
The key advantage of our asymptotic reduction is that by using the
matrix structure in \eqref{large:GCEP} we can readily calculate a
tractable analytical expression for
$\mbox{det}{{\mathcal M}_{\infty}(\lambda)}$, which can be used in
implementing the root-finding condition in \eqref{large:gcep_mat}
numerically. Our first result for \eqref{large:GCEP} in this direction
is as follows:

\vspace{0.2cm}
\begin{prop}\label{lemma:large} Suppose that $b_j\neq 0$ for $j=1\,,
  \ldots,m$ so that ${\mathcal B}$ is invertible. Then,
\begin{equation}\label{large:detB}
  \mbox{det}{\mathcal M}_{\infty}(\lambda) =
  \mbox{det}\left({\mathcal B} + \gamma E + 2\pi \nu {\mathcal G}_0\right) =
  \left( \prod_{i=1}^{m} b_i \right) \kappa_{1\varepsilon}
  \left(1+{\mathcal O}(\nu)\right)\,,
\end{equation}
where
\begin{equation}\label{large:kappa_1eps}
  \kappa_{1\varepsilon} \equiv 1 + \frac{\gamma}{m}
  \pmb{e}^{T}\pmb{v}_1 +  2 \pi \nu \,
  \left(\frac{\pmb{v}^T_1\, \mathcal{G}_0 \,\pmb{v}_1}{\pmb{e}^T\, \pmb{v}_1}
    \right) +
  \mathcal{O}(\nu^2)\,, \qquad \mbox{with} \qquad
  \pmb{v}_1^{T} \equiv \left( {1/b_1},\ldots,{1/b_m}\right) \,.
\end{equation}
\end{prop}

\begin{proof}
Since ${\mathcal B}$ is a diagonal matrix and invertible by assumption, we have
\begin{equation}\label{large:proof1}
  \mbox{det} \left( {\mathcal B} + \gamma E  + 2\pi \nu  {\mathcal G}_0
  \right) = \mbox{det}({\mathcal B}) \det\left( I + \gamma
    {\mathcal B}^{-1}E + 2\pi \nu {\mathcal B}^{-1} {\mathcal G}_0 \right)\,,
\end{equation}
where $\mbox{det}{\mathcal B} = \prod_{i=1}^{m} b_i$. The second term in
\eqref{large:proof1} is calculated by multiplying all the eigenvalues 
$\kappa_{\varepsilon}$ of
\begin{equation}\label{large:aux_eig}
  \left( I + \gamma {\mathcal B}^{-1}E + 2\pi \nu {\mathcal B}^{-1}
    {\mathcal G}_0\right)\, \pmb{v}_{\varepsilon} = \kappa_{\varepsilon}
  \pmb{v}_{\varepsilon} \,.
\end{equation}
To determine ${\mathcal O}(\nu)$ accurate expressions for these eigenvalues
we expand an eigenpair $\kappa_{\varepsilon}$ and $\pmb{v}_{\varepsilon}$ as
\begin{equation}\label{large:L_Expand}
  \kappa_{\varepsilon} = \kappa_{1} + \nu\, \tilde{\kappa} + \cdots \,,\qquad
  \pmb{v}_{\varepsilon} = \pmb{v}_1 + \nu \, \tilde{\pmb{v}} + \cdots\,.
\end{equation}
Upon substituting \eqref{large:L_Expand} into \eqref{large:aux_eig} we
equate powers of $\nu$ to obtain the leading-order problem 
\begin{equation}\label{large:proof_eig0}
(I + \gamma {\mathcal B}^{-1} E)\,\pmb{v}_1 = \kappa_1 \,\pmb{v}_1 \,,
\end{equation}
and the following problem at $\mathcal{O}(\nu)$:
\begin{equation}\label{large:proof_eig1}
  ( I + \gamma {\mathcal B}^{-1}E  - \kappa_1 I)\, \tilde{\pmb{v}} =
  \tilde{\kappa} \, \pmb{v}_1 - 2\pi {\mathcal B}^{-1}{\mathcal G}_0
  \,\pmb{v}_1 \,.
\end{equation}

For the leading order problem \eqref{large:proof_eig0}, we first observe that
$E \, \pmb{q}_j = \pmb{0}$ for $j = 2,\dots,m$, where
$\mbox{span}\lbrace{\pmb{q}_2,\ldots,\pmb{q}_m\rbrace}$ is the $m-1$
dimensional subspace orthogonal to $\pmb{e}$. We readily conclude that
\eqref{large:proof_eig0} has the eigenpairs
\begin{equation}\label{large:0ortho_eigpair}
  \kappa_{1j} = 1  \quad \mbox{and} \quad \pmb{v}_{1j} = \pmb{q}_j \,, \qquad
  \mbox{for} \quad j=2,\ldots,m \,.
\end{equation}
To determine the ${\mathcal O}(\nu)$ correction to these eigenvalues,
we substitute \eqref{large:0ortho_eigpair} into \eqref{large:proof_eig1}
to obtain
\begin{equation}\label{large:1ortho_eigpair}
  \gamma {\mathcal B}^{-1}E \, \tilde{\pmb{v}} =   \tilde{\kappa} \pmb{q}_j
   - 2\pi {\mathcal B}^{-1} {\mathcal G}_0 \,\pmb{q}_j.
\end{equation}
Since ${\mathcal B}\,\pmb{q}_j$, with ${\mathcal B}^T={\mathcal B}$,
is a left null-vector of ${\mathcal B}^{-1}E$, the solvability
condition for \eqref{large:1ortho_eigpair} is
$\pmb{q}_j^T {\mathcal B} \left( \tilde{\kappa} \pmb{q}_j - 2\pi
  {\mathcal B}^{-1} {\mathcal G}_0 \,\pmb{q}_j\right)=0$, which
determines $\tilde{\kappa}$ for each eigenpair. In this way, we obtain
that $m-1$ eigenvalues of \eqref{large:aux_eig} are
\begin{equation}\label{large:eps_ortho_eigpair}
  \kappa_{j\varepsilon}= 1 +  2 \pi \nu \,\frac{\pmb{q}_j^T \, \mathcal{G}_0\,
    \pmb{q}_j}{\pmb{q}_j^T\,{\mathcal B}\,\pmb{q}_j} + \mathcal{O}(\nu^2)\,, 
  \qquad \mbox{where} \quad \pmb{q}_j^T \pmb{e}=0\,, \quad \mbox{for} \quad
  j=2,\dots,m\,.
\end{equation}

To find the remaining eigenpair of the leading order problem we
write \eqref{large:proof_eig0} as
$\pmb{v}_1 (1-\kappa_1) = - \frac{\gamma}{m}
\left(\pmb{e}^T\pmb{v}_1\right) \left( {1/b_1},\ldots, {1/b_m}
\right)^T$. This has the (unnormalized) nontrivial solution
\begin{equation}\label{large:remain0}
  \pmb{v}_1= \left( {1/b_1},\ldots,{1/b_m}\right)^T \,, \qquad
  \mbox{iff} \qquad \kappa_1 = 1 + \frac{\gamma}{m} \pmb{e}^T \pmb{v}_1 =
  1 + \frac{\gamma}{m} \sum_{i=1}^{m} \frac{1}{b_i} \,.
\end{equation}
To determine the $\mathcal{O}(\nu)$ correction to this eigenvalue we
observe, for this eigenpair, that $\pmb{e}$ is a left null-vector of
the matrix in \eqref{large:proof_eig1} since, by
using \eqref{large:remain0} for $\kappa_1$, we calculate
\begin{equation}
  \pmb{e}^T \left(I + \gamma {\mathcal B}^{-1} E -\kappa_1 I\right) =
\pmb{e}^T\left( 1  + \frac{\gamma}{m} \left(\sum_{i=1}^{m} \frac{1}{b_i}\right)
      -\kappa_1 \right)=\pmb{0} \,.
\end{equation}
Therefore, upon left-multiplying \eqref{large:proof_eig1} by
$\pmb{e}^T$, the solvability condition for \eqref{large:proof_eig1} is
$\pmb{e}^T\left(\tilde{\kappa} \pmb{v}_1-2\pi {\mathcal B}^{-1}
  {\mathcal G}_0 \pmb{v}_1\right)=0$, which yields
$\tilde{\kappa}=2\pi {\pmb{v}_1^T {\mathcal G}_0
  \pmb{v}_1/\left(\pmb{e}^T\pmb{v}_1\right)}$ where we used
$\pmb{e}^T {\mathcal B}^{-1}=\pmb{v}_1^T$. We conclude that a two-term
expansion $\kappa_{1\varepsilon}=\kappa_1+ \tilde{\kappa} \nu$ for
this remaining eigenvalue of \eqref{large:aux_eig} is as given in
\eqref{large:kappa_1eps}.  Finally, by multiplying
$\kappa_{1\varepsilon}$ with the other eigenvalues given in
\eqref{large:eps_ortho_eigpair} we obtain
$\det\left( I + \gamma{\mathcal B}^{-1}E + 2\pi \nu {\mathcal B}^{-1}
  {\mathcal G}_0 \right)=\kappa_{1\varepsilon}(1+{\mathcal O}(\nu))$.
In view of \eqref{large:proof1} this completes the derivation of
\eqref{large:detB}.

\end{proof}

The key assumption in Proposition \ref{lemma:large} is that $b_j\neq 0$ for
any $j=1,\ldots,m$. By using \eqref{large:bdef} and \eqref{large:Kmatrix_1}
for $b_j$ and $\mathit{K}_j$, respectively, we conclude that $b_j=0$ if
and only if $\lambda$ is a root of the quadratic equation
${\mathcal Q}_j(\lambda)=0$, where
\begin{equation}\label{large:qj}
  {\mathcal Q}_j(\lambda)\equiv \lambda^2 - \left(
    \mbox{tr}(J_j) - \frac{\eta_j}{\tau} \right) \lambda +
  \mbox{det}(J_j)\left[ 1+ \frac{\eta_j}{\tau}\right]\,, 
 \qquad \mbox{where}  \quad \eta_j \equiv \frac{2\pi d_{2j} D_0}{d_{1j}+D_0}\,,
\end{equation}
with $\mbox{det}(J_j)$ and $\mbox{tr}(J_j)$ as given in
\eqref{large:Kmatrix_jac}.  With this criterion and together with
Proposition \ref{lemma:large} we readily formulate a simple scalar
root-finding problem to identify values of $\lambda$ for which the
reduced GCEP \eqref{large:GCEP} has a nontrivial solution.

\vspace*{0.2cm}

\begin{prop} \label{large:mainprop} Suppose that $\lambda=\lambda^{\star}$
  is a root of ${\mathcal Q}_{s}(\lambda)=0$, where
\begin{equation}\label{large:qs}
   {\mathcal Q}_s(\lambda) \equiv 1 + \frac{\gamma}{m}
  \pmb{e}^{T}\pmb{v}_1 +  2 \pi \nu \,
  \left(\frac{\pmb{v}^T_1\, \mathcal{G}_0 \,\pmb{v}_1}{\pmb{e}^T\, \pmb{v}_1}
    \right)  \qquad \mbox{with} \qquad
  \pmb{v}_1^{T} \equiv \left( {1/b_1},\ldots,{1/b_m}\right) \,.
\end{equation}
Here $\gamma=\gamma(\lambda)$ and $b_j=b_j(\lambda)$ for $j=1,\ldots,m$ are
given in \eqref{large:M_sub} and \eqref{large:bdef}, respectively. Suppose
that $\lambda^{\star}$ satisfies
\begin{equation}\label{large:no_reson}
    \lambda^{\star} \notin \bigcup\limits_{j=1}^{m} \lbrace{\lambda_{j\pm}\rbrace}
    \,, \qquad \mbox{where} \quad {\mathcal Q}_{j}(\lambda_{j\pm})=0 \,,
\end{equation}
and ${\mathcal Q}_j(\lambda)$ is the quadratic defined in
\eqref{large:qj}.  Then,
$\mbox{det}\,{\mathcal M}_{\infty}(\lambda^{\star})=0$ and the
corresponding (unnormalized) nontrivial solution $\pmb{c}$ to the
reduced GCEP \eqref{large:gcep_mat} is
\begin{equation}
  \pmb{c} = \left( \frac{1}{b_1(\lambda^{\star})}, \ldots,
    \frac{1}{b_m(\lambda^{\star})}\right)^T + {\mathcal O}(\nu)\,
\end{equation}
\end{prop}

In this way, we can use Proposition \ref{large:mainprop} to determine Hopf
bifurcation boundaries in the $\tau$ versus $D_0$ parameter plane by simply
letting $\lambda=i\lambda_I$, with $\lambda_I>0$, and setting
\begin{equation}\label{large:hopf}
  \mbox{Re}\left[{\mathcal Q}_{s}(i\lambda_I)\right]=0 \,, \qquad
  \mbox{and} \qquad
    \mbox{Im}\left[{\mathcal Q}_{s}(i\lambda_I)\right]=0 \,,
\end{equation}
while ensuring that the condition \eqref{large:no_reson} holds with
$\lambda^{\star}=i\lambda_I$. A sufficient condition for
\eqref{large:no_reson} to hold along solutions of \eqref{large:hopf}
as parameters are varied is that $\mbox{tr}(J_j)\ne {\eta_j/\tau}$ for
all $j=1,\ldots,m$.

A simple analytically tractable special case of Proposition
\ref{large:mainprop} is when the permeabilities are all identical,
i.e.~$d_{1j}=d_{1c}$, $d_{2j}=d_{2c}$, and when the cell configuration
$\lbrace{\pmb{x}_1,\ldots,\pmb{x}_m\rbrace}$ is such that $\pmb{e}$ is
an eigenvector of the reduced-wave Green's matrix ${\mathcal G}$, and
consequently the Neumann Green's matrix ${\mathcal G}_0$. Therefore,
${\mathcal G}_0\pmb{e}=\beta\pmb{e}$ for some eigenvalue $\beta$. A
ring patterns of cells concentric within the unit disk with common
cell permeabilities and Sel'kov parameters is an example of such a
cell pattern. For this case, \eqref{large:steady-state} admits the
solution ${\mathcal A}=A_c\pmb{e}$ and so $J_j=J_c$ for
$j=1,\ldots,m$.  Since $b_j=b_c$ for $j=1,\ldots,m$, we can take
$\pmb{v}_1=b_{c}^{-1}(1,\ldots,1)^T$ and readily obtain that the
root-finding condition ${\mathcal Q}_{s}(\lambda)=0$ in
\eqref{large:qs} reduces to
\begin{equation}\label{large:equal_sync_1}
  1 + \frac{\gamma}{b_c} + \frac{2\pi\nu}{m b_c} \pmb{e}^T {\mathcal G}_0
  \pmb{e} =0  \,, \qquad \mbox{where} \qquad {\mathcal G}_0 \pmb{e}=
  \beta \pmb{e} \,,
\end{equation}
while from \eqref{large:qj} we have
${\mathcal Q}_{j}(\lambda)={\mathcal Q}_{c}(\lambda)$ for $j=1,\ldots,m$,
where
\begin{equation}\label{large:qc}
  {\mathcal Q}_c(\lambda)\equiv \lambda^2 - \left(
    \mbox{tr}(J_c) - \frac{\eta_c}{\tau} \right) \lambda +
  \mbox{det}(J_c)\left[ 1+ \frac{\eta_c}{\tau}\right]\,, 
 \qquad \mbox{where}  \quad \eta_c \equiv \frac{2\pi d_{2c} D_0}{d_{1c}+D_0}\,,
\end{equation}
By using \eqref{large:bdef} and \eqref{large:Kmatrix_1}, we obtain
after a little algebra that \eqref{large:equal_sync_1} reduces to
\begin{equation}\label{large:equal_root}
   \frac{\lambda + \mbox{det}(J_c)}
  {\lambda^2 - \lambda \, \mbox{tr}(J_c) + \mbox{det}{J_c}} =-\frac{\tau}
  {2\pi d_{2c}}
  \left[ 1 + \frac{d_{1c}}{D_0}\left(1+2\pi\nu \beta\right)
     + \frac{2\pi m d_{1c}}{|\Omega|(1+\tau\lambda)}\right]\,,
\end{equation}
which is a cubic equation in $\lambda$. For this special cell pattern,
we conclude that if $\lambda=\lambda^{\star}$ is a root of
\eqref{large:equal_root} for which $Q_{c}(\lambda^{\star})\neq 0$,
then $\mbox{det}\,{\mathcal M}_{\infty}(\lambda^{\star})=0$. The
corresponding eigenvector of ${\mathcal M}_{\infty}(\lambda^{\star})$
is the in-phase mode $\pmb{c}=\pmb{e}$.

Next, we will show how to determine roots of the reduced GCEP
\eqref{large:gcep_mat} in the case where
${\mathcal
  B}(\lambda)=\mbox{diag}\left(b_1(\lambda),\ldots,b_m(\lambda)
\right)^T$ in \eqref{large:M_sub} is not invertible. We first observe
that a nontrivial $\pmb{c}$ to the leading order problem in
\eqref{large:gcep_mat} (with $\nu=0$), exists if and only if there is a
$\lambda=\lambda^{\star}$ at which {\em at least two} $b_j$ cross through zero
simultaneously. After relabelling the indices as necessary, this
occurs without loss of generality when there is a $\lambda^{\star}$
and an integer $J$ with $2\leq J\leq m$ for which
\begin{equation}\label{degen:ass}
  b_1(\lambda^{\star})=\ldots =b_{J}(\lambda^{\star})=0\,, \quad \mbox{with}
  \quad b_j(\lambda^{\star})\neq 0 \quad \mbox{for} \quad j=J+1,\ldots,m
  \quad \mbox{if} \quad J<m \,.
\end{equation}
Then, for the leading-order problem in \eqref{large:gcep_mat} at
$\lambda=\lambda^{\star}$ the nontrivial solutions $\pmb{c}_0$ are
\begin{equation}\label{degen:order0}
  \left( {\mathcal B}(\lambda^{\star}) + \gamma(\lambda^{\star}) E\right)
  \pmb{c}_0=\pmb{0}\,, \quad \implies
  \quad \pmb{c}_0 \in {\mathcal C}^{\perp} \equiv \Big{\lbrace{}
    \pmb{c}_0 \, \Big{\vert} \,\,  \pmb{e}^{T}\pmb{c}_0=0 \,, \,\, \pmb{c}_0=
\begin{pmatrix}
 \pmb{c}_J\\
 \pmb{0}
\end{pmatrix} \,, \,\, \pmb{c}_J\in \R^{J} \,, \,\, \pmb{0}\in \R^{m-J}
\Big{\rbrace}}\,.
\end{equation}
Next, we introduce an orthonormal basis for the subspace
${\mathcal C}^{\perp}$ of dimension $J-1$ and we decompose $\pmb{c}_0$
as
\begin{equation}\label{degen:c0}
  \pmb{c}_0=\omega_{1}\pmb{v}_1 + \ldots + \omega_{J-1}\pmb{v}_{J-1} \,,
  \qquad \mbox{where} \qquad {\mathcal C}^{\perp} \equiv \mbox{span}\lbrace{
    \pmb{v}_1,\ldots,\pmb{v}_{J-1}\rbrace} \,, \qquad \pmb{v}_j^T\pmb{v}_i
  = \delta_{ij}\,,
\end{equation}
where $\delta_{ij}$ is the Kronecker symbol, and $\omega_j$ for
$j=1,\ldots,J-1$ are scalar coefficients to be found.

To determine the ${\mathcal O}(\nu)$ correction to the leading-order
eigenvalue $\lambda^{\star}$ of the GCEP and identify the constants
$\omega_j$, we look for a nontrivial solution to
${\mathcal M}_{\infty}(\lambda) \pmb{c}=\pmb{0}$, as defined in
\eqref{large:GCEP}, of the form
$\lambda=\lambda^{\star} + {\mathcal O}(\nu)$. We substitute the
expansion
\begin{equation}\label{degen:expan}
  \lambda=\lambda^{\star}+\nu \tilde{\lambda} + \cdots
  \,, \qquad \pmb{c}=\pmb{c}_0 + \nu \pmb{c}_1 + \cdots\,,
\end{equation}
into \eqref{large:GCEP} and then equate ${\mathcal O}(\nu)$ terms to
obtain
\begin{equation}\label{degen:c1}
  \left( {\mathcal B}(\lambda^{\star}) + \gamma(\lambda^{\star}) E\right)
  \pmb{c}_1 = - \tilde{\lambda} \gamma^{\prime}(\lambda^{\star}) E \pmb{c}_0
  - \tilde{\lambda} {\mathcal B}^{\prime}(\lambda^{\star}) \pmb{c}_0 -
  2\pi {\mathcal G}_0 \pmb{c}_0 \,,
\end{equation}
where we observe that $E\pmb{c}_0=\pmb{0}$.  The solvability
condition for \eqref{degen:c1} is that the right-hand side of
\eqref{degen:c1} is orthogonal to each $\pmb{v}_j$ for
$j=1,\ldots, J-1$. In this way, we readily obtain that
$\tilde{\lambda}$ and
$\pmb{\omega}\equiv (\omega_1,\ldots,\omega_{J-1})^T$ are eigenpairs
of the $J-1$ dimensional symmetric generalized matrix eigenvalue
problem
\begin{equation}\label{degen:reduce}
  {\mathcal V}^T {\mathcal G}_0 {\mathcal V} \, \pmb{\omega} = -
  \frac{\tilde{\lambda}}{2\pi} {\mathcal V}^T
  {\mathcal B}^{\prime}(\lambda^{\star}) {\mathcal V} \, \pmb{\omega} \,, \qquad
    {\mathcal V} \equiv 
    \begin{pmatrix}
     \pmb{v}_1\,, & \ldots &\,, \pmb{v}_{J-1}
    \end{pmatrix}\,.
\end{equation}
Here ${\mathcal V}$ is the $m\times J-1$ matrix whose columns provide
an orthonormal basis for ${\mathcal C}^{\perp}$.  In summary, for any
such $\tilde{\lambda}$ satisfying \eqref{degen:reduce} a two-term
expansion for a root of the reduced GCEP
$\mbox{det} {\mathcal M}_{\infty}(\lambda)=0$ is
$\lambda=\lambda^{\star}+ \nu\tilde{\lambda}$ where $\lambda^{\star}$
satisfies \eqref{degen:ass}.

We now illustrate this theory for the special case where $J=2$. This
analysis, given below, will be shown in \S \ref{large:quiet} to be
relevant for analyzing anti-phase instabilities associated with the
cell configuration of Fig.~\ref{fig:tworing} where a pair of isolated
identical cells is spatially segregated from two symmetric ring
clusters. Suppose that cells 1 and 2 have common permeabilities
$d_{1c}\equiv d_{11}=d_{12}$, $d_{2c}\equiv d_{21}=d_{22}$ and that
they have the same intracellular steady-states. Then, from
\eqref{large:Kmatrix} we obtain
$\mbox{tr}{J_c}\equiv\mbox{tr}{J_1}=\mbox{tr}{J_2}$,
$\mbox{det}{J_c}\equiv\mbox{det}{J_1}=\mbox{det}{J_2}$, and so
\eqref{large:bdef} and \eqref{large:Kmatrix_1} yields that
$b_c(\lambda)\equiv b_1(\lambda)=b_2(\lambda)$. Then, in
\eqref{degen:reduce} we can take
${\mathcal V}=\left({1/\sqrt{2}},-{1/\sqrt{2}}\right)^T$, and readily
calculate that
$\tilde{\lambda}=-2\pi {R_c/b_{c}^{\prime}(\lambda^{\star})}$, where
$b_c(\lambda^{\star})=0$. Here
$R_{c}\equiv R_0(\pmb{x}_1)=R_0(\pmb{x}_2)$ is the common value of the
regular part of the Neumann Green's function at the centers
$\pmb{x}_1$ and $\pmb{x}_2$ of the two cells. For $\nu\ll 1$, a simple
perturbation argument shows that
$\lambda\sim\lambda^{\star}+\nu \tilde{\lambda}$ can be identified as
the root to $b_c(\lambda) + 2\pi \nu R_c=0$. Finally, by using
\eqref{large:bdef} for $b_c$, together with \eqref{large:Kmatrix}, we
conclude for $\nu\ll 1$ that $\lambda$ must be a root of the quadratic
\begin{equation}\label{degen:quad_J2}
  \lambda^2 - \lambda \left(\mbox{tr}{J_c} - \frac{\eta_c}{\tau f} \right)
  + \left(1+ \frac{\eta_c}{\tau f}\right) \mbox{det}J_c =0 \,,
  \qquad \mbox{where} \qquad \eta_c \equiv \frac{2\pi d_{2c}D_0}{d_{1c}+D_0} \,,
  \quad f\equiv 1 + \frac{2\pi\nu d_{1c}}{d_{1c}+D_0} R_c \,.
\end{equation}
Since $\mbox{det}J_c>0$, an anti-phase instability occurs for cells 1
and 2, with the other cells remaining quiescent, only when
\begin{equation}\label{degen:quad_J2_crit}
  \mbox{tr}(J_c)- \frac{\eta_c}{\tau f} >0 \,.
\end{equation}
Although this criterion gives a region in the $(D_0,\tau)$ parameter
plane, in \S \ref{large:quiet} we will only implement it for the cell
configuration in Fig.~\ref{fig:tworing} at a fixed value of $\tau$ in
order to determine a threshold in $D_0$.
  
\subsection{Example: \texorpdfstring{$m=10$}{m10} cells in
  the unit disk}\label{subsec:LargeCellExample}

We now apply our simplified theory for the regime
$D = \mathcal{O}(\nu^{-1})$ to a population of $m=10$ cells in the
unit disk. Different spatial configurations of these cells are
considered, and we will focus on three different scenarios: (a) all
cells are identical, (b) some groups of cells are identical and, (c)
none of the cells are identical. For all the examples considered in
this subsection, the cells have a common radius $\varepsilon = 0.05$
and common Sel'kov kinetic parameters as given in
\eqref{Selkov_para}. The heterogeneity in the cells is introduced
through the cell locations and their permeability parameters $d_{1j}$,
$j=1,\ldots,m$, which specifies the rate of feedback of the bulk
signal into the cells. The secretion rate is fixed at $d_2 = 0.2$ for
all the cells. For each spatial configuration of cells and
permeability parameter set $d_{11},\ldots,d_{1m}$ we will compute the
HB bifurcation boundary in the $\tau$ versus $D_0$ plane by solving
\eqref{large:hopf} numerically using Newton's method with arclength
continuation in $D_0$.

In Fig.~\ref{Large:cluster_arbitrary} we plot the HB boundaries (left
panel) in the $\tau$ versus $D_0$ parameter plane together with the
cell pattern (right panel) for a pattern with two clusters of cells
(top row) and a pattern with an arbitrary arrangement of cells (bottom
row). The precise locations and influx rate $d_{1j}$ for the cells for
these two specific configurations are given in
Table~\ref{Table:CellLocation} of Appendix~\ref{Cell_Location}. By
using a numerical winding number computation we have verified that
there are exactly two roots of \eqref{large:qs}, corresponding to two
unstable eigenvalues of the GCEP matrix, inside the lobes spanned by
the HB boundaries.

\begin{figure}[!ht]
  \centering
       \begin{subfigure}[b]{0.45\textwidth}
        \includegraphics[width=\textwidth,height=4.4cm]{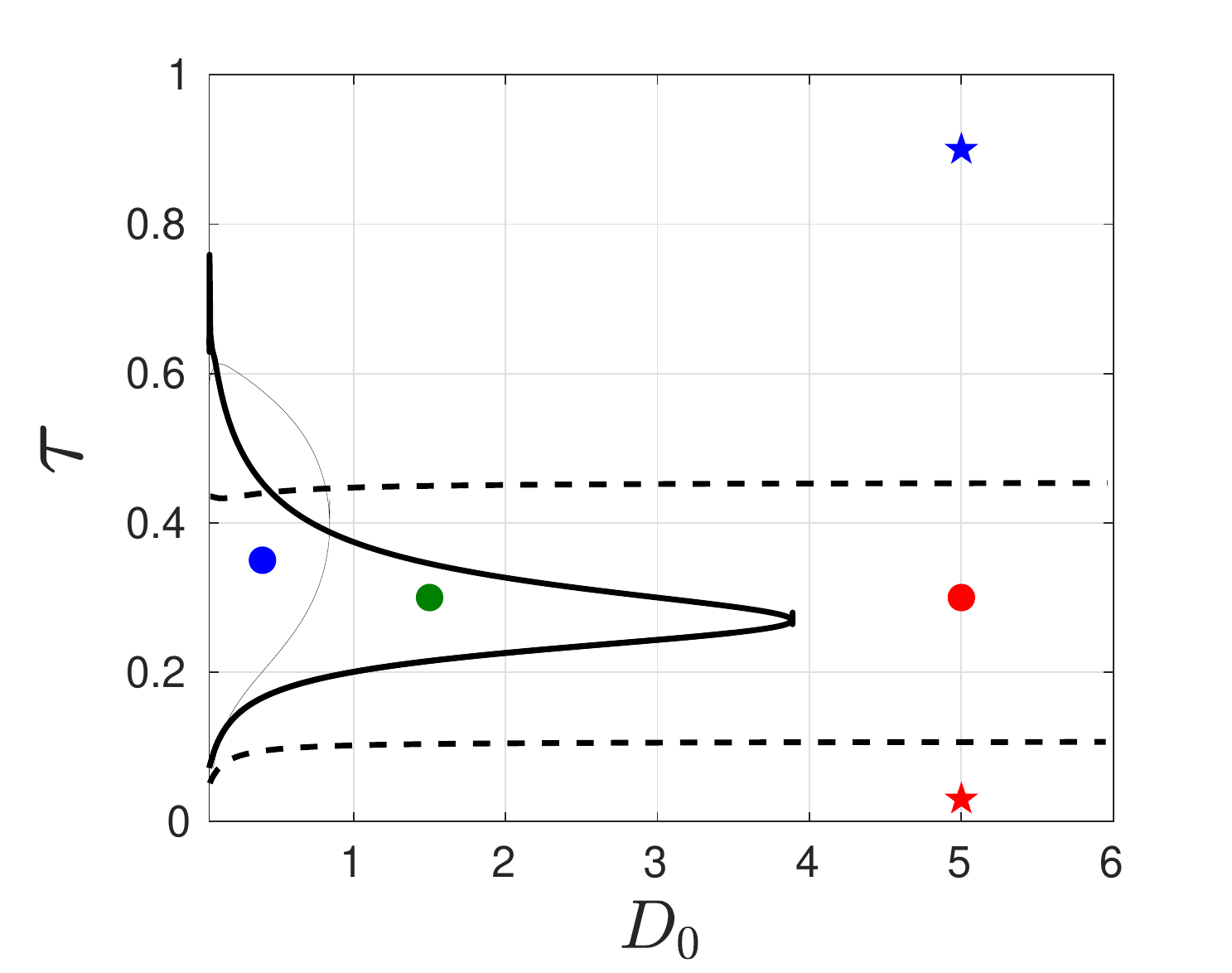}
        \caption{HB boundaries}
        \label{fig:twoclustersHB}
    \end{subfigure}\qquad
    \begin{subfigure}[b]{0.35\textwidth}
      \includegraphics[width=\textwidth,height=4.4cm]{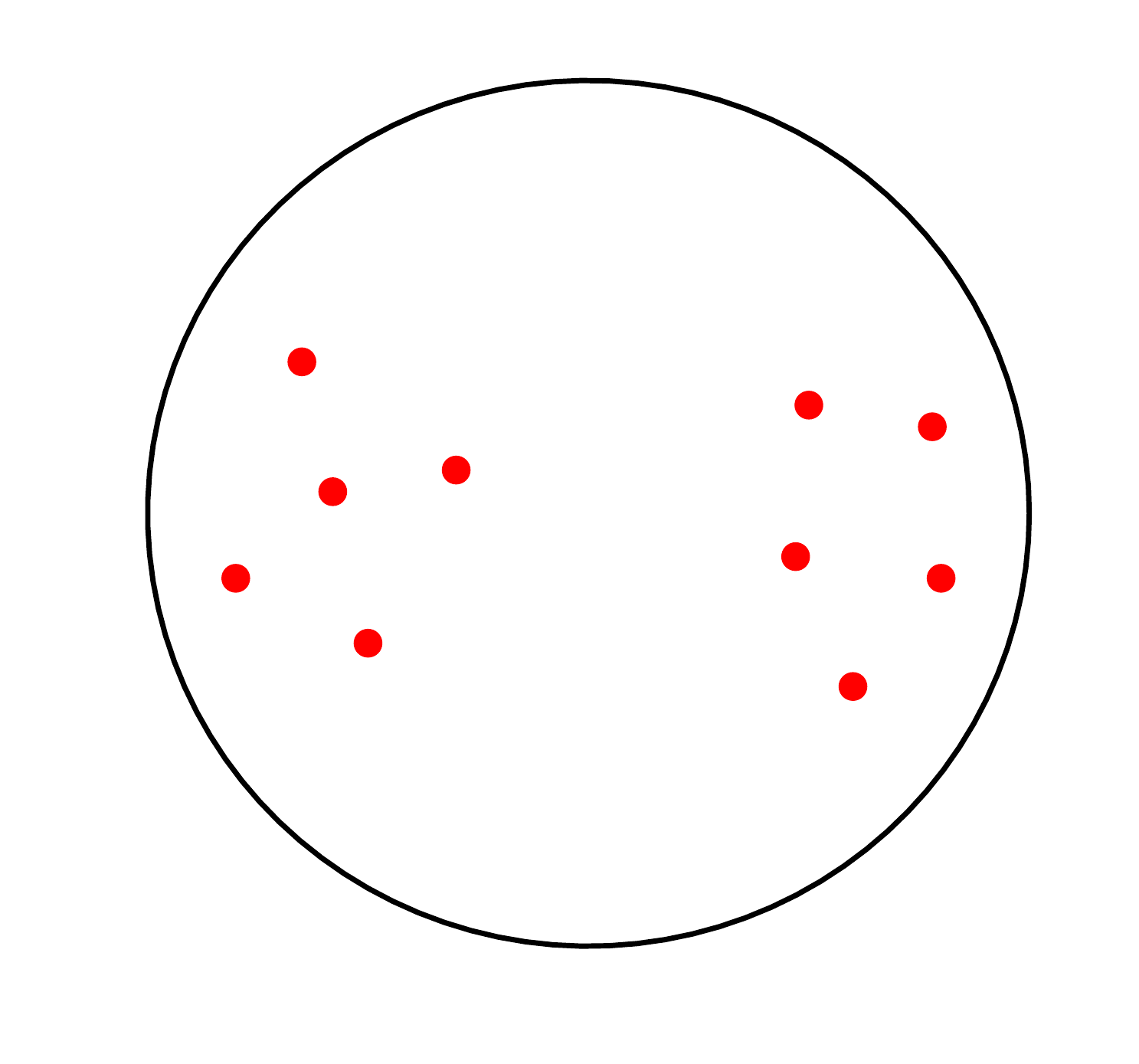}
        \caption{Two spatially segregated clusters of cells}
        \label{fig:twoclusters}
    \end{subfigure}\\
   \begin{subfigure}[b]{0.45\textwidth}
         \includegraphics[width=\textwidth,height=4.4cm]{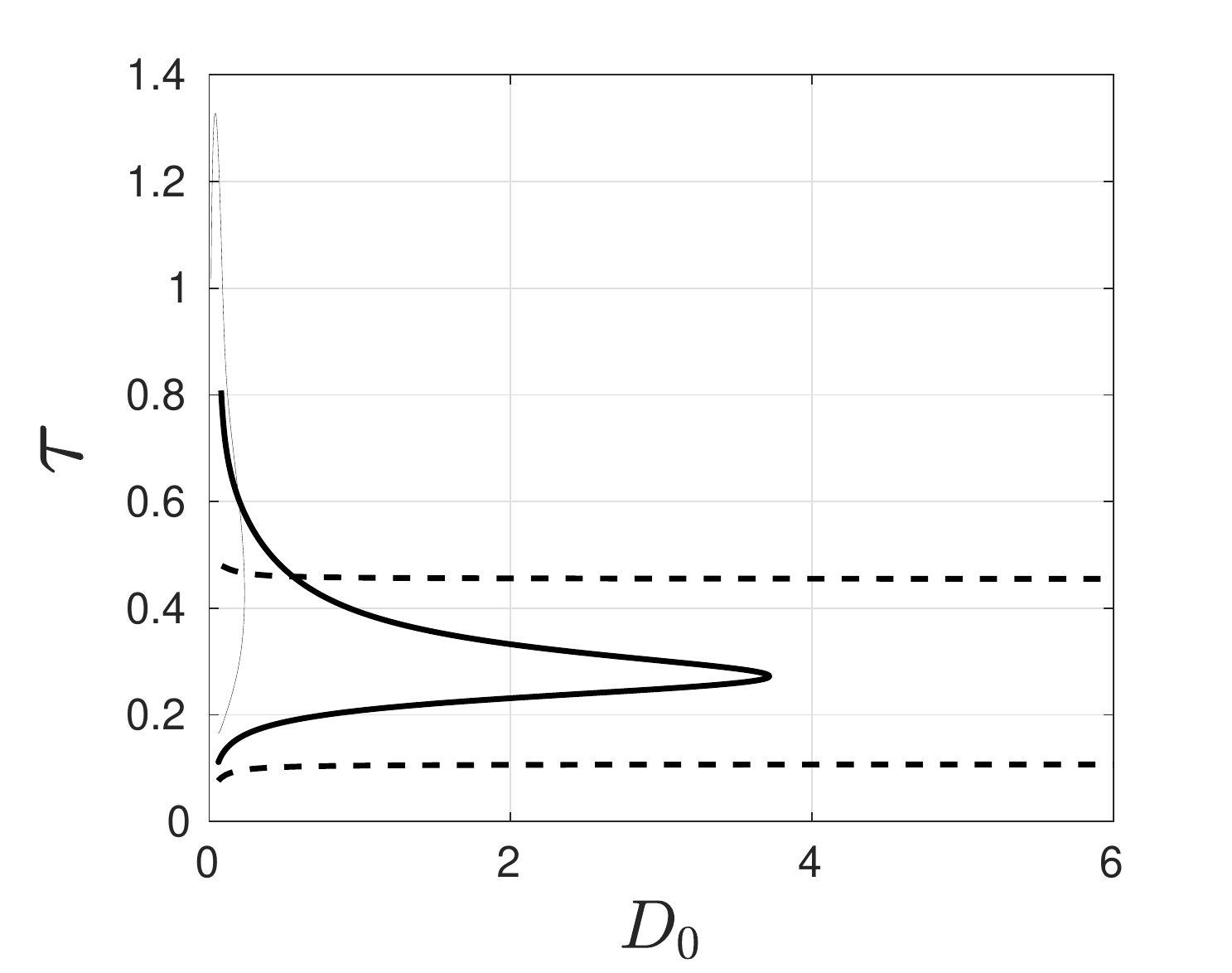}
        \caption{HB boundaries}
        \label{fig:arbitraryHB}
    \end{subfigure}\qquad
    \begin{subfigure}[b]{0.35\textwidth}
        \includegraphics[width=\textwidth,height=4.4cm]{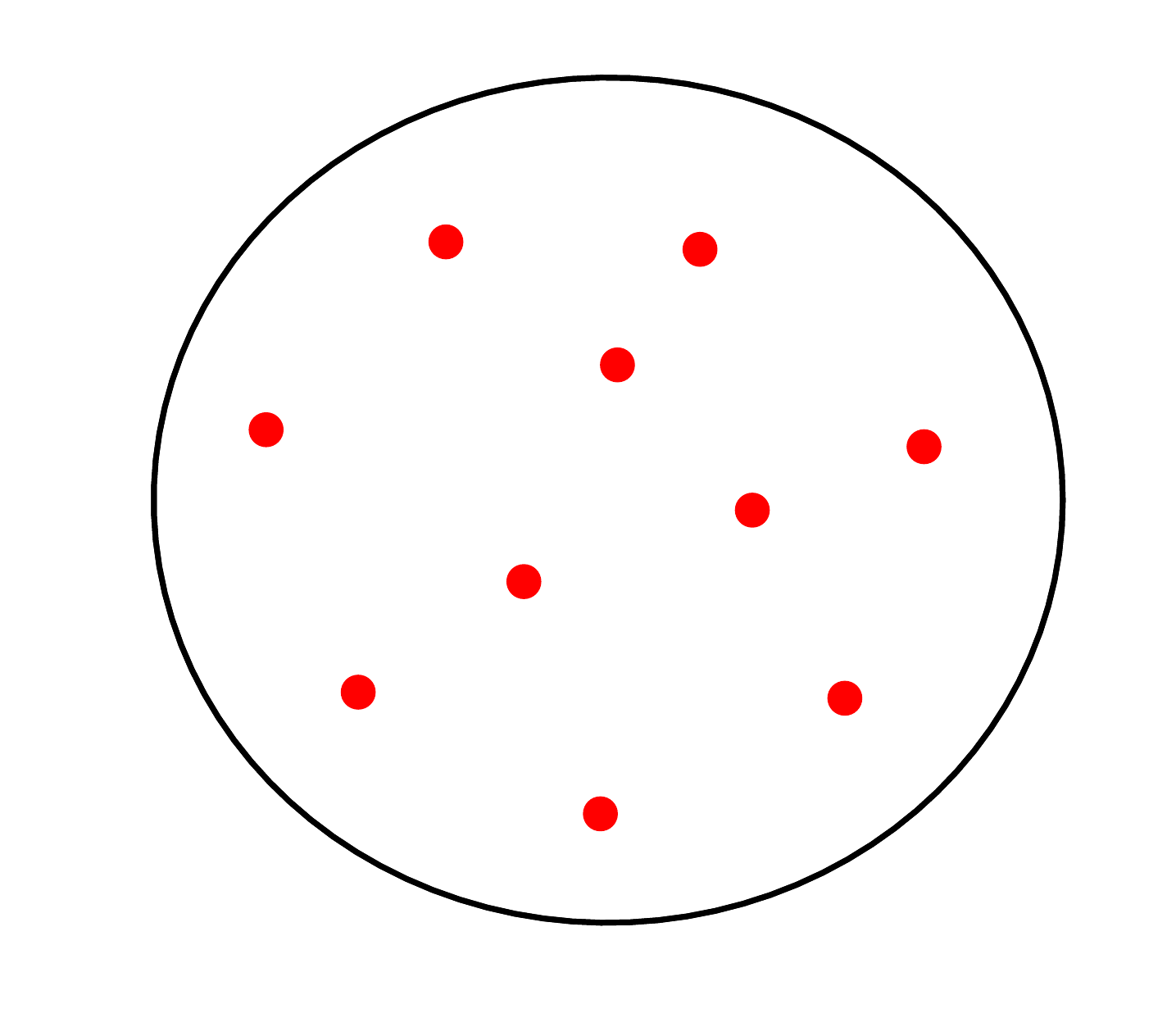}
        \caption{Arbitrary (random) arrangement of cells}
                \label{fig:arbit}
    \end{subfigure}
    \caption{HB boundaries in the $\tau$ versus $D_0$ plane
        for $m=10$ cells with two groups/clusters of cells (top row)
        and arbitrarily placed cells (bottom row).  The dashed curve
        corresponds to identical cells where $d_{1}=0.8$ for
        each cell. The thin solid curve is for when there are two
        groups of identical cells with $d_1 = 0.8$ for the first group
        and $d_1 = 0.4$ for the second group (permeability set
        II). The heavy solid curve is for non-identical cells with
        $d_1$ uniformly selected from the interval
        $0.4 \leq d_1 \leq 0.8$ (permeability set III). The 
        locations and influx rates $d_{1j}$ for the cells are given in
        Table~\ref{Table:CellLocation} of
        Appendix~\ref{Cell_Location}.  The remaining parameters are as
        given in \eqref{Selkov_para}.}\label{Large:cluster_arbitrary}
\end{figure}

By comparing the HB boundaries in Fig.~\ref{fig:twoclustersHB} and
Fig.~\ref{fig:arbitraryHB} we observe from the dashed lines in these
figures that when the cells are all identical (with $d_1=0.8$ for each
cell) the unbounded region of instability in the $(D_0,\tau)$
parameter plane is very similar for the two cell patterns. Therefore,
when $D_0$ is large and when there are a sufficiently large number of
identical cells, this observation suggests that the parameter region
where intracellular oscillations occur should depend only weakly on
the spatial configuration of cells. In addition, by comparing the
heavy solid curves in Fig.~\ref{fig:twoclustersHB} and
Fig.~\ref{fig:arbitraryHB}, we conclude that the HB boundaries for
cell patterns where $d_1$ is uniformly selected from the interval
$0.4 \leq d_1 \leq 0.8$ (permeability set III) also depend only weakly
on the spatial arrangement of the cells, but that there are no
intracellular oscillations when $D_0$ is large. The influx rates
$d_{1j}$ for the cells are given in Table~\ref{Table:CellLocation} of
Appendix~\ref{Cell_Location}. In this case, the bounded instability
region in the $(D_0,\tau)$ plane exists only within a rather narrow
range of $\tau$, which measures the ratio of the time-scale for the
reaction kinetics to the decay rate of the bulk signal. From the thin
solid curve in Fig.~\ref{fig:twoclustersHB}, corresponding to where we
assign the different influx rates $d_1=0.8$ and $d_1=0.4$ to all the
cells in the two different groups (permeability set II in
Table~\ref{Table:CellLocation} of Appendix~\ref{Cell_Location}), we
observe that the parameter region where intracellular oscillations
occur is rather small in $D_0$ but has a larger extent in $\tau$. Even
when $D_0$ is small, identical cells within a cluster are in close
proximity and so are able to synchronize their activities and generate
oscillations within their respective group. However, as $D_0$
increases, the bulk signal diffuses rapidly away from each of the two
clusters and this signal no longer has a sufficient spatial gradient
to coordinate synchronized oscillations between the two spatially
segregated groups of cells. For an arbitrary arrangement of cells, we
observe from the thin solid curve in Fig.~\ref{fig:arbitraryHB} that
when cells are assigned either $d_1=0.8$ or $d_1=0.4$ in such a way
that two neighboring cells are not identical, intracellular
oscillations only occur when $D_0$ is significantly smaller than for
the case when the cells with the same influx rates are clustered.
Qualitatively, this indicates that identical cells within a group of
more closely spaced cells can more readily synchronize their
activity.

For the case of two clustered groups of cells (see
Fig.~\ref{fig:twoclusters}), we will validate our linear stability
predictions with full numerical results computed from the coupled
PDE-ODE model \eqref{DimLess_bulk} using FlexPDE
\cite{flexpde2015solutions} at the indicated points in the HB phase
diagram in Fig.~\ref{fig:twoclustersHB} and for specific permeability
sets. In Table \ref{eigvec:twoclusters} we show the real and imaginary
parts of the components of the normalized eigenvectors $\pmb{c}$ and
${\mathcal K}\pmb{c}$, together with the unique complex conjugate pair
of unstable eigenvalues of the $10\times 10$ GCEP matrix
${\mathcal M}(\lambda)$ in \eqref{Global_System}, as computed from the
root-finding condition $\det{\mathcal M}(\lambda)=0$, at the red,
blue, and green dots shown in the phase diagram in
Fig.~\ref{fig:twoclustersHB}. In Table \ref{Table:CellLocation} of
Appendix \ref{Cell_Location} we indicate the specific permeability set
for the influx rate used at these three pairs of $(D_0,\tau)$.  From
\eqref{nstabform:c}, the components of $\pmb{c}$ measure the diffusive
flux into the cells, while
$\tilde{\pmb{c}}\equiv {\mathcal K}\pmb{c}$, with ${\mathcal K}$ given
in \eqref{large:Kmatrix}, predicts the relative amplitude and phase
shifts of the intracellular oscillations.

\begin{table}[!ht]
\centering
  \begin{tabular}{ c | c | c | c | c |c | c |c} \hline 
\rowcolor{LightCyan}
    Permeability set &$(D_0,\tau)\,;\,\, D={D_0/\nu}$  & $(\mbox{Re}\lambda,\mbox{Im}\lambda)$ & Cell j  &  $\Big( \mbox{Re}(c_j),\mbox{Im}(c_j) \Big)$ & $\theta_j \,(\text{rad})$ &
$\Big( \mbox{Re}(\tilde{c}_j),\mbox{Im}(\tilde{c}_j) \Big)$  \\  \hline \hline    \rowcolor{Cyan}
    &    &     & 1 & $(0.316,0.0000)$  & $0.0000$ & $(-0.309,0.0670)$ \\ 
\rowcolor{Cyan}
    &   &      & 2 & $(0.312,0.00137)$  & $0.0044$ & $(-0.310,0.0669)$ \\ 
\rowcolor{Cyan}
     &   &     & 3 & $(0.309,0.00236)$  & $0.0076$ & $(-0.310,0.0668)$ \\ 
\rowcolor{Cyan}
     &   &     & 4 & $(0.319,-0.00103)$  & $6.284$ & $(-0.308,0.0670)$ \\ 

\rowcolor{Cyan}
 Identical cells&  $ (5,0.3) $ & $(0.0198,0.7802)$  & 5 & $(0.317,-0.000415)$  & $6.28$ & $(-0.309,0.0670)$ \\ 

\rowcolor{Cyan}
    (Set I) & $D\approx 14.98$ &  & 6 & $(0.315,0.000335)$  & $0.0011$ &
    $(-0.309,0.0670)$ \\ 
\rowcolor{Cyan}
    &  Red dot  &    & 7 &  $(0.320,-0.00146)$  & $6.28$  & $(-0.309,0.0671)$
    \\ 
\rowcolor{Cyan}
     &    &    & 8 &  $(0.314,0.00512)$  & $0.0016$  & $(-0.309,0.0670)$ \\
\rowcolor{Cyan}
     &    &    & 9 &  $(0.321,-0.00201)$  & $6.28$   & $(-0.308,0.0671)$\\ 
\rowcolor{Cyan}
     &    &    & 10 &  $(0.321,-0.00183)$  & $6.28$   & $(-0.308,0.0671)$ \\ 
\hline  \hline 
 \rowcolor{Gray}
    &   &      & 1 & $(0.0718,0.0000)$  & $0.0000$ & $(-0.135,0.290)$ \\ 
\rowcolor{Gray}
    & 	&      & 2 & $(0.0751,-0.00850)$  & $6.17$ & $(-0.135,0.273)$ \\ 
\rowcolor{Gray}
     &    &    & 3 & $(0.0693,0.000272)$  & $0.0039$ & $(-0.136,0.270)$ \\ 
\rowcolor{Gray}
     &     &   & 4 & $(0.0802,-0.0145)$  & $6.10$ & $(-0.133,0.289)$ \\ 

\rowcolor{Gray}
Two groups &  $ (0.4,0.35) $ & $(0.00772,0.766)$ & 5 & $(0.0817,-0.0185)$  & $6.06$ & $(-0.133,0.283)$ \\ 

\rowcolor{Gray} 
(Set II)     & $D\approx 1.198$ &   & 6 & $(0.208,-0.369)$  & $5.23$ & $(-0.123,0.307)$ \\ 

\rowcolor{Gray}
     &  Blue dot  &    & 7 &  $(0.204,-0.397)$  & $5.19$  & $(-0.116,0.299)$ \\ 

\rowcolor{Gray}
     &    &    & 8 &  $(0.210,-0.365)$  & $5.23$  & $(-0.125,0.311)$ \\
    
\rowcolor{Gray}
     &    &    & 9 &  $(0.197,-0.4127)$  & $5.16$  & $(-0.109,0.283)$ \\ 
\rowcolor{Gray}
     &    &    & 10 &  $(0.201,-0.406)$  & $5.17$  & $(-0.112,0.292)$ \\ \hline \hline  
 \rowcolor{Cyan} 
    &    &     & 1 & $(0.450,0.000)$  & $0.0000$ & $(-0.246,0.00135)$\\ 
\rowcolor{Cyan}  
    & 	 &     & 2 & $(0.359,0.0635)$  & $0.175$ & $(-0.302,-0.00586)$\\ 
\rowcolor{Cyan}
     &    &    & 3 & $(0.430,0.0277)$  & $0.0643$ & $(-0.269,-0.00193)$ \\ 
\rowcolor{Cyan}
     &    &    & 4 & $(0.242,0.0834)$  & $0.331$ & $(-0.331,-0.00735)$ \\ 

\rowcolor{Cyan} 
Random &   $ (1.5,0.3) $  & $(0.00230,0.765)$ & 5 & $(0.0215,0.0775)$  & $1.30$ & $(-0.368,-0.00610)$ \\ 

\rowcolor{Cyan} 
 (Set III)  & $D\approx 4.494$  &     & 6 & $(0.260,0.0823)$  & $0.307$ & $(-0.324,-0.00917)$ \\ 

\rowcolor{Cyan}
     &  Green dot  &    & 7 &  $(0.117,0.0854)$  & $0.630$ & $(-0.352,-0.00935)$\\ 

\rowcolor{Cyan}
     &    &    & 8 &  $(0.436,0.0192)$  & $0.0439$ & $(-0.259,-0.00170)$ \\
    
\rowcolor{Cyan}
     &    &    & 9 &  $(0.331,0.0697)$  & $0.208$  & $(-0.307,-0.00665)$ \\ 
\rowcolor{Cyan}
     &    &    & 10 &  $(-0.0421,0.0687)$  & $2.12$  & $(-0.376,-0.00609)$ \\ 
  \hline
\end{tabular}
\caption{Real and imaginary parts of the normalized
    eigenvector $\pmb{c}$ of the GCEP matrix \eqref{Global_System}, together
    with $\tilde{\pmb{c}}\equiv{\mathcal K}\pmb{c}$, computed at the
    red, blue, and green dot shown in the phase diagram in
    Fig.~\ref{fig:twoclustersHB} for a two-cluster arrangement of
    cells shown in Fig.~\ref{fig:twoclusters}. The cell locations and
    permeability sets for the influx rate are given in
    Table~\ref{Table:CellLocation} of
    Appendix~\ref{Cell_Location}. The other parameters are as given in
    \eqref{Selkov_para}. The third column gives the unique unstable
    eigenvalue in $\mbox{Re}(\lambda)>0$ of the GCEP matrix at these
    three pairs of $(D_0,\tau)$.}
\label{eigvec:twoclusters}
\end{table}

In the top row of Fig.~\ref{FPDE_2Clusters_ID} we show FlexPDE
simulation results of \eqref{DimLess_bulk} corresponding to the red
dot at $(D_0,\tau)=(5,0.3)$ in the phase diagram of
Fig.~\ref{fig:twoclustersHB} for the case where the cells are all
identical with $d_{1j}=0.8$ for $j=1,\ldots,10$. The eigenpair
$\pmb{c}$ and ${\mathcal K}\pmb{c}$ of the GCEP matrix is given in the
top third of Table \ref{eigvec:twoclusters}. As predicted by the
eigenvector ${\mathcal K}\pmb{c}$ of the linearized theory, the
intracellular oscillations observed in the full simulations are nearly
synchronized both in amplitude and phase. In the bottom row of
Fig.~\ref{FPDE_2Clusters_ID} we show, for this parameter set, that
results computed from the ODE system \eqref{reducedODE}. We observe
that the amplitude and period of intracellular oscillations predicted
by the ODEs \eqref{reducedODE} compare very favorably with
corresponding FlexPDE results computed from \eqref{DimLess_bulk}. To
explain this very favorable comparison, we observe from the surface
plot in Fig.~\ref{FPDE_2Clusters_ID_Surf} that the bulk signal is
roughly spatially uniform when $D_0=5$. Recall that the asymptotic
analysis for the derivation of the ODE system \eqref{reducedODE} in \S
\ref{largeD_ODE} relies on a nearly spatially uniform bulk signal.

\begin{figure}[!httbp]
  \centering
  \begin{subfigure}[b]{0.32\textwidth}
      \includegraphics[width=\textwidth,height=4.2cm]{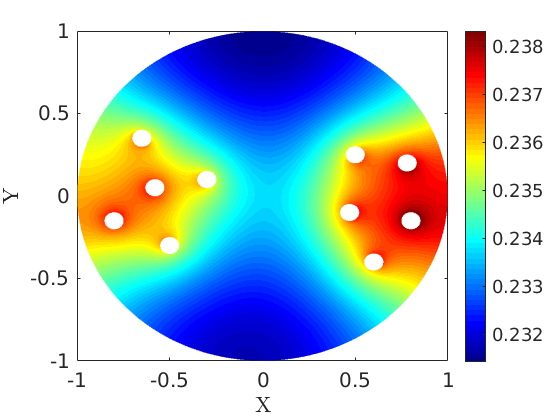}
    \caption{Surface plot at $t=400$}
    \label{FPDE_2Clusters_ID_Surf}
  \end{subfigure}
\begin{subfigure}[b]{0.32\textwidth}
  \includegraphics[width=\textwidth,height=4.2cm]{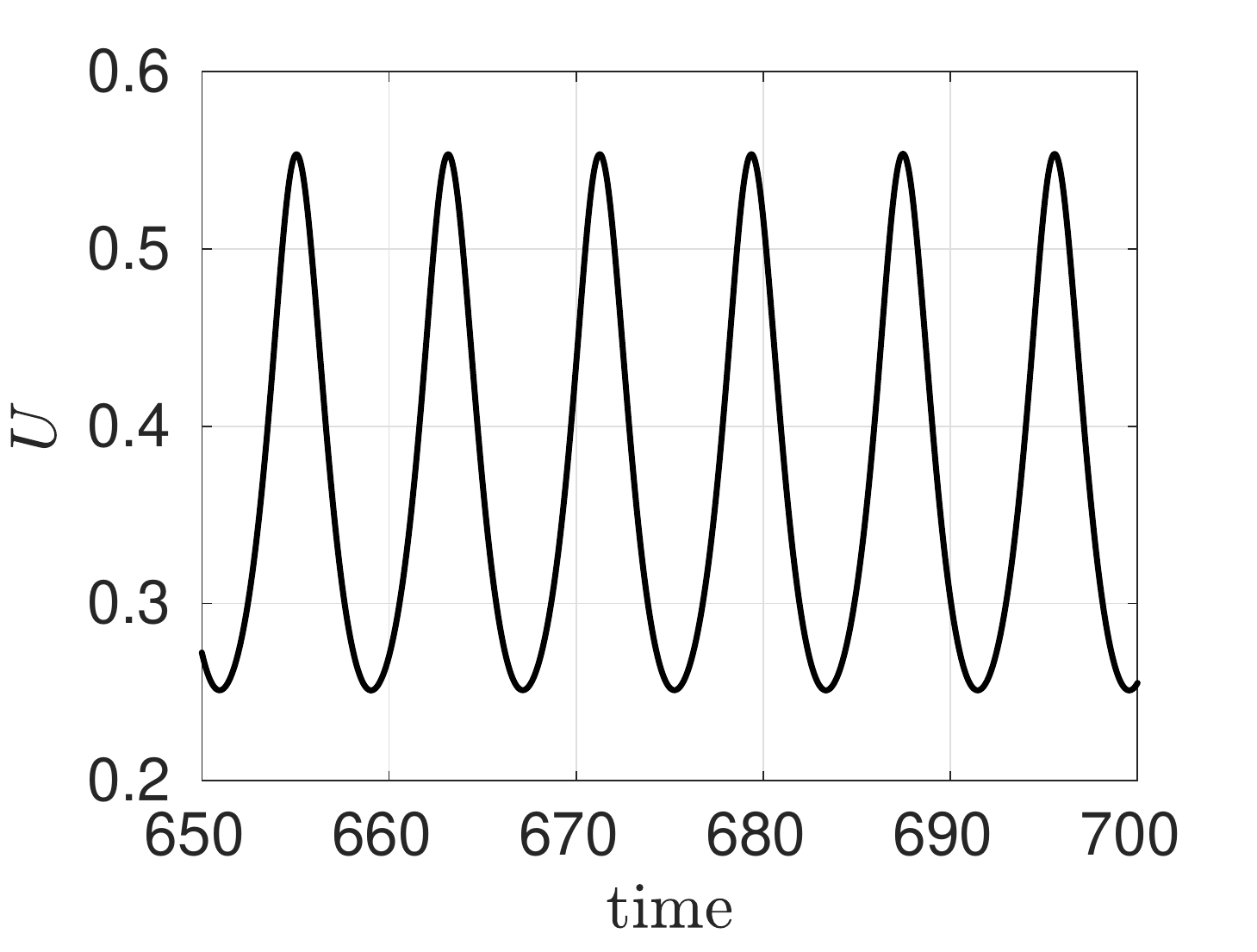}
    \caption{$U$ at $\pmb{x}=(0,0.5)$ (FlexPDE)}
    \label{FPDE_2Clusters_ID_bulk}
  \end{subfigure}
\begin{subfigure}[b]{0.32\textwidth}
  \includegraphics[width=\textwidth,height=4.2cm]{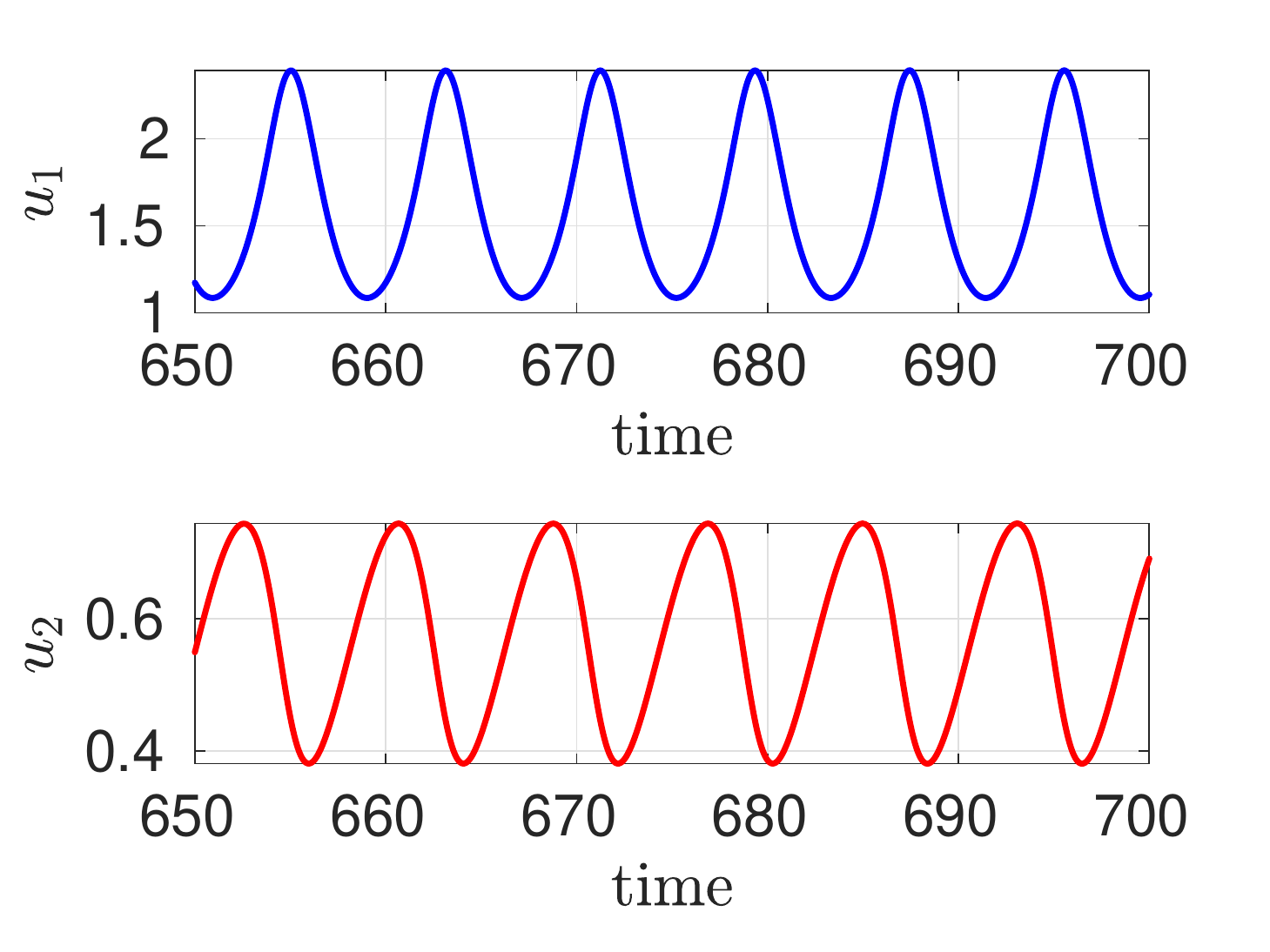}
  \caption{$u_1$ and $u_2$ (FlexPDE)}
   \label{FPDE_2Clusters_ID_cells}
  \end{subfigure}\\
  \begin{subfigure}[b]{0.32\textwidth}
      \includegraphics[width=\textwidth,height=4.2cm]{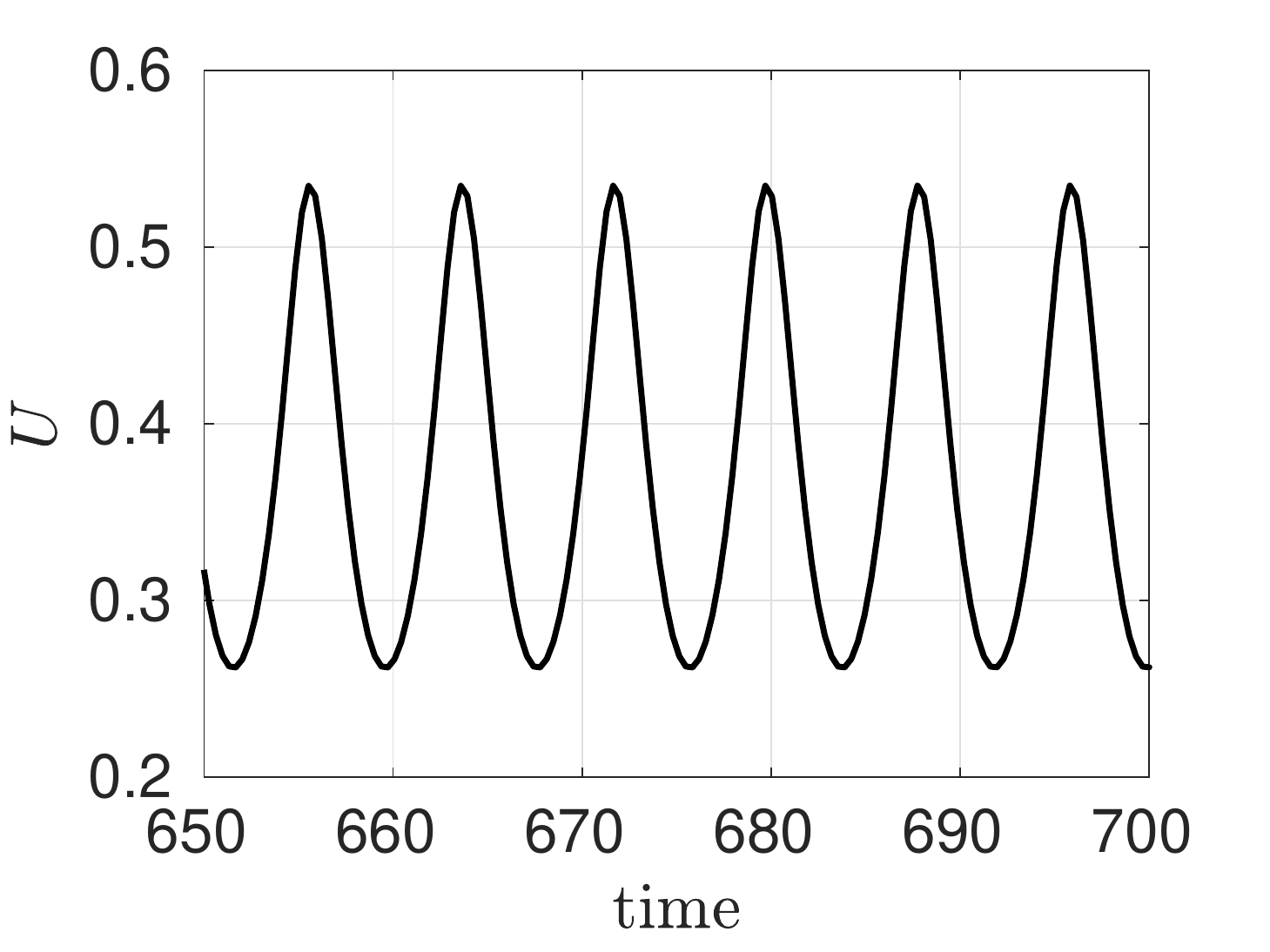}
        \caption{$\bar{U}$ (ODEs)} 
        \label{fig:ubar_reddot}
    \end{subfigure}
    \begin{subfigure}[b]{0.32\textwidth}  
      \includegraphics[width=\textwidth,height=4.2cm]{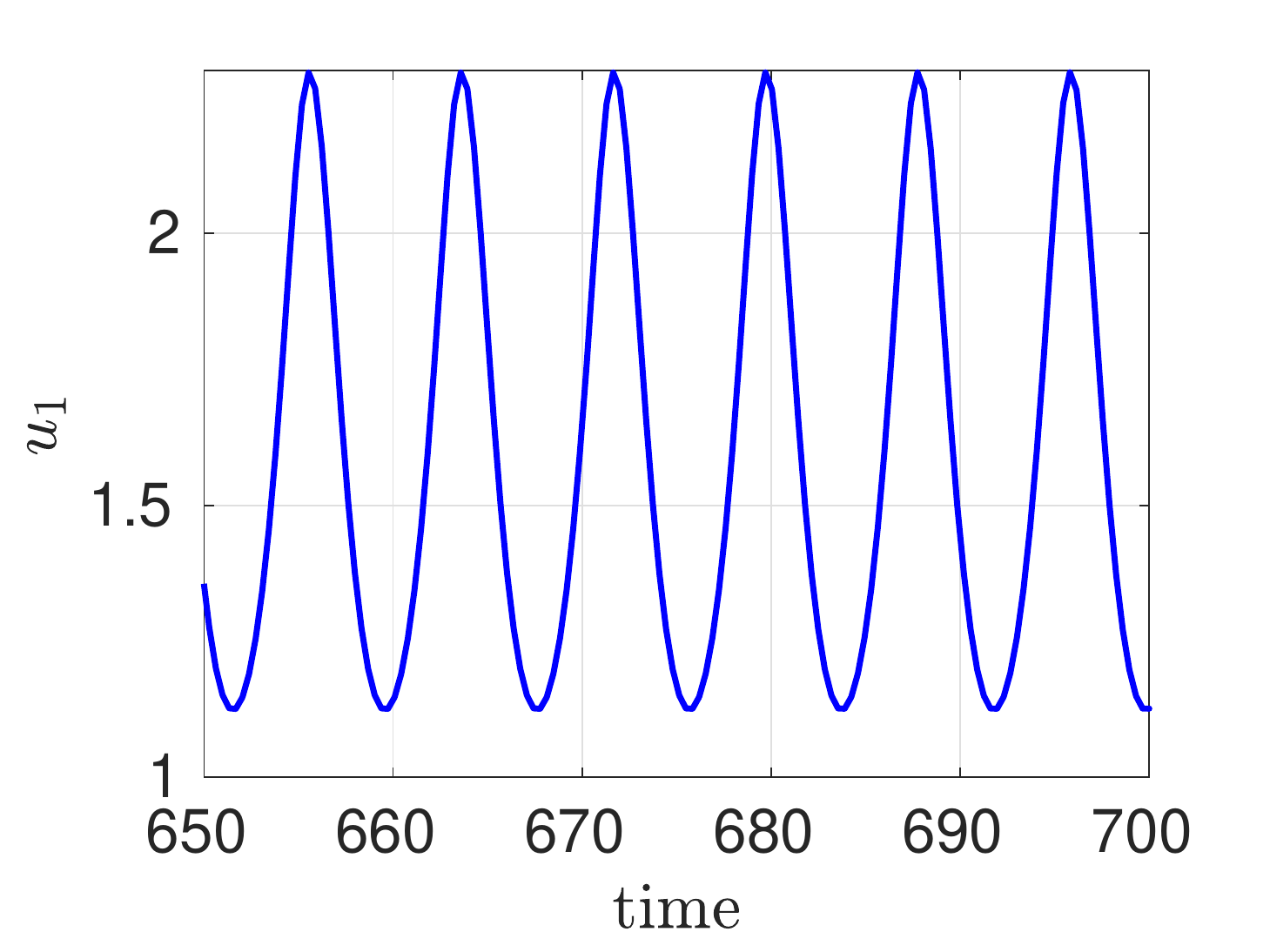}
        \caption{$u_1$ (ODEs)}
        \label{fig:u1_reddot}
      \end{subfigure}
    \begin{subfigure}[b]{0.32\textwidth}  
      \includegraphics[width=\textwidth,height=4.2cm]{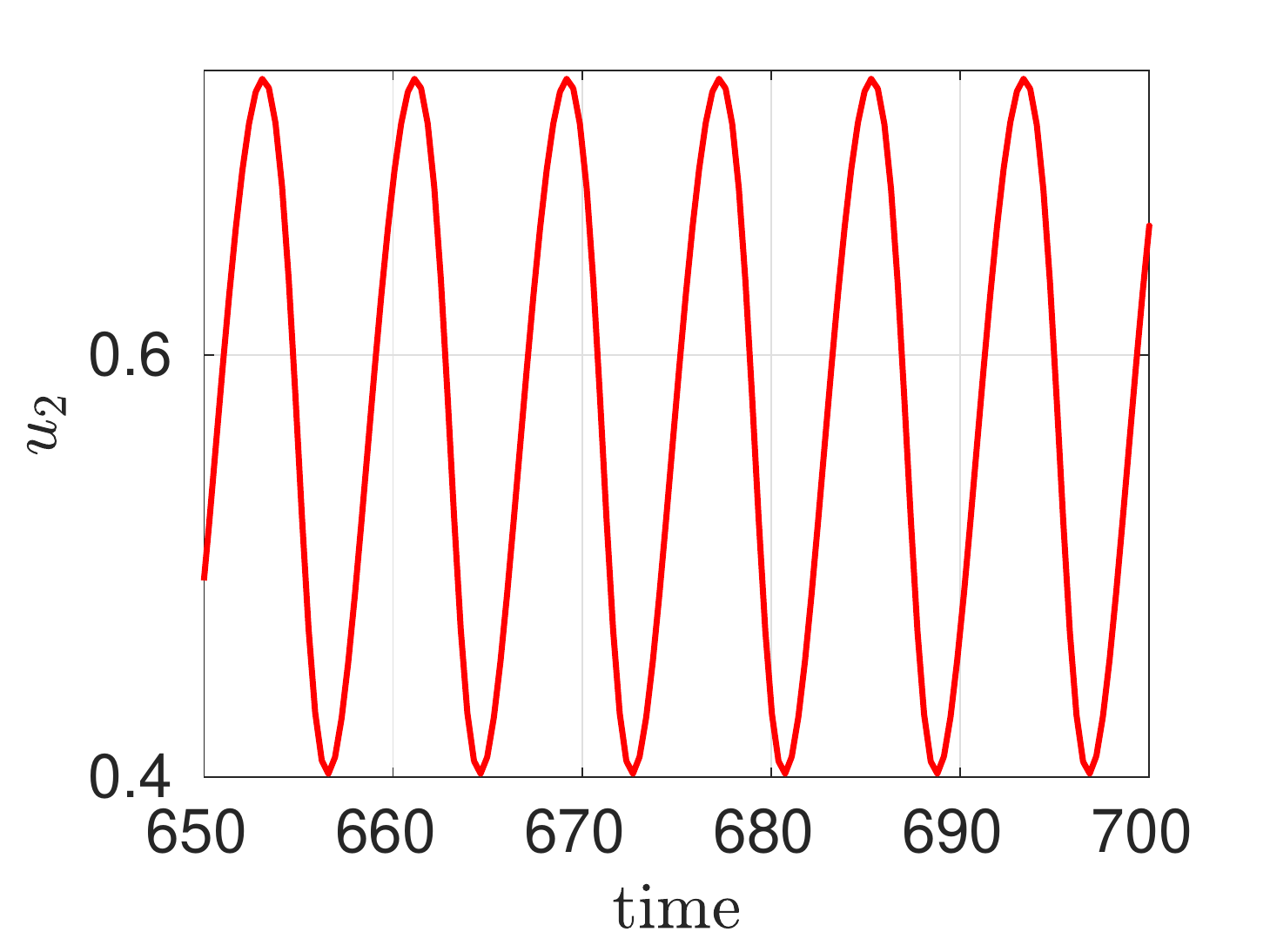}
        \caption{$u_2$ (ODEs)}
        \label{fig:u2_reddot}
    \end{subfigure}
    \caption{Top row: FlexPDE numerical results computed
        from the PDE-ODE model \eqref{DimLess_bulk} for $m=10$
        identical cells, arranged in two clusters (see
        Fig.~\ref{fig:twoclusters}), at the red dot in
        Fig.~\ref{fig:twoclustersHB} where $(D_0,\tau)=(5.0,0.3)$. The
        cells have identical influx rates $d_{1j}=0.8$ for
        $j=1,\ldots,m$ (set I) and almost identical intracellular
        dynamics. The cell locations are in
        Table~\ref{Table:CellLocation} of
        Appendix~\ref{Cell_Location}. Lower row: corresponding results
        for $\bar{U}$, $u_1$, and $u_2$, as computed from the ODE
        system \eqref{reducedODE}. The eigenvector and eigenvalue for
        the GCEP matrix for the linearization are in the top third of
        Table~\ref{eigvec:twoclusters}. The results from the ODEs compare
      well with the FlexPDE simulations.}
  \label{FPDE_2Clusters_ID}
\end{figure}

\begin{figure}[!httbp]
  \centering
    \begin{subfigure}[b]{0.4\textwidth}
      \includegraphics[width=\textwidth,height=4.2cm]{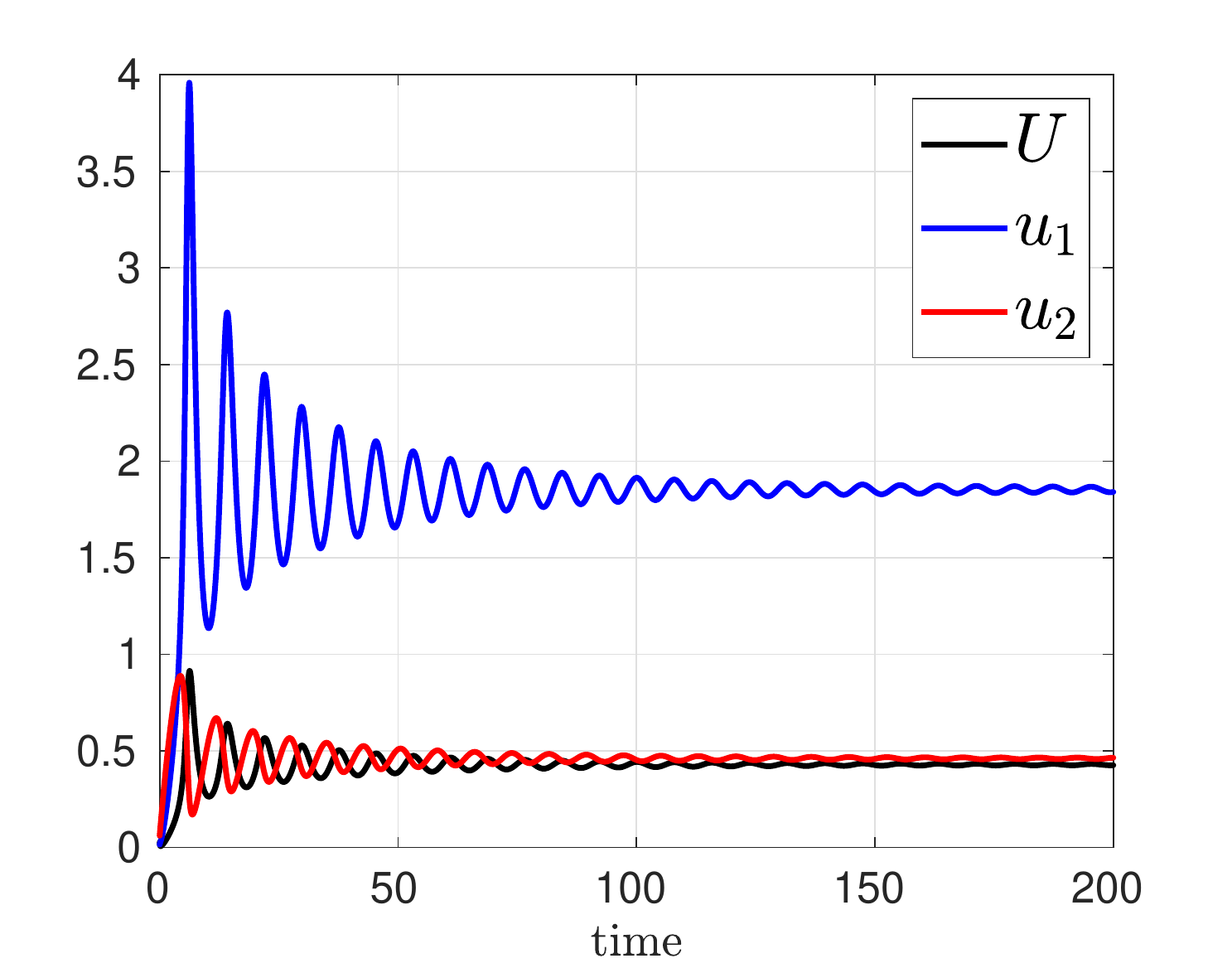}
        \caption{$U$, $u_1$, $u_2$ at $(D_0,\tau)=(5.0,0.9)$ (FlexPDE)} 
        \label{fig:tau09_flexpde}
    \end{subfigure}
    \begin{subfigure}[b]{0.4\textwidth}
      \includegraphics[width=\textwidth,height=4.2cm]{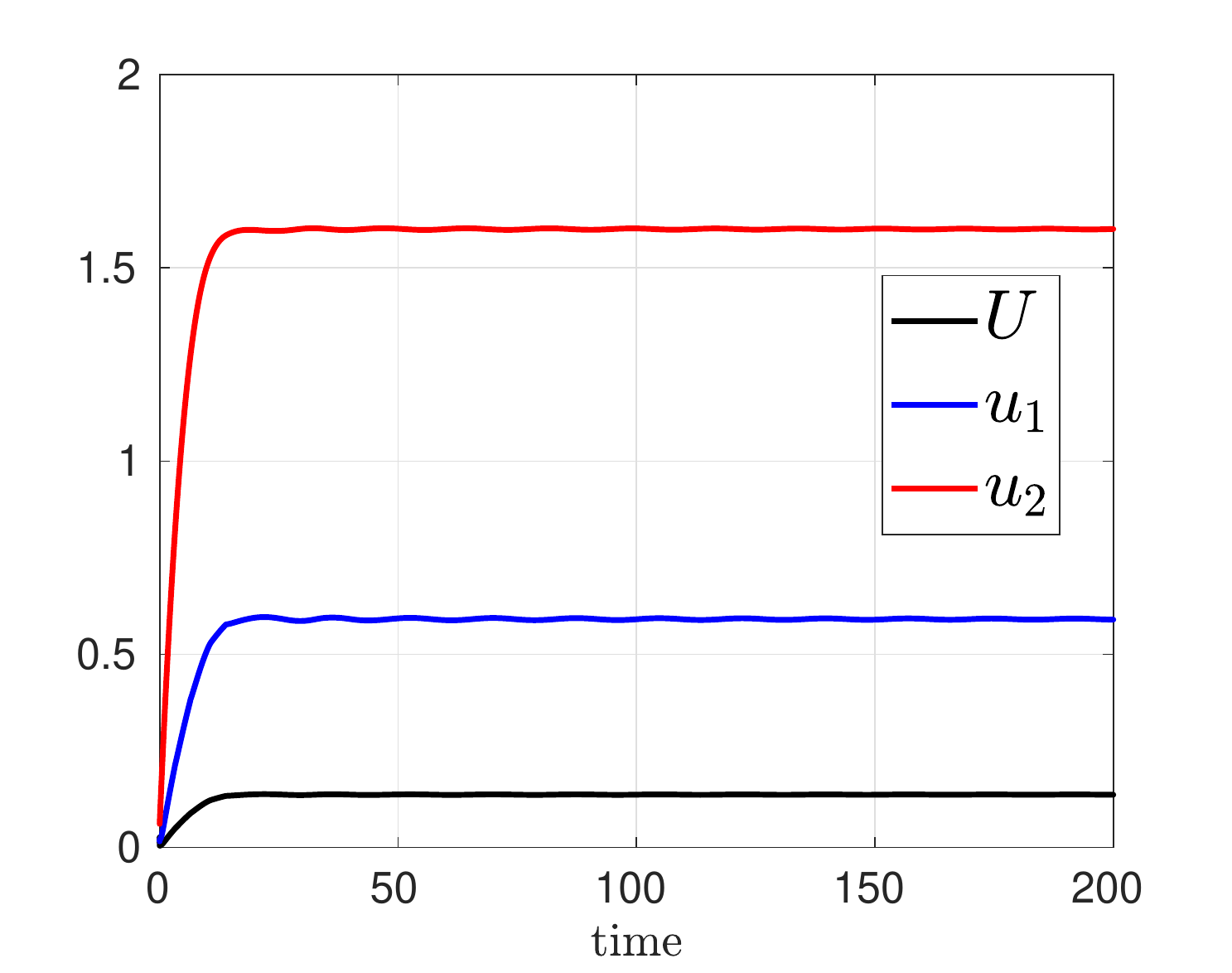}
        \caption{$U$, $u_1$, $u_2$ at $(D_0,\tau)=(5.0,0.03)$ (FlexPDE)} 
        \label{fig:tau03_flexpde}
      \end{subfigure}
   \begin{subfigure}[b]{0.4\textwidth}
      \includegraphics[width=\textwidth,height=4.2cm]{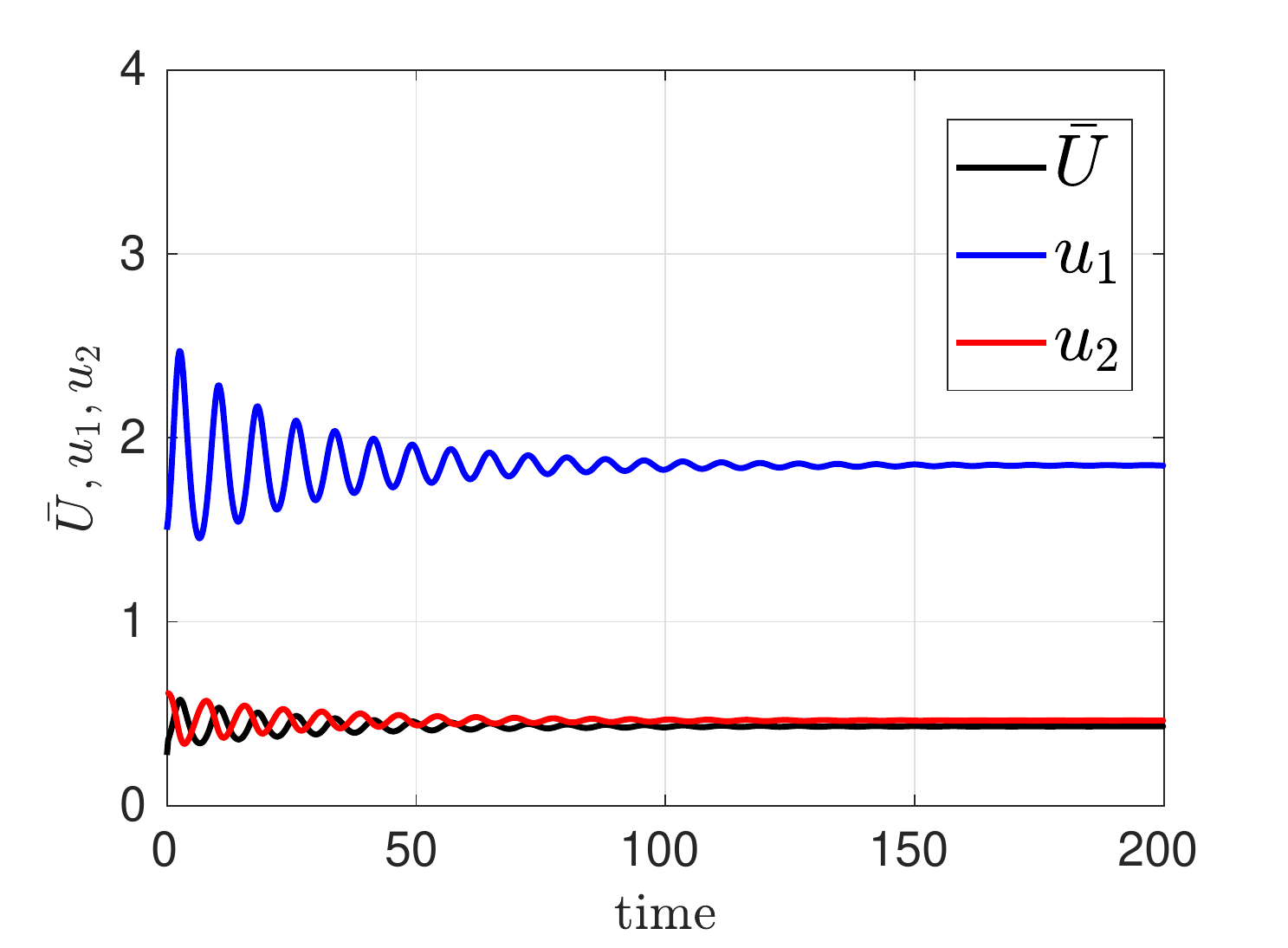}
        \caption{$\bar{U}$, $u_1$, $u_2$ at $(D_0,\tau)=(5.0,0.9)$ (ODEs)} 
        \label{fig:tau09_odes}
    \end{subfigure}
    \begin{subfigure}[b]{0.4\textwidth}
      \includegraphics[width=\textwidth,height=4.2cm]{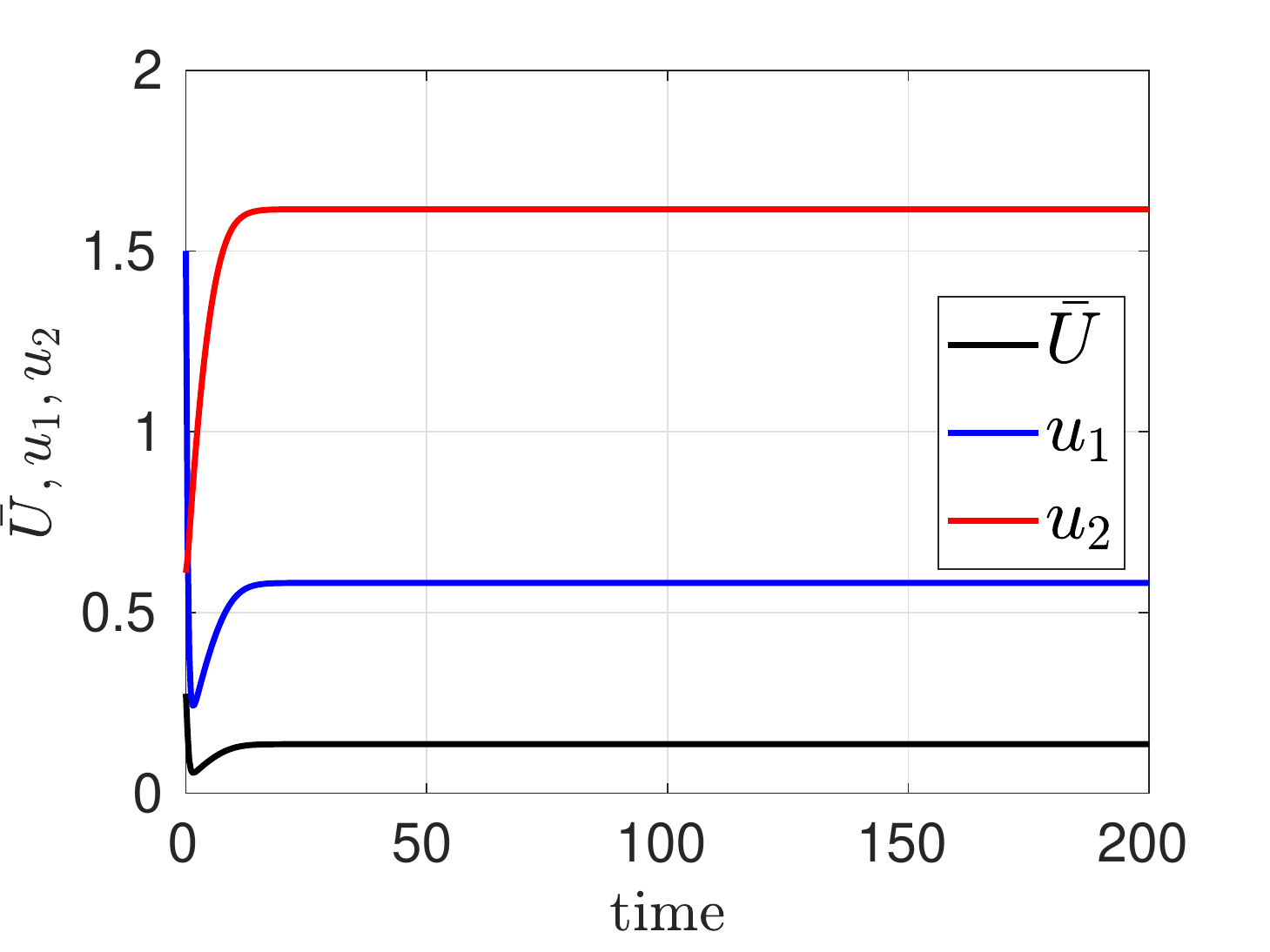}
        \caption{$\bar{U}$, $u_1$, $u_2$ at $(D_0,\tau)=(5.0,0.03)$ (ODEs)} 
        \label{fig:tau03_odes}
    \end{subfigure}
    \caption{Top row: FlexPDE numerical results computed
        from the PDE-ODE model \eqref{DimLess_bulk} for $m=10$
        identical cells, arranged in two clusters (see
        Fig.~\ref{fig:twoclusters}) at the blue star in
        Fig.~\ref{fig:twoclustersHB} where $(D_0,\tau)=(5.0,0.9)$
        (left panel) and at the red star in
        Fig.~\ref{fig:twoclustersHB} where $(D_0,\tau)=(5.0,0.03)$
        (right panel). The cells have identical influx rates
        $d_{1j}=0.8$ for $j=1,\ldots,10$ (set I) and almost identical
        intracellular dynamics. The cell locations are in
        Table~\ref{Table:CellLocation} of Appendix~\ref{Cell_Location}.
        Lower row: corresponding results for
        $\bar{U}$, $u_1$, and $u_2$, as computed from the ODE system
        \eqref{reducedODE}. As expected, there are no sustained oscillations
        and the steady-state is stable.} \label{fig:twoclusters_star}
\end{figure}

\begin{figure}[!httbp]
  \centering
    \includegraphics[width=0.50\textwidth,height=4.2cm]{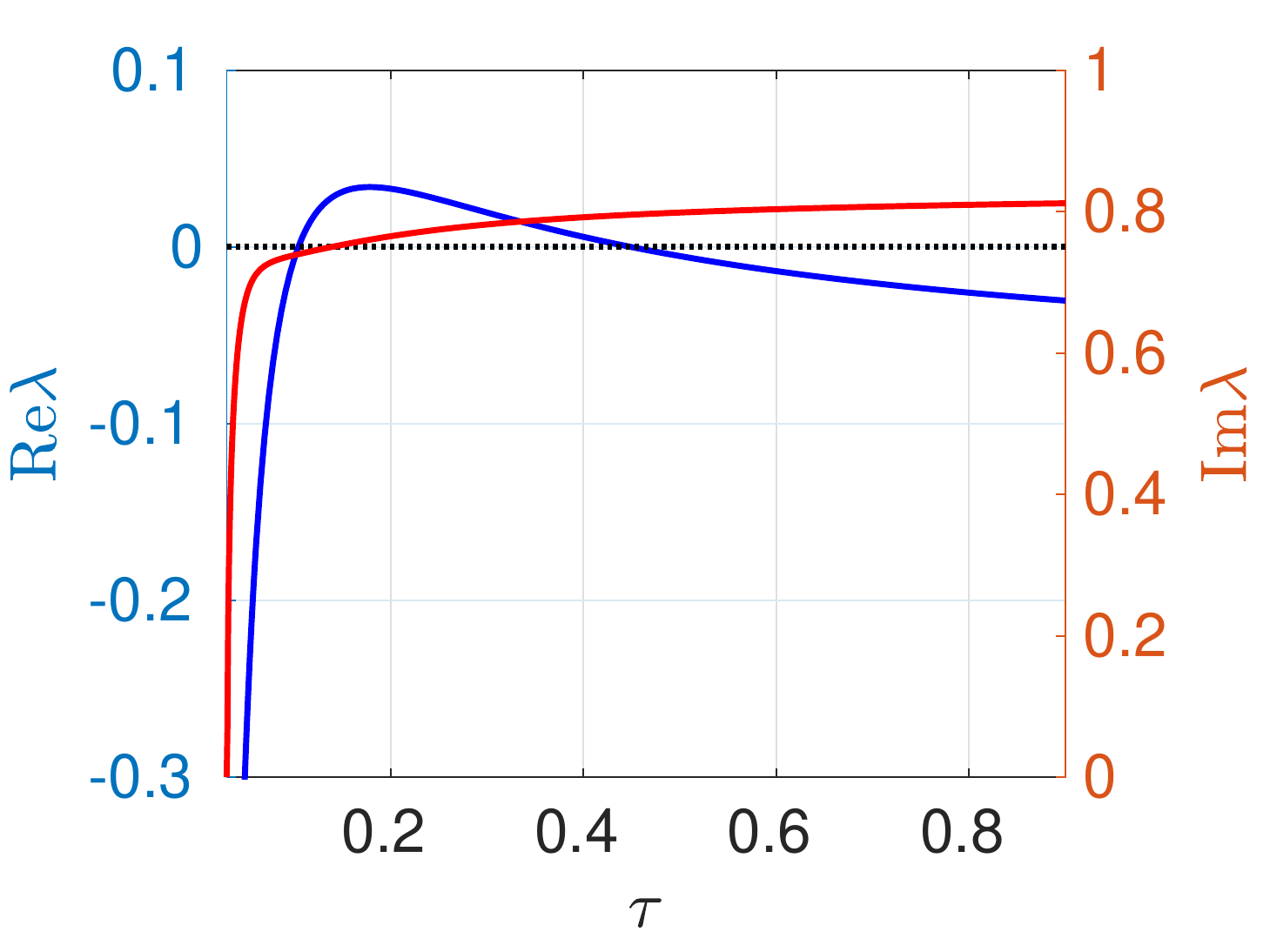}
    \caption{Real (blue curve; left $y$ axis) and imaginary (red
      curve; right $y$ axis) parts of the eigenvalue of the GCEP
      matrix with the largest real part versus $\tau$ along the
      vertical slice $D_0=5$ for $0.03\leq\tau\leq 0.9$ in the phase
      diagram in Fig.~\ref{fig:twoclustersHB}. The eigenvalue is
      computed using root-finding on $\det{\mathcal M}(\lambda)=0$,
      where ${\mathcal M}(\lambda)$ is given in \eqref{Global_System}.
      Observe the two HB values in $\tau$ where $\mbox{Re}(\lambda)=0$
      (dotted black line). When $\tau=0.9$, we have
      $\mbox{Im}(\lambda)\approx 0.81$. However, as $\tau$ decreases
      below the lower HB point, $\mbox{Im}(\lambda)$ decreases to
      zero.} \label{fig:spectrum_tau}
\end{figure}

\begin{figure}[!ht]
  \centering
    \begin{subfigure}[b]{0.40\textwidth}
      \includegraphics[width=\textwidth,height=4.2cm]{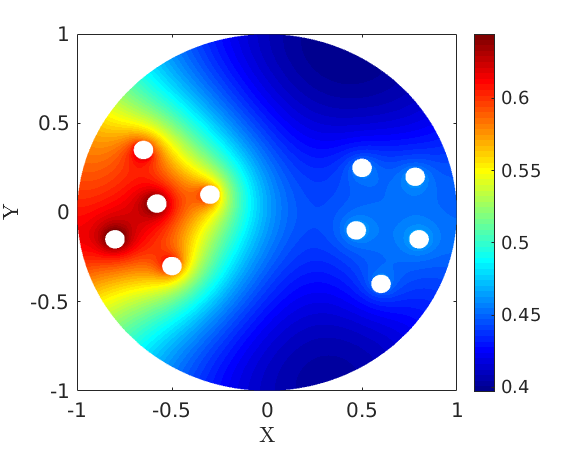}
      \caption{Surface plot at $t=400$}
      \label{FPDE_2Clusters_HalfID_surf}
    \end{subfigure}\\
    \begin{subfigure}[b]{0.32\textwidth}
      \includegraphics[width=\textwidth,height=4.2cm]{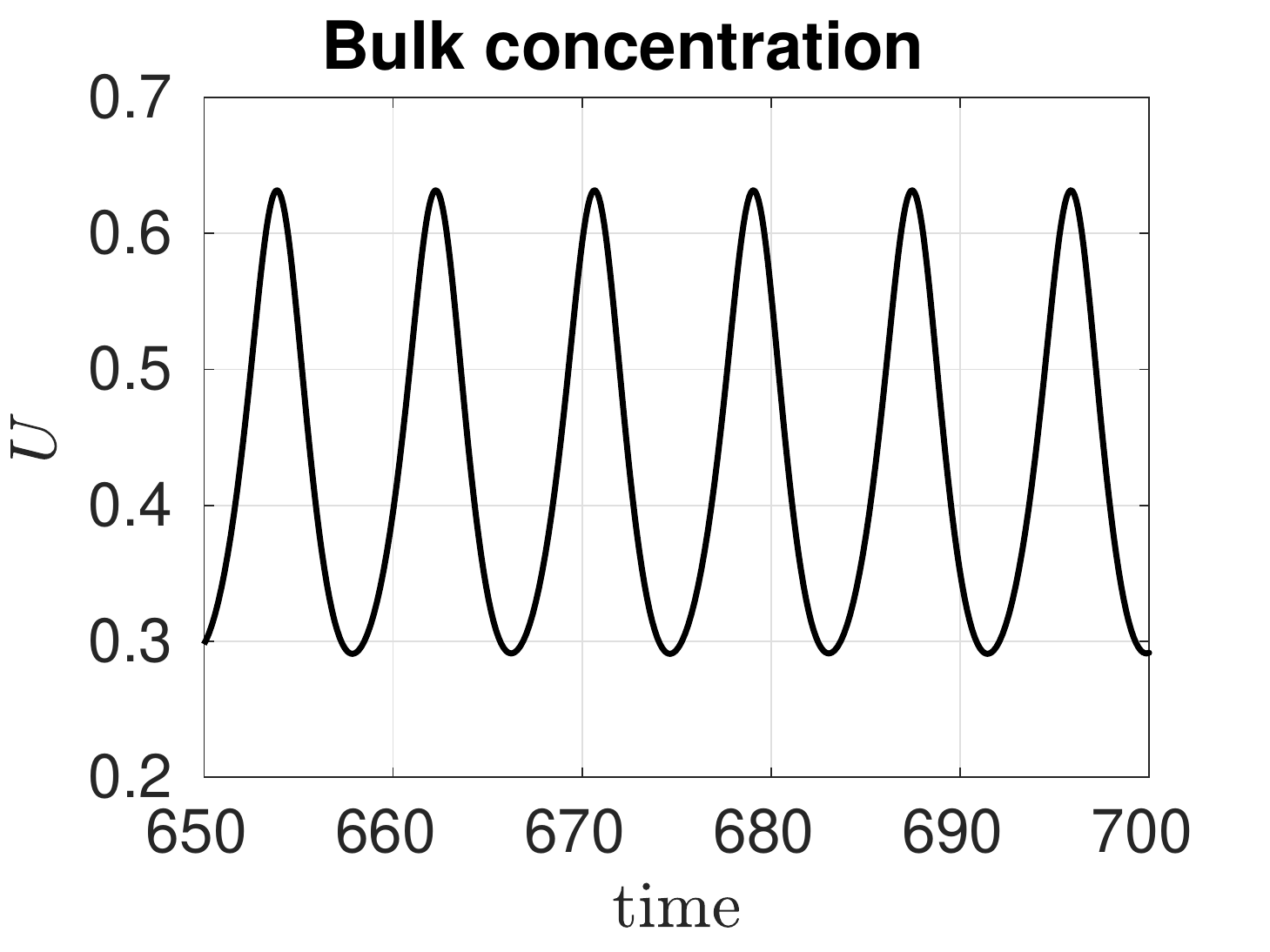}
        \caption{$U$ at ${\pmb x}=(0,0.5)$ (FlexPDE)} 
    \label{FPDE_2Clusters_HalfID_bulk}
    \end{subfigure}
    \begin{subfigure}[b]{0.32\textwidth}
      \includegraphics[width=\textwidth,height=4.2cm]{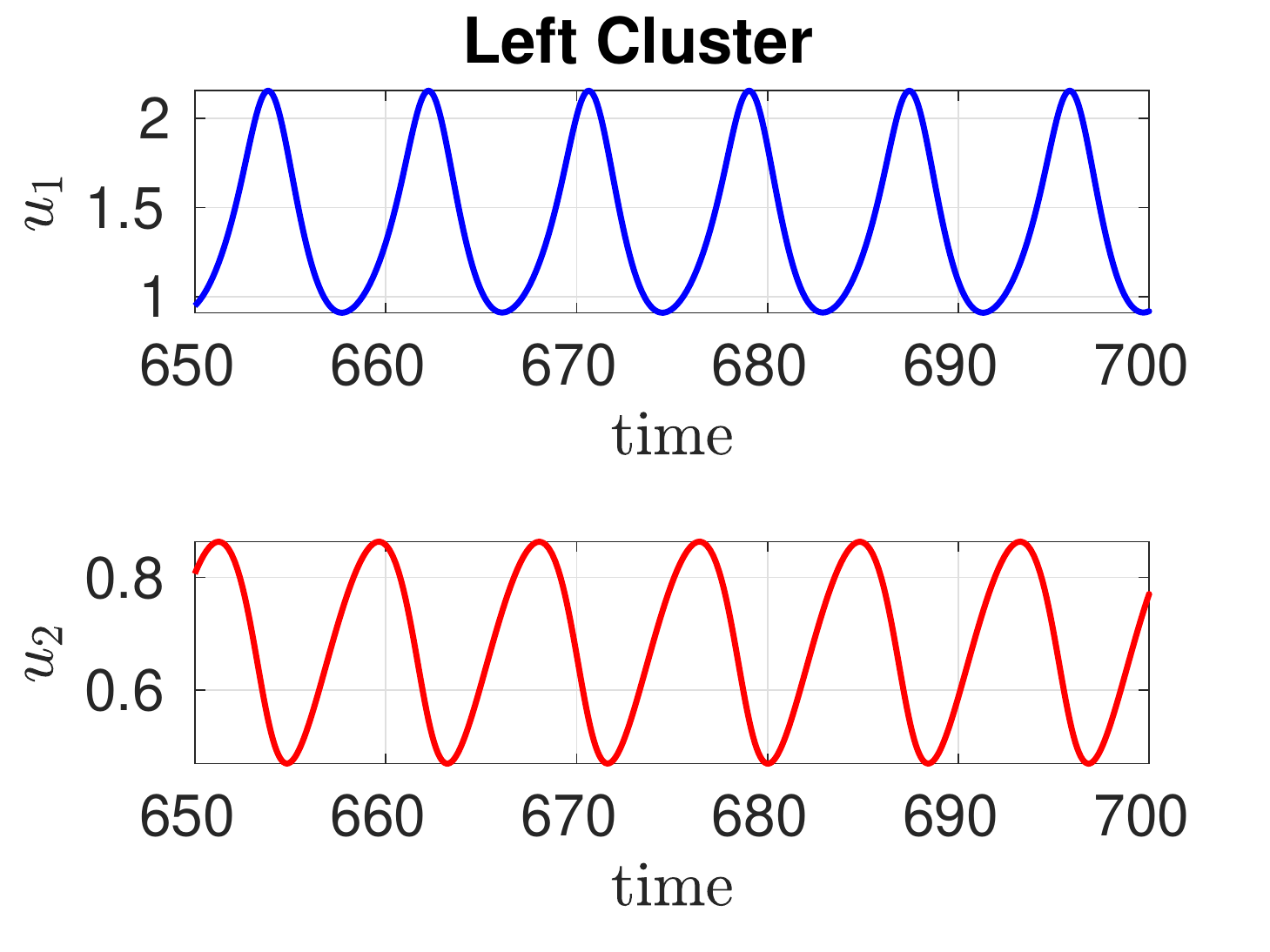}
      \caption{$u_1,u_2$ Left cluster (FlexPDE)}
    \label{FPDE_2Clusters_HalfID_left}
    \end{subfigure}
    \begin{subfigure}[b]{0.32\textwidth}
      \includegraphics[width=\textwidth,height=4.2cm]{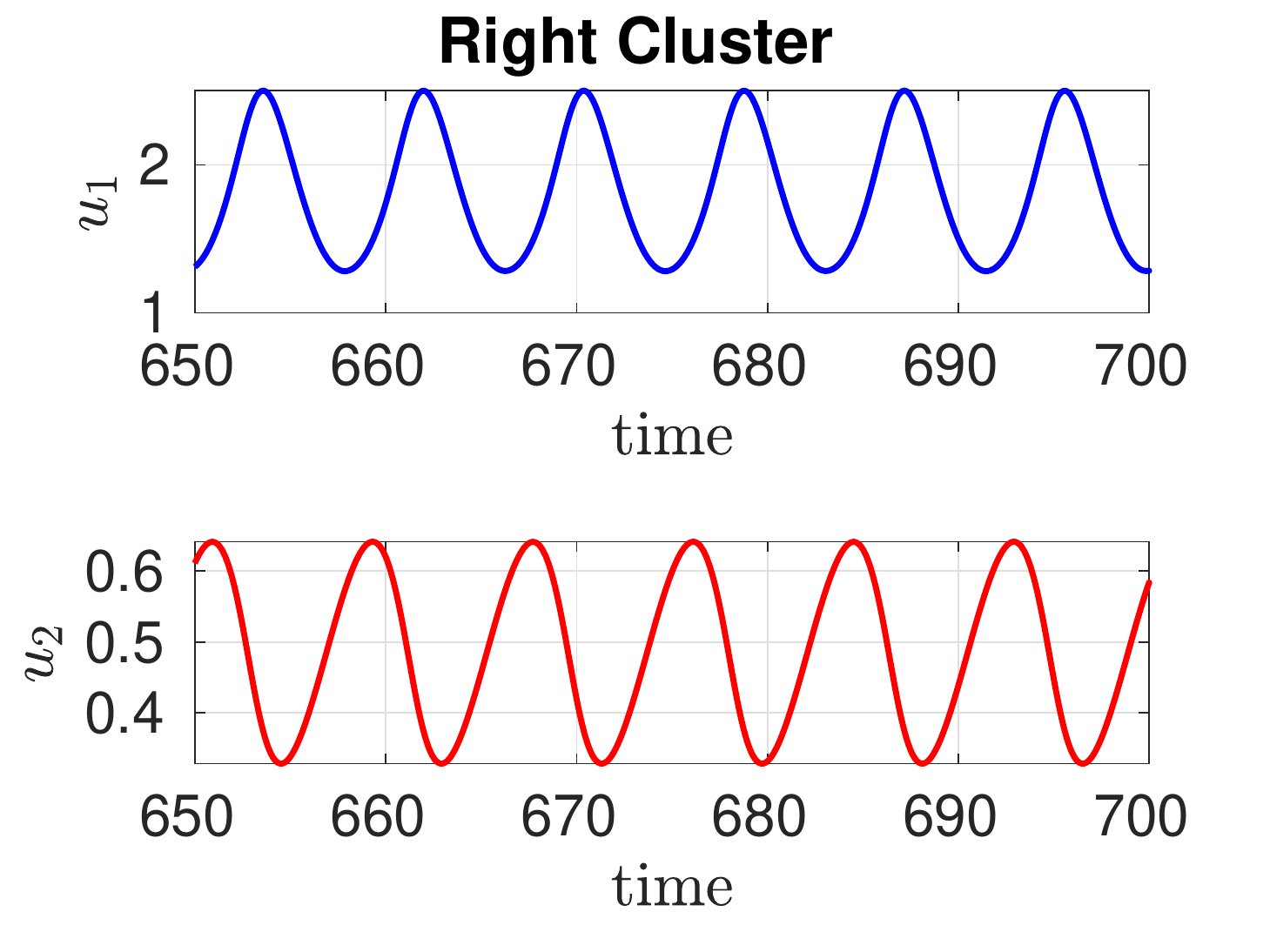}
      \caption{$u_1,u_2$ Right cluster (FlexPDE)}
    \label{FPDE_2Clusters_HalfID_right}
    \end{subfigure}\\
  \begin{subfigure}[b]{0.32\textwidth}
      \includegraphics[width=\textwidth,height=4.2cm]{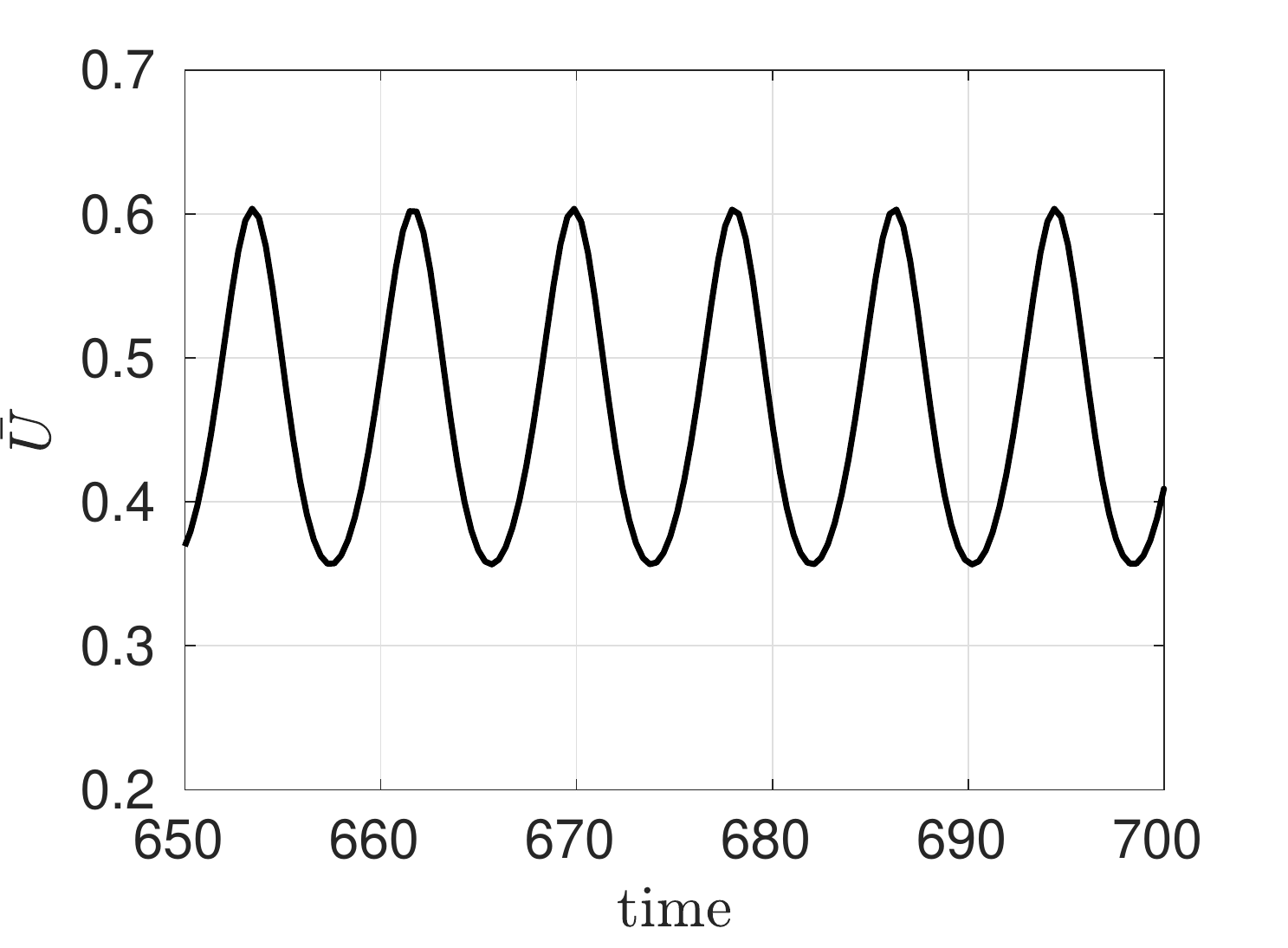}
        \caption{$\bar{U}$ (ODEs)} 
        \label{fig:cluster_bulk_2g}
    \end{subfigure}
    \begin{subfigure}[b]{0.32\textwidth}  
      \includegraphics[width=\textwidth,height=4.2cm]{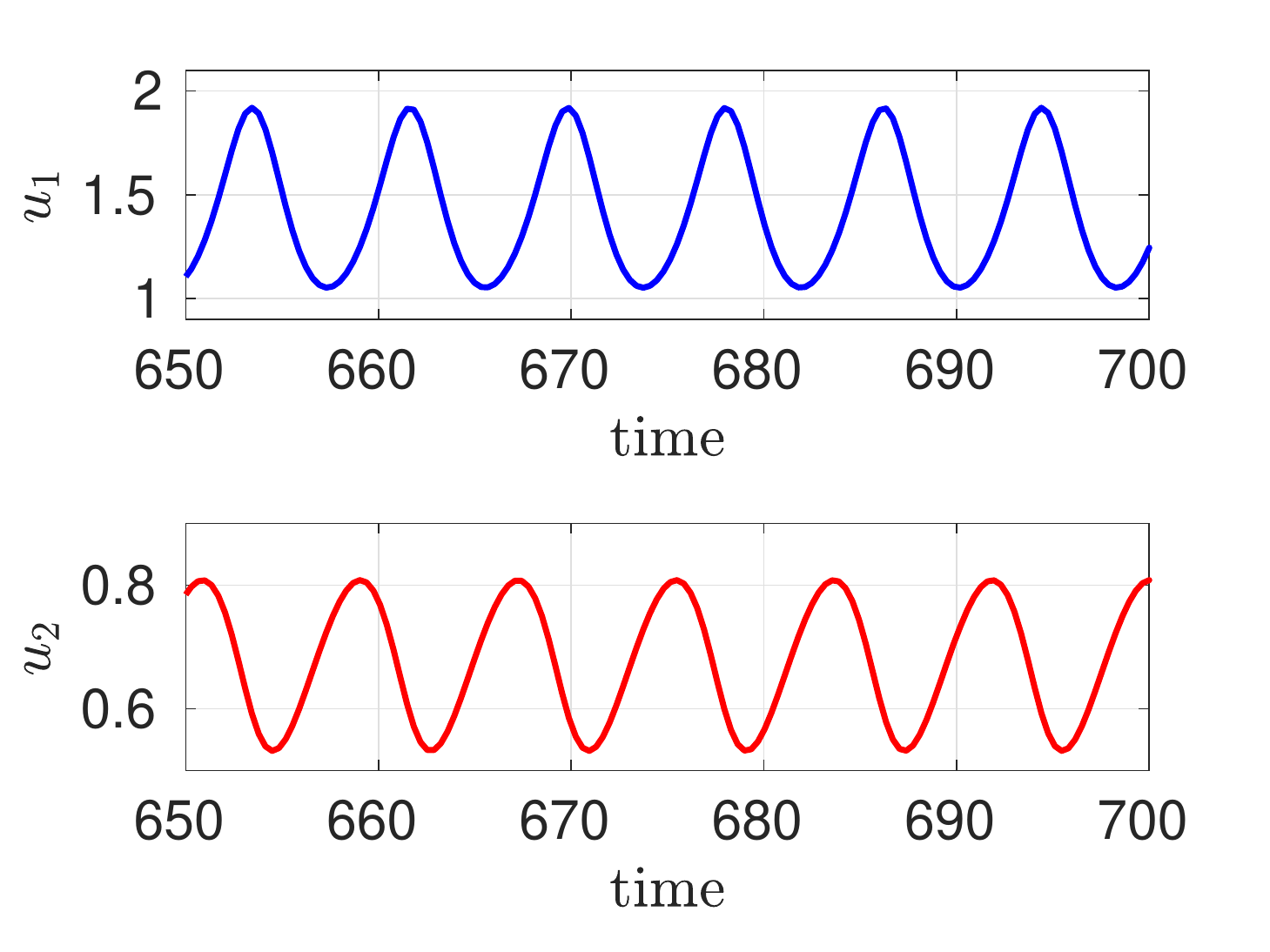}
        \caption{$u_1,u_2$ Left cluster (ODEs)}
        \label{fig:cluster_left_2g}
      \end{subfigure}
    \begin{subfigure}[b]{0.32\textwidth}  
      \includegraphics[width=\textwidth,height=4.2cm]{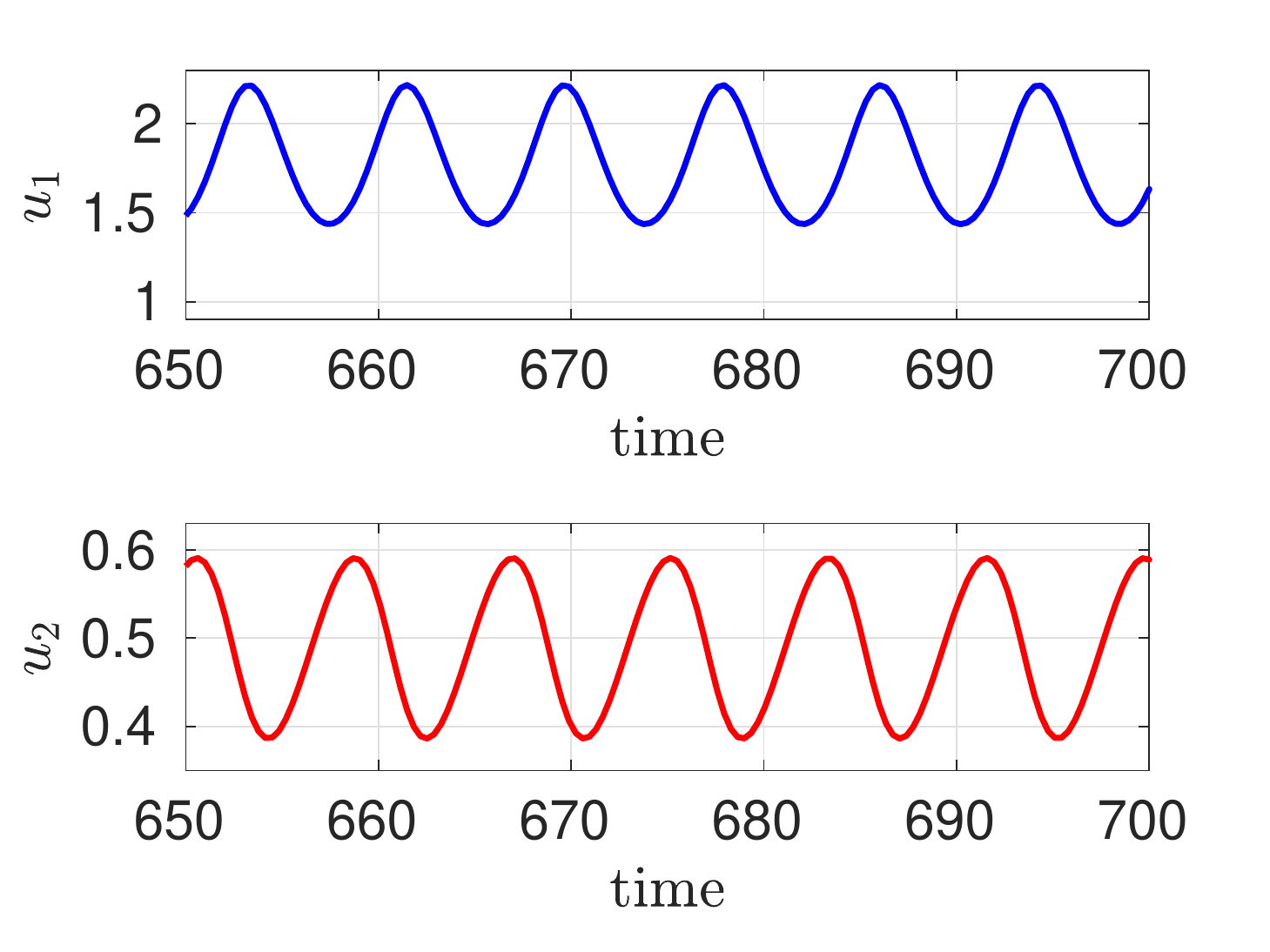}
        \caption{$u_1,u_2$ Right cluster (ODEs)}
        \label{fig:cluster_right_2g}
    \end{subfigure}
  \vspace*{-1ex}
  \caption{Top and middle row: FlexPDE numerical results
      computed from the PDE-ODE model \eqref{DimLess_bulk} for $m=10$
      cells, arranged in two clusters (see Fig.~\ref{fig:twoclusters})
      at the blue dot in Fig.~\ref{fig:twoclustersHB} where
      $(D_0,\tau)=(0.4,0.35)$. The cells in the left cluster have
      identical influx rates $d_{1}=0.4$ while the cells in the right
      cluster have $d_1=0.8$.  The cell locations are in
      Table~\ref{Table:CellLocation} of Appendix~\ref{Cell_Location}.
      Within each cluster there is very similar intracellular
      dynamics. Lower row: corresponding results for $\bar{U}$, $u_1$,
      and $u_2$, as computed from the ODE system
      \eqref{reducedODE}. The eigenvector and eigenvalue for the GCEP
      matrix for the linearization are in the middle third of
      Table~\ref{eigvec:twoclusters}.}
  \label{fig:cluster_all_2g}
\end{figure}

In the top row of Fig.~\ref{fig:twoclusters_star} we show full FlexPDE
results computed from \eqref{DimLess_bulk} at the blue and red stars
in Fig.~\ref{fig:twoclustersHB} corresponding to $(D_0,\tau)=(5,0.9)$
(top left panel) and $(D_0,\tau)=(5,0.03)$ (top right panel),
respectively. In this case, where the cells are all identical with
$d_1=0.8$, the linear stability analysis predicts that the
steady-state is linearly stable and that no sustained intracellular
oscillations should occur. This prediction is confirmed from the
FlexPDE simulations. In the bottom panels of
Fig.~\ref{fig:twoclusters_star} we show that the corresponding results
predicted by the ODE system \eqref{reducedODE} compare very favorably
with the FlexPDE results with regards to the long time limiting
behavior. At the red star point in Fig.~\ref{fig:twoclustersHB} we
calculate from the root-finding condition
$\det{\mathcal M}(\lambda)=0$, where ${\mathcal M}$ is given in
\eqref{Global_System}, that the eigenvalue nearest the origin in
$\mbox{Re}(\lambda)<0$ is $\lambda\approx -1.01 + 0.000202\, i$.  This
predicts a monotone, non-oscillatory, decay to the steady-state. This
feature is observed in the right panels of
Fig.~\ref{fig:twoclusters_star}.  Alternatively, at the blue star
point in Fig.~\ref{fig:twoclusters_star} we calculate that the nearest
eigenvalue in $\mbox{Re}(\lambda)<0$ for the GCEP matrix at
$(D_0,\tau)=(5.0,0.9)$ is $\lambda\approx -0.030 + 0.8104\, i$. This
eigenvalue predicts an oscillatory decay to the steady-state, and is
confirmed by the results in the left panels of
Fig.~\ref{fig:twoclusters_star}. To further explain this transition
between monotone and oscillatory decay to the steady-state, in
Fig.~\ref{fig:spectrum_tau} we plot the real and imaginary parts of
the eigenvalue with the largest real part of the GCEP matrix along the
vertical slice $D_0=5$ for $0.03\leq\tau\leq 0.9$ in the phase diagram
in Fig.~\ref{fig:twoclustersHB}, as obtained by numerically solving
$\det{\mathcal M}(\lambda)=0$. This figure shows that the imaginary
part of this eigenvalue becomes small when $\tau$ decreases below the
lower HB boundary.

In the top and middle rows of Fig.~\ref{fig:cluster_all_2g} we show
FlexPDE simulation results corresponding to the blue dot in
Fig.~\ref{fig:twoclustersHB} where $(D_0,\tau)=(0.4,0.35)$. For this
case, the influx rate for the cells in the right and left clusters
were assigned as $d_1=0.8$ and $d_1=0.4$, respectively (parameter set
II in Table~\ref{Table:CellLocation} of
Appendix~\ref{Cell_Location}). The corresponding normalized eigenpair
${\mathcal K}\pmb{c}$ obtained from the GCEP matrix, as given in the
middle third of Table~\ref{eigvec:twoclusters}, predicts that the
cells will synchronize their oscillations within their respective
groups, but that there will be a slight phase difference in the
intracellular oscillations between the two groups. The results from
the full FlexPDE simulations in the middle row of
Fig.~\ref{fig:cluster_all_2g} confirm these predictions from the
linearized theory. However, from comparing the middle and bottom rows
of Fig.~\ref{fig:cluster_all_2g}, we observe that results from the ODE
system \eqref{reducedODE} do not compare as favorably with FlexPDE
simulations as they do in Fig.~\ref{FPDE_2Clusters_ID} when $D_0=5$.
This poorer agreement is likely due to the fact that there is a
noticeable spatial gradient in the bulk signal at the lower value
$D_0=0.4$, as observed from the surface plot of
Fig.~\ref{FPDE_2Clusters_HalfID_surf}. Moreover,
Fig.~\ref{FPDE_2Clusters_HalfID_surf} shows that the bulk signal can
concentrate, as expected, in the left cluster owing to the lower rate
of chemical influx into the cells in this cluster.

\begin{figure}[!ht]
  \centering
   \begin{subfigure}[b]{0.32\textwidth}  
      \includegraphics[width=\textwidth,height=4.2cm]{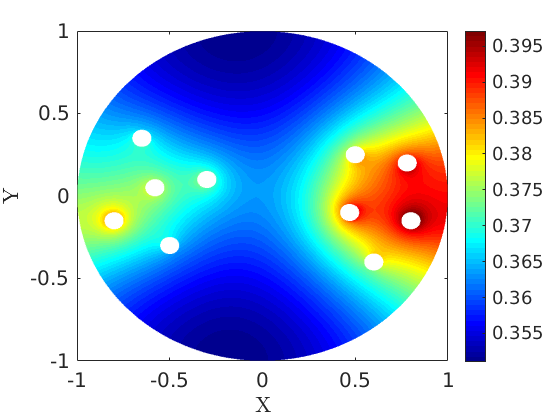}
        \caption{surface plot at $t=400$}
    \label{FPDE_2Clusters_NonID_Surfplot}
  \end{subfigure}
   \begin{subfigure}[b]{0.32\textwidth}  
      \includegraphics[width=\textwidth,height=4.2cm]{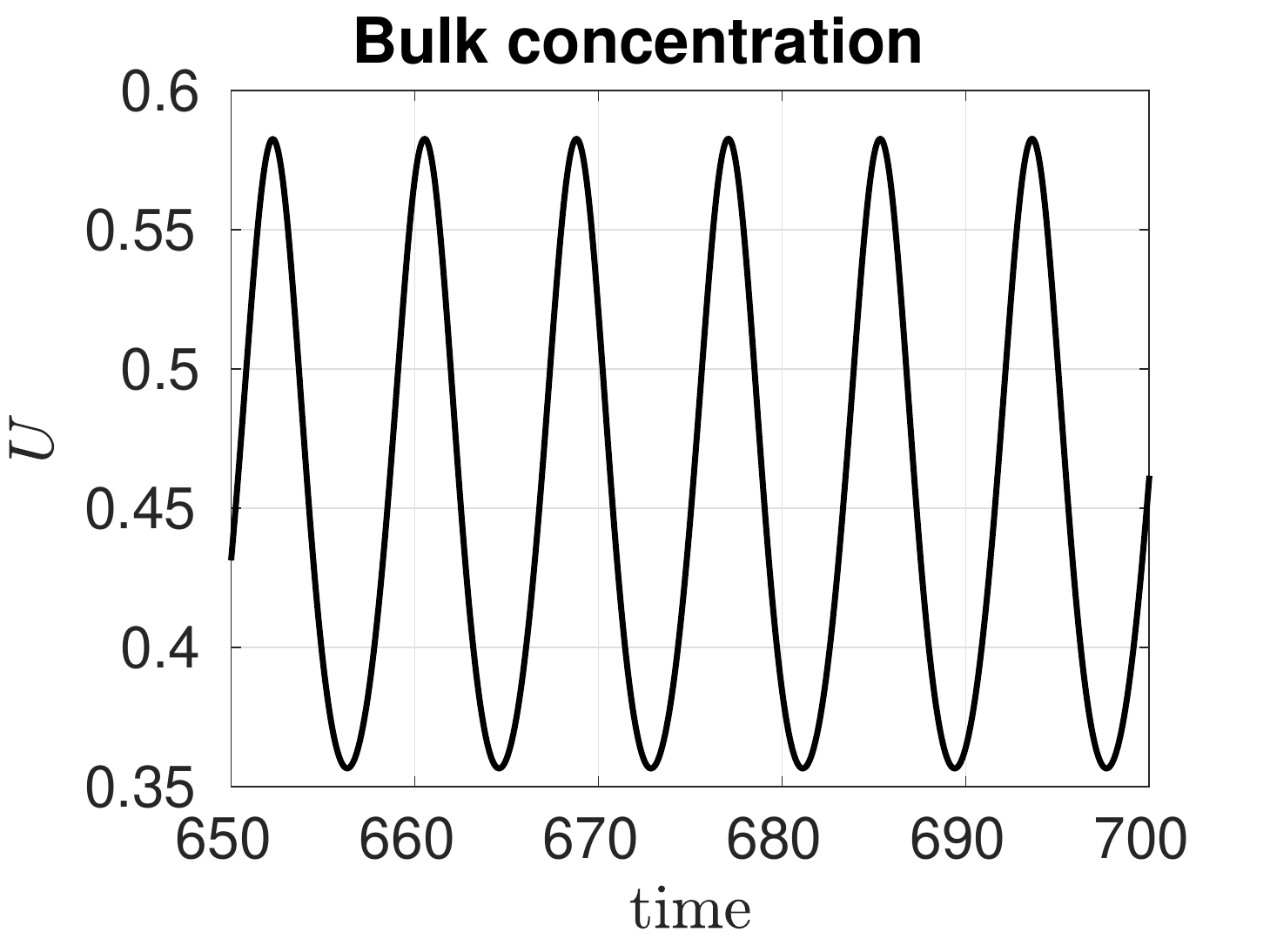}
        \caption{$U$ at $\pmb{x}=(0.0.5)$ (FlexPDE)}
    \label{FPDE_2Clusters_NonID_Surf_bulk}
  \end{subfigure}
   \begin{subfigure}[b]{0.32\textwidth}  
      \includegraphics[width=\textwidth,height=4.2cm]{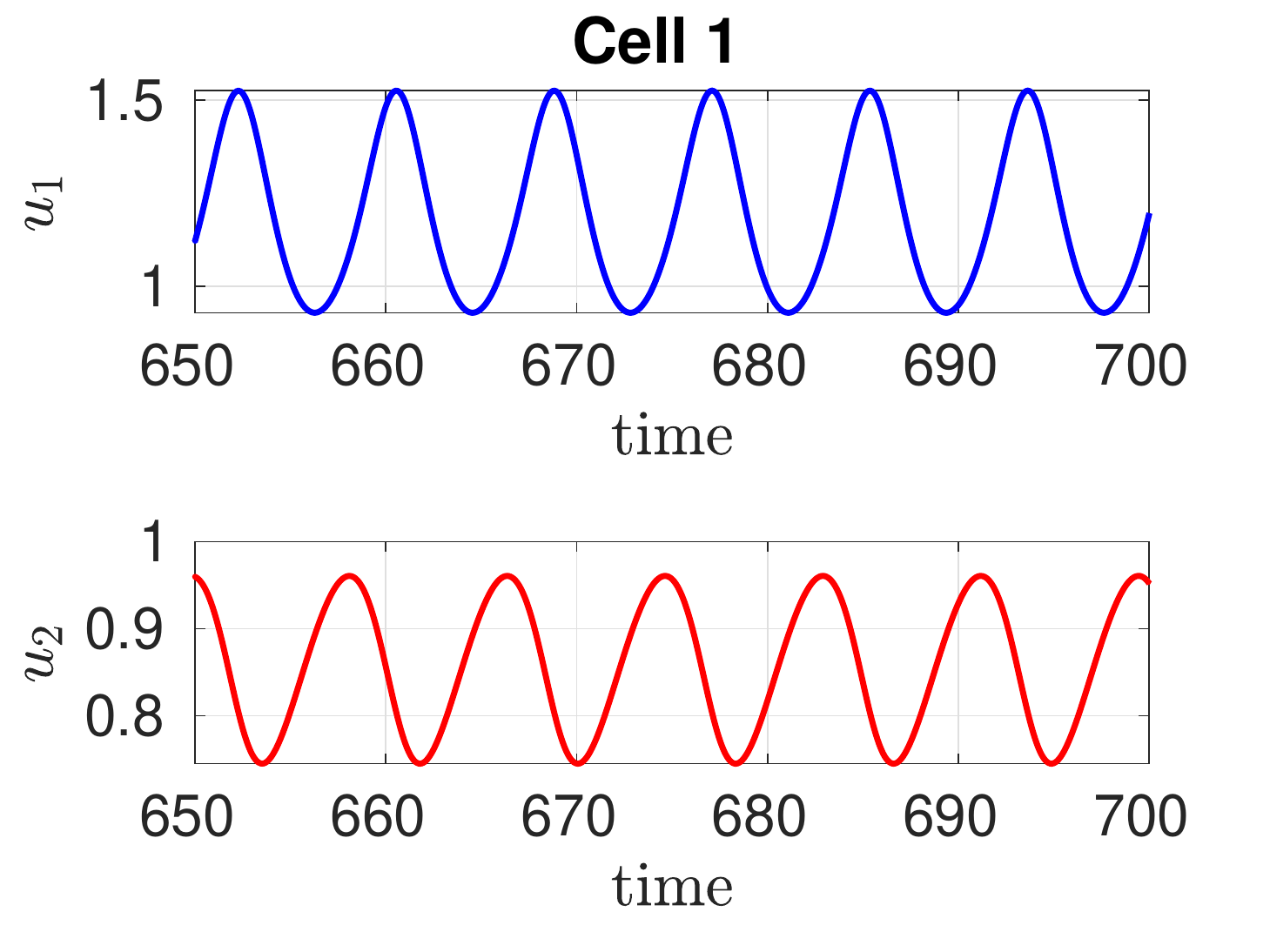}
        \caption{$u_1$, $u_2$ Cell 1 (FlexPDE)}
    \label{FPDE_2Clusters_NonID_cell1}
  \end{subfigure}\\
   \begin{subfigure}[b]{0.32\textwidth}  
      \includegraphics[width=\textwidth,height=4.2cm]{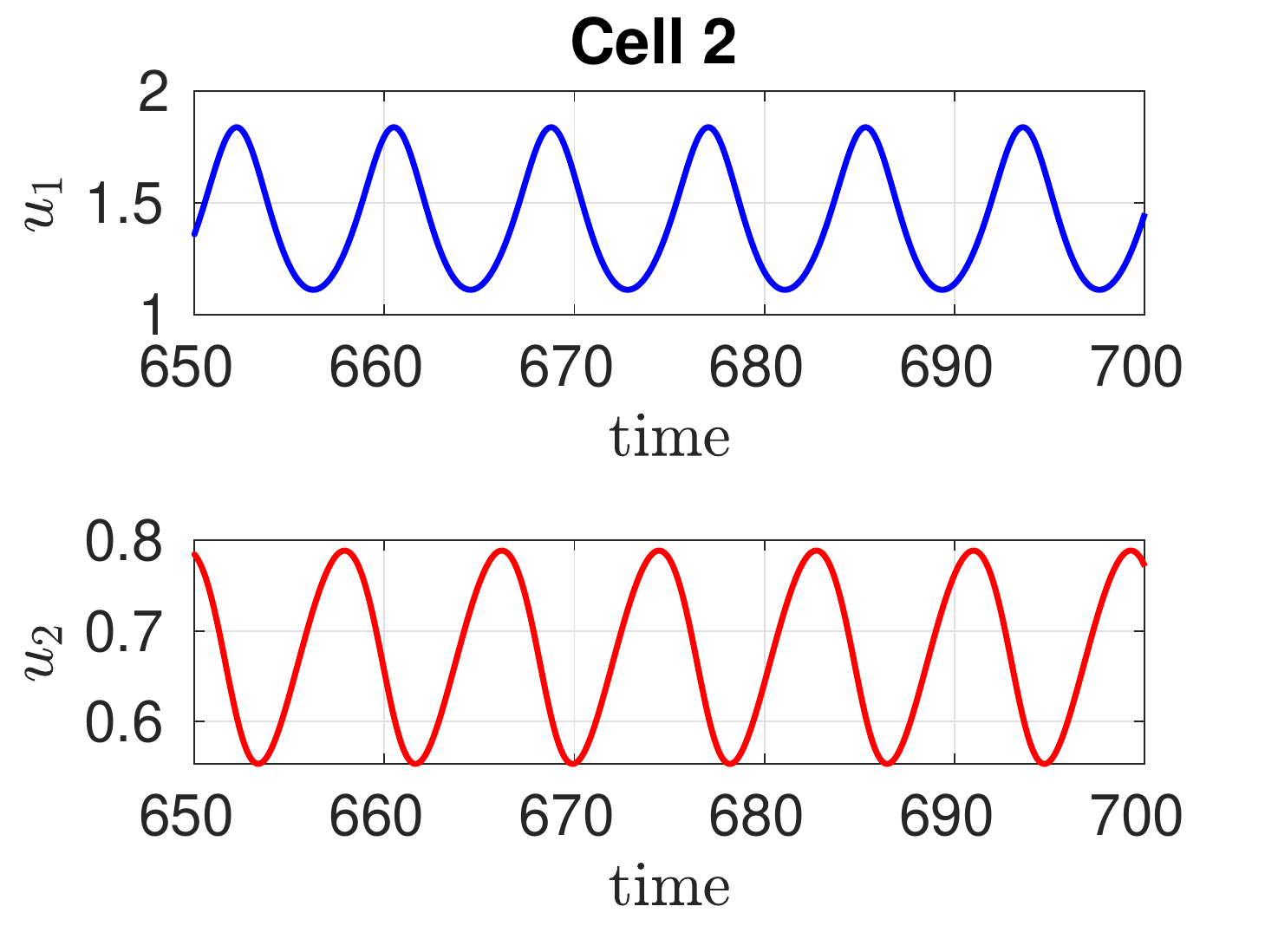}
        \caption{$u_1$, $u_2$ Cell 2 (FlexPDE)}
    \label{FPDE_2Clusters_NonID_cell2}
  \end{subfigure}
   \begin{subfigure}[b]{0.32\textwidth}  
      \includegraphics[width=\textwidth,height=4.2cm]{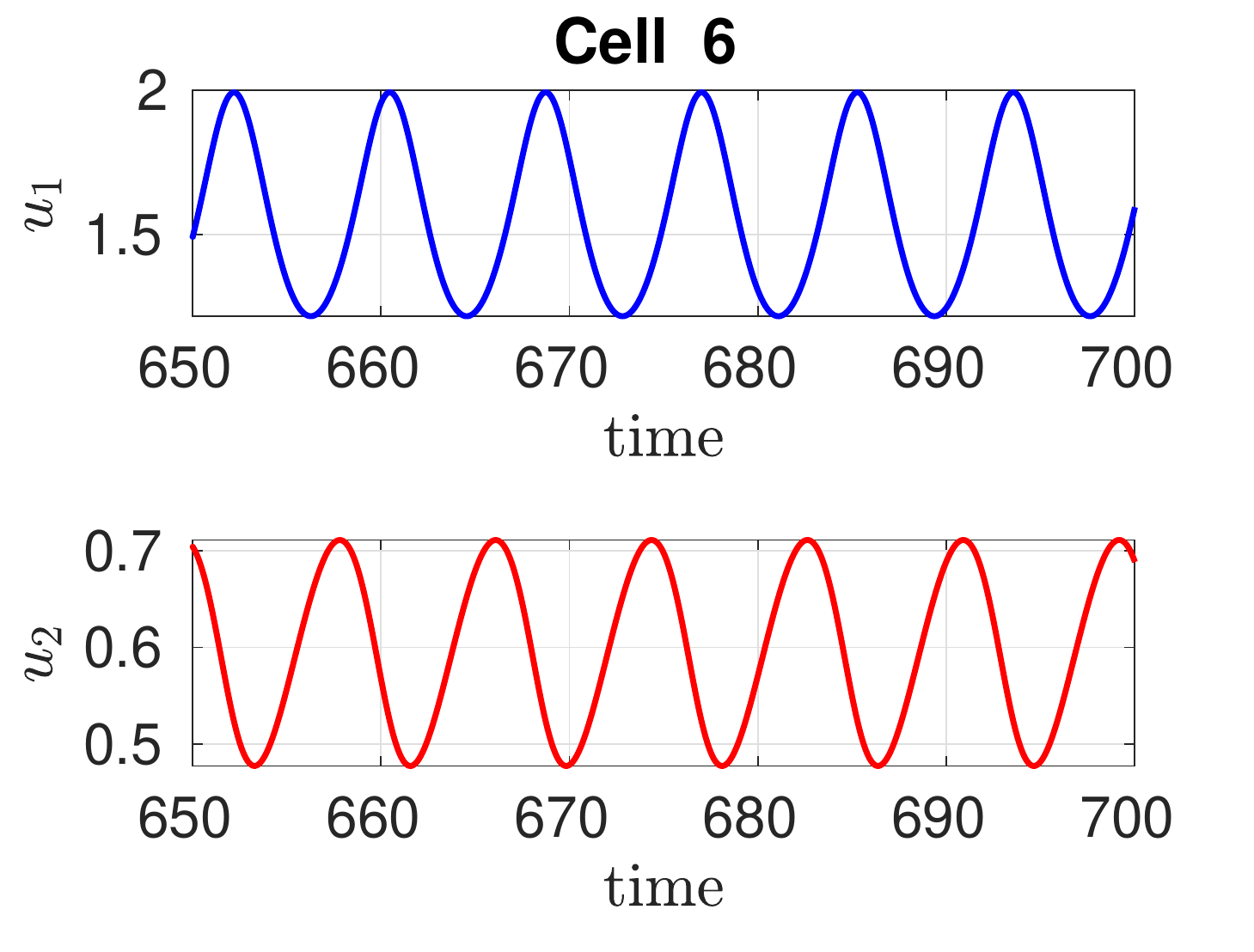}
        \caption{$u_1$, $u_2$ Cell 6 (FlexPDE)}
    \label{FPDE_2Clusters_NonID_cell6}
  \end{subfigure}
   \begin{subfigure}[b]{0.32\textwidth}  
      \includegraphics[width=\textwidth,height=4.2cm]{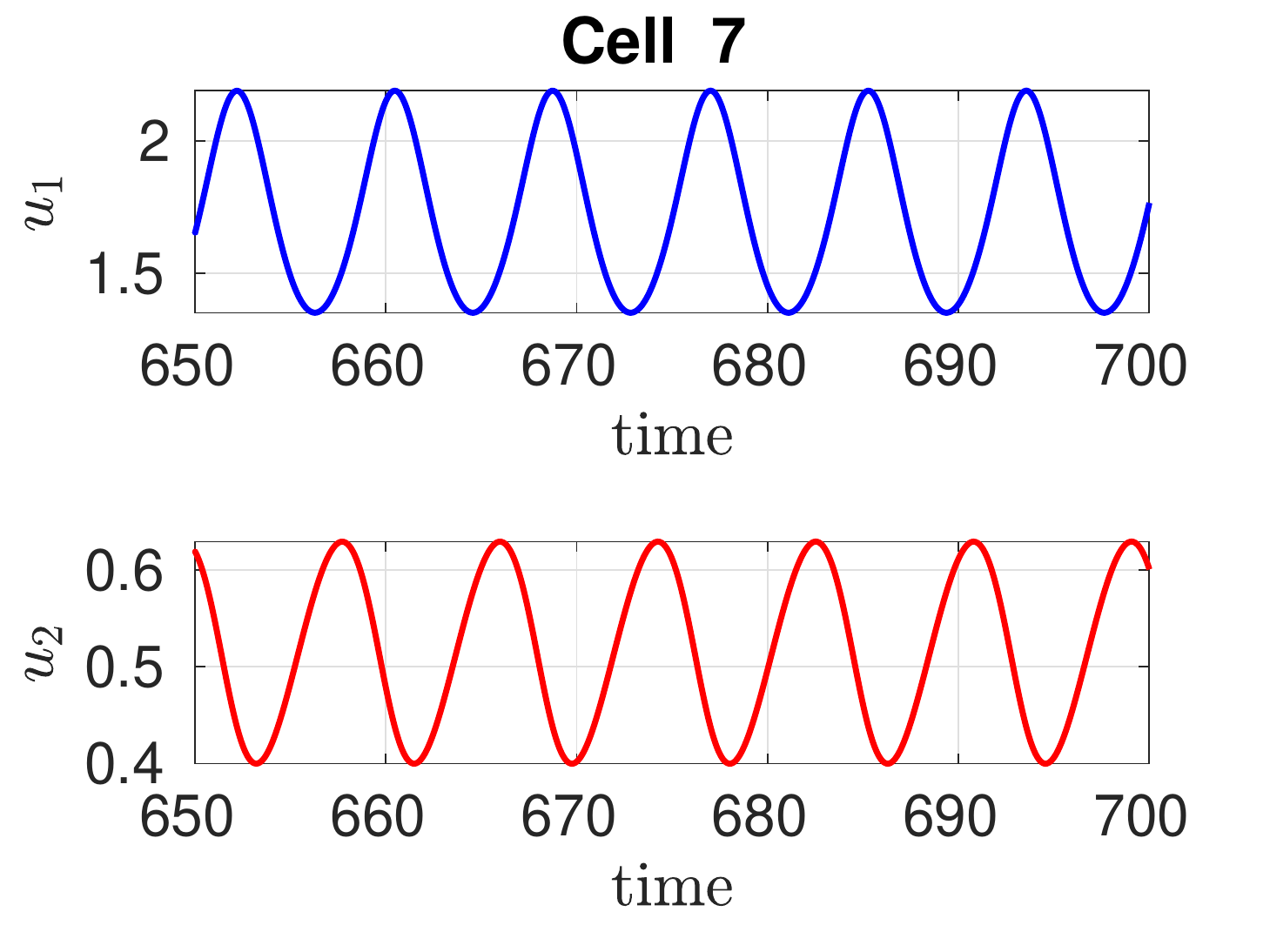}
        \caption{$u_1$, $u_2$ Cell 7 (FlexPDE)}
    \label{FPDE_2Clusters_NonID_cell7}
  \end{subfigure}\\
   \begin{subfigure}[b]{0.32\textwidth}  
      \includegraphics[width=\textwidth,height=4.2cm]{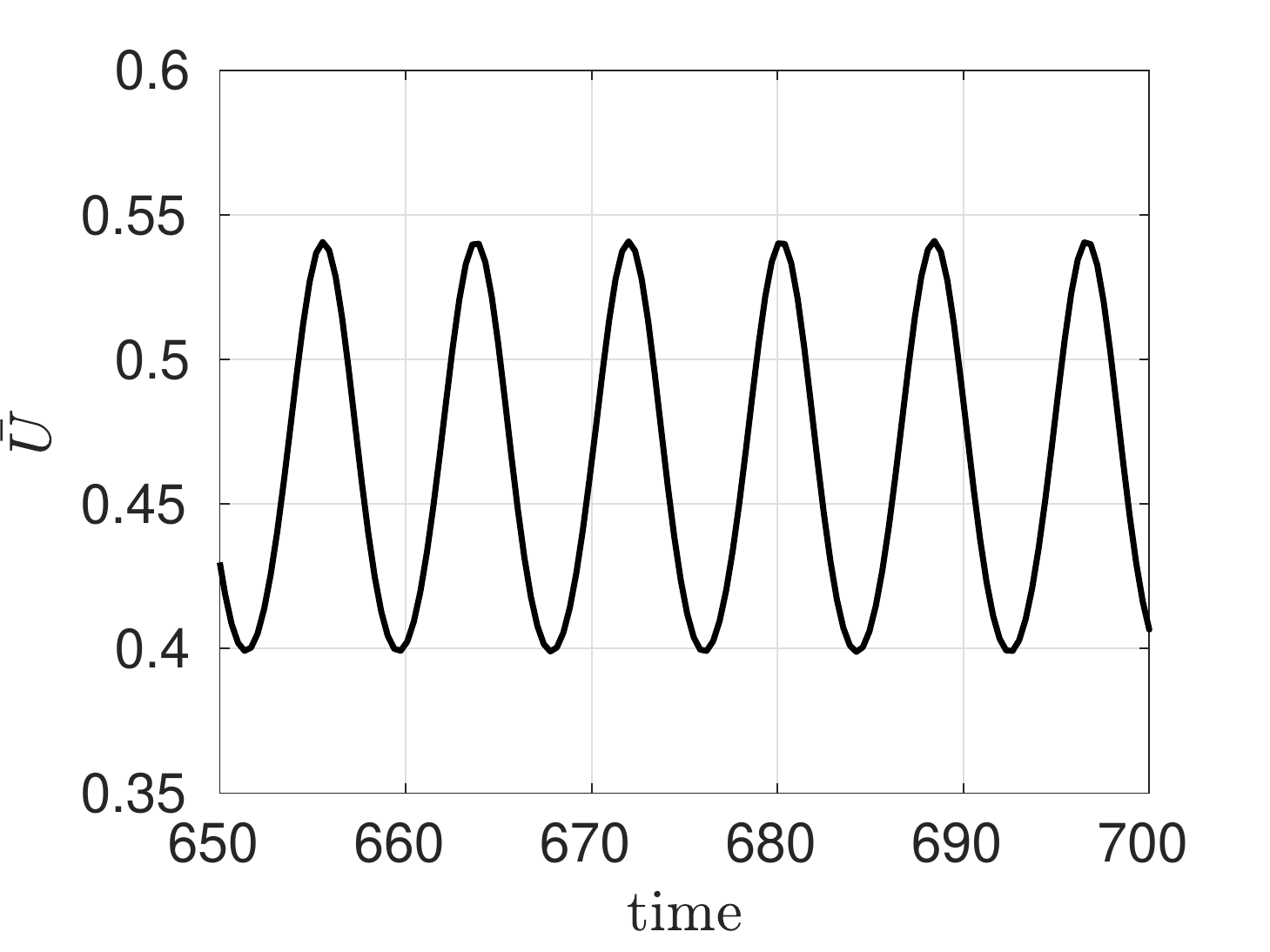}
        \caption{$\bar{U}$ (ODEs)}
    \label{fig:cluster_bulk_arb}
  \end{subfigure}
   \begin{subfigure}[b]{0.32\textwidth}  
      \includegraphics[width=\textwidth,height=4.2cm]{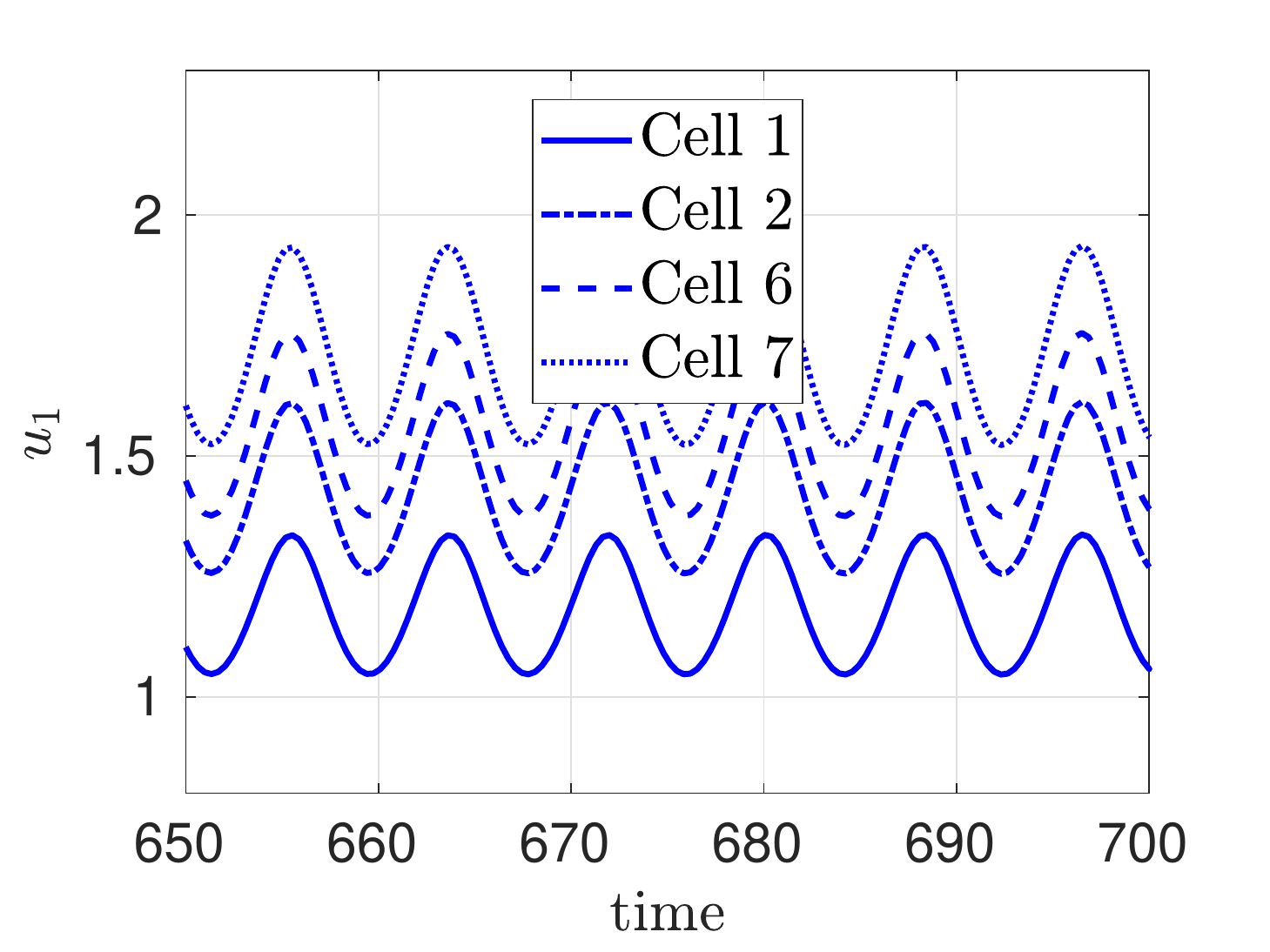}
      \caption{$u_1$, Cells 1,2,6,7 (ODEs)}
          \label{fig:cluster_u1_arb}
  \end{subfigure}
   \begin{subfigure}[b]{0.32\textwidth}  
      \includegraphics[width=\textwidth,height=4.2cm]{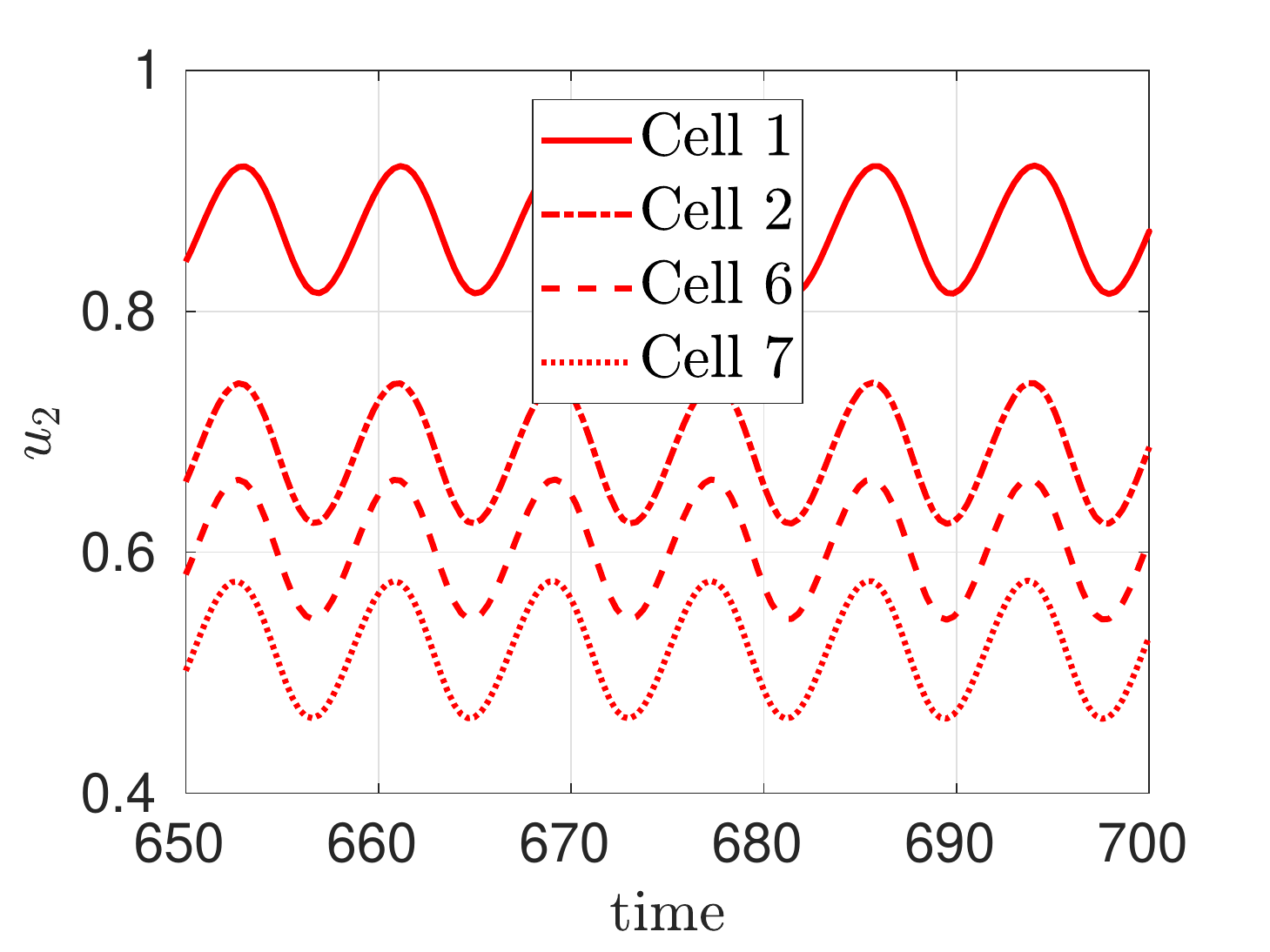}
        \caption{$u_2$, Cells 1,2,6,7 (ODEs)}
              \label{fig:cluster_u2_arb}
  \end{subfigure}\\
  \vspace*{-1ex}
  \caption{Top and middle row: FlexPDE numerical results
      computed from the PDE-ODE model \eqref{DimLess_bulk} for $m=10$
      cells, arranged in two clusters (see Fig.~\ref{fig:twoclusters})
      at the green dot in Fig.~\ref{fig:twoclustersHB} where
      $(D_0,\tau)=(1.5,0.3)$. The influx rates are selected uniformly
      from the interval $0.4 \leq d_1 \leq 0.8$ (see
      Table~\ref{Table:CellLocation} for set III). The cell locations
      are in Table~\ref{Table:CellLocation} of
      Appendix~\ref{Cell_Location}, with Cell 1 at $(0.47,-0.1)$, Cell
      2 at $(0.78,0.2)$, Cell 6 at $(-0.58,0.05)$, and Cell 7 at
      $(-0.65,0.35)$. Each cell has different intracellular dynamics
      owing to the different influx rates. Lower row: corresponding
      results for $\bar{U}$, $u_1$, and $u_2$, as computed from the
      ODE system \eqref{reducedODE}. The eigenvector
      and eigenvalue for the GCEP matrix for the linearization are in
      the lower third of Table~\ref{eigvec:twoclusters}.}
  \label{fig:cluster_arb}
\end{figure}

Next, we give FlexPDE results computed from \eqref{DimLess_bulk}
corresponding to the green dot in Fig.~\ref{fig:twoclustersHB} where
$(D_0,\tau)=(1.5,0.3)$, corresponding to the two cluster cell pattern
of Fig.~\ref{fig:twoclusters}). For this case, the influx rates are
selected randomly from the interval $0.4 \leq d_1 \leq 0.8$ and are
given in Table~\ref{Table:CellLocation} under permeability set III.
In the top and middle row of Fig.~\ref{fig:cluster_arb} we show the
FlexPDE results for the bulk solution at the point $(0,0.5)$ as well
as the intracellular dynamics in cells 1, 2, 6 and 7. From
Table~\ref{Table:CellLocation}, cells 1 and 2 are from the left
cluster, while 6 and 7 belong to the right cluster. The
  corresponding normalized eigenpair ${\mathcal K}\pmb{c}$ of the GCEP
  matrix, as given in the bottom third of
  Table~\ref{eigvec:twoclusters}, predicts that the amplitudes of the
  intracellular dynamics will differ from cell to cell due to the
  different cell influx rates, but that there will only be a small
  phase shift in the intracellular oscillations. Since
  $|({\mathcal K}\pmb{c})_1|< |({\mathcal K}\pmb{c})_7|$ (see
  Table~\ref{eigvec:twoclusters}), the linearized theory predicts
  larger amplitude oscillations for $u_1$ in cell 7 than in cell 1,
  which is expected since cell 7 has a larger influx rate than does
  cell 1 (see Table~\ref{Table:CellLocation}). This is confirmed in
  the FlexPDE simulations in Fig.~\ref{fig:cluster_arb}.  In the
bottom row of Fig.~\ref{fig:cluster_arb} we observe that the
corresponding results from the ODE system \eqref{reducedODE} compare
moderately well with the FlexPDE results regarding the amplitude and
period of oscillations in the bulk medium and within the cells. We
remark that if, instead, we used the simpler ODE system
\eqref{ode:well_mixed}, corresponding to the well-mixed limit
$D\to \infty$, the ODE results would be in very poor agreement with
the full PDE simulations.

 \subsubsection{Isolated cells in a pattern can be quiescent:
  Diffusion-sensing behavior}\label{large:quiet}

In this subsection, we analyze in detail the role of the spatial
configuration of the cells on the triggering of intracellular
oscillations. Specifically, we consider the cell configuration shown
in Fig.~\ref{fig:tworing}, with cell centers given in
Table~\ref{Table:CellLocation_ring} of Appendix~\ref{Cell_Location}.
This pattern consists of two spatially segregated rings of cells
together with two cells that are spatially isolated from the
rings. For this symmetric pattern we will consider five permeability
parameter sets for the cell influx rate, as given in
Table~\ref{Table:CellLocation_ring}, which lead to solution behavior
that can be interpreted both qualitatively and from our linear
stability analysis.

\begin{figure}[!ht]
  \centering
    \begin{subfigure}[b]{0.45\textwidth}
        \includegraphics[width=\textwidth,height=4.4cm]{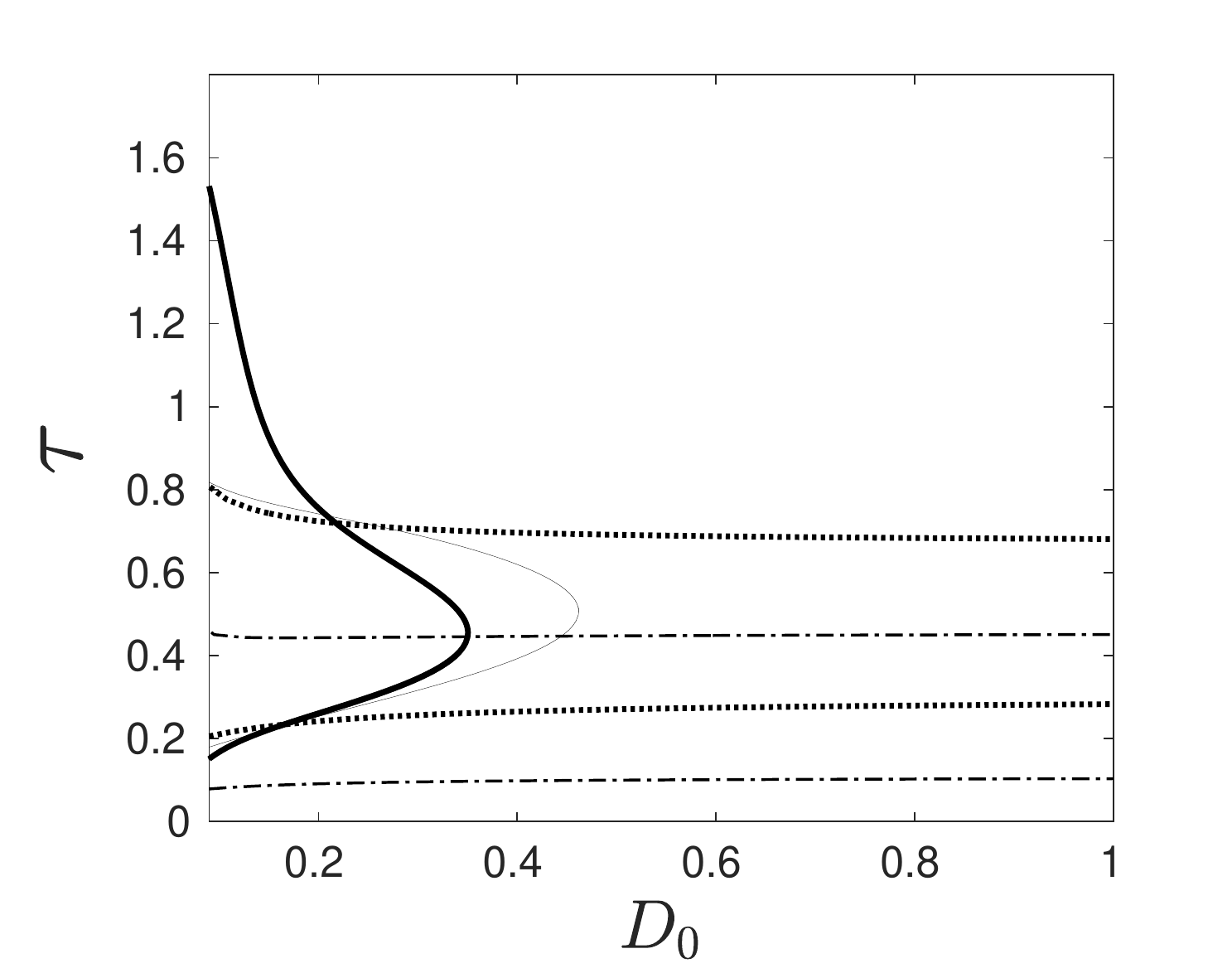}
        \caption{HB boundaries}
        \label{fig:tworingHB}
    \end{subfigure}
    \begin{subfigure}[b]{0.35\textwidth}
        \includegraphics[width=\textwidth,height=4.4cm]{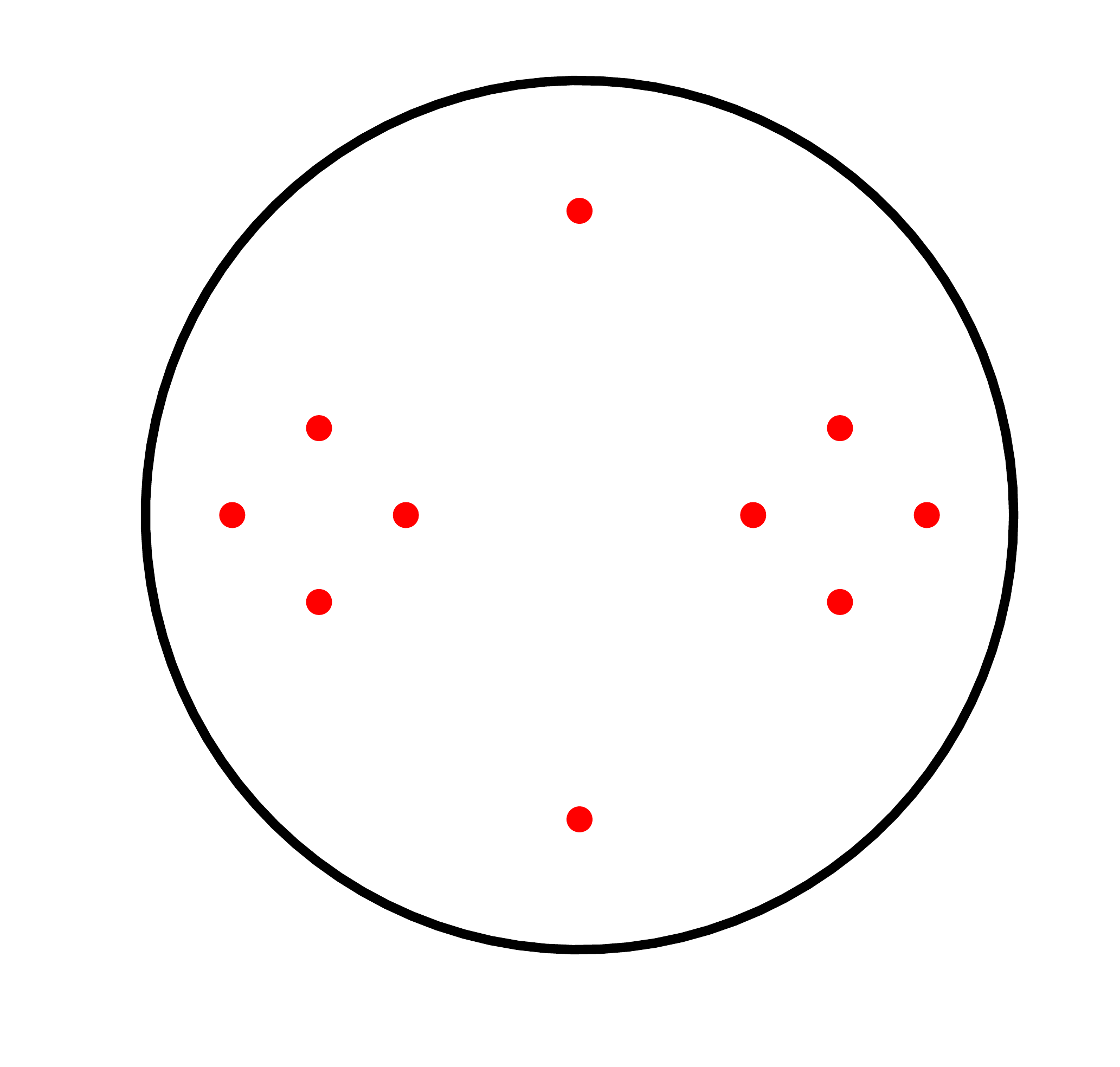}
         \caption{Two rings of cells with two isolated cells}
        \label{fig:tworing}
    \end{subfigure}
    ~\quad
    \caption{HB boundaries (left panel) for $m=10$ cells for
        a pattern with two rings of cells and two isolated cells
        (right panel). The cell locations and permeability parameter
        sets for the influx rates $d_{1j}$ are given in
        Table~\ref{Table:CellLocation_ring} of
        Appendix~\ref{Cell_Location}. Four out of five of the
        permeability sets in Table~\ref{Table:CellLocation_ring} give
        HB boundaries in the $(D_0,\tau)$ plane. The thin solid curve
        (set II) has $d_1=0.3$ for the ring cells and $d_1=0.8$ for
        the isolated cells. The heavy solid curve has $d_1=0.3$ for
        the isolated cells, $d_1=0.8$ for one group of ring cells, and
        $d_1=0.4$ for the cells on the other ring (set III). The
        dash-dotted curve has $d_1=0.8$ for all cells (set IV), while
        the dotted curve has $d_1=0.3$ for all cells (set V). There is
        no HB boundary for set I where $d_1=0.3$ and $d_1=0.8$ for the
        isolated and ring cells, respectively. The remaining
        parameters are given in
        \eqref{Selkov_para}.} \label{Large:ring_pattern}
\end{figure}

In Fig.~\ref{fig:tworingHB}  we plot the  HB boundaries in  the $\tau$
versus $D_0$  parameter plane  for permeability  sets $II-V$  of Table
\ref{Table:CellLocation_ring}, as obtained by solving for the roots of
\eqref{large:hopf}  numerically.  There  is  no HB  boundary  in  this
parameter plane from \eqref{large:hopf}  for permeability set I, where
$d_1=0.8$ for the ring cells and  $d_1=0.3$ for the isolated cells. As
verified from a numerical winding  number computation, only within the
lobes spanned by the HB boundaries  is the steady-state unstable to an
oscillatory instability.  From the dotted and  dashed-dotted curves in
Fig.~\ref{fig:tworingHB} corresponding to where all the cells have the
common influx  rates $d_1=0.3$ or $d_1=0.8$,  respectively, we observe
that the instability  lobe is unbounded in $D_0$.  Therefore, when the
cells are all identical,  intracellular oscillations occur within some
finite  band  of  the  reaction-time parameter  $\tau$  for  all  bulk
diffusivities.  For identical cells  with $d_1=0.8$, the HB boundaries
in  Fig.~\ref{fig:tworingHB}  are  very   similar  to  that  shown  in
Fig.~\ref{Large:cluster_arbitrary} for the  arbitrary cell pattern and
the two-cluster  pattern.  For  identical cells  with the  larger cell
influx rate  $d_1=0.8$ we  observe from  Fig.~\ref{fig:tworingHB} that
the lower  threshold in  $\tau$ where oscillatory  instabilities first
occur is  smaller than when  the common  influx rate is  $d_1=0.3$. To
explain this  qualitatively, we observe  that $\tau$ decreases  as the
bulk  decay  rate $k_B$  increases  (see  \eqref{dim:param}).  With  a
larger bulk decay  rate, the bulk signal becomes less  diffuse and has
larger spatial gradients, which leads to a buildup of the bulk signal,
near the cells. As a result, when  there is a larger cell influx rate,
the  bulk  signal   can  more  readily  enter  the   cell  to  trigger
intracellular oscillations.

\begin{figure}[!ht]
  \centering
    \begin{subfigure}[b]{0.48\textwidth}
        \includegraphics[width=\textwidth,height=6.0cm]{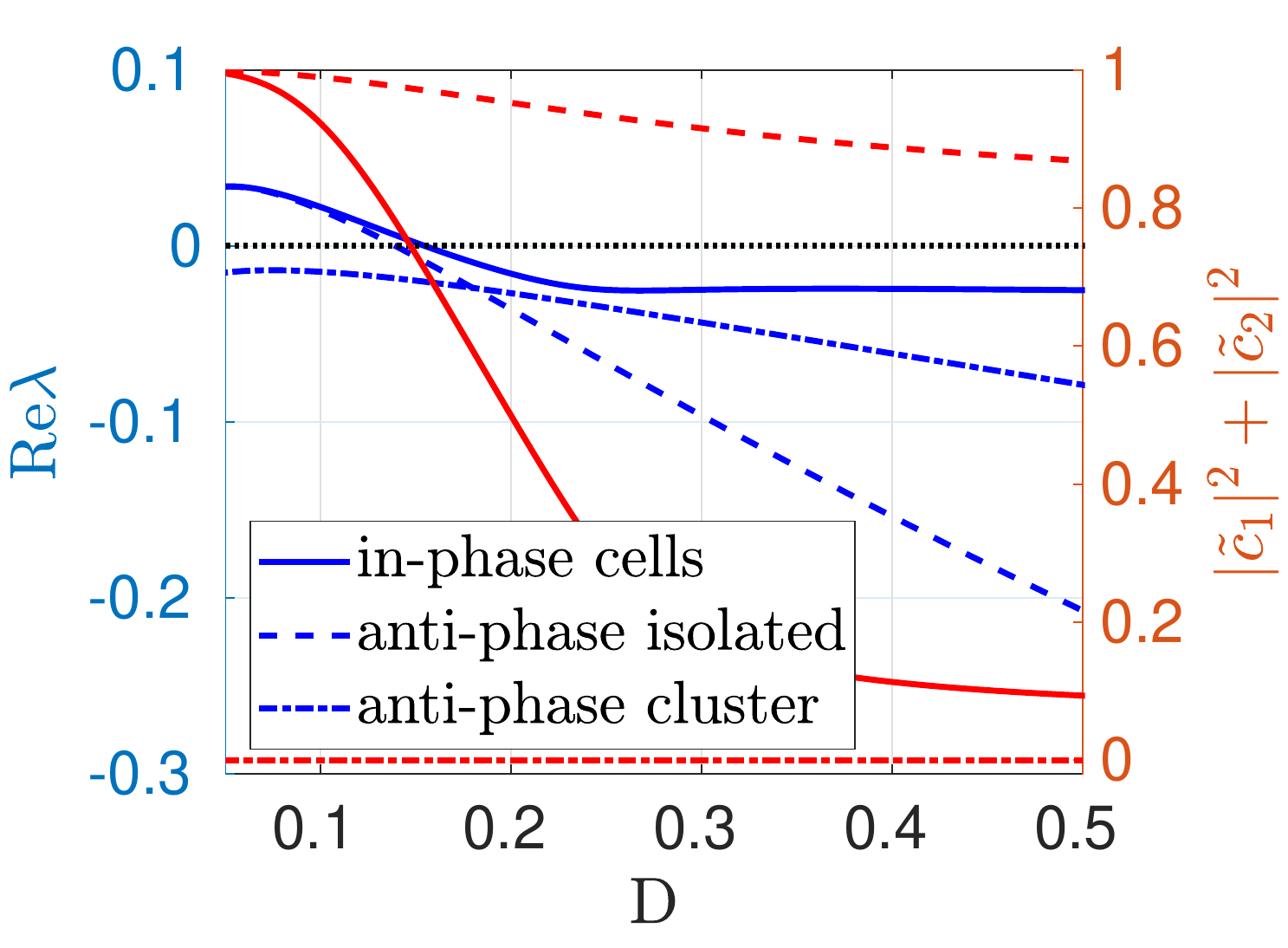}
        \caption{set I: $d_1=0.8$ (rings), $d_1=0.3$ (isolated)}
        \label{fig:cell_plot_I}
    \end{subfigure}
    \begin{subfigure}[b]{0.48\textwidth}
        \includegraphics[width=\textwidth,height=6.0cm]{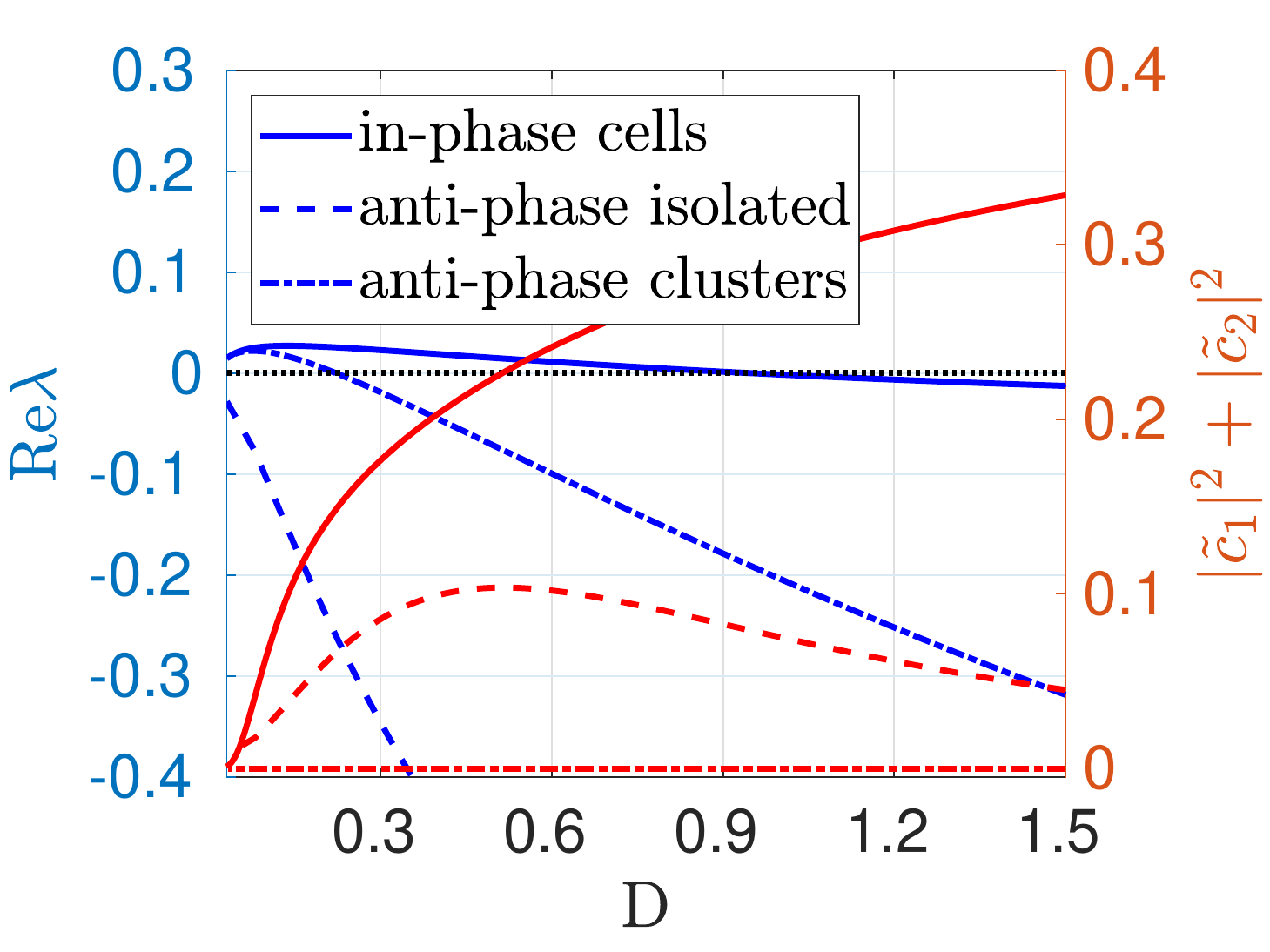}
         \caption{set II: $d_1=0.8$ (isolated), $d_1=0.3$ (rings)}
        \label{fig:cell_plot_II}
    \end{subfigure}
    ~\quad
    \caption{Spectral information from the eigenvalues
        $\lambda$ and normalized eigenvectors
        $\tilde{\pmb{c}}={\mathcal K}\pmb{c}$, with
        $\sum_{i}|\left({\mathcal K}\pmb{c}\right)_i|^2=1$, obtained
        from the GCEP matrix \eqref{Global_System} and
        \eqref{large:Kmatrix_1} for ${\mathcal K}$, as computed using
        root-finding on $\det{\mathcal M}(\lambda)=0$. The cell
        configuration is the two-ring and two-isolated cell pattern of
        Fig.~\ref{fig:tworing} for permeability set I (left panel) and
        permeability set II (right panel) for the fixed value
        $\tau=0.35$. Left $y$-axis: Real parts (blue curves) of the
        three eigenvalues of the GCEP matrix with the largest real
        parts versus $D$. Right $y$-axis: The sum
        $|\tilde{c}_1|^2+|\tilde{c}_2|^2$ (red curves) of the first
        two components of the normalized eigenvector
        ${\mathcal K}\pmb{c}$, which measures the relative amplitude
        of oscillations in the two isolated cells in comparison to the
        cells on the ring. The linetype (solid, dashed, dot-dashed) of
        the blue and red curves correspond to the same eigenpair of
        the GCEP matrix.}\label{small:ring_pattern}
\end{figure}

Since from Fig.~\ref{Large:ring_pattern} the steady-state is always
linearly stable for permeability set I when $D={D_0/\nu}$, this
motivates studying the full GCEP matrix ${\mathcal M}(\lambda)$ in
\eqref{full_gcep} that is valid for $D={\mathcal O}(1)$.  In the left
and right panels of Fig.~\ref{small:ring_pattern} we plot spectral
information versus $D$, for fixed $\tau=0.35$, obtained from the roots
of $\mbox{det}{\mathcal M}(\lambda)=0$ for the two-ring and isolated
cell pattern for permeability sets I (left panel) and II (right
panel). In the left $y$-axes of Fig.~\ref{small:ring_pattern} we plot
the real parts of the three eigenvalues of the GCEP matrix that have
the largest real parts.  These three eigenvalue branches are
associated with different spatial modes for the cells as indicated in
the figure legends and discussed below. In the right $y$-axes of
Fig.~\ref{small:ring_pattern} we plot the sum
$|\tilde{c}_1|^2+|\tilde{c}_2|^2$ (red curves) of the first two
components of the normalized eigenvector
$\tilde{\pmb{c}}\equiv {\mathcal K}\pmb{c}$, which measures the
relative amplitude of oscillations in the two isolated cells in
comparison to the cells on the ring (see \eqref{nstabform:c}). The
linetype (solid, dashed, dot-dashed) of the blue and red curves in
Fig.~\ref{small:ring_pattern} correspond to the same eigenpair of the
GCEP matrix.
  
For permeability set I, where the influx rate $d_1=0.3$ on the
isolated cells is lower than that on the ring cells ($d_1=0.8$), we
observe from the solid and dashed lines in Fig.~\ref{fig:cell_plot_I}
that there can be two unstable modes when $D$ is small enough. On the
range $D<0.138$, there is an unstable mode given by the
blue dashed curve where anti-phase oscillations occur for the two
isolated cells (cells 1 and 2), which has a much higher amplitude than
for the ring cells. On the two rings, cells 4 and 6 as well as cells 8
and 10 oscillate out of phase. The other ring cells (cells 3, 5, 7,
and 9) are essentially quiescent for this mode. Due to the low influx
rate into the isolated cells, the isolated cells communicate
imperfectly with their images across the domain boundary when $D$ is
small, which leads to the anti-phase instability.  Observe that as $D$
increases above $D\approx 0.138$ this mode becomes stable and the
amplitude in the isolated cells decreases (decreasing dashed red
curve) since the bulk signal near the isolated cells is washed
away. The other possible unstable mode, given by the blue solid curve
in Fig.~\ref{fig:cell_plot_I}, is one for which the two isolated cells
are in-phase. The ring cells, which have smaller amplitude
oscillations than do the isolated cells, are all roughly in phase and
have only a slight phase difference with the isolated cells. This
in-phase mode is unstable only for $D<0.155$. We remark that this low
stability threshold value of $D$ is consistent with the observations
in Fig.~\ref{fig:tworingHB} that the steady-state is always linearly
stable for $D={\mathcal O}(\nu^{-1})\gg 1$. Finally, the dashed-dotted
blue curve in Fig.~\ref{fig:cell_plot_I} corresponds to a linearly
stable mode where the isolated cells are quiescent, with each ring
having cell oscillations that are (roughly) synchronized in amplitude
and phase. However, for this mode there is a large phase shift between
the oscillations in the two rings clusters. With the larger influx
rate for the ring cells, there is effective communication between the
two spatially segregated rings as $D$ increases, which precludes any
anti-phase instability between the two ring clusters.

  \begin{figure}[!ht]
  \centering
    \begin{subfigure}[b]{0.40\textwidth}
      \includegraphics[width=\textwidth,height=4.2cm]{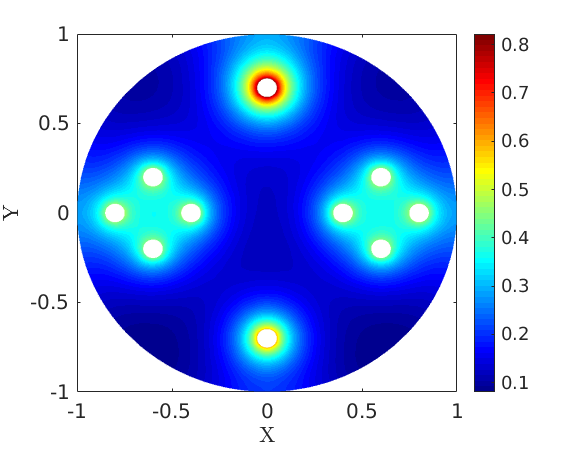}
      \caption{Surface plot at $t=400$}
      \label{fig:flexpde:ring_surf_I}
    \end{subfigure}
    \begin{subfigure}[b]{0.35\textwidth}
      \includegraphics[width=\textwidth,height=4.2cm]{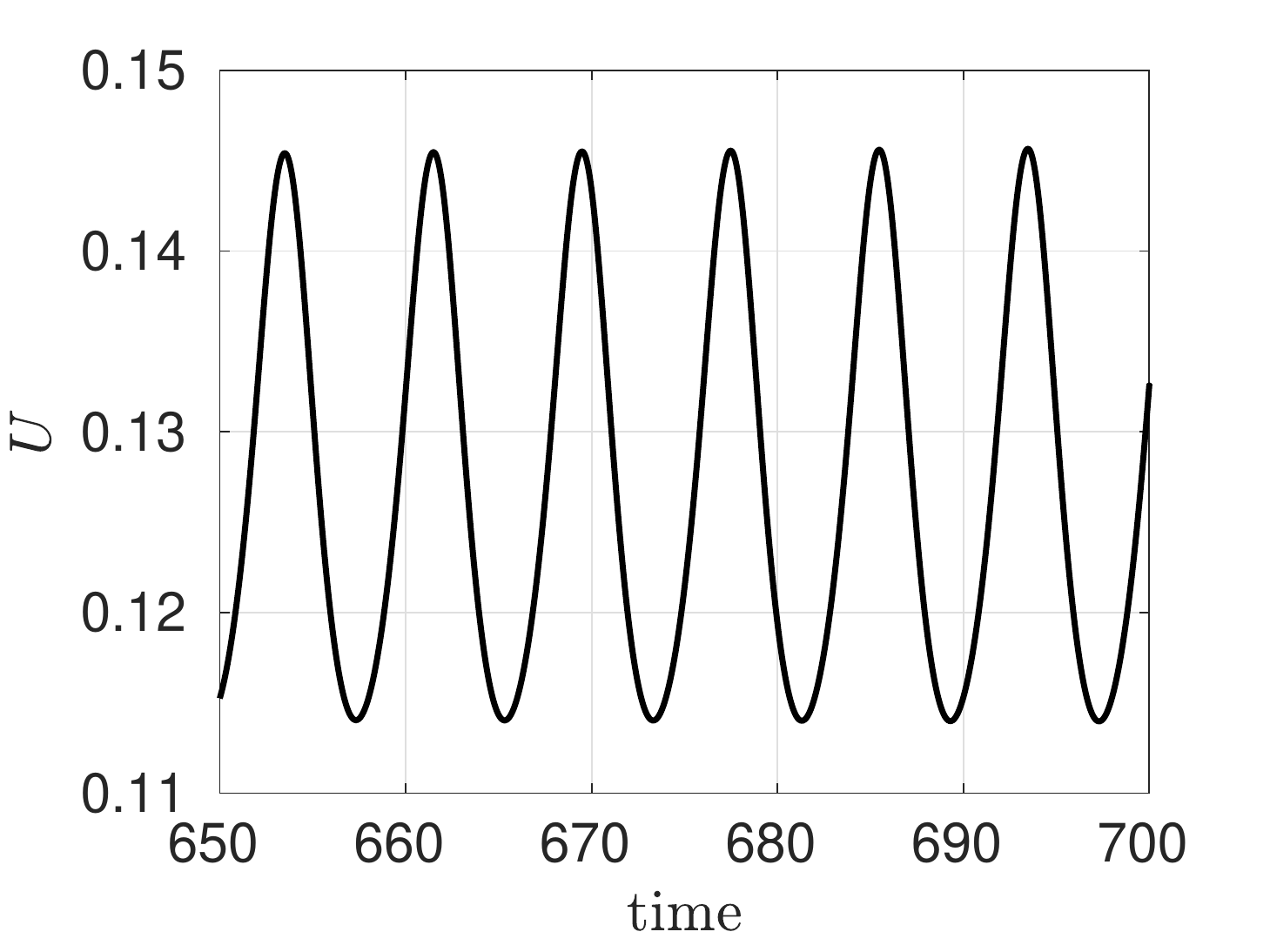}
        \caption{$U$ at ${\pmb x}=(0,0)$ (FlexPDE)} 
    \label{fig:flexpde:ringI_bulk}
    \end{subfigure}\\
    \begin{subfigure}[b]{0.32\textwidth}
      \includegraphics[width=\textwidth,height=4.2cm]{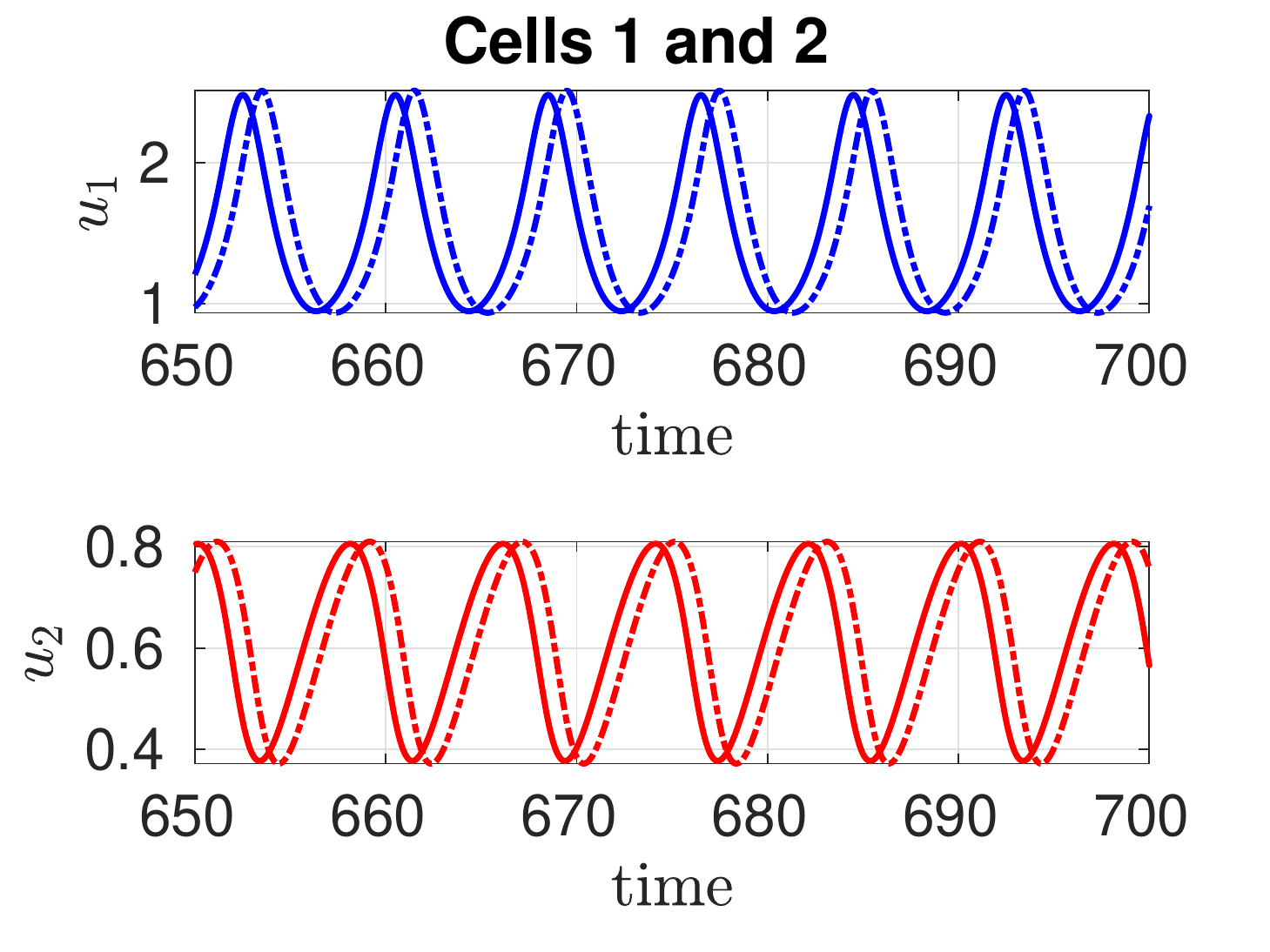}
      \caption{$u_1,u_2$ isolated cells 1 and 2 (FlexPDE)}
    \label{fig:flexpde:ringI_cell12}
    \end{subfigure}
    \begin{subfigure}[b]{0.32\textwidth}
      \includegraphics[width=\textwidth,height=4.2cm]{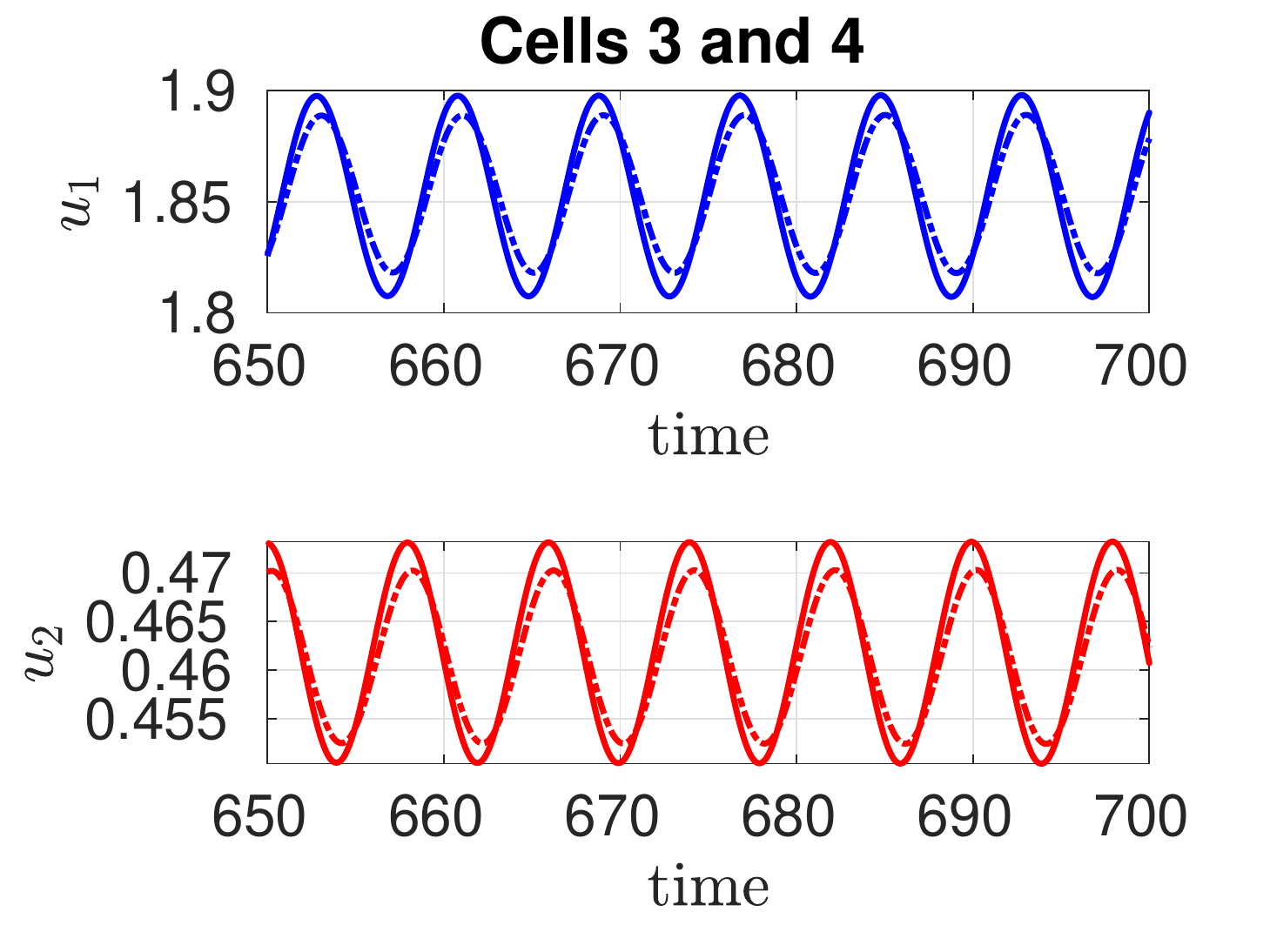}
      \caption{$u_1,u_2$ cells 3 and 4 (left ring)}
    \label{fig:flexpde:ringI_cell34}
  \end{subfigure}
    \begin{subfigure}[b]{0.32\textwidth}
      \includegraphics[width=\textwidth,height=4.2cm]{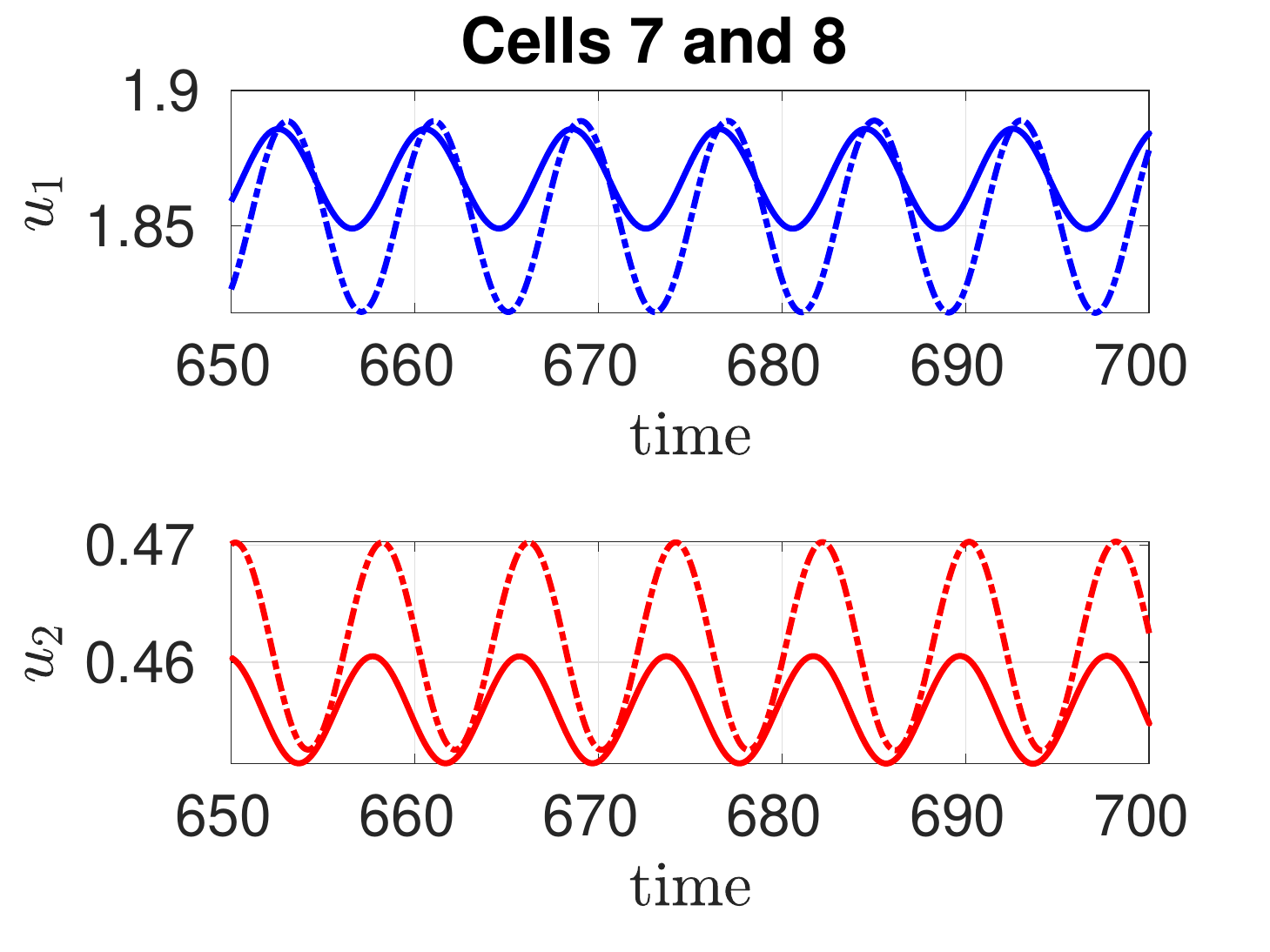}
      \caption{$u_1,u_2$ cells 7 and 8 (right ring)}
    \label{fig:flexpde:ringI_cell78}
    \end{subfigure}
  \vspace*{-1ex}
  \caption{FlexPDE numerical results computed from the
      PDE-ODE model \eqref{DimLess_bulk} for $m=10$ cells, arranged in
      two rings with two isolated cells (see Fig.~\ref{fig:tworing})
      when $D=0.05$ and $\tau=0.35$ for the permeability set I and
      cell locations given in Table
      \ref{Table:CellLocation_ring}. Other parameters are given in
      \eqref{Selkov_para}. The linear stability predictions are given
      in Fig.~\ref{fig:cell_plot_I}.} \label{fig:ring_permaI}
\end{figure}

\begin{figure}[!ht]
  \centering
    \begin{subfigure}[b]{0.40\textwidth}
      \includegraphics[width=\textwidth,height=4.2cm]{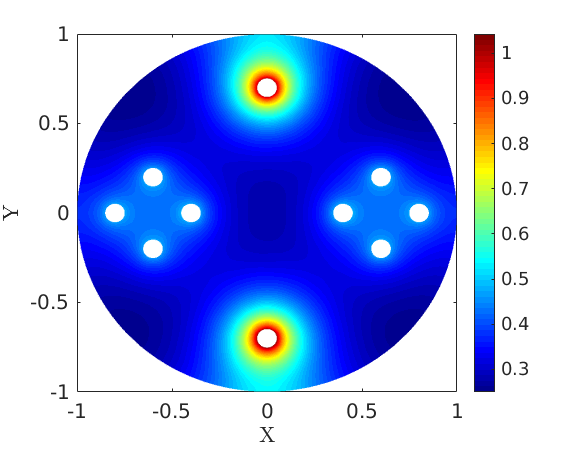}
      \caption{Surface plot at $t=400$}
      \label{fig:flexpde:ring_surf_ID13}
    \end{subfigure}
    \begin{subfigure}[b]{0.35\textwidth}
      \includegraphics[width=\textwidth,height=4.2cm]{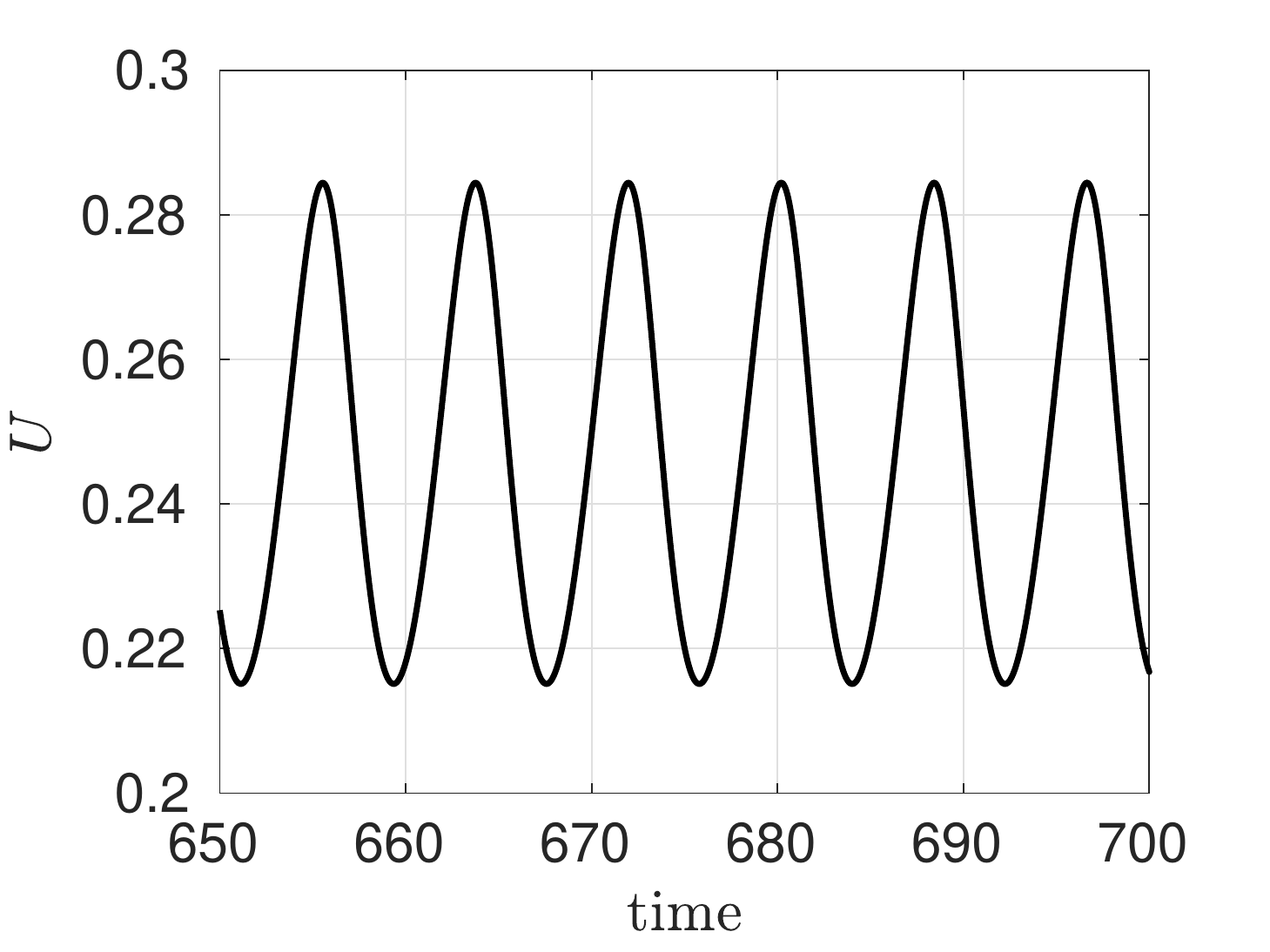}
        \caption{$U$ at ${\pmb x}=(0,0)$ (FlexPDE)} 
    \label{fig:flexpde:ringI_bulkD13}
    \end{subfigure}\\
    \begin{subfigure}[b]{0.32\textwidth}
      \includegraphics[width=\textwidth,height=4.2cm]{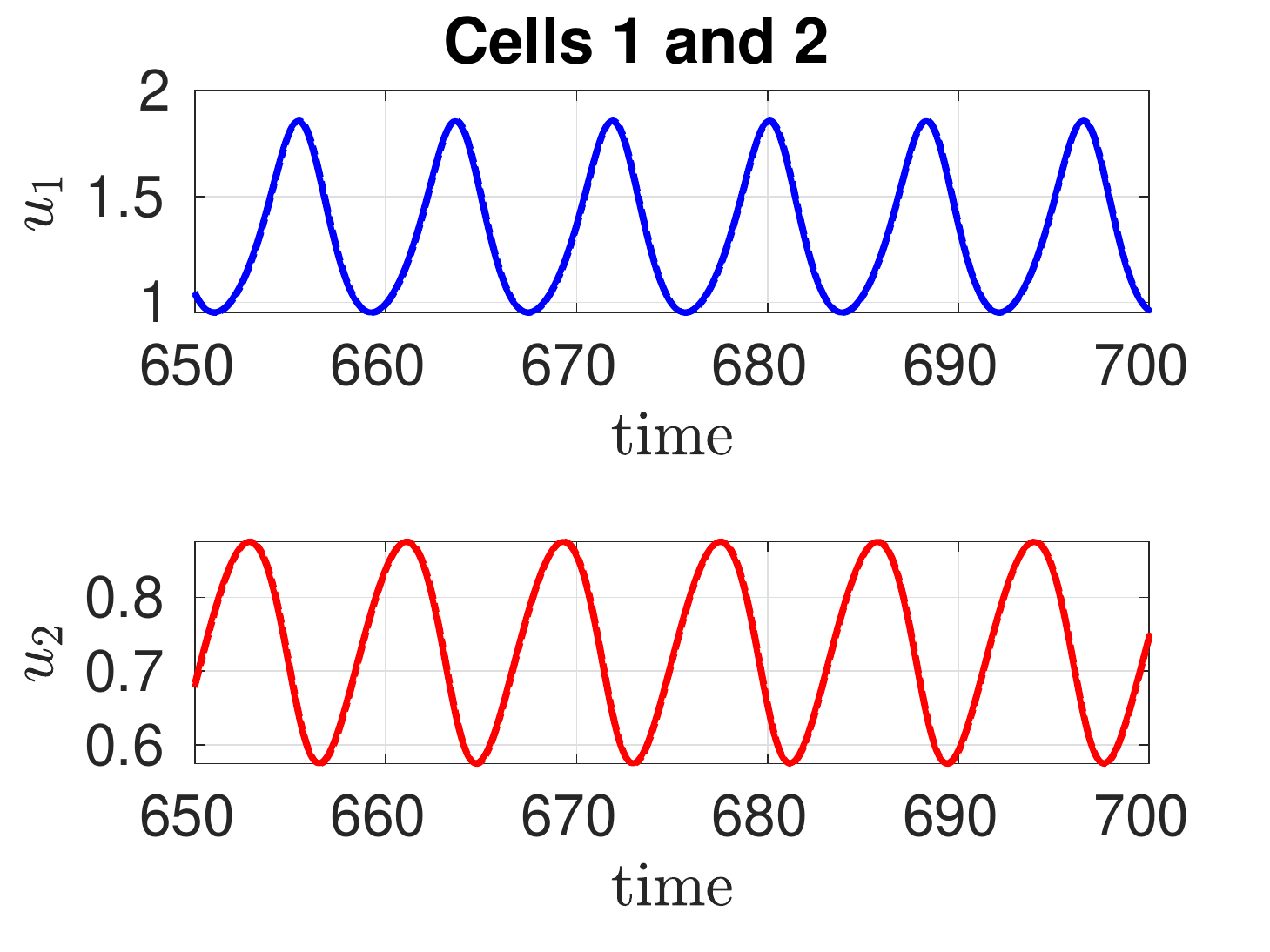}
      \caption{$u_1,u_2$ isolated cells 1 and 2 (FlexPDE)}
    \label{fig:flexpde:ringI_cell12D13}
    \end{subfigure}
    \begin{subfigure}[b]{0.32\textwidth}
      \includegraphics[width=\textwidth,height=4.2cm]{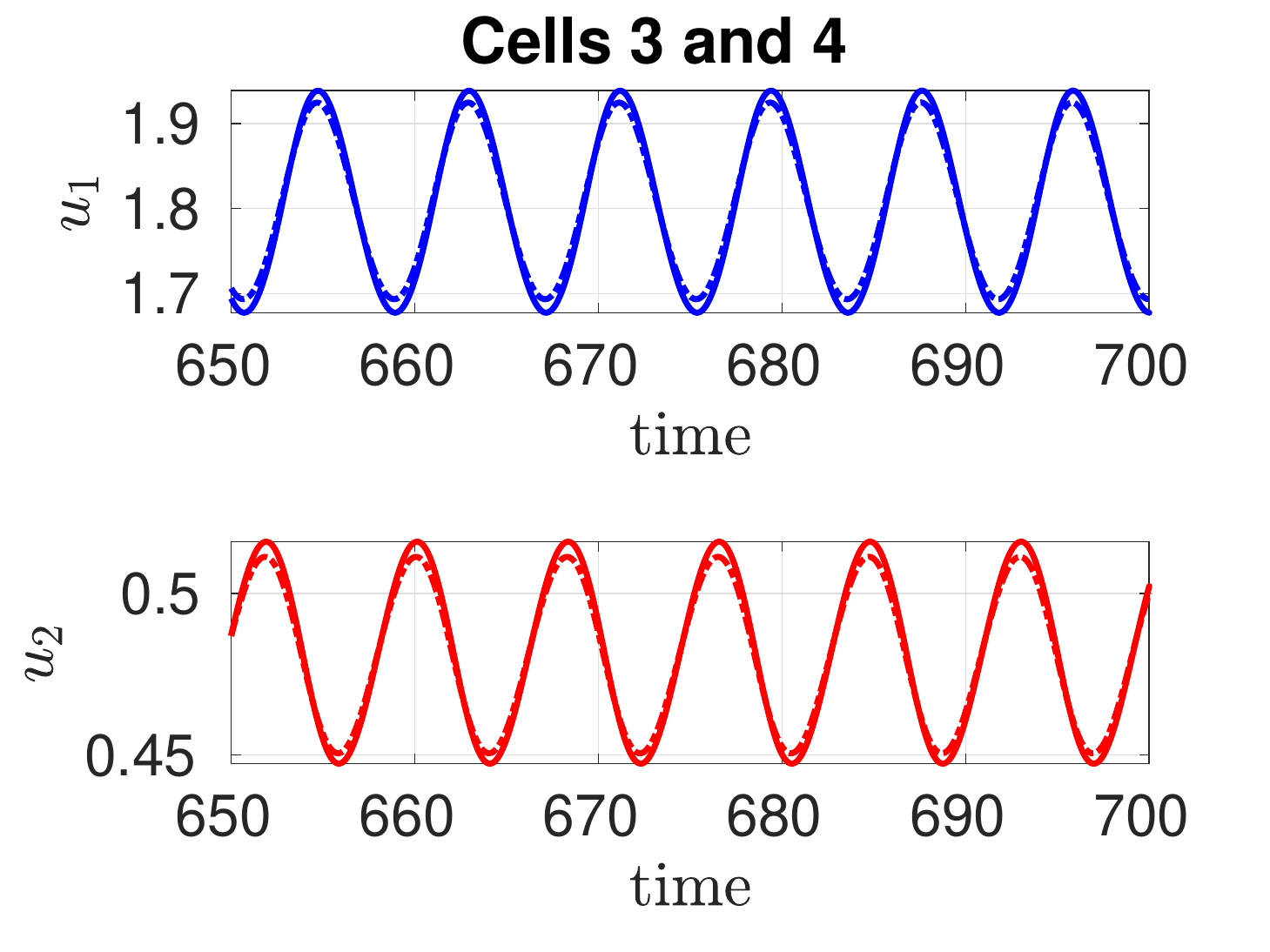}
      \caption{$u_1,u_2$ cells 3 and 4 (left ring)}
    \label{fig:flexpde:ringI_cell34D13}
  \end{subfigure}
    \begin{subfigure}[b]{0.32\textwidth}
      \includegraphics[width=\textwidth,height=4.2cm]{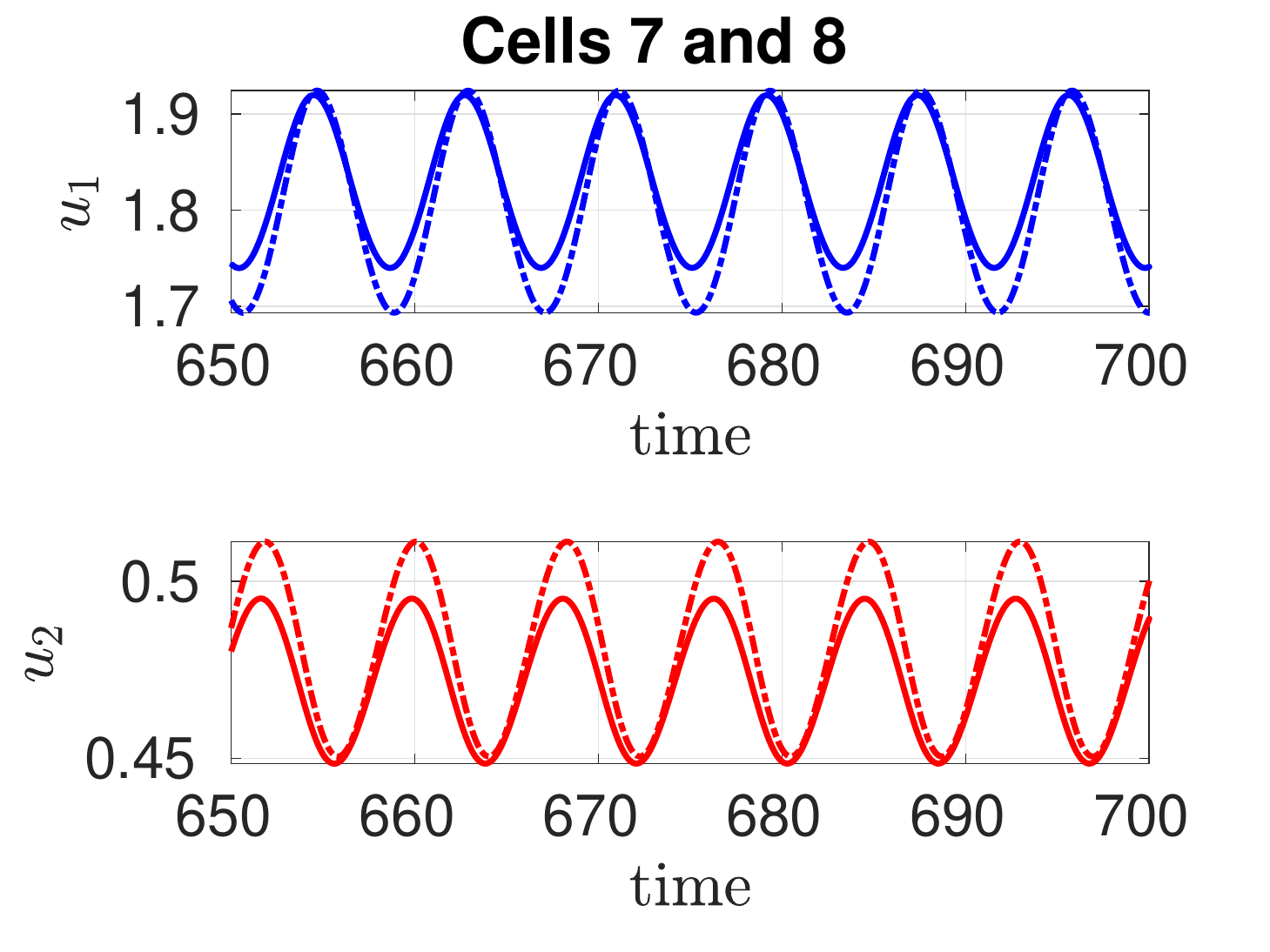}
      \caption{$u_1,u_2$ cells 7 and 8 (right ring)}
    \label{fig:flexpde:ringI_cell78D13}
    \end{subfigure}
  \vspace*{-1ex}
  \caption{Same caption as in Fig.~\ref{fig:ring_permaI}
      except that now $D$ is increased to $D=0.13$. Observe that
      there is large in-phase signaling gradient near the isolated
      cells in the surface plot (a), and that the isolated cells have
      now synchronized their dynamics. The oscillations in the ring
      cells are more synchronized than when $D=0.05$, but still have
      smaller amplitude for $u_1$ than do the isolated cells.}
  \label{fig:ring_permaID13}
\end{figure}

As a test of the predictions of the linear stability theory, as
summarized by Fig.~\ref{fig:cell_plot_I}, in
Fig.~\ref{fig:ring_permaI} and Fig.~\ref{fig:ring_permaID13} we show
FlexPDE results for \eqref{DimLess_bulk} for permeability set I when
$D=0.05$ and $D=0.13$, respectively. For $D=0.05$ where both the
in-phase and anti-phase modes (solid and dashed curves in
Fig.~\ref{fig:cell_plot_I}) are unstable with comparable growth rates,
we observe from
Figs.~\ref{fig:flexpde:ringI_cell12}--\ref{fig:flexpde:ringI_cell78}
that, as predicted, the intracellular dynamics in cells 1 and 2 are
not in-phase and that the amplitude of oscillations in the ring cells
is much smaller than in the isolated cells. The strong anti-phase bulk
signaling gradient at the isolated cells, as shown in
Fig.~\ref{fig:flexpde:ring_surf_I}, is also consistent with the linear
stability theory. In contrast, when $D=0.13$, the in-phase mode is the
dominant instability as seen from Fig.~\ref{fig:cell_plot_I}. For this
larger value of $D$, we observe from
Fig.~\ref{fig:flexpde:ring_surf_ID13} and
Fig.~\ref{fig:flexpde:ringI_cell12D13} that the bulk signaling
gradient and the intracellular oscillations are now in-phase at the
two isolated cells.  Moreover, by comparing
Figs.~\ref{fig:ring_permaI} and \ref{fig:ring_permaID13}, we observe
that the oscillations within the ring cells are more synchronized and
have a larger amplitude, while the isolated cells have a smaller
amplitude, when $D=0.13$ as compared to when $D=0.05$. These
observations are all consistent with the linear stability predictions
shown in Fig.~\ref{fig:cell_plot_I}.

\begin{figure}[!ht]
  \centering
    \begin{subfigure}[b]{0.40\textwidth}
      \includegraphics[width=\textwidth,height=4.2cm]{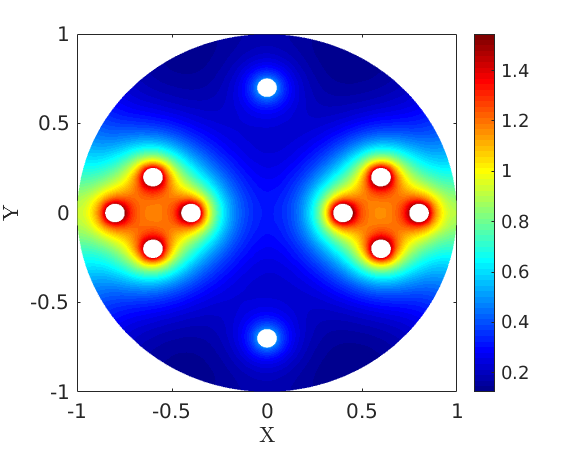}
      \caption{Surface plot at $t=400$}
      \label{fig:flexpde:ring_surf_II}
    \end{subfigure}
    \begin{subfigure}[b]{0.35\textwidth}
      \includegraphics[width=\textwidth,height=4.2cm]{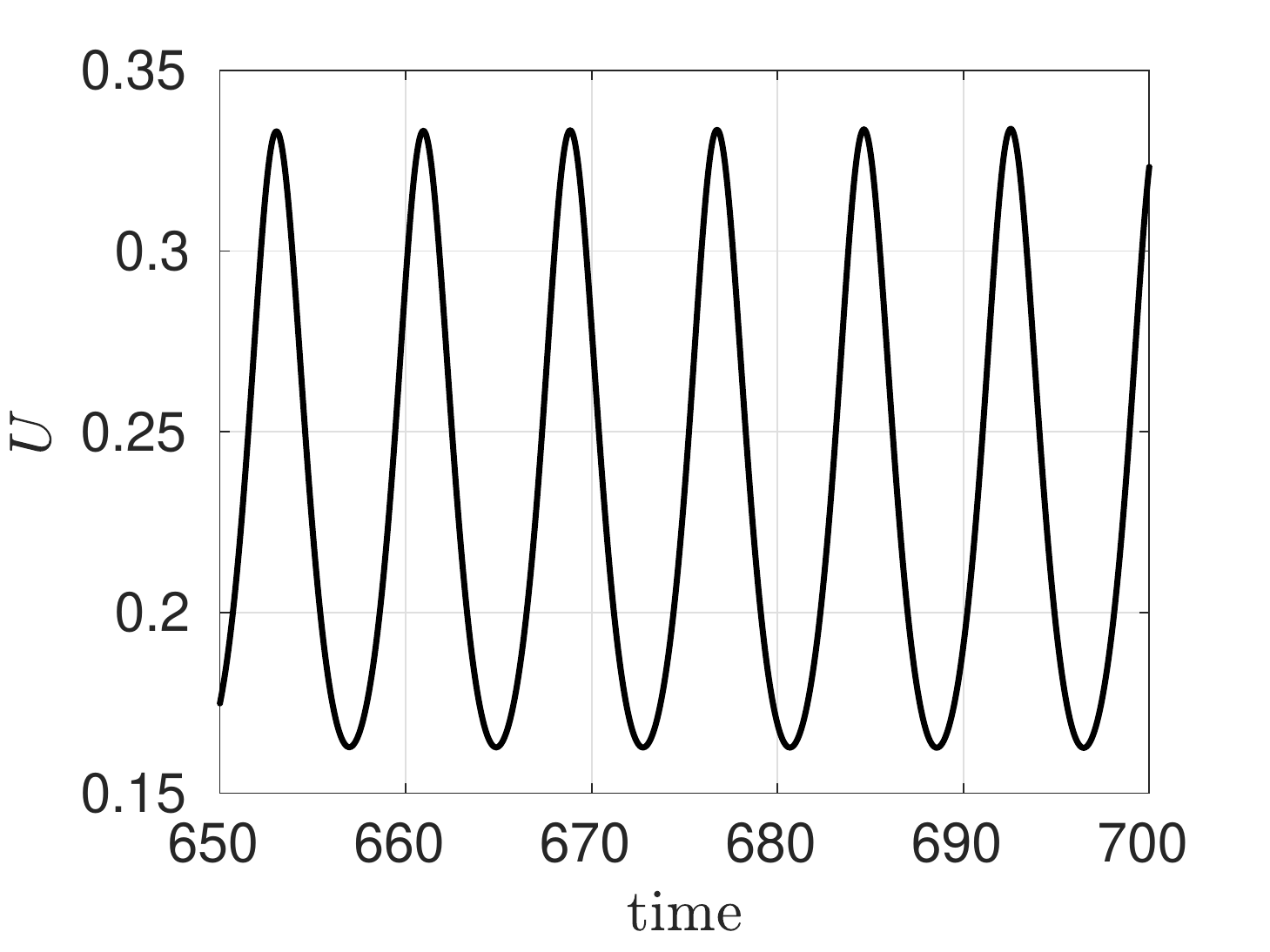}
        \caption{$U$ at ${\pmb x}=(0,0)$ (FlexPDE)} 
    \label{fig:flexpde:ringII_bulk}
    \end{subfigure}\\
    \begin{subfigure}[b]{0.32\textwidth}
      \includegraphics[width=\textwidth,height=4.2cm]{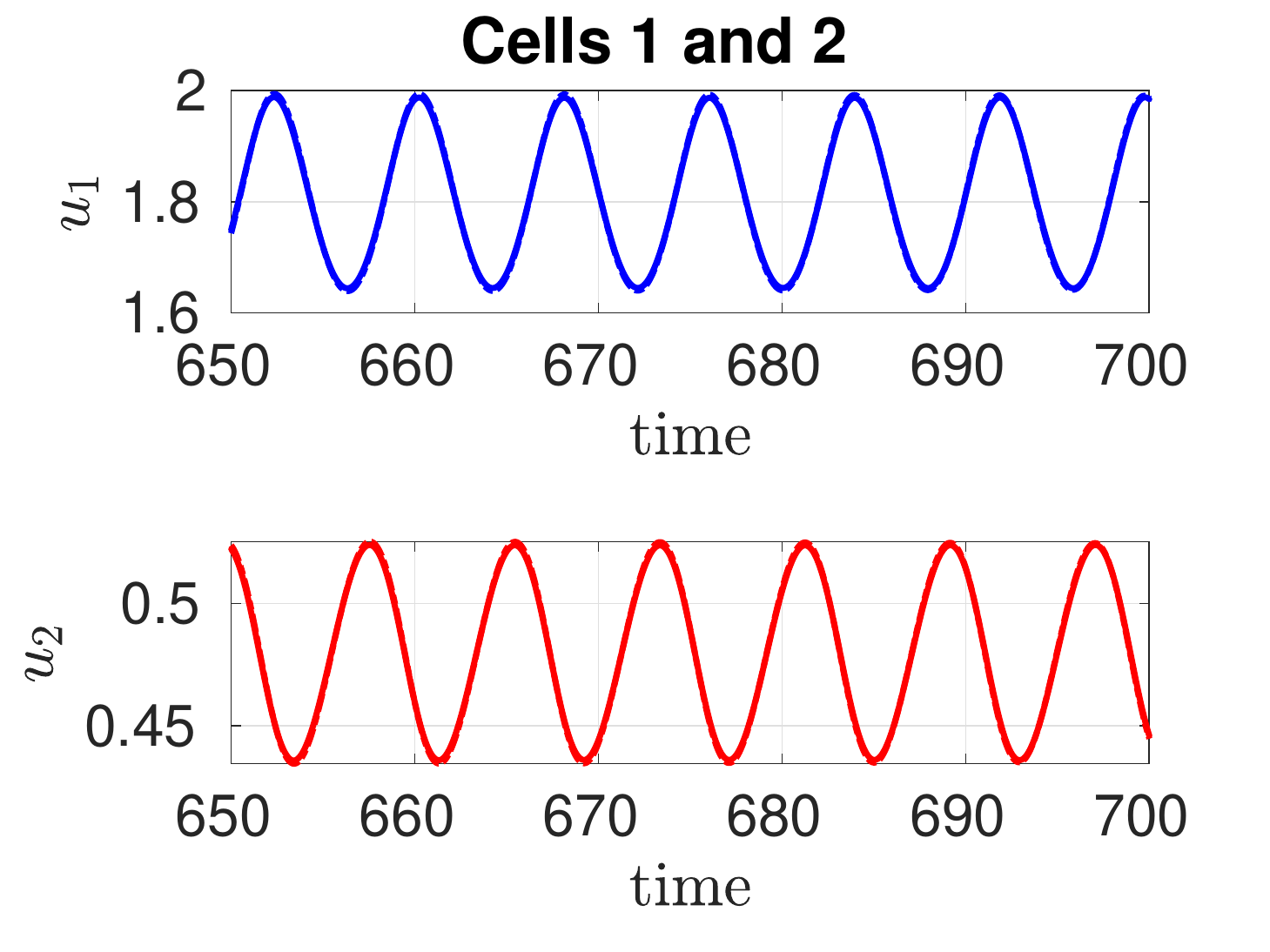}
      \caption{$u_1,u_2$ isolated cells 1 and 2 (FlexPDE)}
    \label{fig:flexpde:ringII_cell12}
    \end{subfigure}
    \begin{subfigure}[b]{0.32\textwidth}
      \includegraphics[width=\textwidth,height=4.2cm]{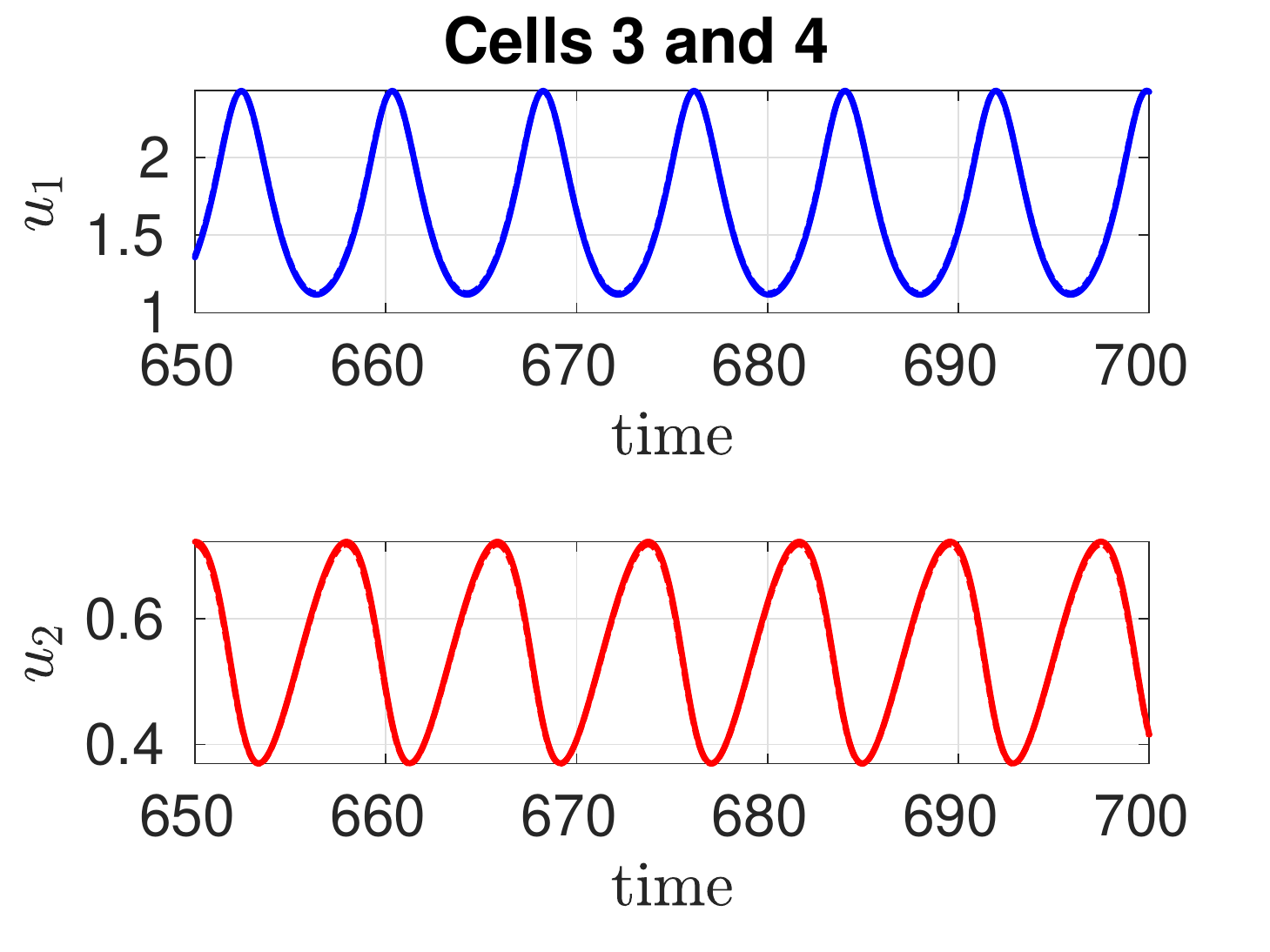}
      \caption{$u_1,u_2$ cells 3 and 4 (left ring)}
    \label{fig:flexpde:ringII_cell34}
  \end{subfigure}
    \begin{subfigure}[b]{0.32\textwidth}
      \includegraphics[width=\textwidth,height=4.2cm]{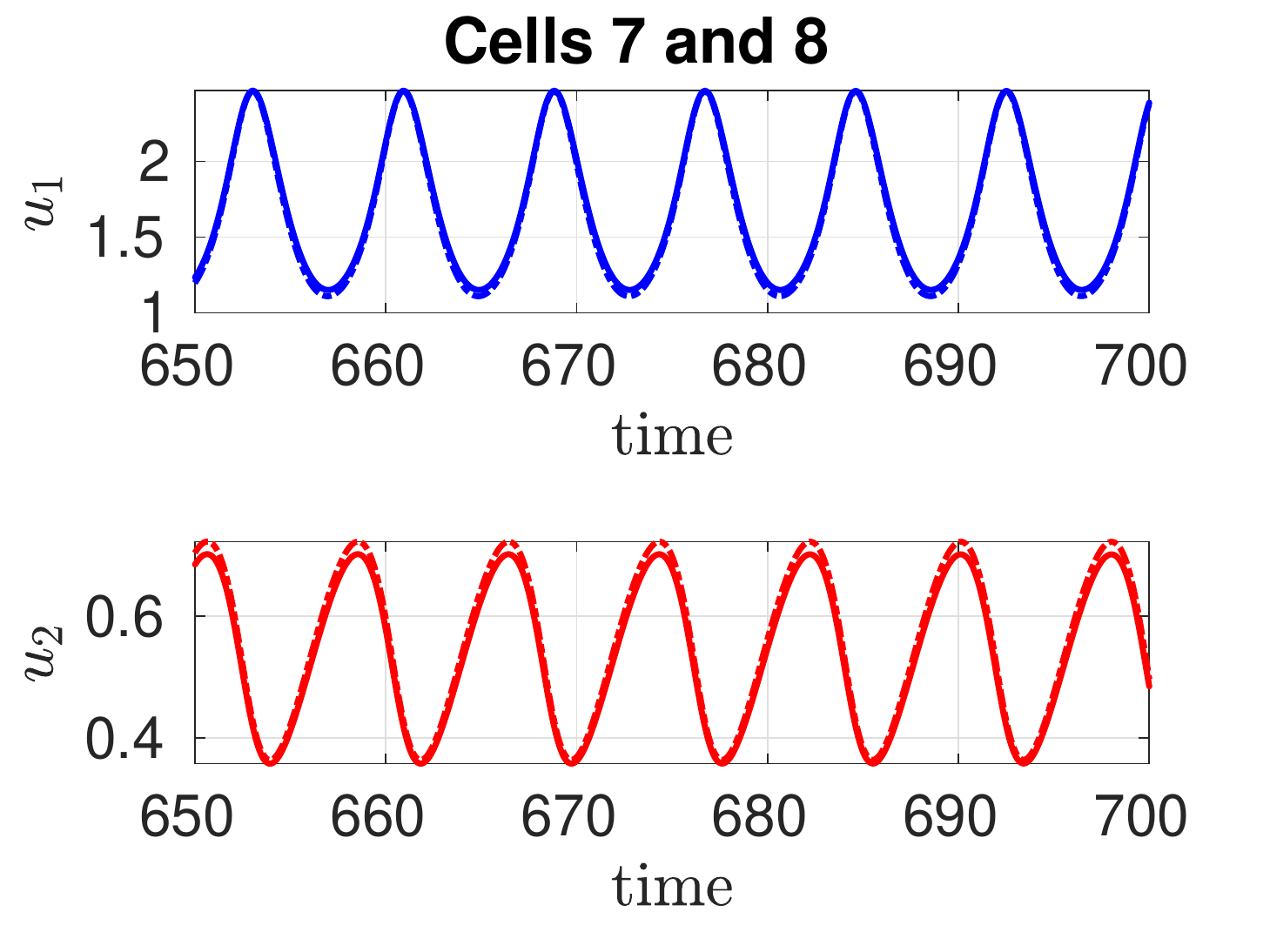}
      \caption{$u_1,u_2$ cells 7 and 8 (right ring)}
    \label{fig:flexpde:ringII_cell78}
    \end{subfigure}
  \vspace*{-1ex}
  \caption{FlexPDE numerical results computed from the
      PDE-ODE model \eqref{DimLess_bulk} for $m=10$ cells, arranged in
      two rings with two isolated cells (see Fig.~\ref{fig:tworing})
      when $D=0.05$ and $\tau=0.35$ for the permeability set II and
      cell locations given in Table
      \ref{Table:CellLocation_ring}. Other parameters are given in
      \eqref{Selkov_para}. The linear stability predictions are given
      in Fig.~\ref{fig:cell_plot_II}.} \label{fig:ring_permaII}
\end{figure}

In Fig.~\ref{fig:cell_plot_II} we show the corresponding spectral
plot, as computed from the GCEP \eqref{full_gcep}, for permeability
set II where the isolated cells have a higher influx rate
  $d_1=0.8$ than do the ring cells $d_1=0.3$. In contrast to the case
for permeability set I, we now observe that in addition to the usual
in-phase mode that is unstable for $D<0.95$, the only other possible
unstable mode corresponds to an anti-phase instability between the two
ring clusters. Since the influx rate into the ring cells has been
decreased, this unstable mode is due to a relatively poorer
communication between the two rings when the bulk diffusivity is
small. This anti-phase cluster mode is unstable for $D<0.22$.  In
contrast, the anti-phase mode for the isolated cells, which was
unstable for permeability set I when $D$ is small, is now always
linearly stable for set II. Another key feature from
Fig.~\ref{fig:cell_plot_II} is that, for the in-phase mode, the
intracellular oscillations in the isolated cells are much smaller when
$D$ is small than that in Fig.~\ref{fig:cell_plot_II}. However, we
observe from the increasing solid red curve in
Fig.~\ref{fig:cell_plot_II} that the amplitude of the oscillations in
the isolated cells grows as $D$ increases. The interpretation of this
observation is that the secretion from the ring cells can more easily
diffuse across the domain to the isolated cells, where it is readily
absorbed, as $D$ is increased. In Fig.~\ref{fig:ring_permaII} we show
results from FlexPDE simulations of \eqref{DimLess_bulk} for
permeability set II when $D=0.05$. The strong signalling gradient near
the ring cells where the influx rate is small, the smaller amplitude
oscillations in the isolated cells than in the rings, and the slight
phase shift between the oscillations in the two ring clusters (cells 3
and 4 versus cells 7 and 8) are all consistent with the predictions in
Fig.~\ref{fig:cell_plot_II} of the linear stability theory.

Finally, to obtain some analytical insight into the possibility of an
unstable mode consisting of anti-phase oscillations in the isolated
cells together with quiescent ring cells, we implement the degenerate
perturbation theory on the reduced GCEP matrix \eqref{large:GCEP} from
the $D={D_0/\nu}$ regime. By using the simple criterion in
\eqref{degen:quad_J2_crit}, based on the quadratic equation
\eqref{degen:quad_J2}, we calculate for permeability set I that this
anti-phase mode for the isolated cells is unstable when $\tau=0.35$
only when $D_0<0.0817$. This corresponds to $D<0.245$ using
$D={D_0/\nu}$ with $\nu={-1/\log\varepsilon}$ and
$\varepsilon=0.05$. Although this simple criterion
\eqref{degen:quad_J2_crit} predicts the existence of an unstable
anti-phase mode for the isolated cells, the threshold value is not so
accurate (see Fig.~\ref{fig:cell_plot_I}) owing to the fact that $D$
is not asymptotically large. In contrast, for permeability set II, we
calculate from the criterion \eqref{degen:quad_J2_crit} that the
anti-phase mode for the isolated cells is always stable, as is
consistent with Fig.~\ref{fig:cell_plot_II}.

\setcounter{equation}{0}
\setcounter{section}{5}
\section{Discussion and Outlook}\label{sec:discussion}

We have analyzed diffusion-sensing behavior for the coupled cell-bulk
PDE-ODE system \eqref{DimLess_bulk} that was used to model the
switch-like onset of intracellular oscillations for a collection of
heterogeneous cells, as mediated through a passive bulk diffusion
field with finite bulk diffusivity $D$. For the case of Sel'kov
reaction kinetics, we have studied how the onset of intracellular
oscillations depends on the spatial configuration of cells, on the
membrane permeability parameters, and on the triggering effect of a
single ``defective'' cell with a different kinetic parameter. In
contrast, in \cite{smjw_quorum} QS behavior resulting from
increase in cell density were studied using only the ODE system
\eqref{ode:well_mixed} with global coupling, which pertains to the
well-mixed regime of infinite bulk diffusivity.

There are several open theoretical and computational challenges that
should be explored for \eqref{DimLess_bulk}. One key numerical issue
concerns developing efficient and well-conditioned numerical
techniques to implement the full linear stability theory based on the
root-finding condition $\mbox{det}({\mathcal M}(\lambda))=0$ for the
GCEP given in \eqref{full_gcep} for a large number, i.e.~$m\geq 100$,
of randomly located cells with arbitrary permeabilities. The solution
strategies for such nonlinear matrix eigenvalue problems are typically
restricted to matrices wih special structure, such as Hermitian
matrices, matrices with low-rank dependence on $\lambda$, or
matrices that are quadratic or rational in the eigenvalue parameter
$\lambda$ (cf.~\cite{guttel} and \cite{betcke}). In contrast, the GCEP
matrix in \eqref{full_gcep} is not Hermitian when $\lambda$ is
complex-valued, and is not of low-rank owing to its dependence on the
full eigenvalue-dependent Green's matrix ${\mathcal G}_\lambda$.

From a mathematical viewpoint, a second open direction would be to
develop a weakly nonlinear theory for the Hopf bifurcations of
steady-state solutions of \eqref{DimLess_bulk}. The numerical results
shown in \S \ref{SpecConfigSec} and \ref{sec:LargePopulation} have
suggested that the intracellular oscillations are supercritical for
the Sel'kov kinetics, and that their relative amplitude and phases
within the cells are well-predicted by the eigenvector of the GCEP
matrix ${\mathcal M}(\lambda)$ in \eqref{full_gcep}. Also of interest
would be to explore whether ODE phase-reduction techniques, such as
surveyed in \cite{sync1} and \cite{sync2}, can be extended to the
PDE-ODE setting of \eqref{DimLess_bulk}.

From a modeling perspective, the PDE-ODE system \eqref{Dim_bulk} is
well-suited for investigating triggered intracellular oscillations and
collective dynamics associated with specific microbial systems for
which detailed models of the signaling pathways are available, the
autoinducer is known, and where membrane permeabilities can be
estimated from biological data. Such calibrated biological models are
readily incorporated into the theoretical framework \eqref{Dim_bulk}.
In particular, it would be worthwhile to analyze \eqref{Dim_bulk} with
the detailed model of \cite{wolf} and \cite{henson} for the glycolytic
pathway in yeast cells and the model of \cite{mina} for bacterial
communication of {\em Escherichia coli} cells. For non-oscillatory QS
systems, it would be interesting to use the intracellular
Lux-signaling pathways of \cite{qs_jump} into our PDE-ODE model
\eqref{Dim_bulk} to analyze sudden jumps between small and large
amplitude steady-states for bistable QS systems as the cell density
increases past a saddle-node bifurcation value. Moreover,
\eqref{Dim_bulk} can be used to study the effect of spatial diffusion
on the parallel QS signaling pathways associated with the marine
bacterium {\em V. harveyi}, which were recently modeled in
\cite{bressloff} through an ODE well-mixed limit. Our analysis in \S
\ref{largeD_ODE} has shown that the dimensionless PDE-ODE system
\eqref{DimLess_bulk} can be reduced to a limiting ODE system only when
the bulk diffusivity satisfies $D={\mathcal O}(-\log\varepsilon)$,
where $\varepsilon$ is the dimensionless ratio of the cell radius to
the length scale of the domain. The classic ODE system for the
well-mixed limit corresponds to the regime
$D\gg {\mathcal O}(-\log\varepsilon)$.

{Finally, it would be worthwhile to extend the 2-D model
\eqref{Dim_bulk} to allow for two passive bulk diffusion fields, and
to consider 3-D domains, where the cell-cell interaction is weaker
than in 2-D owing to the rapid decay of the 3-D free-space Green's
function. By including two bulk species, and with a re-formulation of
the class of lattice-based models introduced in \cite{rauch} into the
framework of \eqref{Dim_bulk}, it should be possible to show
analytically that changes in the permeability parameters of a
collection of cells can induce a Turing or transcritical bifurcation
to a patterned steady-state for an activator-inhibitor system even
when the two bulk diffusing species have very similar diffusivities.
In a spatially homogeneous medium without cells, Turing bifurcations
only occur for activator inhibitor systems when there is a
sufficiently large diffusivity ratio, which is not typical for most
biological systems. The PDE-ODE framework of \eqref{Dim_bulk}, where
the small signaling compartments are modeled explicitly, also provides
an alternative approach for studying large-scale self-organized
structures that have previously been modeled through PDE-ODE lattice
dynamics in which the dynamically active ``cells'' are treated as
point sources restricted to lattice sites and where a discrete
Laplacian averaging over nearby lattice points replaces the Laplace
operator (cf.~\cite{ross_pde}, \cite{wolf2},
\cite{RosslerHeter}). Such lattice PDE-ODE systems have been used to
model traveling waves of yeast activation due to substrate addition of
glucose at a localized source (cf.~\cite{wolf2}), and the emergence of
spiral waves resulting from the coupling of Fitzhugh-Nagumo or Rossler
cell kinetics to discrete lattice-based diffusion
(cf.~\cite{RosslerHeter}, \cite{ross_pde}).}
 
\section*{Acknowledgements}\label{sec:ak}
Michael Ward gratefully acknowledges the financial support from the
NSERC Discovery grant program. We are grateful to Justin Tzou
(Macquarie U.) for his initial help with the FlexPDE simulations.

\begin{appendices}
\renewcommand{\theequation}{\Alph{section}.\arabic{equation}}
\setcounter{equation}{0}

\section{The reduced-wave Green's function for the unit disk}\label{app:green}

In this appendix we provide explicit expressions for the matrix
spectrum of \eqref{cen:glam_eig} corresponding to the matrix block of
the eigenvalue-dependent reduced-wave Green's matrix in
\eqref{ring:eig_gblock} representing the interactions of $m-1\geq 2$
cells on a ring of radius $r_0$, with $0<r_0<1$, concentric within the
unit disk. For notational convenience we define $N$ by $N=m-1$. The
$N$ cells on the ring of radius $r_0$ are centered at
$\pmb{x}_k = r_0 \left( \cos\psi_k,\sin\psi_k\right)^T$, where
$\psi_k\equiv {2\pi(k-1)/N}$ for $k=1,\ldots,N$. In the unit disk, the
reduced-wave Green's function $G_{\lambda}(\pmb{x};\pmb{\xi})$ and its
regular part, satisfying \eqref{EigGreen}, can be calculated using
separation of variables as (see equations (6.10) and (6.11) of
\cite{jia2016})
\begin{subequations} \label{cyclic:g}
\begin{align} \label{cyclic:Glam}
G_\lambda(\pmb{x};\pmb{\xi}) &= \frac{\,1}{\,2 \pi}
K_0\left(\varphi_\lambda|\pmb{x}-\pmb{\xi}|\right) -\frac{\,1}{\,2 \pi}
\sum^{\infty}_{n=0} \beta_n \cos\left(n(\psi-\psi_0)\right)
\frac{K^{\prime}_n(\varphi_\lambda)} {I^{\prime}_n(\varphi_\lambda)}
  I_n\left(\varphi_\lambda r\right) I_n\left(\varphi_\lambda \rho\right)\,,\\
 R_{\lambda}(\pmb{\xi}) &= \frac{1}{2\pi} \left[\log\left(2\sqrt{D}\right) -
  \gamma_e -\frac{1}{2}\log{\left(1+\tau\lambda\right)} \right]
 -\frac{\,1}{\,2 \pi} \sum^{\infty}_{n=0} \beta_n
 \frac{K^{\prime}_n(\varphi_\lambda)} {I^{\prime}_n(\varphi_\lambda)} \left[
   I_n\left(\varphi_\lambda \rho\right) \right]^2 \,, \label{cyclic:Rlam}
\end{align}
\end{subequations}
where $\varphi_{\lambda}\equiv \sqrt{(1+\tau \lambda)/D}$ is the principal
branch of $\varphi_{\lambda}$, $\gamma_e=0.5772$ is Euler's constant, and
$I_n(z)$ and $K_n(z)$ are the modified Bessel functions of the
first and second kind of order $n$, respectively. In \eqref{cyclic:g},
$\beta_0\equiv 1$, $\beta_n\equiv 2$ for $n\geq 1$, while
$\pmb{x}\equiv r(\cos\psi,\sin\psi)^T$, and $\pmb{\xi} \equiv
\rho(\cos\psi_0,\sin\psi_0)^T$. 

For a ring pattern, the $N\times N$ symmetric Green's matrix
${\mathcal G}_{\lambda
  N}$ is also cyclic and can be generated by a cyclic permutation of
its first row $\pmb{a}_\lambda \equiv
(a_{\lambda,1},\ldots,a_{\lambda,N})^T$, which is defined term-wise by
\begin{equation}\label{cyclic:G_lam_row}
  a_{\lambda,1} \equiv R_{\lambda}(\pmb{x}_1) \,; \qquad a_{\lambda,k} =
  G_{\lambda} (\pmb{x}_k;\pmb{x}_1) \,, \qquad k=2,\ldots, N \,.
\end{equation}
The matrix spectrum
${\mathcal G}_{\lambda N}\pmb{v}_j=\omega_{\lambda
  j}\pmb{v}_j$ in \eqref{cen:glam_eig} with
$N=m-1$ is calculated as in \S 6 of \cite{jia2016}.  The
in-phase eigenpair is
\bsub \label{cyclic:spectrum_all}
\begin{equation}\label{cyclic:spectrum_1}
\omega_{\lambda N} = \sum_{k=1}^N a_{\lambda,k}\,, \qquad \pmb{v}_N =
(1, \ldots, 1)^T\,,
\end{equation}
while the other eigenvalues, corresponding to the anti-phase modes
for which $\pmb{v}_j^T\pmb{v}_N=0$ for $j=1,\ldots,N-1$, are
\begin{equation}\label{cyclic:spectrum_j}
  \omega_{\lambda j} = \sum_{k=0}^{N-1} \cos\left( \frac{2 \pi
      j k}{N}\right) a_{\lambda, k+1}\,, \qquad j=1,\ldots,N-1 \,.
\end{equation}
Since $\omega_{\lambda j}=\omega_{\lambda,N-j}$ for $j=1,\ldots,
\lceil {N/2} \rceil-1$, there are $\lceil {N/2} \rceil -1$ pairs of
degenerate eigenvalues for ${\mathcal G}_{\lambda N}$. Here the ceiling
function $\lceil x \rceil$ is defined as the smallest integer not less
than $x$.  When $N$ is even, there is an eigenvalue of multiplicity
one given by $\omega_{\lambda,\frac{N}{2}}=\sum_{k=0}^{N-1} (-1)^{k}
a_{\lambda, k+1}$. The other eigenvectors for $j=1,\ldots, \lceil {N/2}-1 \rceil$
are
\begin{equation}\label{cyclic:eigvec}
 \pmb{v}_j = \left(1, \cos \left(\frac{2 \pi j}{N}\right), \ldots,
 \cos \left( \frac{2 \pi j(N-1)}{N} \right)\,\right)^T \,,\qquad
 \pmb{v}_{N-j} = \left(0, \sin \left( \frac{2 \pi j}{N}\right),
 \ldots, \sin \left( \frac{2 \pi j(N-1)}{N}\right)\,\right)^T \,.
\end{equation}
\esub Finally, when $N$ is even, there is an additional eigenvector
$\pmb{v}_{\frac{N}{2}}=(1,-1,\ldots,-1)^T$. By using the explicit
formulae for $G_{\lambda}$ and its regular part from \eqref{cyclic:g},
the eigenvalues $\omega_{\lambda j}$ of the Green's matrix
${\mathcal G}_{\lambda N}$ are then easily computed from
\eqref{cyclic:spectrum_all}.

\vspace*{0.2cm}

\begin{remark}\label{app:degen}
The symmetric and cyclic matrix ${\mathcal G}_{\lambda N}$ has
${N/2}$ distinct anti-phase modes if $N$ is even and ${(N-1)/2}$ distinct
anti-phase modes if $N$ is odd.
\end{remark}

Finally, in order to calculate the quantities in \eqref{ring:gblock} and
\eqref{cent:g_eig}, which are needed in \eqref{Linear_reduceSys} for
determining the steady-state solution, we need only set $\lambda=0$
in \eqref{cyclic:g} and \eqref{cyclic:spectrum_1}.

\section{Cell locations and permeability parameters} \label{Cell_Location}

The cell locations and permeability parameter sets for the influx rate
$d_{1j}$, for $j=1,\ldots,m$, are given in the next two tables for the
cell patterns studied in \cref{sec:LargePopulation}.

\begin{table}[!ht]
\centering
  \begin{tabular}{ | c || c | c | c | c || c | c | c| c| } \hline
\rowcolor{LightCyan}  & \multicolumn{4}{| c ||}{Two clusters of cells} & \multicolumn{4}{| c |}{ Arbitrary cell locations}\\ 
\rowcolor{LightCyan}
\multirow{-2}*{Cell $i$}
  &  $x_i$  &   $y_i$  &   $d_{1i}$(II) &  $d_{1i}$(III) & $x_i$  &  $y_i$ &  $d_{1i}$(II) & $d_{1i}$(III)  \\  \hline  \hline  \rowcolor{Gray}
  $1$ & $0.4700$  & $-0.1000$  & $0.8$  & $0.4000$   &    $0.5205$ &   $-0.4687$ & $0.4$   &  $0.4000$\\   
\rowcolor{Gray}
  $2$	 & $0.7800$ & $0.2000$ & $0.8$  & $0.5168$   &   $-0.1856$ &   $-0.1927$ & $0.8$ & $0.5168$ \\ 
  \rowcolor{Cyan}
 $3$  &  $0.8000$ & $-0.1500$ &  $0.8$  &  $0.4378$   &  $0.3170$  & $-0.0236$ & $0.8$  & $0.4378$ \\ 
\rowcolor{Cyan}$4$  & $0.5000$ & $0.2500$ & $0.8$   & $0.6120$       &  $-0.3571$ & $0.6112$  & $0.4$  & $0.6120$ \\ 
\rowcolor{Gray}  $5$   & $0.6000$ & $-0.4000$ & $0.8$  & $0.7421$    &  $0.2019$ &   $0.5935$ & $0.4$  & $0.7421$   \\ 
\rowcolor{Gray}
$6$   & $-0.5800$ & $0.0500$ & $0.4$ & $0.6069$       & $-0.7526$  & $0.1667$ & $0.8$  &  $0.6069$ 
  \\ 
  \rowcolor{Cyan} 
  $7$  & $-0.6500$ & $0.3500$ & $0.4$ & $0.7062$ &  $-0.5500$  & $-0.4543$ & $0.4$  & $0.7062$ \\
  \rowcolor{Cyan}
$8$   & $-0.8000$  & $-0.1500$ &$0.4$  & $0.4365$ &   $0.0203$ &    $0.3201$ & $0.4$  &  $0.4365$ \\
  \rowcolor{Gray}
  $9$   & $-0.3000$ & $0.1000$ & $0.4$ & $0.5629$ &   $0.6947$ & $0.1265$ & $0.8$  & $0.5629$ \\ 
  \rowcolor{Gray}
  $10$ & $-0.5000$  & $-0.3000$ & $0.4$ & $0.8000$   & $-0.0172$ & $-0.7422$ & $0.8$  & $0.8000$\\ \hline
\end{tabular}
\caption{Locations of the cell centers and influx
    permeability rates for parameter sets II and III corresponding to
    the two cluster arrangement of cells in
    Fig.~\ref{fig:twoclusters} (columns 2--5) and the arbitrary
    arrangement of cells in Fig.~\ref{fig:arbit} (columns 6--9).
    The permeability parameter set I (not shown) is for identical
    cells where $d_{1j}=0.8$ for $j=1,\ldots,m$.}
\label{Table:CellLocation}
\end{table}

\begin{table}[!ht]
\centering
  \begin{tabular}{ | c || c | c | c | c | c | c | c | } \hline
\rowcolor{LightCyan}  & \multicolumn{7}{| c |}{Two rings with two isolated cells} \\ 
\rowcolor{LightCyan}
\multirow{-2}*{Cell $i$} &  $x_i$  &   $y_i$  &   $d_{1i}$(I)  &  $d_{1i}$(II) &  $d_{1i}$(III) &  $d_{1i}$(IV) &  $d_{1i}$(V)  \\  \hline  
 \rowcolor{Gray}
  $1$ &  $0.0000$  & $-0.7000$   &   $0.3$ &   $0.8$ & $0.3$ & $0.8$ & $0.3$ \\   
  \rowcolor{Gray}
  $2$ &  $0.0000$  & $0.7000$   &   $0.3$ &   $0.8$ & $0.3$ & $0.8$ & $0.3$ \\   
  \rowcolor{Cyan}
  $3$ &  $-0.4000$ & $0.0000$   &  $0.8$  & $0.3$ &  $0.8$ & $0.8$ & $0.3$ \\ 
  \rowcolor{Cyan}
 $4$  &  $-0.6000$ & $0.2000$ & $0.8$ & $0.3$  & $0.8$ & $0.8$ & $0.3$ \\ 
   \rowcolor{Gray}
 $5$  & $-0.8000$ & $0.0000$ & $0.8$  & $0.3$   & $0.8$ & $0.8$ & $0.3$\\ 
   \rowcolor{Gray}
 $6$   & $-0.6000$ & $-0.2000$ & $0.8$  & $0.3$ & $0.8$ & $0.8$ & $0.3$    \\ 
 \rowcolor{Cyan}
  $7$  &  $0.8000$ & $0.0000$ & $0.8$ & $0.3$  & $0.4$ & $0.8$ & $0.3$ \\
  \rowcolor{Cyan}
  $8$   & $0.6000$ & $0.2000$ & $0.8$  & $0.3$   & $0.4$ & $0.8$ & $0.3$ \\ 
  \rowcolor{Gray}
  $9$   & $0.4000$ & $0.0000$ & $0.8$  & $0.3$   & $0.4$ & $0.8$ & $0.3$ \\ 
  \rowcolor{Gray}
  $10$ & $0.6000$ & $-0.2000$ & $0.8$  & $0.3$   & $0.4$ & $0.8$ & $0.3$ \\ \hline
\end{tabular}
\caption{Locations of the cell centers and influx
    permeability rates for parameter sets I--V corresponding
    to the pattern of two rings of cells together with two isolated
    cells, as shown in Fig.~\ref{Large:ring_pattern}.}
\label{Table:CellLocation_ring} 
\end{table}

\end{appendices}

\bibliographystyle{plain}
\bibliography{ReferenceProposal.bib}

\begin{thebibliography}{10}

\bibitem{TestCon}
{Continuation Test:} a {MATLAB} library which defines test functions for
  continuation codes.
\newblock
  \url{http://people.math.sc.edu/Burkardt/m_src/test_con/test_con.html}.
\newblock Accessed: 2020-02-26.

\bibitem{alciatore1995winding}
D.~Alciatore and R.~Miranda.
\newblock A winding number and point-in-polygon algorithm.
\newblock {\em Glaxo Virtual Anatomy Project Research Report, Department of
  Mechanical Engineering, Colorado State University}, 1995.

\bibitem{betcke2}
T.~Betcke, N.~G. Highan, V.~Mehrmann, G.~M.~N. Porzio, C.~Schr\"oder, and
  Tisseur F.
\newblock {NLEVP}: A collection of nonlinear eigenvalue problems. users' guide.
\newblock {\em MIMS EPring 2011.117, Manchester Institue for Mathematical
  Sciences, The University of Manchester, UK}, page 10 pages, updated 2019.

\bibitem{betcke}
T.~Betcke, N.~G. Highan, V.~Mehrmann, C.~Schr\"oder, and Tisseur F.
\newblock {NLEVP}: A collection of nonlinear eigenvalue problems.
\newblock {\em ACM Trans. Math. Software}, 39(2):7.1--7.28, 2013.

\bibitem{bressloff}
P.~Bressloff.
\newblock Ultrasensitivity and noise amplification in a model of v. harveyi
  quorum sensing.
\newblock {\em Phys. Rev. E}, 93:062418, 2016.

\bibitem{ross_pde}
X.~Z. Cao, H.~Yuan, and B.~W. Li.
\newblock Selection of spatiotemporal patterns in arrays of spatially
  distributed oscillators indirectly coupled via a diffusive environment.
\newblock {\em Chaos}, 29:043104, 2019.

\bibitem{hess1}
M.~A.~J. Chaplain, M.~Ptashnyk, and M.~Sturrock.
\newblock Hopf bifurcation in a gene regulatory network model: Molecular
  movement causes oscillations.
\newblock {\em Math. Mod. Meth. Appl. Sci.}, 25(6):1179--1215, 2015.

\bibitem{dano}
S.~Dan{\o}, M.~F. Madsen, and P.~G. S{\o}rensen.
\newblock Quantitative characterization of cell synchronization in yeast.
\newblock {\em Proceedings of the National Academy of Sciences},
  104(31):12732--12736, 2007.

\bibitem{dano2}
S.~Dan{\o}, P.~G. S{\o}rensen, and F.~Hynne.
\newblock Sustained oscillations in living cells.
\newblock {\em Nature}, 402(6759):320--322, 1999.

\bibitem{de2007dynamical}
S.~De~Monte, F.~d'Ovidio, S.~Dan{\o}, and P.~G. S{\o}rensen.
\newblock Dynamical quorum sensing: {P}opulation density encoded in cellular
  dynamics.
\newblock {\em Proceedings of the National Academy of Sciences},
  104(47):18377--18381, 2007.

\bibitem{dilanji_diffusion}
G.~E. Dilanji, J.~B. Langebrake, P.~De~Leenheer, and S.~J. Hagen.
\newblock Quorum activation at a distance: Spatiotemporal patterns of gene
  regulation from diffusion of an autoinducer signal.
\newblock {\em J. Am. Chem. Soc.}, 6:34695, 2016.

\bibitem{dock}
J.~D. Dockery and J.~P. Keener.
\newblock A mathematical model for quorum sensing in {\em {p}seudomonas
  {a}eruginosa}.
\newblock {\em Bull Math Biol.}, 63(1):95--116, 2001.

\bibitem{dunny1997cell}
G.~M. Dunny and B.~Leonard.
\newblock Cell-cell communication in gram-positive bacteria.
\newblock {\em Annual review of microbiology}, 51(1):527--564, 1997.

\bibitem{flexpde2015solutions}
PDE FlexPDE.
\newblock Solutions inc.
\newblock {\em URL http://www. pdesolutions. com}, 2015.

\bibitem{gao_diffusion}
M.~Gao, H.~Zheng, Y.~Ren, R.~Lou, F.~Wu, W.~Yu, X.~Liu, and X.~Ma.
\newblock A crucial role for spatial distribution in bacterial quorum sensing.
\newblock {\em Scientific Reports}, 6(34695), 2016.

\bibitem{goldbeter1997biochemical}
A.~Goldbeter.
\newblock {\em Biochemical oscillations and cellular rhythms: the molecular
  bases of periodic and chaotic behaviour}.
\newblock Cambridge university press, 1997.

\bibitem{gomez2007}
A.~Gomez-Marin, J.~Garcia-Ojalvo, and J.~M. Sancho.
\newblock Self-sustained spatiotemporal oscillations induced by membrane-bulk
  coupling.
\newblock {\em Phys. Rev. Lett.}, 98:168303, Apr 2007.

\bibitem{gou2017}
J.~Gou, W.-Y. Chiang, P.-Y. Lai, M.~J. Ward, and Y.-X. Li.
\newblock A theory of synchrony by coupling through a diffusive chemical
  signal.
\newblock {\em Physica D}, 339:1--17, 2017.

\bibitem{gou2015}
J.~Gou, Y.~X. Li, W.~Nagata, and M.~J. Ward.
\newblock Synchronized oscillatory dynamics for a 1-{D} model of membrane
  kinetics coupled by linear bulk diffusion.
\newblock {\em SIAM J. Appl. Dyn. Sys.}, 14(4):2096--2137, 2015.

\bibitem{jia2016}
J.~Gou and M.~J. Ward.
\newblock An asymptotic analysis of a 2-d model of dynamically active
  compartments coupled by bulk diffusion.
\newblock {\em Journal of Nonlinear Science}, 26(4):979--1029, 2016.

\bibitem{gou2016}
J.~Gou and M.~J. Ward.
\newblock Oscillatory dynamics for a coupled membrane-bulk diffusion model with
  {F}itzhugh-{N}agumo kinetics.
\newblock {\em SIAM J. Appl. Math.}, 76(2):776--804, 2016.

\bibitem{gregor2010}
T.~Gregor, K.~Fujimoto, N.~Masaki, and S.~Sawai.
\newblock The onset of collective behavior in social amoebae.
\newblock {\em Science}, 328(5981):1021--1025, 2010.

\bibitem{guttel}
S.~G\"uttel and F.~Tisseur.
\newblock The nonlinear eigenvalue problem.
\newblock {\em Acta Numerica}, 26(1):1--94, 2017.

\bibitem{DSQS}
B.~A. Hense, C.~Kuttler, J.~M{\"{u}}ller, M.~Rothballer, A.~Hartmann, and J.~U.
  Kreft.
\newblock Does efficiency sensing unify diffusion and quorum sensing?
\newblock {\em Nature Reviews. Microbiology}, 5:230--239, 2007.

\bibitem{henson}
M.~A. Henson, M\"uller D., and M.~Reuss.
\newblock Cell population modelling of yeast glycolytic oscillations.
\newblock {\em Biochem J.}, 368:433--446, 2002.

\bibitem{smjw_quorum}
S.~Iyaniwura and M.~J. Ward.
\newblock Localized signaling compartments in 2-d coupled by a bulk diffusion
  field: Quorum sensing and synchronous oscillations in the well-mixed limit.
\newblock {\em to appear, Europ. J. Appl. Math.}, 2020.

\bibitem{review-yeast}
K.~Kamino, K.~Fujimoto, and Sawai S.
\newblock Collective oscillations in developing cells: Insights from simple
  systems.
\newblock {\em Develop. Growth Differ.}, 53:503--517, 2011.

\bibitem{KTW2005}
T.~Kolokolnikov, M.~S Titcombe, and M.~J. Ward.
\newblock Optimizing the fundamental {N}eumann eigenvalue for the {L}aplacian
  in a domain with small traps.
\newblock {\em Europ. J. Appl. Math.}, 16(2):161--200, 2005.

\bibitem{RosslerHeter}
B.~W. Li, X.~Z. Cao, and C.~Fu.
\newblock Quorum sensing in populations of spatially extended chaotic
  oscillators coupled indirectly via a heterogeneous environment.
\newblock {\em Journal of NonLinear Science}, 27(6):1667--1686, 2017.

\bibitem{Rossler}
B.~W. Li, C.~Fu, H.~Zhang, and X.~Wang.
\newblock Synchronization and quorum sensing in an ensemble of indirectly
  coupled chaotic oscillators.
\newblock {\em Physical Review E}, 86(4):046207, 2012.

\bibitem{hess2}
C.~K. Macnamara and M.~A.~J. Chaplain.
\newblock Spatio-temporal models of synthetic genetic oscillations.
\newblock {\em Math. Biosc. Eng.}, 14:249--262, 2017.

\bibitem{qs_jump}
P.~Melke, P.~Sahlin, A.~Levchenko, and H.~Jonsson.
\newblock A cell-based model for quorum sensing in heterogeneous bacterial
  colonies.
\newblock {\em PLoS Computational Biology}, 6(6):e1000819, 2010.

\bibitem{mina}
P.~Mina, M.~di~Benardo, N.~J. Savery, and K.~Tsaneva-Atanasova.
\newblock Modeling emergence of oscillations in communicating bacteria: a
  structured approach from one to many cells.
\newblock {\em J. Royal Society Interface}, 10:20120612, 2012.

\bibitem{muller2006}
J.~M{\"u}ller, C.~Kuttler, B.~A. Hense, M.~Rothballer, and A.~Hartmann.
\newblock Cell--cell communication by quorum sensing and dimension-reduction.
\newblock {\em Journal of mathematical biology}, 53(4):672--702, 2006.

\bibitem{muller2013}
J.~M{\"u}ller and H.~Uecker.
\newblock Approximating the dynamics of communicating cells in a diffusive
  medium by odes—homogenization with localization.
\newblock {\em Journal of mathematical biology}, 67(5):1023--1065, 2013.

\bibitem{sync2}
H.~Nakao.
\newblock Phase reduction approach to synchronisation of nonlinear oscillators.
\newblock {\em Contemporary Physics}, 57(2):188--214, 2015.

\bibitem{nandy1998}
V.~Nanjundiah.
\newblock Cyclic {AMP} oscillations in {D}ictyostelium {D}iscoideum: {M}odels
  and observations.
\newblock {\em Biophysical chemistry}, 72(1-2):1--8, 1998.

\bibitem{glass}
F.~Naqib, T.~Quail, L.~Musa, H.~Vulpe, J.~Nadeau, J.~Lei, and L.~Glass.
\newblock Tunable oscillations and chaotic dynamics in systems with localized
  synthesis.
\newblock {\em Phys. Rev. E}, 85:046210, Apr 2012.

\bibitem{dicty}
J.~Noorbakhsh, D.~J. Schwab, A.~E. Sgro, T.~Gregor, and P.~Mehta.
\newblock Modeling oscillations and spiral waves in {\em dictyostelium}
  populations.
\newblock {\em Phys Rev. E.}, 91(6):062711, 2015.

\bibitem{paquin_1d}
F.~Paquin-Lefebvre, W.~Nagata, and M.~J Ward.
\newblock Pattern formation and oscillatory dynamics in a two-dimensional
  coupled bulk-surface reaction-diffusion system.
\newblock {\em SIAM J. Appl. Math.}, 80(3):1520--1545, 2020.

\bibitem{sync1}
B.~Pietras and A.~Daffertshofer.
\newblock Network dynamics of coupled oscillators and phase reduction
  techniques.
\newblock {\em Physics Reports}, 819:1--105, 2019.

\bibitem{rauch}
E.~M. Rauch and M.~Millonas.
\newblock The role of trans-membrane signal transduction in turing-type
  cellular pattern formation.
\newblock {\em J. Theor. Biol.}, 226:401--407, 2004.

\bibitem{wolf2}
J.~Sch{\"{u}}tze and J.~Wolf.
\newblock Spatio-temporal dynamics of glycolysis in cell layers: {A}
  mathematical model.
\newblock {\em Biosystems}, 99(2):104--108, 2010.

\bibitem{Selkov}
E.~E. Sel'kov.
\newblock Self-oscillations in glycolysis 1. a simple kinetic model.
\newblock {\em European Journal of Biochemistry}, 4(1):79--86, 1968.

\bibitem{taga2003chemical}
M.~E. Taga and B.~L. Bassler.
\newblock Chemical communication among bacteria.
\newblock {\em Proceedings of the National Academy of Sciences}, 100(suppl
  2):14549--14554, 2003.

\bibitem{taylor2}
A.~F. Taylor, M.~Tinsley, and K.~Showalter.
\newblock Insights into collective cell behavior from populations of coupled
  chemical oscillators.
\newblock {\em Phys. Chemistry Chem Phys.}, 17(31):20047--20055, 2015.

\bibitem{taylor1}
A.~F. Taylor, M.~Tinsley, F.~Wang, Z.~Huang, and K.~Showalter.
\newblock Dynamical quorum sensing and synchronization in large populations of
  chemical oscillators.
\newblock {\em Science}, 323(5914):614--6017, 2009.

\bibitem{tinsley2}
M.~R. Tinsley, A.~F. Taylor, Z.~Huang, and K.~Showalter.
\newblock Emergence of collective behavior in groups of excitable
  catalyst-loaded particles: {S}patiotemporal dynamical quorum sensing.
\newblock {\em Phys. Rev. Lett.}, 102:158301, 2009.

\bibitem{tinsley1}
M.~R. Tinsley, A.~F. Taylor, Z.~Huang, F.~Wang, and K.~Showalter.
\newblock Dynamical quorum sensing and synchronization in collections of
  excitable and oscillatory catalytic particles.
\newblock {\em Physica D}, 239(11):785--790, 2010.

\bibitem{reflect}
A.~Trovato, F.~Seno, M.~Zanardo, S.~Alberghini, A.~Tondello, and A.~Squartini.
\newblock Quorum vs. diffusion sensing: {A} quantitative analysis of the
  relevance of absorbing or reflecting boundaries.
\newblock {\em FEMS Microbiology Letters}, 352:198--203, 2014.

\bibitem{Mueller2014uecke}
H.~Uecker, J.~M{\"u}ller, and B.~A. Hense.
\newblock Individual-based model for quorum sensing with background flow.
\newblock {\em Bulletin of mathematical biology}, 76(7):1727--1746, 2014.

\bibitem{king2001}
J.~P. Ward, J.~R. King, A.~J. Koerber, P.~Williams, J.~M. Croft, and R.~E.
  Sockett.
\newblock Mathematical modelling of quorum sensing in bacteria.
\newblock {\em Mathematical Medicine and Biology}, 18(3):263--292, 2001.

\bibitem{ward2018spots}
M.~J Ward.
\newblock Spots, traps, and patches: {A}symptotic analysis of localized
  solutions to some linear and nonlinear diffusive systems.
\newblock {\em Nonlinearity}, 31(8):R189, 2018.

\bibitem{WHK1993}
M.~J. Ward, W.~D. Henshaw, and J.~B. Keller.
\newblock Summing logarithmic expansions for singularly perturbed eigenvalue
  problems.
\newblock {\em SIAM J. Appl. Math.}, 53(3):799--828, 1993.

\bibitem{QS_whiteley}
M.~Whiteley, S.~P. Diggle, and E.~P. Greenberg.
\newblock Bacterial quorum sensing: the progress and promise of an emerging
  reserch areas.
\newblock {\em Nature}, 551:7680, November 15 2017.

\bibitem{wolf}
J.~Wolf and R.~Heinrich.
\newblock Effect of cellular interaction on glycolytic oscillations in yeast: a
  theoretical investigation.
\newblock {\em Biochem J.}, 345:321--334, 2000.

\bibitem{xu_bressloff}
B.~Xu and P.~Bressloff.
\newblock A {PDE-DDE} model for cell polarization in fission yeast.
\newblock {\em SIAM J. Appl. Math.}, 76(3):1844--1870, 2016.

\end{thebibliography}

\end{document}